\documentclass[12pt,a4paper,twoside]{book}
\usepackage[utf8]{inputenc}
\usepackage[T1]{fontenc}
\usepackage{setspace}
\usepackage{amsmath,amsfonts,amssymb,amsthm}
\usepackage{algorithm,algpseudocode}
\usepackage{bibentry}
\usepackage[bottom]{footmisc}
\usepackage[
a4paper,
twoside,
bindingoffset=1cm,
inner=3cm,
outer=2.5cm,
top=2.5cm,
bottom=3cm,
headsep=1cm
]{geometry}
\usepackage{charter}
\usepackage[sort,numbers]{natbib}
\usepackage{booktabs,caption,multicol,hhline,tikz,multirow,array}
\usetikzlibrary{arrows.meta}
\usepackage{colortbl}
\usepackage{arydshln}
\usepackage{url}
\usepackage{subfig} 
\usepackage{caption}
\usepackage{graphicx, multirow, array}
\usepackage{tabularx}
\usepackage{hyperref}
\usepackage{enumitem}
\usepackage[acronym, nopostdot]{glossaries} 

\usepackage{xfrac}    
\usepackage{nicefrac}
\usepackage{physics} 

\usepackage[titletoc]{appendix}

\usepackage{nomencl}
\makenomenclature

\usepackage{ifpdf}
\usepackage{url}
\usepackage{rotating}

\usepackage{times} 

\usepackage{siunitx} 

\newcolumntype{C}[1]{>{\centering\arraybackslash}p{#1}} 

\newcolumntype{M}[1]{>{\centering\arraybackslash}m{#1}}
\newcolumntype{N}{@{}m{0pt}@{}}
\usepackage{fancyhdr}
\fancypagestyle{plain}{%
	\fancyhead{}%
	\fancyfoot[C]{\bfseries\thepage}}
\fancypagestyle{mainmatter}{%
	\fancyhf{}
	\fancyhead[LE,RO]{\bfseries\thepage}
	\fancyhead[LO]{\bfseries\nouppercase\rightmark}
	\fancyhead[RE]{\bfseries\nouppercase\leftmark}}
\fancypagestyle{frontmatter}{%
	\fancyhead{}%
	\fancyfoot[C]{\bfseries\thepage}}

\newtheorem{theorem}{Theorem}[chapter] 
\newtheorem{corollary}[theorem]{Corollary}

\newtheorem{conjecture}{Conjecture}[chapter]

\newtheorem{definition}[theorem]{Definition}

\newcommand{\reffig}[1]{\mbox{Fig.~\ref{#1}}}

\newcommand{\efb}{E_\text{FB}}

\newcommand{\mI}{\mathbb{I}}

\newcommand{\dH}{\delta H_1}

\newcommand{\vpsi}{\vec{\psi}}
\newcommand{\bpsi}[1]{\bra{\psi_{#1}}}
\newcommand{\kpsi}[1]{\ket{\psi_{#1}}}
\newcommand{\evpsi}[3]{\mel{\psi_{#1}}{#2}{\psi_{#3}}}

\newcommand{\vphi}{\vec{\varphi}}
\newcommand{\bphi}[1]{\bra{\varphi_{#1}}}
\newcommand{\kphi}[1]{\ket{\varphi_{#1}}}
\newcommand{\evphi}[3]{\mel{\varphi_{#1}}{#2}{\varphi_{#3}}}

\newcommand{\kev}[1]{\ket{e_{#1}}}

\newcommand{\PCLS}{\Psi_\text{CLS}}
\newcommand{\EFB}{E_\text{FB}}

\newcommand{\tw}{\tilde{w}}

\makeglossaries

\begin{document}

\setstretch{1.3}

\nobibliography*
\newcounter{cont}


\thispagestyle{empty}

\begin{center}
	\vspace*{1cm}
	
	\begin{spacing}{2}
		{\large \textbf{Ph.D. Thesis}}
	\end{spacing}
	
	\vspace*{3cm}
	
	{\Huge \textbf{Flatband generators}}
	
	\vfill
	
	{\large Wulayimu Maimaiti}
	
	\vspace{0.7cm}
	
	
	
	{\large Basic Science}
	
	
	
	
	
	\vspace{0.7cm}
	
	{\Large UNIVERSITY OF SCIENCE AND TECHNOLOGY}
	
	\vspace{0.7cm}
	
	
	
	{\large December 2019}
	
\end{center}

\newpage
\thispagestyle{empty}
\mbox{}

\thispagestyle{empty}

\begin{center}
	\vspace*{3cm}

	{\Huge \textbf{Flatband generators}}
	
	\vspace{4cm}
	
	{\large Wulayimu Maimaiti}
	
	\vspace{4cm}
	
	{\large A Dissertation Submitted in Partial Fulfillment of Requirements\\ For the Degree of Doctor of Philosophy}
	
	\vspace{2.5cm}
	
	{\large December 2019}
	
	\vspace{2.5cm}
	
	
	

	{\Large UNIVERSITY OF SCIENCE AND TECHNOLOGY} 
	
	{\large Major of: \textbf{Basic Science}}
	
	\vspace{0.7cm}
	
	
	
	{\large Supervisor: Sergej Flach \\ Co-supervisor: Alexei Andreanov}
	
\end{center}

\newpage
\thispagestyle{empty}
\mbox{}

\thispagestyle{empty}

\begin{center}
	
    \vspace*{1cm}%
    {\huge \textbf{We hereby approve the Ph.D.
    \\ thesis of "Wulayimu Maimaiti".} } \\ %
   \vspace{3cm}
    {\large December 2019}%
      \vspace{3cm} 
      
      \includegraphics[width=1\linewidth]{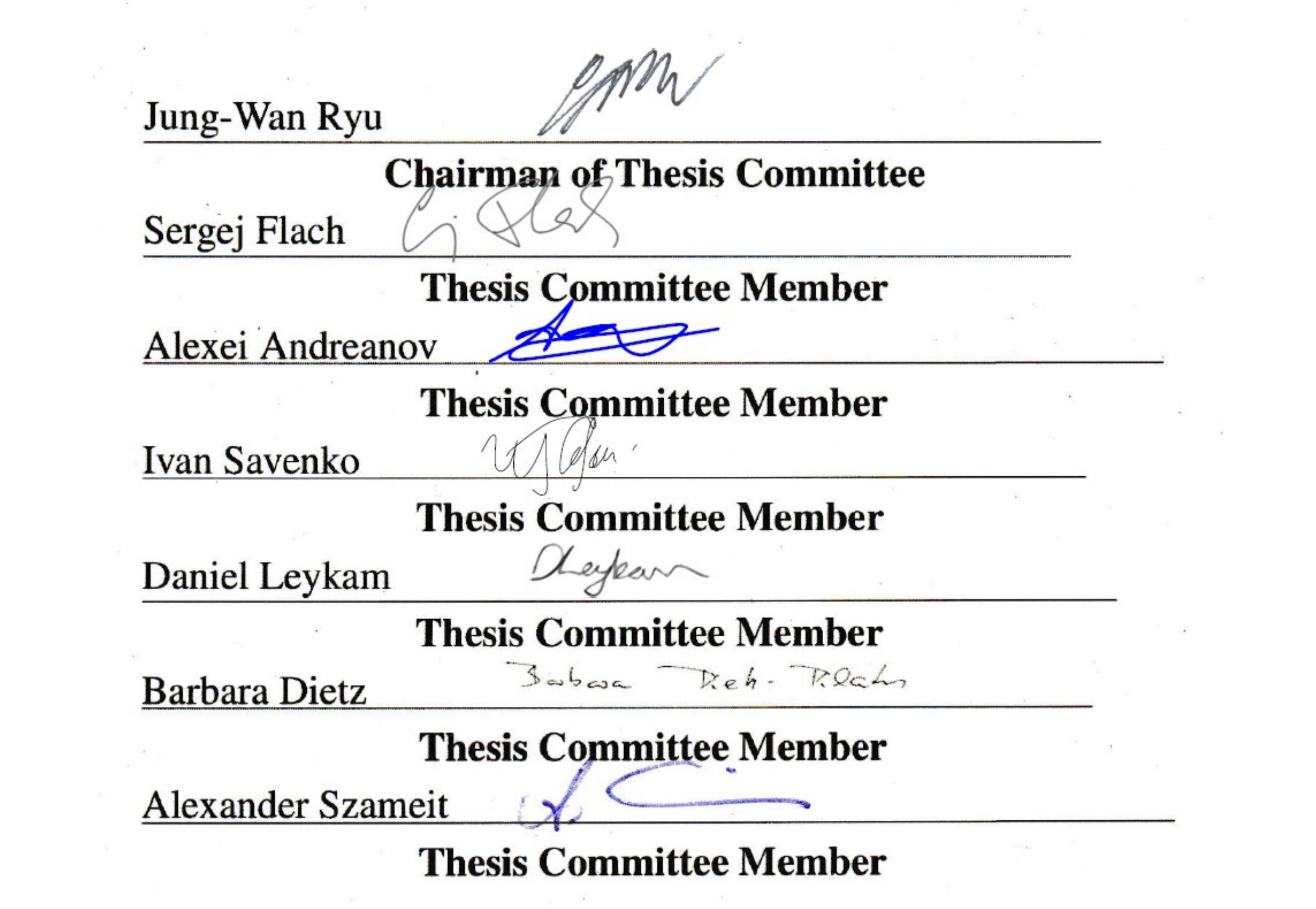}
	
	
	\vfill
	

	{\Large UNIVERSITY OF SCIENCE AND TECHNOLOGY}

\end{center}

\newpage
\thispagestyle{empty}
\mbox{} 

\thispagestyle{empty} 

\vspace*{1cm}

\begin{center}
    {\Huge ACKNOWLEDGEMENT}
\end{center} 

\vspace{2cm}

Professor Sergej Flach and Professor Alexei Andreanov have been great advisors during these years, and I am really fortunate to work with them. I thank them for providing me with nice ideas and great suggestions when I met difficult problems. They showed great persistance in supporting me with patient explanations when I didn't understand some points. The broad vision and accurate hints from Sergej, and the mathematical skills and smart ideas from Alexei were highly valuable to me, while at the same time I was given the freedom to think and work independently.

I thank all my colleagues and friends who supported me during my PhD studies. I am really lucky to work in the Center for Theoretical Physics of Complex Systems, where I received great support from everyone. I had nice discussions with my officemates Dr. Ajith Ramachandran and Dr. Carlo Danieli over the years; these discussions were a great help. I would also like to thank Professor Barbara Dietz for her support during our collaboration.  And I thank Professor Joel Rasmussen for his help to revise my thesis and strengthen the language.

These years of PhD studies have been very challenging for me. I have gone through difficult and stressful times, and during such times, my colleagues and friends gave a lot of support. Sergej always offered great guidance, and I especially thank my friend Dr. Sadiq Seytniyaz who provided me great emotional support during his time in Korea. Their encouragement always gave me energy and courage. 

I acknowledge the financial assistance of the Institute for Basic Science, and the great support and outstanding research environment of the Center for Theoretical Physics of Complex Systems.

At the end, and most importantly, I thank my parents and all my family members for their endless support from the very beginning. I also thank my wife, who has always been my supporter, for all her help during my studies. My family has been my pillar to lean on, and without their support, it would be impossible for me to reach this stage. I ascribe my achievements to my parents and my family. Thank you.

\newpage
\thispagestyle{empty}
\mbox{} 

\vspace*{1cm}

\begin{center}
    {\Huge ABSTRACT} \\ 
    \vspace{2cm} 
    {\large\textbf{Flatband generators}}
\end{center} 

\vspace{2cm} 

Flatbands (FBs) are dispersionless energy bands in the single-particle spectrum of a translational invariant tight-binding network. The FBs occur due to destructive interference, resulting in macroscopically degenerate eigenstates living in a finite number of unit cells, which are called compact localized states (CLSs). Such macroscopic degeneracy is in general highly sensitive to perturbations, such that even slight perturbation lifts the degeneracy and leads to various interesting physical phenomena. 

In this thesis, we develop an approach to identify and construct FB Hamiltonians in 1D, 2D Hermitian, and 1D non-Hermitian systems. First, we introduce a systematic classification of FB lattices by their CLS properties, and propose a scheme to generate tight-binding Hamiltonians having FBs with given CLS properties---a FB generator. Applying this FB generator to a 1D system, we identify all possible FB Hamiltonians of 1D lattices with arbitrary numbers of bands and CLS sizes. Extending the 1D approach, we establish a FB generator for 2D FB Hamiltonians that have CLSs occupying a maximum of four unit cells in a $2\times2$ plaquette. Employing this approach in the non-Hermitiaon regime, we realize a FB generator for a 1D non-Hermitian lattice with two bands. Ultimately, we apply our methods to propose a tight-binding model that explains the spectral properties of a microwave photonic crystal. 

Our results and methods in this thesis further our understanding of the properties of FB lattices and their CLSs, provide more flexibility to design FB lattices in experiments, and open new avenues for future research. 

\vfill

\noindent \underline{\makebox[15.5cm][s]{}} \\ %
   $^\ast${\small A thesis submitted to committee of the University of Science and 
  Technology in a partial fulfillment of the requirement for the degree of
  Doctor of Philosophy conferred in February 2020. } 
  
  \vspace{0.5cm}

\newpage
\thispagestyle{empty}
\mbox{} 

\begin{center}
    \includegraphics[width=1\linewidth]{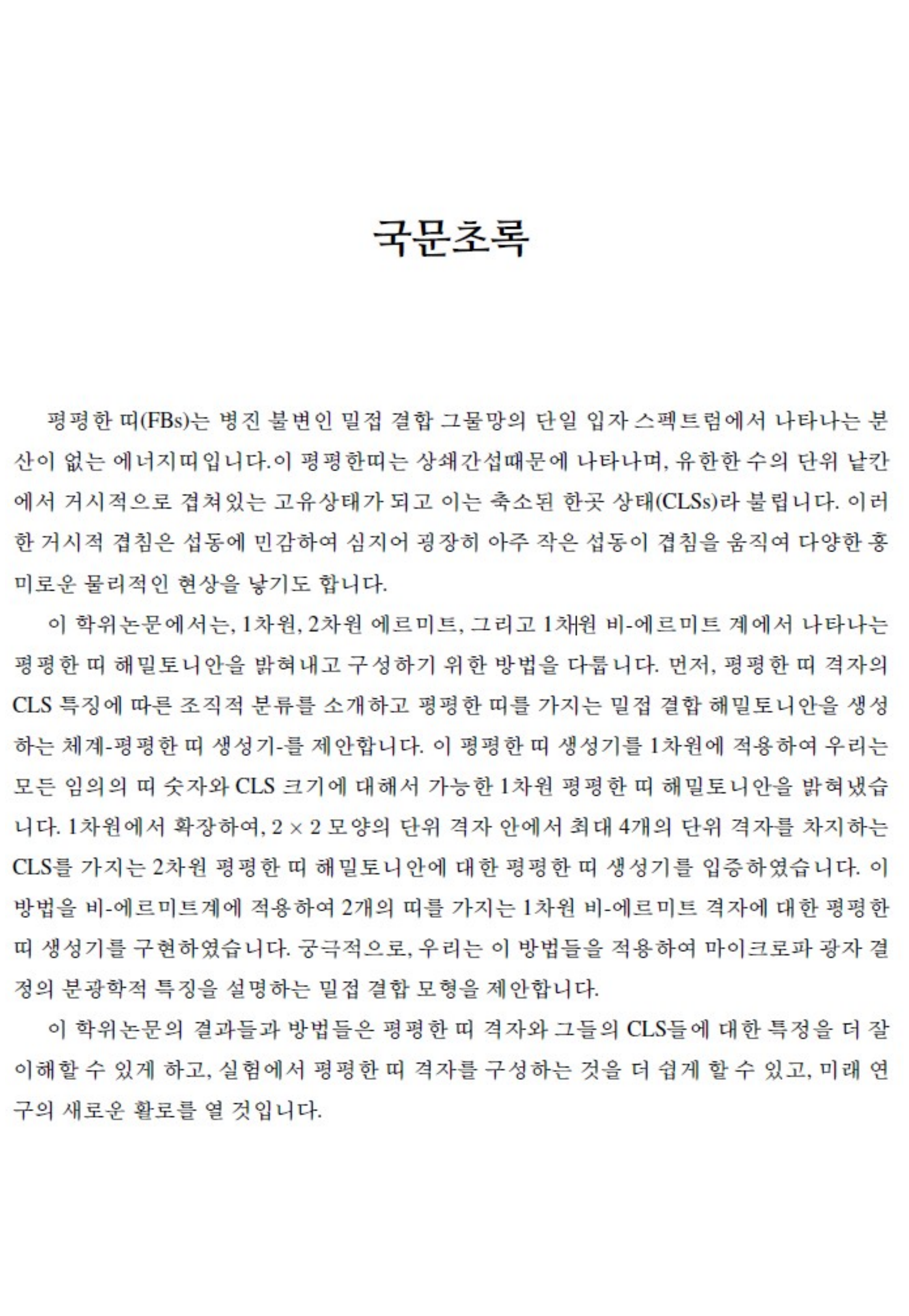}
\end{center}

\newpage
\thispagestyle{empty}
\mbox{} 



\frontmatter
\renewcommand{\chaptermark}[1]{\markboth{#1}{}}
\renewcommand{\sectionmark}[1]{\markright{\thesection\ #1}}
\pagestyle{frontmatter}

\tableofcontents 


\newpage
\thispagestyle{empty}
\mbox{}



\listoffigures
\listoftables



\printnomenclature

\mainmatter
\renewcommand{\chaptermark}[1]{\markboth{#1}{}}
\renewcommand{\sectionmark}[1]{\markright{\thesection\ #1}}
\pagestyle{mainmatter}


\chapter{Motivation and outline}

\ifpdf
    \graphicspath{{Chapter1/Figs/Raster/}{Chapter1/Figs/PDF/}{Chapter1/Figs/}}
\else
    \graphicspath{{Chapter1/Figs/Vector/}{Chapter1/Figs/}}
\fi


Localization is widespread in condensed matter physics, such as Anderson localization in the presence of disorder, states localized to the edges of topological insulators, and so on. In general, disorder is needed to ensure localization, yet there is a special type of localization in translational invariant tight-binding networks where lattice geometry or symmetry leads to destructive interference, which results in the compact localization of wave functions in a finite number of lattice sites. Such states are called \emph{compact localized states} (CLSs), and their most striking property is their ability to generate dispersionless energy bands, i.e. flatbands (FB), which are macroscopically degenerate. 

Systems with macroscopic degeneracies are rare in physical systems, since the high degree of symmetry or fine-tuning needed to support them is easily destroyed by weak perturbations. However, this fragility is the very reason that makes macroscopic degeneracies attractive. Small perturbations in such a fine-tuned system will typically lift the degeneracy and yield uniquely defined eigenstates, thereby leading to exotic physical phenomena that may qualitatively differ for different perturbations (see Fig. \ref{fig:fb-vs-perturbations}). Thus, macroscopic degeneracies could host endpoints of various phase transition lines, and promise rich physics in their neighborhood. Therefore, searching for FB lattice models has been the focus of a wide range of theoretical and experimental research. Nowadays, manufacturing technologies are approaching the realization of such fine-tunings---perhaps not with exact precision, but with some level of control.

\begin{figure}[htb!]
    \centering
    \includegraphics[width=0.7\linewidth]{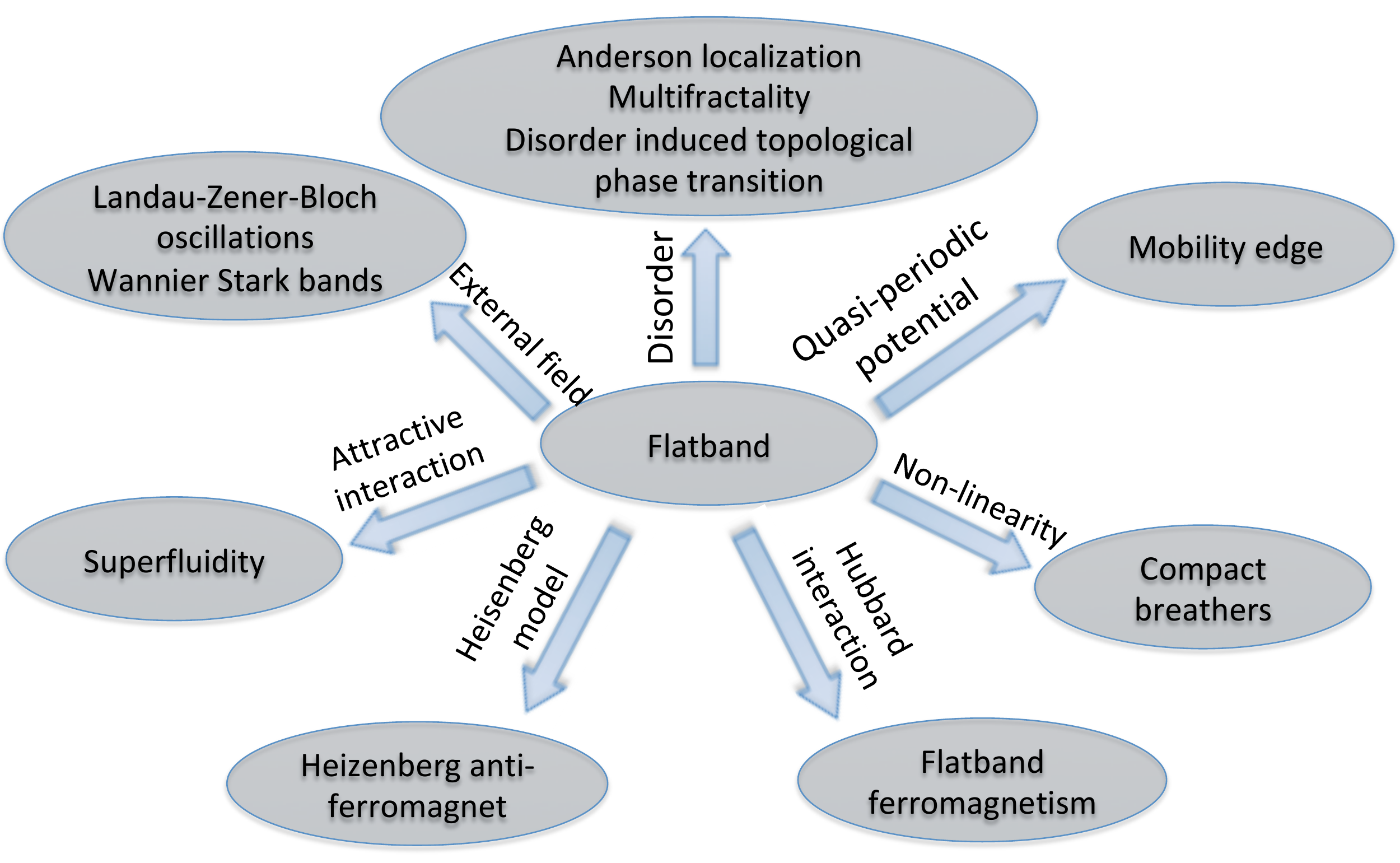}
    \caption[Flatband systems and perturbations]{Various physical phenomena in flatband systems under different perturbations.}
    \label{fig:fb-vs-perturbations}
\end{figure}

Studies of FBs started in the early 1990s, when Mielke showed that a special type of lattices (known as line graphs in mathematics) can have a ferromagnetic ground state if the lowest band is flat~\cite{mielke1991ferromagnetism}. Subsequently, the same type of ferromagnetism associated with FBs was found in further types of lattices by Tasaki ~\cite{tasaki1992ferromagnetism}. Flatband models have also been applied in Heisenberg spin models to show the existence of fully polarized anti-ferromagnetic ground states under an external magnetic field~\cite{Schnack2001independent,schulenburg2002macroscopic,richter2004exact}. Currently, FBs are being studied in various fields of condensed matter physics, and have been experimentally realized in 1D, 2D, and 3D systems (see Fig. \ref{fig:fb-in-various-systems}). For example, on the atomic scale and in the ultra low temperature region, FBs have been shown in cold atoms in optical lattices~\cite{apaja2010flat,taie2017spatial}. In the room temperature regime, FBs are studied theoretically and experimentally in photonic crystal waveguide networks and optical waveguide arrays~\cite{schulz2017photonic,vicencio2015observation,mukherjee2015observation}. Flatbands have also been observed in electronic systems such as superconducting wire networks and nano-engineered atomic lattices on metallic surfaces~\cite{naud2001aharonov,slot2017experimental,drost2017topological}. In cavity quantum electrodynamics (QED) setups, FBs have been realized in exciton--polariton condensates~\cite{baboux2016bosonic,klembt2017polariton}.

\begin{figure}[htb!]
    \centering
    \includegraphics[width=0.7\linewidth]{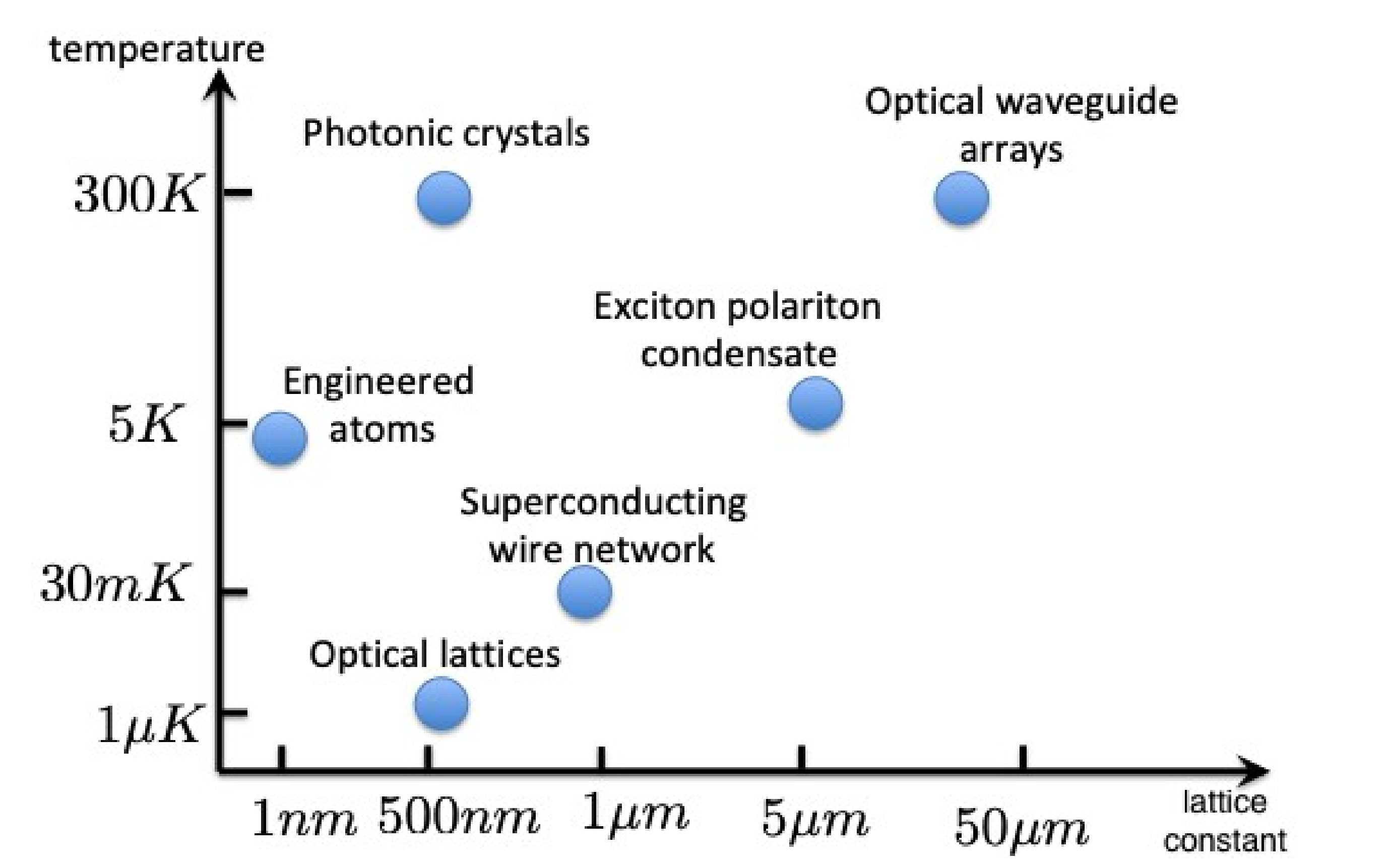}
    \caption[Flatbands in different setups]{Flatbands in various setups on different scales and in different fields. }
    \label{fig:fb-in-various-systems}
\end{figure}

Theoretically, many different approaches have been devised to construct FB lattices, such as Mielke's line graph approach~\cite{mielke1991ferromagnetism}, Tasaki's cell construction method~\cite{tasaki1992ferromagnetism}, origami rules~\cite{dias2015origami}, local unitary transformation~\cite{flach2014detangling}, local symmetry partitioning~\cite{roentgen2018compact}, chiral symmetry~\cite{ramachandran2017chiral}, repetition of mini arrays~\cite{morales2016simple}, etc. Most of these construction methods are limited to either exploiting specific symmetries, like chiral symmetry, or considering specific geometric properties, like those in line graphs. Therefore, known FB models, in most cases, correspond to some highly symmetric points in the parameter space of FB Hamiltonians. By adjusting system parameters though, it becomes possible to achieve FBs without the support of symmetry or special geometry, leading to increased flexibility in both model design and experimental realization. In terms of CLSs, while they play a decisive role in the behaviour of FB models, little is known about how they are linked to lattice properties, or how they determine the response to various perturbations, etc. In this regard, systematically classifying and constructing FB Hamiltonians based on CLSs is highly desirable.

In this thesis, we consider a single particle (non-interacting) Schr\"{o}dinger equation (Hamiltonian) for a discrete translational invariant tight-binding network that possesses at least one FB. We explore the properties of the CLSs and systematically construct FB Hamiltonians that have the required CLS properties. The FB generator is a scheme that generates all possible FB Hamiltonians possessing CLSs with specific sizes and shapes. This involves solving a set of eigenvalue problems subjected to some constraints.

In Chapter \ref{chapter2}, we introduce the basic concepts used in the rest of the thesis. Starting with a discrete translational invariant system, we review the tight-binding model, Bloch theorem, and band structure. Then, using a simple 1D example, we elucidate the origins of FBs, CLSs, and macroscopic degeneracy. Next, we cover the existing FB construction methods and discuss the related applications and experimental realizations. 

In Chapter \ref{chapter3}, we introduce a method to classify FB lattices according to their CLSs. Then we show various properties of CLSs, and propose a block matrix representation of tight-binding Hamiltonians, through which we mathematically formulate the conditions for the existence of CLSs. Based on these conditions, we present our core idea of FB generation. 

In Chapter \ref{chapter4}, we study FB generators in 1D systems. Starting from a two-band problem, we show the full parameterization of FB Hamiltonians and identify the parameter space in which we can obtain FBs. Moving to a higher number of bands, we introduce an inverse eigenvalue method that yields analytic and numerical solutions for FB Hamiltonians. 

In Chapter \ref{chapter5}, we extend the 1D FB generator to 2D. We introduce a classification scheme for CLSs that occupy a maximum of 4 unit cells. Borrowing the methods from the 1D case, we derive analytic solutions for FB Hamiltonians for different CLS classes. Our results cover all known examples, and a vast amount of new models can be generated from our solutions. 

In Chapter \ref{chapter6}, we extend the idea of FB generation to the non-Hermitian regime, where we only consider a 1D non-Hermitian lattice with two bands. Using $k$-space analysis, we achieve the full parameterization of non-Hermitian FB Hamiltonians with both complete and partial FBs, as well as bands in which the modulus is flat.

In Chapter \ref{chapter7}, we apply our methodologies for FB generation to the design of a tight-binding model that explains the experimental results from a microwave photonic crystal. We develop a numerical algorithm to fit the experimental data, and explain the singularity in the experimentally observed spectrum using the FB in our tight-binding model. 

We conclude the thesis with a summary of our results along with their significance, before discussing interesting open problems and future research directions.

\chapter{Introduction: Flatbands in discrete systems}
\label{chapter2}

\ifpdf
    \graphicspath{{Chapter2/Figs/}{Chapter2/Figs/PDF/}{Chapter2/Figs/}}
\else
    \graphicspath{{Chapter2/Figs/}{Chapter2/Figs/}}
\fi

Systems with translational symmetry are common throughout physics, for example periodic arrangements of atoms in a solid crystal. A single particle in such a system can be described by a Schr\"{o}dinger equation with periodic potential, as  
\begin{equation}
    \left( - \frac{\hbar^2 \nabla^2}{2 m} + V(\vec{r}) \right) \psi(\vec{r}) = E \psi(\vec{r}),
    \label{eq:sch-eq-per-pot}
\end{equation} 
where the potential $V$ is periodic with periodicity $\vec{R}$, such that
\begin{equation}
    V(\vec{r}+\vec{R})=V(\vec{r}) \;.
\end{equation} 
For a 1D system, $\vec{R}$ is an integer multiple of the smallest periodicity. For higher dimensions, $\vec{R}$ is a vector sum of the integer multiples of the smallest periodicities in each spacial dimension. Periodic systems are described by a Bloch wave~\cite{Ashcroft76,Kittel2004,singleton2001band}, which is the solution of Eq. \eqref{eq:sch-eq-per-pot}.

In Bloch representation, the energies of the system depend on the wave vector, which is termed dispersion relation and forms an infinite number of energy bands. In practice though, physicists are interested in finite numbers of bands that can be addressed separately from the rest of the bands (e.g. gapped bands). When a finite number of bands is considered, the system is well-described by the tight-binding model~\cite{Ashcroft76,slater1954simplified}. In this model, particles are considered to be localized around a discrete set of periodically arranged points in space, which is called a lattice. In general, lattices have translational invariance, with the periodically arranged points referred to as lattice sites. We focus on a special type of lattice that possesses dispersionless energy bands, which we call flatband (FB) lattices. 

This thesis is dedicated to identifying FB Hamiltonians in discrete translational invariant systems under tight-binding approximation, i.e. lattice systems. In this chapter, we cover the fundamentals of band structure in such systems as well as FBs. We begin with the tight-binding model and the Bloch theorem in Section \ref{section2.2}, before turning to FB lattices in Section \ref{section2.3} where we introduce FBs, compact localized states (CLSs), and macroscopic degeneracy. Then, in Section \ref{section2.4}, we briefly describe FB construction methods, and in Section \ref{section2.5} we review applications of FBs in various fields, from solid-state physics to photonics. The chapter is summarized in Section \ref{section2.6}.


\section{Tight-binding model for discrete systems}
\label{section2.2} 

\nomenclature[FB]{FB}{Flatband}   

\nomenclature[CLS]{CLS}{Compact localized states}
\nomenclature[$U$]{$U$}{Size of compact localized states}

\nomenclature[$\mathbf{U}$]{$\mathbf{U}$}{Shape vector of a compact localized state}   

\nomenclature[$\nu$]{$\nu$}{Number of sites per unit cell (number of bands)}   
\nomenclature[$\mu$]{$\mu$}{Band index}  

\nomenclature[CFB]{CFB}{Chiral flatband}



Generally, in discrete translational invariant lattice systems, the lattice sites are well separated such that the following approximation is commonly applied.

\begin{definition}
\emph{Tight-binding approximation}: When the overlaps of the wave functions of all neighboring sites are small, the overlaps can be neglected. As a result, the wave functions are considered to be well-localized around their lattice sites.
\label{def:tb-aprox}
\end{definition}

%
%

In this approximation, the wave functions live in discrete points in space, and thus can be labeled by the index of their lattice sites.



\subsection{Tight-binding model}

The time-independent Schr\"{o}dinger equation for a translational invariant lattice is written as 
\begin{equation}
    H \vert \Psi \rangle = E \vert \Psi \rangle 
    \label{eq:t-indp-schr-eq},
\end{equation}
where $H$ is the Hamiltonian matrix of the system and $\vert \Psi \rangle$ is a state vector. If there are $N$ lattice sites, then $H$ is an $N \times N$ matrix and $\vert \Psi \rangle$ is an $N$-component vector.


As previously stated, the overlaps between the wave functions of neighboring sites are zero in the tight-binding approximation. This allows us to introduce a basis vector
\begin{equation}
    \vert n \rangle = \begin{pmatrix} \vdots\\ 0 \\ 1 \\ 0 \\ \vdots \end{pmatrix} ,
\end{equation} 
which represents the particle living in the $n$th lattice site, where $1$ is the $n$th element of $\vert n \rangle$. Therefore, $\vert n \rangle$ is a {\bf position eigenstate}, whose dimension is equal to the total number of sites in the lattice. Then the set $\{\vert n \rangle \}$ forms a complete basis of the Hilbert space of the lattice as
\begin{equation}
    \langle n \vert m \rangle = \delta_{n,m}\ , \quad \sum_n \vert n \rangle \langle n \vert = 1.
\end{equation}
Since the wave functions are localized in the tight-binding approximation, transport only occurs via hoppings between neighboring lattice sites. 
Hopping can be understood as a quantum tunneling effect, and is defined as a matrix element of the tight-binding Hamiltonian:
\begin{equation}
    t_{n m} = \langle n \vert H \vert m \rangle. 
\end{equation} 
This equation gives the hopping strength between the $n$th and $m$th sites, where $H$ is the Hamiltonian of the lattice. Therefore, the Hamiltonian in Eq. \eqref{eq:t-indp-schr-eq} can be written in tight-binding form as 
\begin{equation}
    H = \sum_{n} \epsilon_{n} \vert n \rangle \langle n \vert + \sum_{n,m} t_{n m } \vert n \rangle \langle m \vert \; ,
    \label{eq:TB-Hamiltonian}
\end{equation} 
where $\epsilon_n=\langle n \vert n \rangle$ is the onsite energy of the $n$th lattice site. In this tight-binding Hamiltonian, if only nearest neighbor hoppings are considered, $m$ runs over all nearest neighboring sites.


We can also write the state vector $\vert \Psi \rangle$ of the entire lattice as linear combinations of the basis $\vert \psi_n \rangle$ as
\begin{equation}
    \vert \Psi \rangle =  \sum_n \phi_n \vert n \rangle = \frac{1}{\sqrt{N}} \begin{pmatrix} \vdots\\ \phi_{n-1}\\ \phi_{n}\\ \phi_{n+1}\\ \vdots \end{pmatrix}, \quad \sum_n \vert \phi_n \vert^2 = 1
    \label{eq:state-vec}
\end{equation} 
where $\phi_n $ is the wave function of the $n$th site. Thus, the probability of finding a particle in the $n$th site is $\vert \phi_n \vert^2$.

Putting Eqs. \eqref{eq:TB-Hamiltonian} and \eqref{eq:state-vec} into Eq. \eqref{eq:t-indp-schr-eq}, we get single-particle wave function $\phi_n$ at the $n$th site that satisfies
\begin{equation}
    \epsilon_n \phi_{n} + \sum_{m} t_{n m} \phi_{m} = E \phi_{n} 
    \label{eq:TB-eig-prob}.
\end{equation} 

If there are $\nu$ sites per unit cell, it is convenient to label each site with a unit cell index. The basis vectors then become 
\begin{equation}
    \vert n, j \rangle = \begin{pmatrix} \vdots\\ 0 \\ 1 \\ 0 \\ \vdots \end{pmatrix},\quad \langle n,j \vert m,j^\prime \rangle = \delta_{n,m}\delta{j,j^\prime}, \quad \sum_{n,j} \vert n,j \rangle \langle n,j \vert = 1 \;,
\end{equation} 
which represents the particle living in the $j$th site of the $n$th unit cell. Then it is helpful to introduce a state vector for a single unit cell, as
\begin{equation}
    \vert \Phi_n \rangle = \sum_{j} \phi_{n,j} \vert n,j \rangle = \begin{pmatrix} \vdots\\ 0 \\ \vec{\psi}_n \\ 0 \\ \vdots \end{pmatrix}\;, 
    \label{eq:single-cell-state-vec}
\end{equation}
where the single unit cell wave function $\vec{\psi}_n$ is a $\nu$-component vector 
\begin{equation}
    \vec{\psi}_n = \begin{pmatrix} \phi_{n,1} \\ \phi_{n,2} \\ \vdots \\ \phi_{n,\nu} 
    \end{pmatrix} \; ,
    \label{eq:single-uc-wf}
\end{equation} 
whose component $\phi_{n,j=1,\dots,\nu}$ is the wave function of the $j$th site in the $n$th unit cell. As before, then the probability of finding the particle in the $j$th site of the $n$th unit cell is $\vert \phi_{n,j} \vert^2$. 

The set $\{ \vert \Phi_n \rangle \}$ now forms the orthonormal basis 
\begin{equation}
    \langle \Phi_n \vert \Phi_m \rangle = \delta_{n,m}, \quad \sum_n \vert \Phi_n \rangle \langle \Phi_n \vert = 1 \;. 
    \label{eq:orth-nrm-cond-uc-rep}
\end{equation}
The state vector of the entire lattice can be written in terms of the single unit cell state vectors $\vert \Phi_n \rangle$ as
\begin{equation}
    \vert \Psi \rangle =  \sum_n \vert \Phi_n \rangle = \begin{pmatrix} \vdots\\ \vec{\psi}_{n-1}\\ \vec{\psi}_{n}\\ \vec{\psi}_{n+1}\\ \vdots \end{pmatrix} \; .
    \label{eq:state-vec-uc-rep}
\end{equation} 
Then, the tight-binding Hamiltonian reads 
\begin{equation}
    H = \sum_{n,j} \epsilon_j \vert n,j \rangle \langle n,j \vert  + \sum t_{n,j,m,j^\prime} \vert n,j \rangle \langle m,j^\prime \vert, 
    \label{eq:TB-Ham-uc-rep}
\end{equation}
where 
\begin{equation}
    t_{n,j,m,j^\prime} = \langle n,j \vert H \vert m,j^\prime \rangle \; 
\end{equation}
is the hopping between the $j$th site of the $n$th unit cell and the $j^\prime$th site of the $m$th unit cell. Note that $t_{n,j,n,j^\prime}$ is intracell hopping, or the hoppings between sites inside the same unit cell. 

Putting Eqs. \eqref{eq:state-vec-uc-rep} and \eqref{eq:TB-Ham-uc-rep} into Eq. \eqref{eq:t-indp-schr-eq}, we arrive at the eigenvalue problem for the tight-binding model:
\begin{equation}
    \epsilon_j \phi_{n,j} + \sum_{m,j^\prime} t_{n,j, m,j^\prime} \phi_{m,j^\prime} = E \phi_{n,j} , \quad j=1,\dots,\nu.
    \label{eq:TB-eig-prob-1}
\end{equation} 
We will use this notation extensively throughout the thesis.

\subsection{Bloch theorem and band structure} 

The Bloch theorem states that the solution of the Schr\"{o}dinger equation \eqref{eq:t-indp-schr-eq} for discrete translational invariant systems with $\nu$ sites per unit cell can be written as Bloch waves (plane waves multiplied by a periodic function with a periodicity equaling lattice periodicity) as 
\begin{equation}
    \vec{\psi}_n (\vec{k}) = \vec{u}(\vec{k}) e^{i \vec{k} \cdot \vec{R}_n}
    \label{eq:bloch-func},
\end{equation} 
where $\vec{\psi}_n$ and Bloch function $\vec{u}(\vec{k})$ are $\nu$-component vectors, $\vec{R}_n$ is the lattice translation vector for the $n$th unit cell, and $\vec{k}$ is the wave vector. The Bloch function is invariant under lattice translations, i.e. $\vec{u}(\vec{k},\vec{R}_0)=\vec{u}(\vec{k},\vec{R}_0+\vec{R}_n)$. 

Since single-site wave functions form a complete basis, the extended Bloch function can be approximated by a linear combination of single-site wave functions. For simplicity, we take a lattice with a single site per unit cell, then  
\begin{equation}
  u (\vec{k}) = \frac{1}{\sqrt{A}} \sum_n \phi_n e^{i \vec{k} \cdot \vec{R}_n}
  \label{eq:bloch-func-r-space-rep},
\end{equation} 
where $\vec{R}_n$ is the lattice translation vector to the $n$th site, $A$ is a normalization factor, and $\phi_n$ is the wave function of the $n$th site. Note that Eq. \eqref{eq:bloch-func-r-space-rep} is a Fourier transform of $\phi_n$. Then, $\phi_n$ can also be written as an inverse Fourier transform of $u (\vec{k})$ as
\begin{equation}
    \phi_n = \frac{1}{\sqrt{B}} \sum_{\vec{k}} u (\vec{k}) e^{- i \vec{k} \cdot \vec{R}_n}, 
    \label{eq:phi_n-k-representation}
\end{equation}
where $B$ is a normalization factor. Putting Eq. \eqref{eq:phi_n-k-representation} into Eq. \eqref{eq:TB-eig-prob} and canceling out the normalization factor and $\sum_{\vec{k}} u (\vec{k}) e^{- i \vec{k} \cdot \vec{R}_n}$, we get 
\begin{equation}
    E (\vec{k}) = \epsilon + \sum_{m} t_{n m} e^{-i \vec{k} \cdot \vec{R}_{m-n}}
    \label{eq:dispersion-generic},
\end{equation} 
where $\epsilon$ is onsite energy, $\vec{R}_{m-n}=\vec{R}_m -\vec{R}_n$, and we use $e^{-i \vec{k} \cdot \vec{R}_{m}}=e^{-i \vec{k} \cdot \vec{R}_{n}} e^{-i \vec{k} \cdot \vec{R}_{m-n}}$.  In  Eq. \eqref{eq:dispersion-generic}, $E(\vec{k})$ for all $\vec{k}$s forms the {\bf energy band} of the lattice. For example, if we have a 1D lattice with nearest neighbor hopping $t$, then Eq. \eqref{eq:dispersion-generic} gives
\begin{equation}
    E (\vec{k}) = \epsilon + 2 t \cos(k).
\end{equation}  

When there are multiple $\nu$ sites per unit cell, one can also rewrite Eqs. \eqref{eq:bloch-func-r-space-rep} and \eqref{eq:phi_n-k-representation} in terms of the single unit cell state vector from Eq. \eqref{eq:single-cell-state-vec}. In this case, the Bloch function for each band also becomes a $\nu$-component vector, $\vec{u}(\vec{k})$. The Bloch function $\vec{u}_{\mu=1,\dots,\nu}(\vec{k})$ of the $\mu$th band can be written as a Fourier transform of single unit cell wave functions by
\begin{equation}
  \vec{u}_{\mu} (\vec{k}) = \frac{1}{\sqrt{A}} \sum_{n,j} \vec{\psi}_{n} e^{i \vec{k} \cdot \vec{R}_n} \;,
  \label{eq:bloch-func-r-space-rep-1}
\end{equation} 
where $\nu$-component vector $\vec{\psi}_n$ is the single unit cell wave function of the $n$th unit cell. Now, $\vec{\psi}_n$ can be written as an inverse Fourier transform of $\vec{u}_{\mu} (\vec{k})$,
\begin{equation}
    \vec{\psi}_n = \frac{1}{\sqrt{B}} \sum_{\vec{k}} u_{\mu} (\vec{k}) e^{- i \vec{k} \cdot \vec{R}_n}, \quad \mu=1,\dots,\nu.
    \label{eq:phi_n-k-representation-1}
\end{equation}
Putting Eq. \eqref{eq:phi_n-k-representation-1} into Eq. \eqref{eq:TB-eig-prob-1}, we obtain the following equation for every $\vec{k}$: 
\begin{equation}
    E_{\mu} (\vec{k}) u_{\mu j}(\vec{k}) = \sum_j \epsilon_j u_{\mu j}(\vec{k}) + \sum_{m,j^\prime} t_{n,j, m,j^\prime} u_{\mu j^\prime}(\vec{k}) e^{-i \vec{k} \cdot \vec{R}_{m-n}}, \quad j=1,\dots,\nu \;,
    \label{eq:disc-cent-eq}
\end{equation}
where $u_{\mu j} (\vec{k})$ is the $j$th component of $\vec{u}_{\mu} (\vec{k})$, $\epsilon_j$ is onsite energy of the $j$th site in a unit cell, and $\vec{R}_{m-n} = \vec{R}_m - \vec{R}_n$. The eigenvalue $E_{\mu} (\vec{k})$ of Eq. \eqref{eq:disc-cent-eq} for all $\vec{k}$ forms the $\mu$th energy band of the lattice, with $\mu=1,\dots,\nu$. 

For finite lattices, only the set of wave vectors $\vec{k}$ is allowed, determined by Born--von Karman periodic boundary conditions, which are written as 
\begin{equation}
    \vec{\psi}_n(\vec{k})=\vec{u}(\vec{k}) e^{i \vec{k} \cdot \vec{R}_n}=\vec{u}(\vec{k}) e^{i (\vec{k} \cdot \vec{R}_n + N_j \vec{a}_j) } \;,
\end{equation} 
where $\psi_n(\vec{k})$ is the Bloch wave, $j$ runs over the dimensions of the lattice, $\vec{a}_j$ are the primitive vectors of the lattice, and $N_j$ are integers (assuming the lattice has $N$ cells where $N=N_1 N_2 N_3$). This implies that 
\begin{equation}
    e^{i N_j \vec{k} \cdot \vec{a}_j} = 1 \;,
\end{equation}
and so the allowed wave vectors are 
\begin{equation}
	\label{eq:allowed-ks}
    	\vec{k} = 2\pi \sum_{j=1}^3 \frac{m_j}{N_j a_j}, \quad m_j=1,2,\dots N_j  \; .
\end{equation}
The set of allowed $\vec{k}$s forms a space called reciprocal space or \emph{$k$-space}, which is also defined as the Fourier transform of a direct lattice. As it can be seen from Eq.\eqref{eq:allowed-ks}, there are as many number of allowed $\vec{k}$s as the number of unit cells in the lattice.

\section{Flatbands, compact localized states, and macroscopic degeneracy}  
\label{section2.3} 


Flatbands are dispersionless energy bands ($E_{\mu}(k)=const$) of translational invariant tight-binding networks, which are the consequence of destructive interference. In this section, we take a 1D cross-stitch lattice as an example to show how FBs are formed.


We start from the most trivial case of isolated lattice sites, i.e. all hoppings are zero. If the unit cell contains a single site, then each site is described by the following Schr\"{o}dinger equation, 
\begin{equation}
    H_{a} \phi_a = \epsilon \phi_a,
    \label{eq:single-site-sch-eq}
\end{equation} 
where $\epsilon$ is onsite energy and $\phi_a$ is the local wave function. The Hamiltonian of this system is a diagonal matrix with a highly degenerate eigenvalue $\epsilon$. Therefore, the whole system has only one energy $E=\epsilon$ that gives a dispersionless energy band---a flatband. Note that this is the only possibility where one can get a FB in a single band system. Similarly, lattices with isolated unit cells with $\nu$ sites per unit cell give $\nu$ FBs. For example, Fig. \ref{fig:trivial-fb} (a) shows a lattice with isolated unit cells of two sites, giving two trivial FBs. When the unit cells get closer and the hoppings between them are non-zero due to tunneling, as in Fig. \ref{fig:trivial-fb} (b), the bands generally become dispersive. 

\begin{figure}[htb!]
    \centering
    \includegraphics[width=0.6\columnwidth]{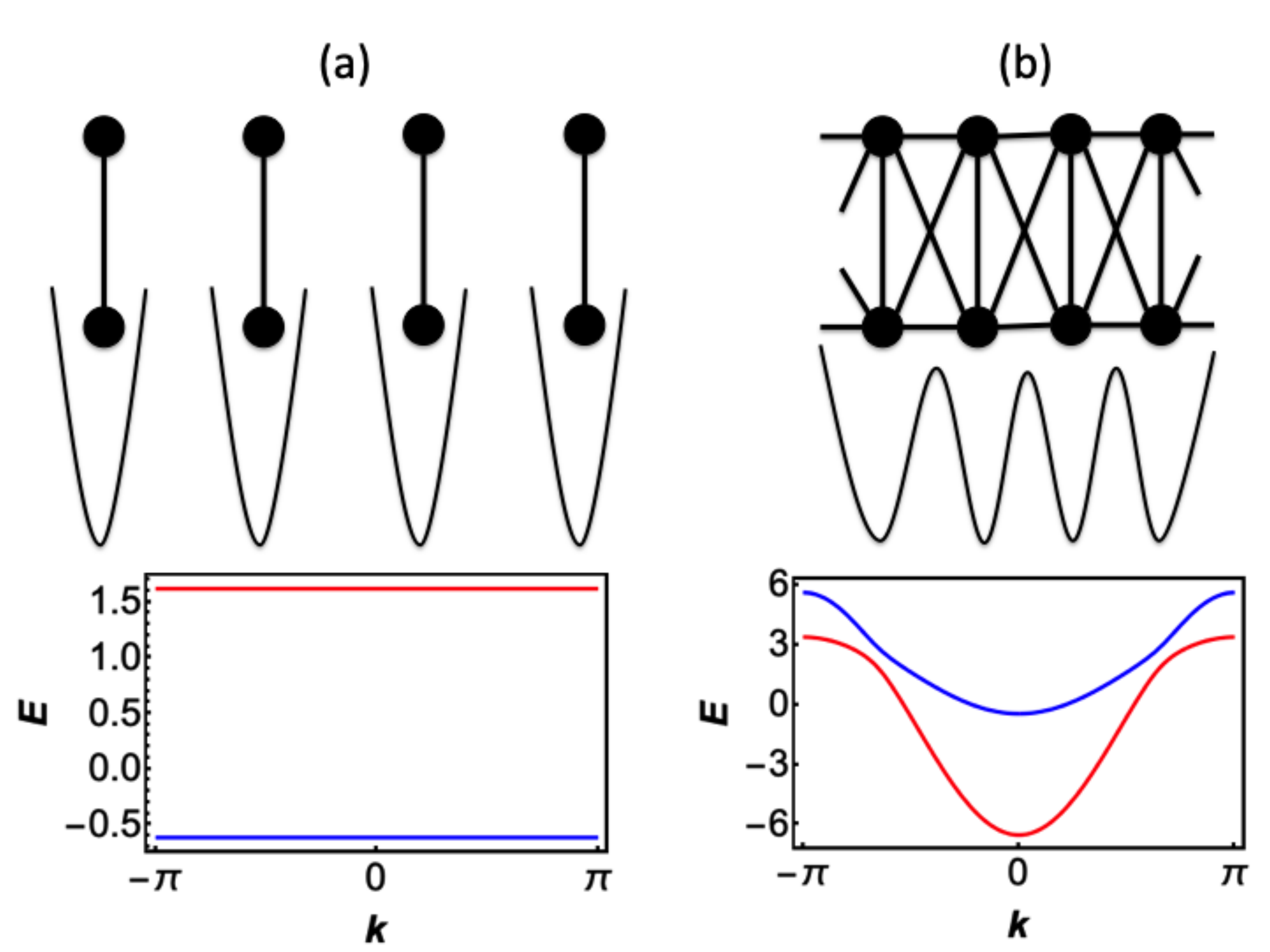}
    \caption[Lattices with isolated unit cells and interacting unit cells]{Lattices and corresponding band structures for (a) isolated unit cells and (b) interacting unit cells. The curved lines below the lattices illustrate the potential around the sites. In (a), the band structure is obtained with $\epsilon=0,1$ and intracell hopping $1$. In (b), the band structure is obtained with $\epsilon=0,1$, intra-cell hoppings $1$, horizontal hoppings $-2$, and diagonal hoppings $1$. }
    \label{fig:trivial-fb}
\end{figure}

Besides this trivial case of isolated unit cells, there are lattices with non-zero hoppings between unit cells that still have FBs. In such systems, due to lattice geometry or symmetry, wave functions hopping from neighboring sites possess opposite phases with equal amplitudes, and so interfere destructively at certain sites---this is called \emph{destructive interference}. As a result, the wave functions are trapped in a finite number of lattice sites and are strictly zero elsewhere (see Fig. \ref{fig:cross-stich-fb}).  Such strictly localized states are called {\bf compact localized states} (CLS), which cause net hoppings to vanish, giving an analogous effect of isolated unit cells. We explain this situation in a simple example below.


\begin{figure}[htb!]
    \centering
    \includegraphics[width=0.8\linewidth]{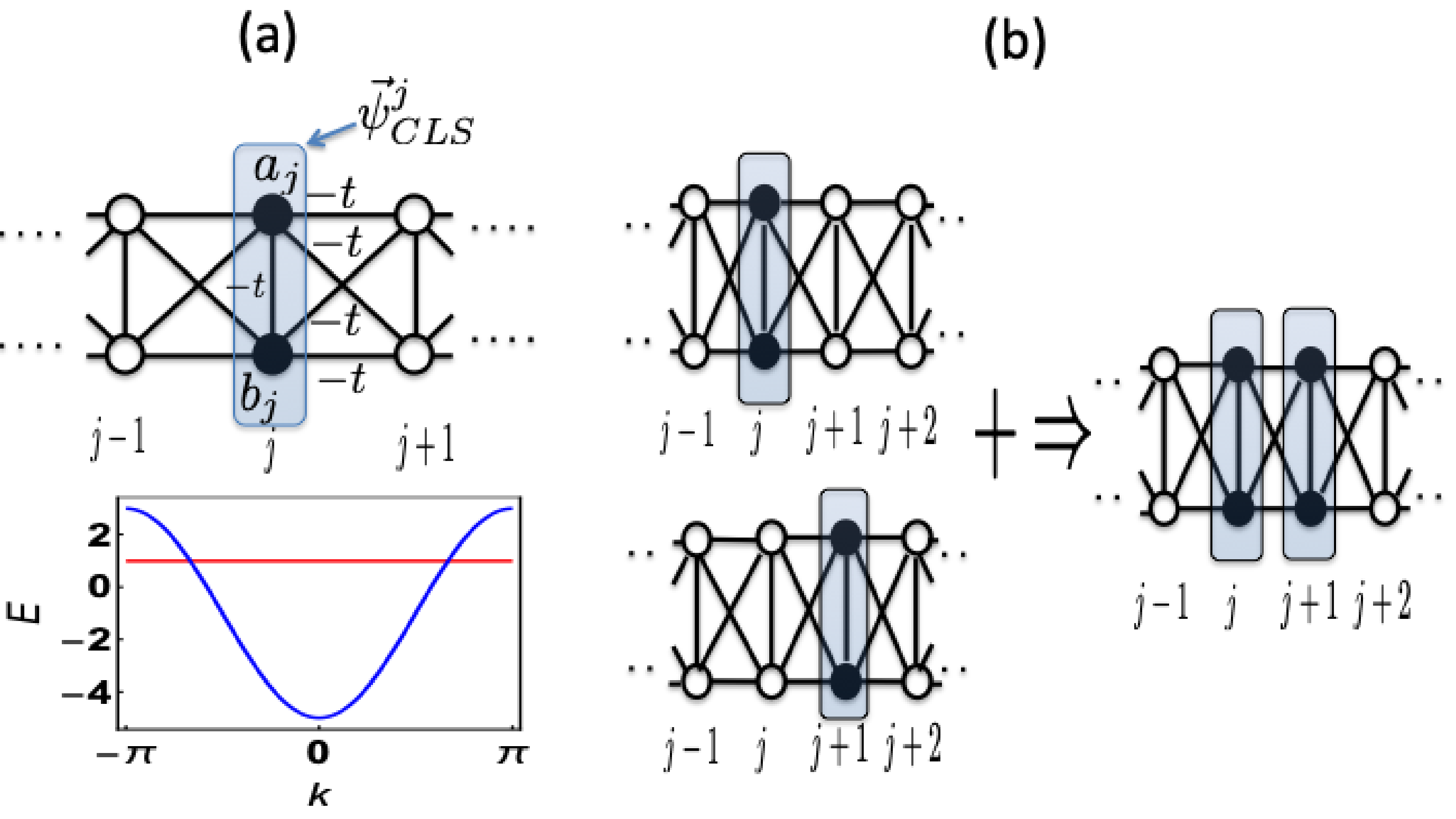}
    \caption[Cross-stitch lattice with one flatband]{A cross-stitch lattice with one FB. The filled circles show the locations of the CLSs, and the empty circles show the sites with destructive interference where the wave function is strictly zero. $j$ is the unit cell index, and $a_j,b_j$ are the wave amplitudes at each lattice site inside the $j$th unit cell. (a) The CLS is located in the $j$th unit cell, and its band structure corresponds to $t=1, a_j=-b_j=1$. (b) Superposition of two lattice translations of the CLS in (a). }
    \label{fig:cross-stich-fb}
\end{figure}

 Consider a 1D cross-stitch lattice, as shown in Fig. \ref{fig:cross-stich-fb} (a). If we set onsite energies to zero and tune the hoppings as
 \begin{equation}
     t_{nm}=-t \delta_{|n-m|,1} \;,
 \end{equation}
then Eq. \eqref{eq:TB-eig-prob-1} gives the following eigenvalue problem
\begin{equation}
\begin{aligned}
    E a_j &= \epsilon_j^a a_j - t a_{j-1} - t a_{j+1} - t b_{j-1} - t b_{j+1} - t b_j,\\
    E b_j &= \epsilon_j^a b_j - t b_{j-1} - t b_{j+1} - t a_{j-1} - t a_{j+1} - t a_j.
\end{aligned}
\label{eq:cross-stich-fb-eig-prob}
\end{equation}
This has one FB at energy $E_{FB} = t$. As shown in Fig. \ref{fig:cross-stich-fb} (a), one possible eigenstate of the tight-binding Hamiltonian corresponding to this FB energy is 
\begin{equation}
    \vert \Psi_j \rangle = \begin{pmatrix} \vdots \\ 0 \\ 1 \\ -1 \\ 0 \\ \vdots \end{pmatrix} =  \begin{pmatrix} \vdots\\ 0 \\ \vec{\psi}_{CLS}^j \\ 0 \\ \vdots \end{pmatrix},
    \label{eq:cross-stich-u1-eig-st}
\end{equation} 
where $\vec{\psi}_{CLS}^j=(a_j,b_j)^T=(1,-1)^T$, and index $j$ refers to the location of the non-zero amplitudes at the $j$th unit cell. 

In Fig. \ref{fig:cross-stich-fb} (a), the CLS (in Eq. \eqref{eq:cross-stich-u1-eig-st}) is located in the $j$th unit cell and is the consequence of destructive interference in the neighboring sites. More precisely, the sum of the hopping amplitudes from $a_j,b_j$ to the neighboring sites is zero, i.e. $a_{j+1} = -t a_j - t b_j = 0$, which indicates that wave amplitudes at the $j$th unit cell must satisfy $a_j=-b_j$ in order to have a CLS. 

The lattice translations of the CLS in Eq. \eqref{eq:cross-stich-u1-eig-st} are also solutions of Eq. \eqref{eq:cross-stich-fb-eig-prob}; i.e. for $\forall j$, these lattice translations are eigenstates of the tight-binding Hamiltonian \eqref{eq:TB-Ham-uc-rep} corresponding to $E_{FB}=t$. Equation \eqref{eq:t-indp-schr-eq} then reads
\begin{equation}
    H \vert \Psi_j \rangle = E_{FB} \vert \Psi_j \rangle \rightarrow \quad \text{for}\ \forall j.
\end{equation} 


The superposition of all lattice translations of $\vert \Psi_j \rangle$ is also an eigenstate of the tight-binding Hamiltonian corresponding to $E_{FB}=t$ (see Fig. \ref{fig:cross-stich-fb} (b)), as
\begin{equation}
    H \vert \Psi \rangle = E_{FB} \vert \Psi \rangle,\quad \vert \Psi \rangle = \sum_j c_j \vert \Psi_j \rangle.
    \label{eq:degen-eig-st}
\end{equation} 
Therefore, the whole system has one $k$-independent energy at $E_{FB}=t$, which forms the FB. 

Now suppose there are $N$ unit cells in this example lattice, giving $N$ lattice translations of the CLS which are all eigenstates of the FB. The superpositions of these copies are also eigenstates, thereby forming \emph{macroscopic degeneracy}.


We note that a given FB eigenstate could be a linear combination of smaller eigenstates, as shown in Eq. \eqref{eq:degen-eig-st}. However, one can always define an irreducible eigenstate for a given FB lattice that cannot be decomposed into the superposition of smaller eigenstates. For example, in the above 1D example, Eq. \eqref{eq:cross-stich-u1-eig-st} is an irreducible eigenstate. In the next chapter, we will discuss this irreducibility in detail, as it determines most of the properties of FB lattices.  

More generally, the parameters in the example shown in Fig. \ref{fig:cross-stich-fb} are not the only parameters that give a FB in a cross-stitch lattice. For example, let's suppose the horizontal hoppings are $t_1$, diagonal hoppings are $t_2$, and $t_1 \ne t_2$. In order to achieve destructive interference in the $a_{j+1}$ site, we need $t_1 a_j + t_2 b_j=0$. In the same way, we can achieve destructive interference at all other neighboring unit cells such that the state is strictly localized at the $a_j,b_j$ sites (i.e. the $j$th unit cell). Therefore, FB networks are \emph{fine-tuned} such that we can control the hoppings as well as the wave amplitudes to achieve destructive interference and thus the FB. This is one of the reasons why we aim in this thesis at identifying all possible FB Hamiltonians for a given type of lattice. For example, in Chapter \ref{chapter4} we identify all possible 1D FB Hamiltonians with two bands.

Above, we took a simple 1D example to explain the origin of FBs, which naturally extends to higher lattice dimensions, where FBs also originate from the same mechanism (destructive interference). Our explanations on the concepts of CLSs, macroscopic degeneracy, and fine-tuning also apply to higher lattice dimensions.


\section{Flatband construction methods} 
\label{section2.4} 

In this section, we briefly review existing methods for designing FB lattices. The central idea in modeling FB lattices is to achieve destructive interference that supports CLSs. 

\subsection{Geometrical methods}

One direct approach to achieving destructive interference is through lattice geometry. This approach is shown through the following four examples.

\begin{figure}[htb!]
    \centering
    \includegraphics[width=0.8\linewidth]{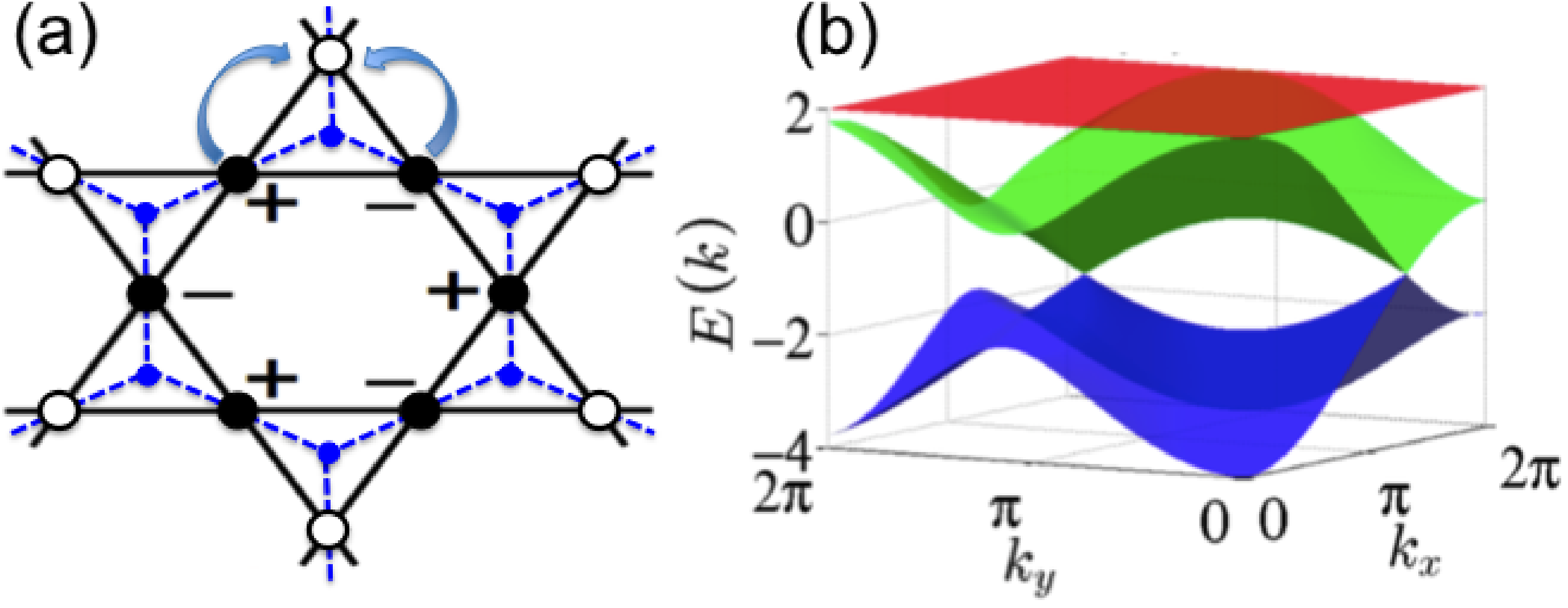}
    \caption[Kagome lattice]{(a) Kagome lattice and (b) its band structure. The original honeycomb lattice is shown in blue dotted lines and blue circles, and the kagome lattice is shown in black solid lines and circles. Filled and empty black circles correspond to lattice sites with non-zero and zero wave amplitudes, respectively. The non-zero amplitudes of the CLS are located in an elementary cycle with alternating "+" and "--" signs, and arrows indicate the two wave functions interfering destructively. The band structure has one FB (red). }
    \label{fig:linegraph-kagome}.
\end{figure}

\begin{itemize}
    \item {\bf Line graph approach:} The line graph is a special type of mapping of an original lattice, which is basically a bond-site exchange. More precisely, given an original lattice, a line graph is obtained by first assigning a site for each of the bonds in the original lattice, and then connecting the sites in the line graph if the corresponding bonds in the original lattice share a common site. Mielke~\cite{mielke1991ferromagnetic,mielke1991ferromagnetism,mielke1991ferromagnetic,mielke1993ferromagnetism} pointed out that destructive interference is a common feature of line graphs. Compact localized states are located in a vertex-disjoint elementary cycle with alternating signs of amplitudes~\cite{mielke1992exact,motruk2012bose}, which gives destructive interference (see Fig. \ref{fig:linegraph-kagome}). The vertex-disjoint cycles in graph theory are cycles that do not share common sites. From 1991--1993, Mielke~\cite{mielke1991ferromagnetism,mielke1991ferromagnetic,mielke1993ferromagnetism} showed that, in line graphs with Hubbard interaction, these CLSs can form a highly degenerate ferromagnetic ground state for a certain electron density. One typical example is a 2D Kagome lattice, which is a line graph of a honeycomb lattice, as shown in Fig.  \ref{fig:linegraph-kagome}. 
    
    \item {\bf Cell construction:} As proposed by Tasaki \cite{tasaki1992ferromagnetism,tasaki2008hubbard,tasaki1998fbferromagnetism}, cell construction starts from an elemental cell consisting of a single internal site and two or more external sites. These cells are then assembled to form a lattice by sharing the external sites, where destructive interference occurs. One typical example of the cell construction method is a 1D sawtooth chain, in which the elemental cell is a triangle, as shown in Fig. \ref{fig:cell_const_sawtooth}. These triangular elemental cells are arranged in such a way that they share the external sites with neighboring elemental cells.
    \begin{figure}[htb!]
        \centering
        \includegraphics[width=0.6\linewidth]{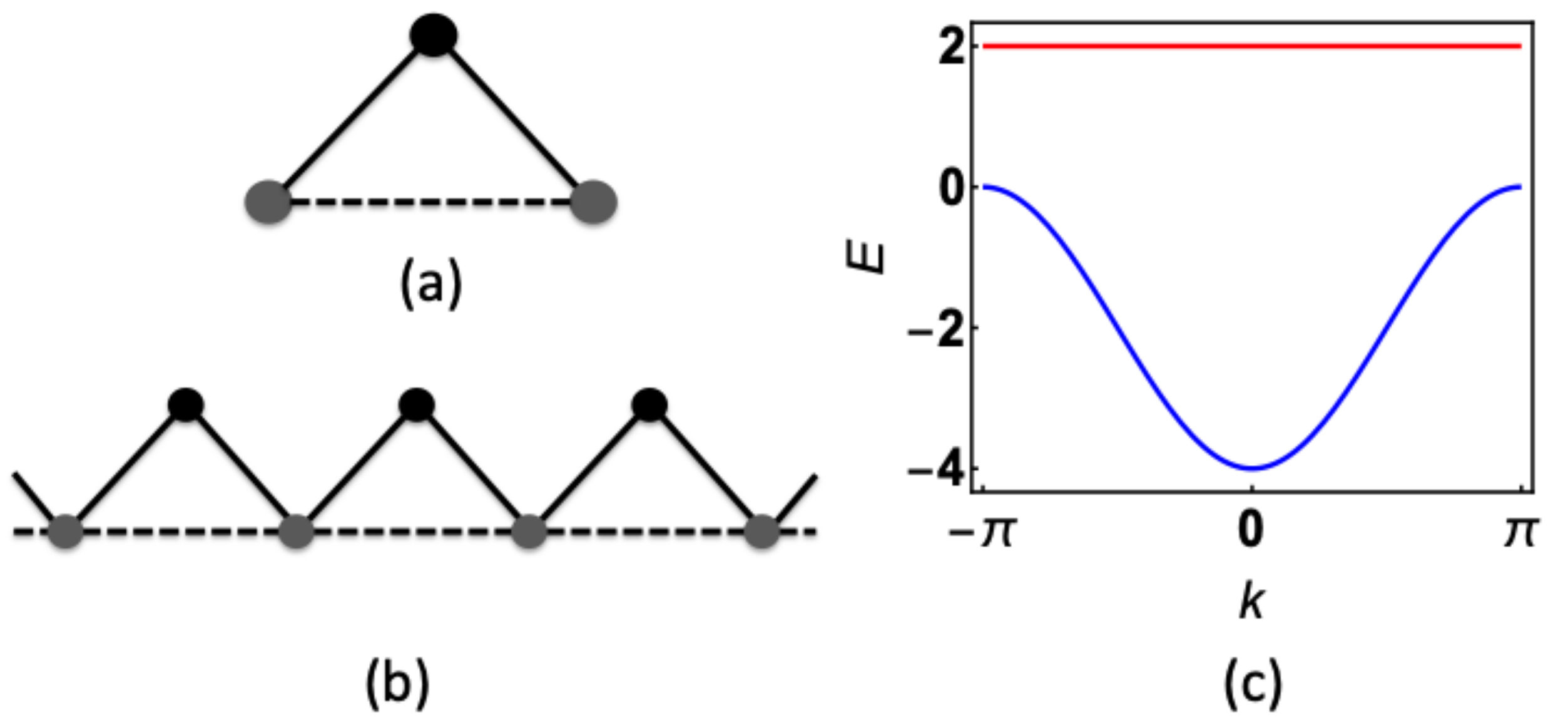}
        \caption[Example of a cell construction method -- sawtooth chain]{(a) Elemental cell, (b) 1D sawtooth chain, and (c) its band structure. The elemental cell has two external sites (gray circles) and one internal site (black circle). The sawtooth chain is constructed by assembling elemental cells with shared external sites (gray circles). When onsite energies are the same and the ratio between diagonal hoppings (solid lines) and baseline hoppings (dashed line) is $\sqrt{2}$, the sawtooth lattice has one FB, as shown in (c).}
        \label{fig:cell_const_sawtooth}
    \end{figure}
    
    \item {\bf Origami rules:} Traditionally a technique to fold paper into beautiful shapes, origami was popularized in mathematics by Fus\`{e} through "modular" or unit origami  \cite{fuse1990unit,fuse1998fabulous}. In 2015, Dias et al.~\cite{dias2015origami} introduced a method to construct localized eigenstates in a Hubbard model using origami rules. In this method, they started from a plaquette or a set of plaquettes with a higher symmetry than that of the whole lattice, with compact localized modes living on these plaquettes that result in FBs. Then a simple set of rules were applied to divide, fold, or unfold the tight-binding localized states in such plaquette(s) to new plaquette geometries.
    
    \item {\bf Repetition of miniarrays:} In 2016, Morales-Inostroza et al.~\cite{morales2016simple} introduced a simple way to construct FB lattices using a repetition of miniarrays. This method starts from a miniarray and consecutively adds new miniarrays using a connector site to form 1D and 2D lattices having single or multiple FBs. Using lattice geometry and by tuning the hoppings and wave functions of each site, destructive interference is achieved at the connector sites.
\end{itemize}


A common feature among the above methods is that destructive interference takes place at the corner sites in lattices composed of corner-sharing triangles. As a result, CLSs surrounded by these triangles are formed that give rise to FBs. 

\subsection{Flatband generation from local unitary transformations} 
\label{section2.4.2}

For FB lattices with a CLS occupying the $U=1$ unit cell, it is always possible to detangle FB eigenstates from the dispersive part of the spectrum~\cite{flach2014detangling}. Reverting this procedure leads to one of the most generic FB construction methods for a CLS occupying the $U=1$ unit cell---the entangling method~\cite{flach2014detangling}. More precisely, the method starts from a lattice with isolated sites in each unit cell, then applies a unitary transformation to couple these sites to the rest of the lattice. This procedure preserves the FB that is formed by these isolated sites~\cite{flach2014detangling}. Following Ref. \cite{flach2014detangling}, we present an example to illustrate this method. 

\begin{figure}[htb!]
    \centering
    \includegraphics[width=0.6\linewidth]{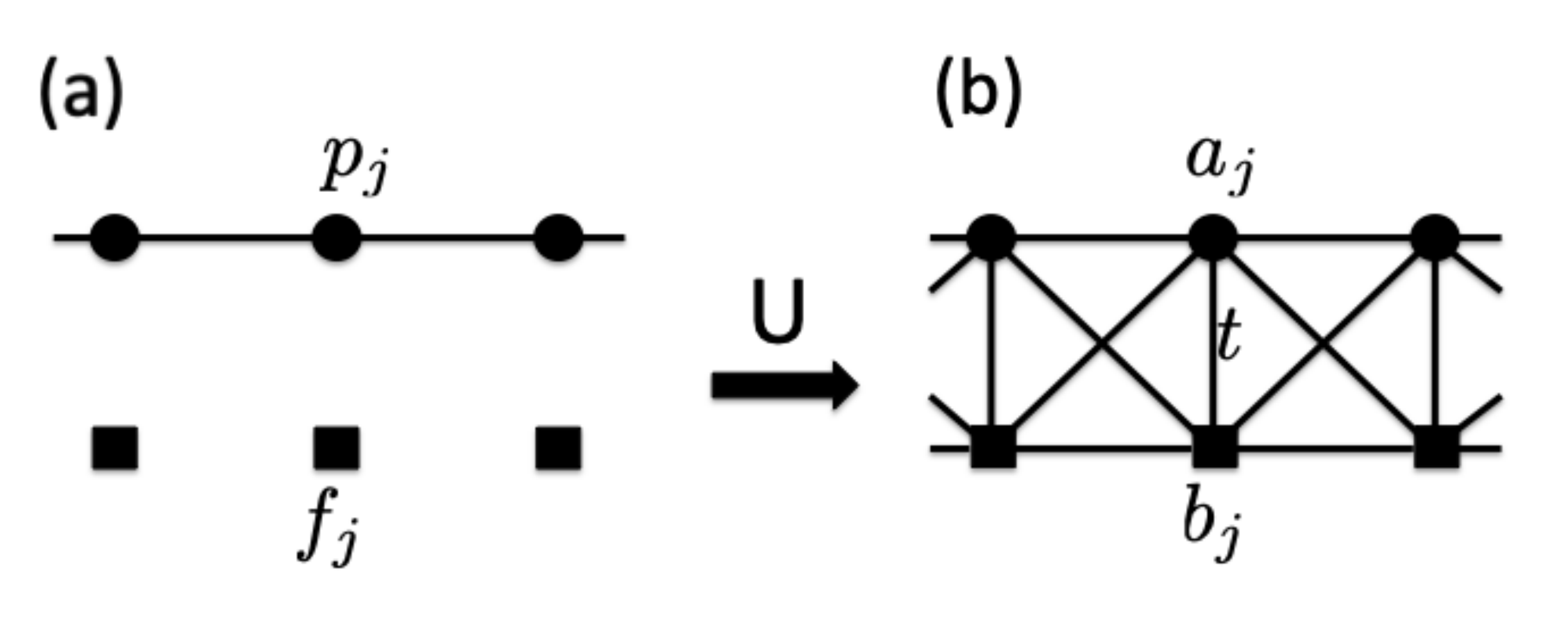}
    \caption[Coupling an isolated lattice to a dispersive chain]{Coupling (a) an isolated site to a dispersive chain to achieve (b) a cross-stitch lattice with a FB.}
    \label{fig:fano-ent-cross-stich}
\end{figure}

Consider a 1D chain with one isolated site nearby each lattice site in the chain, as shown in Fig. \ref{fig:fano-ent-cross-stich} (a). If we set the hoppings between nearest neighboring sites to be $-1$, then the eigenvalue problem reads 
\begin{equation}
\begin{aligned}
    E p_j &= \epsilon_p p_j - p_{j-1} - p_{j+1},\\
    E f_j &= \epsilon_f f_j \;.
\end{aligned}
\label{eq:u1-cross-stich-det}
\end{equation} 
It is convenient to introduce the following hopping matrices, 
\begin{equation}
    H_0 = \begin{pmatrix} \epsilon_p & 0 \\ 0 & \epsilon_f \end{pmatrix}, \quad H_1 = \begin{pmatrix} -1 & 0 \\ 0 & 0 \end{pmatrix}, \quad \vec{\psi}_j = \begin{pmatrix} p_j \\ f_j \end{pmatrix},
\end{equation}
where $H_0$ gives the hoppings and onsite energies inside the unit cell, $H_1$ gives the hoppings between nearest neighboring unit cells, and $\vec{\psi}_j$ is wave function of the $j$th unit cell. The eigenvalue problem in Eq. \eqref{eq:u1-cross-stich-det} becomes 
\begin{equation}
    H_0 \vec{\psi}_j + H_1^\dagger \vec{\psi}_{j-1} + H_1 \vec{\psi}_{j+1} = E \vec{\psi}_j.
    \label{eq:u1-cross-stich-det-mat-form}
\end{equation} 
The lattice described by this equation has one dispersive band formed by the connected sites and one FB formed by the isolated sites, as 
\begin{equation}
    E_{FB} = \epsilon_f,\quad E(k) =\epsilon _p - 2 \cos (k)\;.
    \label{eq:det-cross-stitch-band}
\end{equation}

We can apply unitary transformation to the system, which keeps band structure unchanged. More precisely, if we apply local rotation in the space $\{p_n,f_n\}$,
\begin{equation}
\begin{aligned}
    \tilde{H}_0 &= \hat{U}  H_0  \hat{U}^\dagger,\quad \tilde{H}_1 = \hat{U}  H_1  \hat{U}^\dagger,\quad  \tilde{\psi}_j = \hat{U} \vec{\psi}_j\\
    \hat{U} &= \begin{pmatrix} \cos (\theta) & -\sin (\theta) \\
                              \sin (\theta) & \cos (\theta)        
            \end{pmatrix}\;,
\end{aligned} 
\label{eq:cross-stich-rotation}
\end{equation}
the eigenvalue problem after unitary transformation reads, 
\begin{equation}
    E \tilde{\psi}_j = \tilde{H}_0 \tilde{\psi}_j + \tilde{H}_1^\dagger \tilde{\psi}_{j-1} + \tilde{H}_1 \tilde{\psi}_{j+1},
    \label{eq:rot-crsstch-eig-prob}
\end{equation}
which gives the same band structure as in Eq. \eqref{eq:det-cross-stitch-band}. 

If we choose $\theta = \nicefrac{\pi}{4}$, we have 
\begin{equation}
    \tilde{H}_0 = \frac{1}{2} \begin{pmatrix} 
 \left(\epsilon _f+\epsilon _p\right) &  \left(\epsilon _p-\epsilon _f\right) \\
  \left(\epsilon _p-\epsilon _f\right) &  \left(\epsilon _f+\epsilon _p\right)  \end{pmatrix} , \quad 
 \tilde{H}_1 = \frac{1}{2} \begin{pmatrix}
 -1 & -1 \\ -1 & -1  \end{pmatrix} , \quad 
 \tilde{\psi}_j = \frac{1}{\sqrt{2}} \begin{pmatrix} p_j+f_j \\ p_j - f_j \end{pmatrix} \; .
 \label{eq:rotated-crsstch-hopp-mat}
\end{equation} 
If we make the following variable replacements,
\begin{equation}
    \begin{aligned}
        a_j &= \frac{1}{\sqrt{2}}(p_j + f_j), \quad b_j = \frac{1}{\sqrt{2}}(p_j - f_j) , \\
        \epsilon &= \epsilon_f + \epsilon_p \ \quad t = \epsilon_p - \epsilon_f \; ,
    \end{aligned}
\end{equation} 
we get a cross-stitch lattice as shown in Fig. \ref{fig:fano-ent-cross-stich} (b).

This procedure transformed a 1D chain with isolated sites into a cross-stitch lattice using local unitary operator $U$, while keeping the same band structure, and thus the FB. A similar procedure can be applied to other lattices with CLS size $U=1$;
this procedure can also be generalized to any dimension, as described below. 

\begin{figure}[htb!]
    \centering
    \includegraphics[width=0.6\linewidth]{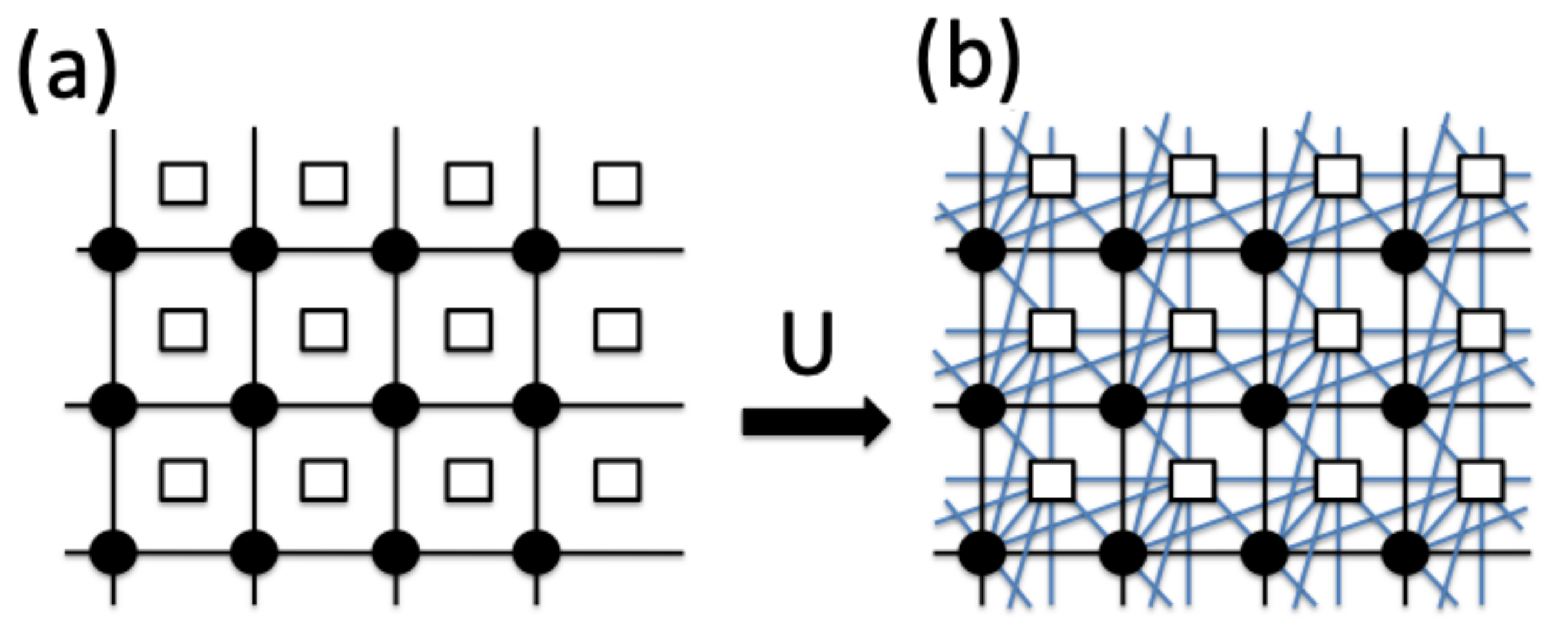}
    \caption[Generalization of the entangling procedure]{Generalization of the entangling procedure. (a) A 2D square lattice with isolated sites, and (b) rotated version. Bonds show only the connectivity, not the actual values. This figure is taken from \cite{flach2014detangling}.}
    \label{fig:fano-ent-gen}
\end{figure}

{\bf Entangling method for flatband construction}: Consider a $d$-dimensional dispersive lattice with $m$ bands ($m$ sites per unit cell), then assign $n$ decoupled sites to each unit cell. By applying local unitary transformation, we can couple these isolated sites to the remaining sites and form a lattice with $n$ FBs, which has a CLS of size $U=1$~\cite{flach2014detangling}. The graphical illustration of the simplest transformation in 2D with $m=n= 1$ is shown in Fig. \ref{fig:fano-ent-gen} (b).  

This method is the most generic method to construct FB lattices with a CLS occupying a single unit cell. For CLS sizes larger than 1, one cannot decouple the FB into isolated sites because the CLSs do not form an orthogonal basis~\cite{flach2014detangling}; this is discussed in the next chapter. Therefore, the inverse procedure (entangling procedure) is not possible for CLS size $U\ge2$.
In later chapters, we introduce other methods for constructing FB lattices with CLS size $U \ge 2$.

\subsection{Flatband construction from symmetry}
\label{section2.4.3} 

As mentioned previously, in most known cases, the existence of FBs depends on special lattice geometry or symmetry, and we have already reviewed the lattice geometry-based FB construction methods. In this section, we discuss symmetry-based FB construction methods.

\subsubsection*{Chiral symmetry}

Bipartite lattices host FBs that are protected by chiral sysmmetry, and a systematic way to construct chiral FBs has been proposed by Ramachandran et al.~\cite{ramachandran2017chiral}. Based on this work, we discuss this chiral FB-generating principle.

{\bf Bipartite lattice}: A lattice consisting of two sublattices A and B, such that lattice sites in A are coupled only to sites in B, and vice versa. 

{\bf Chiral symmetry}: In lattice systems, chiral symmetry refers to both particle-hole and time-reversal symmetry. This concept originated from quantum chromodynamics (QCD)~\cite{schwinger1957theory,schwinger1967chiral,kantor1968chiral}. The consequence of chiral symmetry is: if $\{ \psi^A, \psi^B \}$ is an eigenvector of eigenenergy $E$, then there is an eigenvector $\{ \pm \psi^A, \mp \psi^B \}$ corresponding to eigenenergy $-E$.

A bipartite lattice with chiral symmetry is called a chiral lattice. From Lieb's theorem~\cite{lieb1989two}, it can be inferred that chiral lattices with an odd number of bands always possess at least one chiral FB at energy $E=0$. More precisely, if the number $N_A$ of sublattice A sites is larger than the corresponding number $N_B$ of sublattice B sites, then there are at least $N=|N_A - N_B|$ states $\{ \psi^A, 0\}$ at energy $E=0$~\cite{lieb1989two,sutherland1986localization}, which only live in sublattice A. The sublattice A (B) with larger (smaller) number of sites is called the majority (minority) sublattice, with non-zero amplitudes of the eigenstates of energy $E=0$ occupying the majority sublattice only. This result leads to a systematic classification of chiral FBs. 

Consider a $d$-dimensional translational invariant bipartite lattice with odd number $\nu=\mu_A+\mu_B$ of sites per unit cell, with $\mu_A$ sites belonging to sublattice A and $\mu_B$ sites belonging to sublattice B. Suppose sublattice A is the majority sublattice and B is the minority sublattice, such that $1<\mu_B<\mu_A<\nu$. The $\mu_A$ sites in any unit cell are only connected to the remaining $\mu_B$  sites (some of them possibly belonging to other unit cells) by nonzero hopping terms $t_{lm}$. Then, the eigenvalue problem reads
\begin{equation}
    E \psi_l^{A,B} = -\sum_m t_{lm} \psi_m^{B,A},
    \label{eq:bipart-eig-eq}
\end{equation}
 and chiral symmetry leads to a symmetric band structure, i.e. any band $E_{\mu}(k)$ which does not cross $E=0$ is either positive or negative, and has a symmetry-related partner band $E_{\mu^\prime}(k)=-E_{\mu}(k)$. Due to the odd number of bands, there is at least one band without a symmetric partner band. Therefore, under chiral symmetry action, this \emph{unpaired} band must transform into itself, becoming a FB at energy $E=0$. 
 
Since $\nu$ is odd, the difference in the number of sites on sublattices A and B is $\Delta N = N_{uc} (2 \mu_A - \nu) \ne 0$, where $N_{uc} = L^d$ is the number of unit cells, and $L$ is the linear dimension. This indicates a macroscopic degeneracy at $E = 0$, which is possible only when there are $(2 \mu_A - \nu)$ FBs at $E =0$. Classification of chiral FBs by the imbalance of minority and majority sites follows from this observation, and can be used to generate chiral FBs~\cite{ramachandran2017chiral}. While this work considered Hermitian systems, the concepts can apply to non-Hermitian systems as well. 
 
 
 \subsubsection*{Local symmetry partitioning} 
 
Schmelcher et al.~\cite{roentgen2018compact} proposed a framework to design CLSs using local symmetry properties of a discrete Hamiltonian, which can be used to construct tunable FB lattices in 1D and 2D. In this framework, the Hamiltonian of the system can be block partitioned using two recent theorems from graph theory~\cite{barrett2017equitable,francis2017extensions,fritscher2016exploring}. Such partitioning of the Hamiltonian is induced by the symmetry of a given system under local site permutations. Based on Ref.~\cite{roentgen2018compact}, we briefly discuss this local symmetry partitioning method for constructing FB lattices.

Consider a three-site system as shown in Fig. \ref{fig:local-sym-part} (a). The Hamiltonian of this system is invariant under the permutation of sites 2 and 3. When the wave functions at sites 2 and 3 have opposite phases, they interfere destructively at site 1, leading to the localization of wavefunctions at sites 2 and 3. If an arbitrary system is connected to site 1, as shown in Fig. \ref{fig:local-sym-part} (b), the wave functions still localize at sites 2 and 3. The eigenvalue corresponding to this localized state in the enlarged system is the same as the original three-site system. Sites 2 and 3 together are called a symmetric subsystem, and site 1 with the attached system is called a non-symmetric subsystem. Then, the Hamiltonian of the system can be partitioned into blocks. 
\begin{figure}[htb!]
	\centering
	\includegraphics[width=0.8\linewidth]{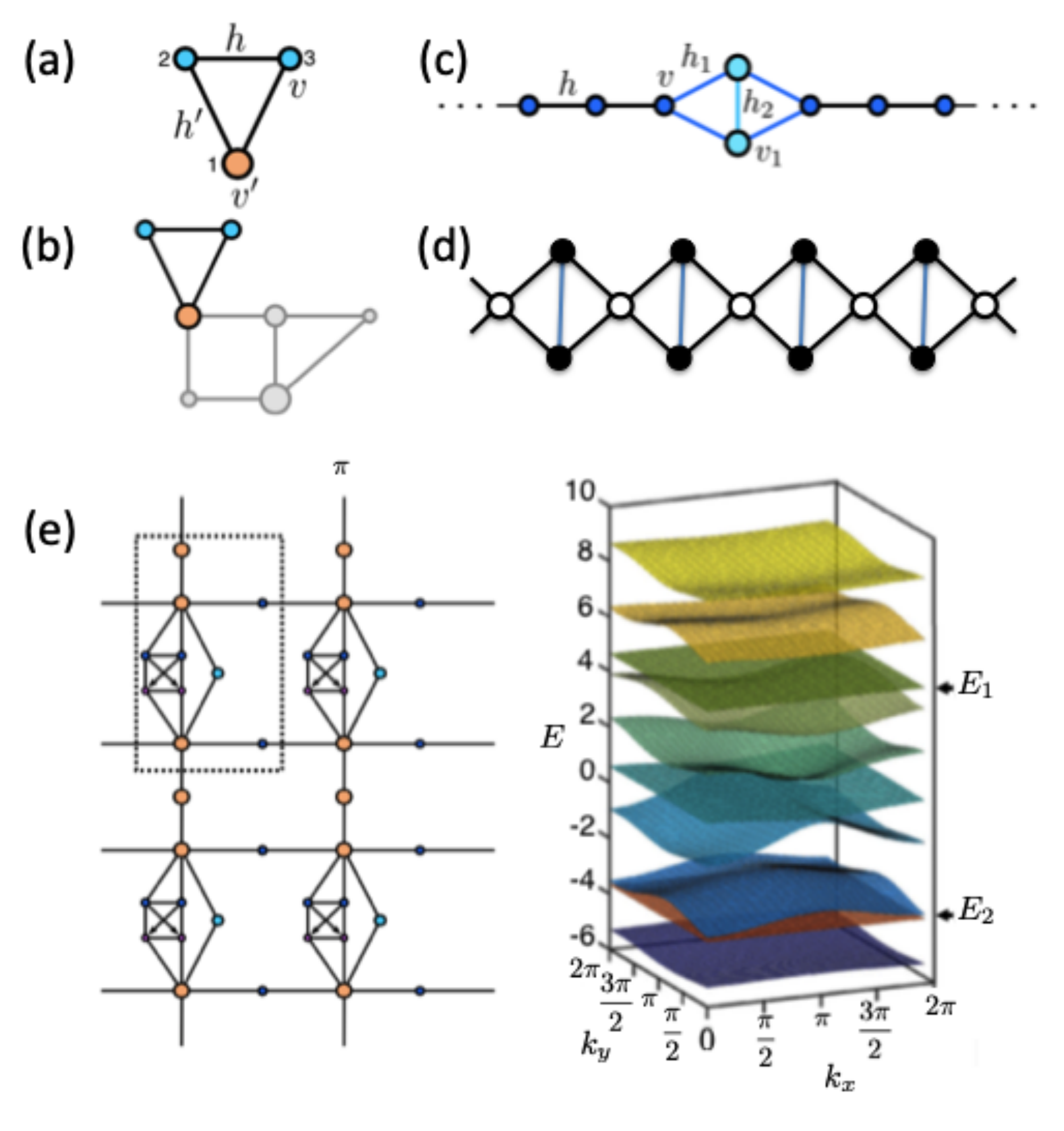}
	\caption[Flatband construction by local symmetry partitioning]{Flatband construction by local symmetry partitioning. (a) A three-site system, which is symmetric under the permutation of sites 2 and 3. (b) The extended system with an arbitrary subsystem (gray) attached to site 1 (which is fixed under permutation). (c) A dispersive chain locally perturbed by a symmetric dimer defect with indicated onsite and hopping elements. (d) A diamond chain formed by periodically adding dimers to the dispersive chain in (c). (e) A 2D FB lattice constructed by the local symmetry partitioning method and its band structure having two FBs at $E_1,\ E_2$. Figures (a)--(c) and (e) are taken from Ref.~\cite{roentgen2018compact}. }
	\label{fig:local-sym-part}
\end{figure}

Now, if a dimer is added to a dispersive chain, with the hoppings and onsite energies as shown in Fig. \ref{fig:local-sym-part} (c), such that the dimer has permutation symmetry, there will be a CLS living in the dimer. Then the system Hamiltonian can be partitioned into blocks corresponding to the symmetric part (dimer) and the non-symmetric part (dispersive chain). When the system parameters are tuned such that the spectrum of the non-symmetric part of the Hamiltonian coincides with the spectrum of the original dispersive chain, the spectrum of the system consists of the band structure of the unperturbed chain augmented by the energy of the CLS. Since, in the present case, the CLS does not interact with the extended states, one can add such dimers periodically to form a lattice, as shown in Fig. \ref{fig:local-sym-part} (d), whose spectrum consists of a FB from the eigenenergy of the CLS.

As discussed in Ref.~\cite{roentgen2018compact}, the local symmetry (permutation) operators can be identified as commutative and non-commutative, depending on the commutation of these operators with the system Hamiltonian. As in the above case, generally, the Hamiltonians of such systems with local symmetry can be partitioned into blocks corresponding to symmetric and non-symmetric subsystems. Following a similar construction method as above, more complicated FB lattices can be designed in 1D and 2D, as shown in Fig. \ref{fig:local-sym-part} (e).

\section{Applications and experimental realizations of flatband systems}
\label{section2.5}

Flatband lattices are fine-tuned, and the eigenstates of a FB are macroscopically degenerate. Fine-tuning of the parameters of the FB Hamiltonian suggests their extreme sensitivity to perturbations, which might change drastically the properties of the system and even trigger phase transitions leading to a plethora of interesting physical phenomena. However as is demonstrated in the following chapters there are actually many perturbations that preserve the flatness of the band. These perturbation induced reach physical phenomena has led to FB lattices being the focus of many experiments in different scales, different disciplines, and different energy scales as well. In this section, we provide a brief review of such phenomena and experimental realizations.


\subsection{Perturbation as a playmaker}


Amidst the search for different FB models, other theoretical studies have been focusing on the effects of different perturbations in FB lattices.
In one early example, the effect of repulsive Coulomb interaction in FB models was studied by Mielke and Tasaki, who independently showed that some FB Hubbard models host a ferromagnetic ground state \cite{mielke1991ferromagnetic,mielke1991ferromagnetism,mielke1993ferromagnetism,tasaki1992ferromagnetism}. Flatband ferromagnetism has been an interesting topic of various studies \cite{tasaki1994stability,tasaki2008hubbard,maksymenko2012flatband}; as discussed in the previous section, their approaches also provide a method for constructing FB models.  

Disorder is another important class of perturbation that has been widely studied in FB systems. Except for symmetry-protected FBs, such as chiral FBs, the majority of them are sensitive to disorder. Adding onsite disorder to a FB lattice will induce Anderson localization, which breaks the FB and causes the CLS to become exponentially localized. In Ref. \cite{flach2014detangling}, the authors added onsite disorder to FB lattices, and found energy-dependent localization length scaling in terms of the Fano resonances. When locally correlated disorder and quasiperiodic potentials are added in FB lattices, they show vanishing localization lengths for arbitrarily weak disorder and mobility edges for quasiperiodic perturbations \cite{bodyfelt2014flatbands,danieli2015flatband}.

In higher dimensions, more interesting phenomena have been found. In 2D, Chalker et al. reported multifractality of the FB eigenstates in the weak disorder limit~\cite{chalker2010anderson}, which is induced by the long-range decay of the projected interaction. In 3D disordered FBs, Goda et al.~\cite{goda2006inverse,nishino2007flat} numerically demonstrated an \emph{inverse} Anderson transition, where all states are localized at first, then delocalize at a critical disorder strength before localizing again when a second critical disorder strength is reached. 
Such transitions are also seen in the level spacing statistics of certain 2D FBs~\cite{shukla2017criticality}.
Recently, a number of new phenomena in disordered FB systems are attracting more attention, including FBs under nonquenched (evolving) disorder~\cite{radosavljevic2017light}, disorder-induced topological phase transitions~\cite{chen2017disorder}, and the temporal dynamics of disordered FB states~\cite{gneiting2018lifetime}. 

The effect of external fields has also been studied. Khomeriki et al. reported that FB lattices under an external field show sinusoidal Bloch oscillations in the band structures. This study was conducted on a 1D diamond chain \cite{khomeriki2016landau}, and it was observed that the FB almost completely stops the Bloch oscillations for a substantial time when the wavefunction is trapped in the original FB part of the unperturbed bands, made possible by Landau--Zener tunneling \cite{khomeriki2016landau}. 

\begin{figure}[htb!]
    \centering
    \includegraphics[width=0.8\linewidth]{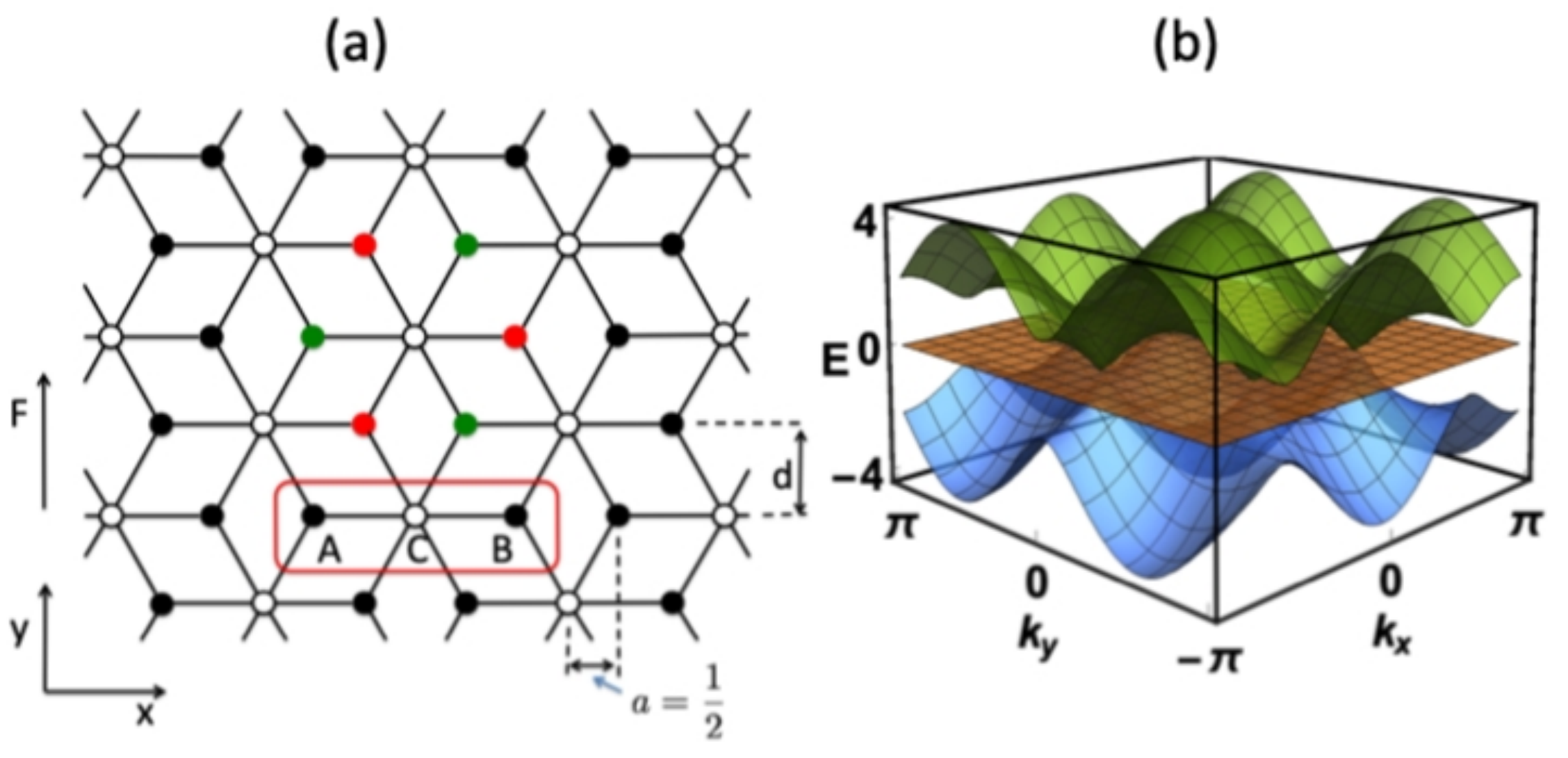}
    \caption[Dice lattice and its band structure]{(a) Dice lattice and (b) its band structure without an external DC field. The unit cell consists of three sites denoted by A, B, and C. The red and green circles represent the amplitudes $+ \nicefrac{1}{6}$ and $- \nicefrac{1}{6}$, respectively, of a CLS. The sites in the red box indicate an uncoupled trimer in the limit of infinitely strong DC bias~\cite{kolovsky2018topological}.}
    \label{fig:dice-lattice}
\end{figure}

More interestingly, there are cases where FBs can be partially preserved under an external DC field. Kolovsky et al.~\cite{kolovsky2018topological} showed that, under such field, the energy spectrum of a 2D dice lattice (see Fig. \ref{fig:dice-lattice}) forms a new band structure consisting of the periodic repetition of 1D energy band multiplets, with one of them being flat (see Fig. \ref{fig:dice-ws-band}). Without an external field, the dice lattice has a FB supported by a CLS. When an external field is applied, the FB in each multiplet is supported by noncompact exponentially localized states~\cite{kolovsky2018topological}. 


\begin{figure}[htb!]
    \centering
    \includegraphics[width=0.6\linewidth]{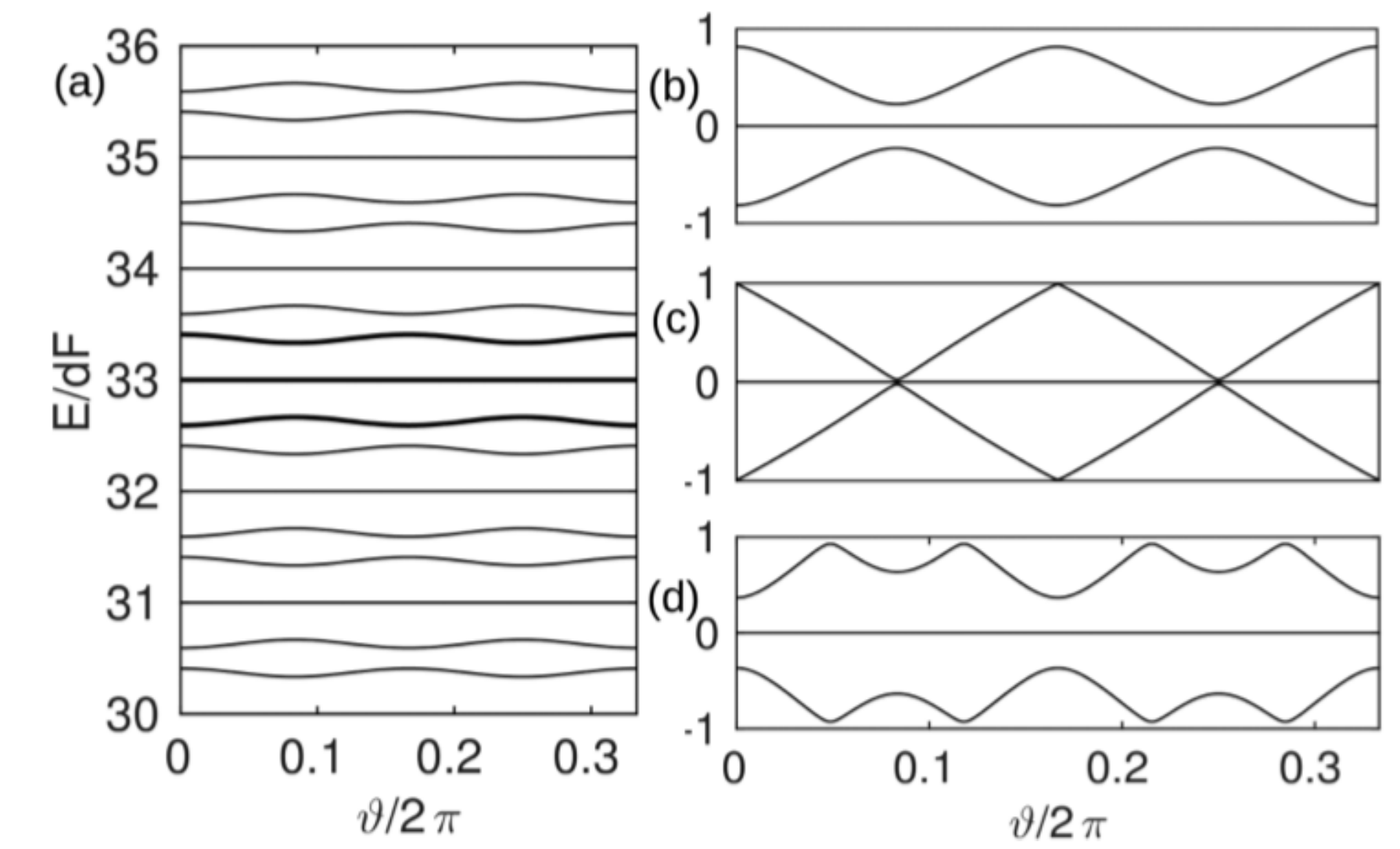}
    \caption[Wannier--Stark ladder of a dice lattice]{(a) Wannier--Stark (WS) band ladder of the biased dice lattice in Fig. \ref{fig:dice-lattice} with ${\bf F}$ along the $y$ axis. The thick lines highlight the irreducible triplet of three WS bands. (b)--(d) Irreducible triplets for different field strengths~\cite{kolovsky2018topological}. This figure is courtesy of Ramachandran et al. \cite{kolovsky2018topological}}
    \label{fig:dice-ws-band}
\end{figure}

As it is known, adding nonlinearity lifts the spectral band structure of linear networks. 
Interestingly though, it has recently been shown that local Kerr nonlinearity may preserve the destructive interference of a FB network, yielding \emph{compact breather solutions}. 
These solutions are periodic in time and compact in space, and they have been discovered in several sample nonlinear FB networks~\cite{johansson2015compactification, belicev2017localized}. 
Moreover, compact breathers have also been observed in the framework of ultra-cold atoms in a diamond chain lattice network, where nonlinear terms and spin-orbit coupling are simultaneously present \cite{gligoric2016nonlinear}.
Likewise, Perchikov and Gendelman studied compact time-periodic solutions in nonlinear mechanical cross-stitch networks \cite{perchikov2017flat}. 
In 2018, Danieli et al. established the necessary and sufficient conditions for the continuability of linear CLSs as compact breathers in the nonlinear regime of FB networks \cite{danieli2018breather}. 
Proven to be linearly stable once their frequency is tuned off-resonance with linear dispersive states and overlapping linear CLS, 
compact breather solutions have been recently shown to act as Fano scatterers for dispersive waves \cite{ramachandran2018fano}.

 
Further, recent studies have shown that, if local attractive interactions are considered, FB systems exhibit superfluidity\cite{murad2018performed,peotta2015superfluidity,julku2016geometric,tovmasyan2016effective}. In 2015, Peotta et al.~\cite{peotta2015superfluidity} demonstrated the promising applications of topologically nontrivial FBs to increase the critical temperature of the superconducting transition. In 2016, Julku et al.~\cite{julku2016geometric} studied the transport properties of the Lieb lattice FB in the presence of an attractive Hubbard interaction, and showed that topologically trivial FBs also have the potential to achieve high-$T_c$ superconductivity.

\subsection{Experimental realizations and applications} 

The fine-tuned nature of FB lattices requires special configurations of lattice geometry and hoppings. Therefore, real materials with macroscopically degenerate FBs are rare in nature, which makes the search for such materials an attractive and ambitious task for many experimentalists. With the development of fabrication technologies, artificial lattices with FBs have been designed and tested in various systems. The effort to observe FBs in laboratory was started soon after theoretical models for FB ferromagnetism were first established. Here, we briefly review experimental schemes to realize FBs in different frameworks; more details regarding experimental realizations of FBs in artificial lattices are given in a review paper by Leykam et al.~\cite{leykam2018artificial}.

As mentioned, the earliest efforts to realize FBs systems started from ferromagnetism. Even though no direct observation of a FB ferromagnet has been reported so far, it is worth bringing up some important theoretical proposals and experiments that demonstrate the possibility of realizing FB ferromagnetism.  
Tamura et al.\cite{tamura2002ferromagnetism} theoretically showed the potential to observe FB ferromagnetism in semiconductor dot arrays using existing fabrication technology.
Other theoretical models have also suggested quantum dot arrays and quantum atomic wires formed on solid surfaces for this purpose~\cite{kimura2002magnetic,shiraishi2004theoretical,ichimura1998deltachain}. Experimental realization remains challenging though because these setups require ultra-cold temperatures. Later, the theoretical proposal to experimentally realize organic FB ferromagnets based on polymers~\cite{arita2002gateinduced,arita2003flatband,suwa2003flatband,aoki2004design} also turned out to be very challenging. In 2015, FBs were observed in the electronic structure of silicene~\cite{Hatsugai2015flatband}, with which FB ferromagnetism may be possible. More recently, hole-doped pyrochlore oxides $Sn_2Nb_2O_7$ and $Sn_2Ta_2O_7$ were reported by I. Hase et al.~\cite{hase2018possibility} to be candidates to realize FB ferromagnetism. For more information about FBs in spin systems, readers can refer to the review paper by Derzhko et al.~\cite{derzhko2015strongly}.  

Although no FB ferromagnet has yet been found, experimental FBs have been realized in many other platforms, such as electronic systems~\cite{tadjine2016from,qiu2016designing,drost2017topological,slot2017experimental}, optical lattices~\cite{taie2015coherent,apaja2010flat,ozawa2017interaction,an2017flux}, photonic lattices~\cite{vicencio2015observation,mukherjee2015observation,weimann2016transport}, and exciton-polariton condensates~\cite{masumoto2012exciton,jacqmin2014direct,baboux2016bosonic}.

\paragraph{Flatbands in electronic systems:}

In 1998, Vidal et al.~\cite{vidal1998aharonov} found a completely flat spectrum induced by a critical magnetic field strength in certain periodic electronic networks.
This flat spectrum is the result of an interference effect, analogous to the Aharonov--Bohm (AB) effect~\cite{aharanov1959significance}, which is induced by a $\pi$ magnetic flux threaded through each unit cell. This effect is thus referred to as AB caging, and the paper suggested the the possibility of observing this effect in superconducting wire networks. 
Even though the direct observation of single-particle band structure is not possible in such networks, and measurements are limited to transport properties, the indirect observation of AB caging was reported by Abilio et al.~\cite{abilio1999magnetic}. Later, in 2001, indications of AB caging were observed in the magnetoresistance oscillations of a low-temperature $(30\ mK)$ normal metal dice lattice~\cite{naud2001aharonov}, where a lattice period of $\approx \mu m$ was used to reach the required magnetic field.

Progress in nano-fabrication methods, such as lithography and atomic manipulation techniques~\cite{tadjine2016from,qiu2016designing}, have enabled experimentalists to devise artificial FB lattices for electrons on the nano-scale. Some works have employed a scanning tunnelling microscope (STM) to embed a 2D Lieb lattice onto a substrate surface. In one example, Slot et al.~\cite{slot2017experimental} designed an electronic Lieb lattice formed by the surface state electrons of Cu(111), where they used an array of carbon monoxide molecules positioned with an STM to confine the electrons and observed the predicted characteristic electronic structure of the Lieb lattice (see Fig. \ref{fig:exp-real-atomic-optical-fb} (b)).
Using a low-temperature STM, Drost et al.~\cite{drost2017topological} created a vacancy lattice by removing atoms from a chlorine monolayer, and designed analogues of the poly-acetylene (dimer) chain and Lieb lattice. 

\begin{figure}[hbt!]
    \centering
    \includegraphics[width=0.9\linewidth]{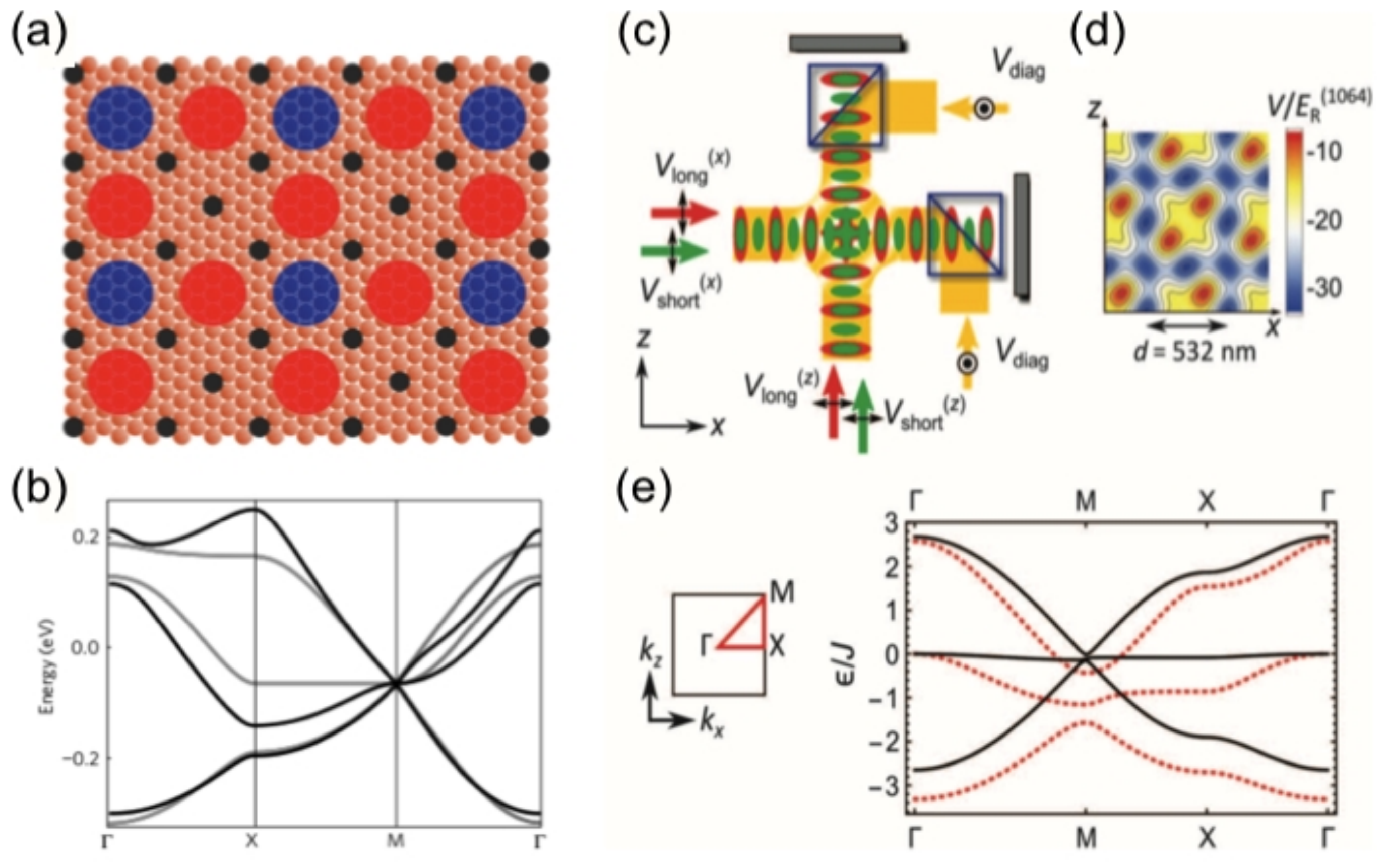}
    \caption[Experimental realizations of flatbands in engineered atomic lattices and optical lattices.]{Experimental realization of FBs in an engineered atomic lattice and optical lattice as taken from Refs.~\cite{slot2017experimental,taie2015coherent}. (a) Schematics of a Lieb lattice formed from surface state electrons of Cu(111) (red and blue circles) confined by an array of carbon monoxide molecules (black circles) from~\cite{slot2017experimental}, with (b) corresponding band structure. Due to next nearest neighboring interactions, the FB become dispersive. (c) Experimental realization of a Lieb optical lattice from~\cite{taie2015coherent}. Black arrows indicate the polarizations of the lattice beams. (d) Lattice potential profile, and (e) band structures for different lattice potentials.}
    \label{fig:exp-real-atomic-optical-fb}
\end{figure}

\paragraph{Flatbands in optical lattices:}

The development of laser cooling techniques, ion trapping, and optical lattices has provided a new platform for the realization of FB lattices. Tunability of the potential in these techniques has made it possible to devise FB lattices by arranging atoms and tuning the hoppings. Lieb lattices are the most-studied FB model in optical lattices due to their relatively simple geometry and novel properties under interactions. 

An optical Lieb lattice was first realized by Shen et al. \cite{shen2010single} in order to experimentally observe massless Dirac fermionic behavior in the vicinity of the Dirac cone. In optical lattices with relatively shallow optical potential, next nearest neighbor hoppings are unavoidable, leading to a non-zero width of their 'flat' band. The exact Bloch wave spectrum for such optical lattices has been computed numerically by Apaja et al.~\cite{apaja2010flat}, who found that a 'flat' band with a width of only $1.5 \%$ of the total bandwidth could be achieved. 
Later, Takahashi et al.~\cite{taie2015coherent} reported the experimental realization of an optical Lieb lattice with bosonic condensation (see Fig. \ref{fig:exp-real-atomic-optical-fb} (c--e). 
By manipulating the optical potential, they were able to engineer the population and phase of each lattice site, which enabled them to coherently transfer atoms into the FB and observe the CLS. Manipulating system parameters enabled them to control the delocalization transition of the CLS, and detect the presence of FB-breaking perturbations. As predicted by Apaja et al.~\cite{apaja2010flat}, they observed interactions that induced a decay of the condensate into the lower dispersive band.

Later, in 2017, Ozawa et al.~\cite{ozawa2017interaction} found that FB energy at the edge of the Brillouin zone is shifted by interactions, leading to a non-zero dispersion. These measurements were in accordance with the tight-binding limit of the Gross--Pitaevski equation. Subsequently, Taie et al.~\cite{taie2017spatial} used fermionic cold atoms in a Lieb lattice to demonstrate the transport of particles adiabatically between the horizontal and vertical rim sites through a dark state.

Aside from Lieb lattices, optical lattices have been realized in the form of 1D sawtooth chains by An et al. in ~\cite{an2017flux}.

\begin{figure}[htb!]
    \centering
    \includegraphics[width=0.9\linewidth]{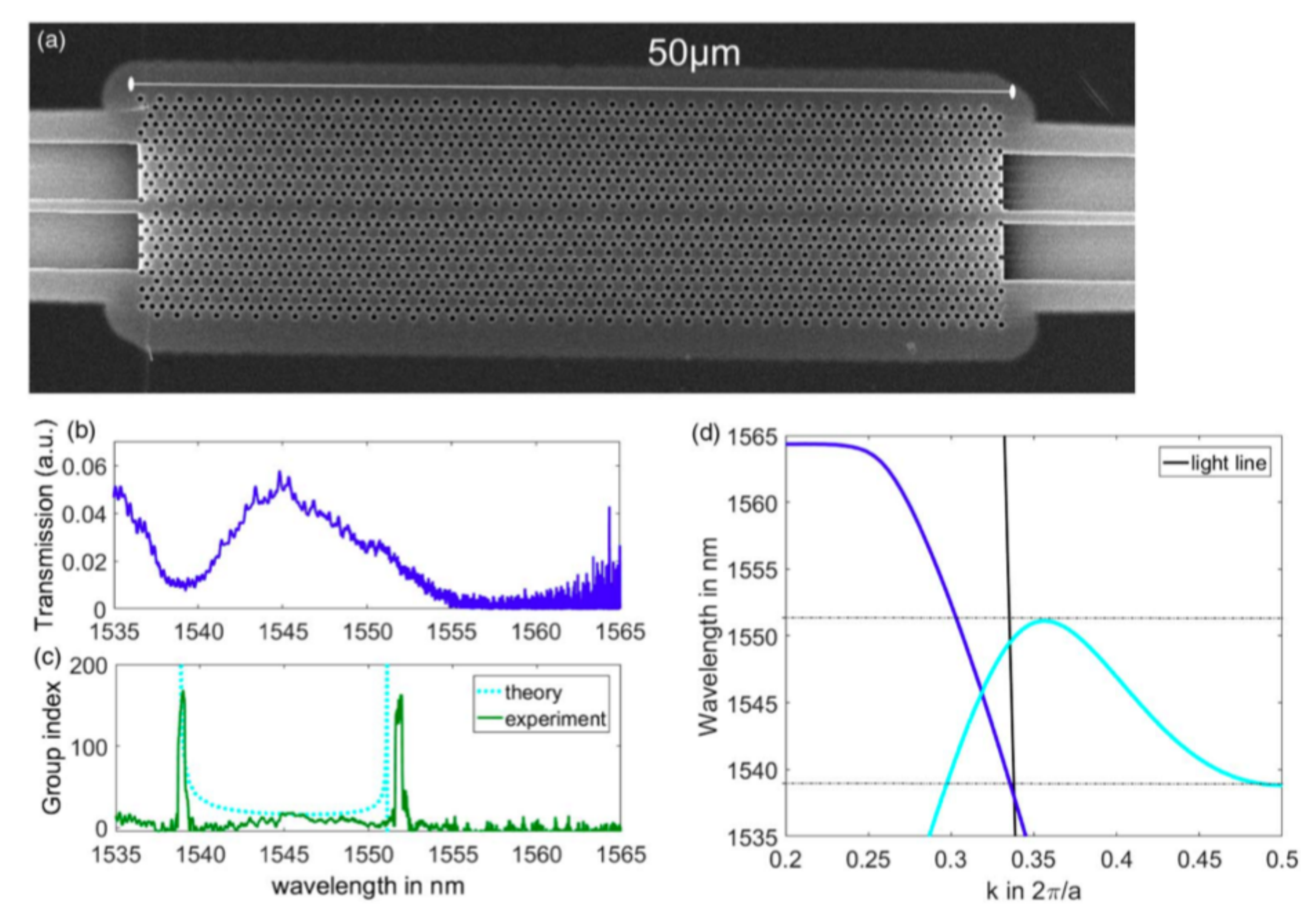}
    \caption[Photonic crystal waveguide of a kagome lattice.]{(a) Scanning electron microscope image of an air-bridge kagome-lattice photonic crystal waveguide. (b) Experimental transmission spectrum of the waveguide shown in (a). (c) Experimental (solid green) and theoretical (dashed light-blue) group index spectra. The corresponding band structure is shown in (d). This figure is from Schulz et al. \cite{schulz2017photonic}.}
    \label{fig:pcw-slowlight-kagome}
\end{figure}

\paragraph{Flatbands in photonic systems:} 

In photonics, FBs have significant applications in the technologically important concept of slow light, which could be advantageous for the buffering and time-domain processing of optical signals. From 2001, it has been known that slow light can be generated by strong dispersion close to the photonic band edge of a photonic crystal waveguide (PCW)~\cite{letartre2001group,notmi2001extremely}. However, slow light in PCWs has two issues to be considered: the frequency bandwidth of the effect and higher-order dispersion. To begin with, a balance must always be achieved between reducing group velocity and maintaining a useful operation bandwidth. Furthermore, higher-order dispersion must be avoided, as it severely distorts optical signals. Flatbands can solve higher-order dispersion, making them valuable in PCW-based slow light applications. 

An early proposal for achieving FBs in photonic crystals came from Takeda et al.~\cite{takeda2004flat}. Their model was composed of circular rods, with electromagnetic waves localizing at the rods that formed photonic FBs. This proposal was not successfully pursued though due to fabrication difficulties. Instead, in slow light schemes, previous research was mostly focused on 1D photonic waveguide arrays or 2D lattices designed through defects in triangular photonic crystals and optimized by numerical or intuitive approaches~\cite{baba2008slow,li2008systematic,xu2015design}. In 2017, Schulz et al. presented an alternative to the standard triangular-structured waveguides, based on the kagome lattice, that demonstrated an improved reduction of group velocity and even stopped light~\cite{schulz2017photonic} (see Fig. \ref{fig:pcw-slowlight-kagome}). The kagome-lattice waveguides inspired further slow light engineering in PCWs~\cite{feigenbaum2010resonant,endo2010tight,nixon2013observing,nakata2012observation,kajiwara2016observation}.

Another active approach to realize FBs is through the use of the femtosecond laser writing technique to fabricate optical waveguide networks ~\cite{szameit2010discrete}. The advantages of this technique over other photonic systems is the long propagation distance and arbitrarily tunable coupling between waveguides, which allows for the exploration of time-periodic lattice dynamics, and Bloch oscillations induced by weak potential gradients. 
In 2014, Guzm\'{a}n-Silva et al.~\cite{guzman2014experimental} were the first to fabricate a Lieb lattice based on phonic waveguide arrays, and studied the effect of the FB on the transport properties of the lattice. However, the presence of the FB could only be deduced indirectly, because a superposition of all bands is excited by the single waveguide input in this original experiment. In subsequent studies, Vicencio and Mukherjee et al.~\cite{vicencio2015observation,mukherjee2015observation} used a multi-waveguide input beam to directly excite the CLS and observe nondiffracting photons (see Fig. \ref{fig:ow-lieb-and-ep-stub} (a)). Weimann et al. \cite{weimann2016transport} studied the transport properties of a photonic waveguide sawtooth lattice. The optical induction technique~\cite{efremidis2002optinduction,makasyuk2006optinduction,schwartz2007photonic} was used by Xia et al. \cite{xia2016demonstration} to create photonic Lieb lattices, where they demonstrated distortion-free image transmission using the CLS. Further related studies are ongoing.

\begin{figure}[htb!]
    \centering
    \includegraphics[width=0.9\linewidth]{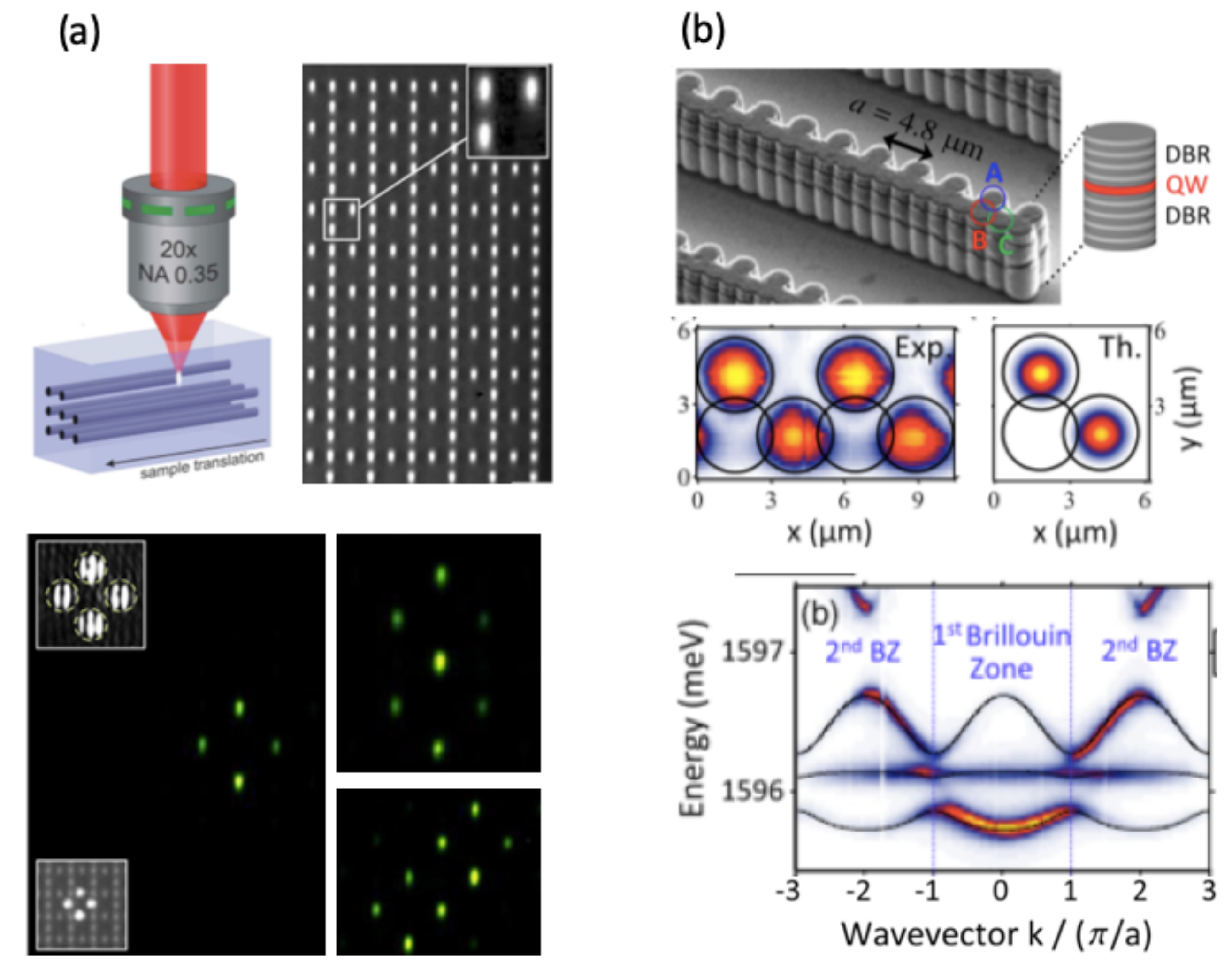}
    \caption[Optical waveguide Lieb lattice and exciton-polariton stub lattice.]{(a) Optical waveguide Lieb lattice fabricated by the femtosecond laser writing technique and the observed compact localized modes. This figure is from~\cite{vicencio2015observation}. (b) A 1D stub lattice made by exciton-polariton condensate and its band structure. This figure is from~\cite{baboux2016bosonic}.}
    \label{fig:ow-lieb-and-ep-stub}
\end{figure}

\paragraph{Flatbands in exciton-polariton condensates:} 

Microcavity exciton-polariton condensates provide another way of achieving FB lattices. In 2012, Masumoto et al. \cite{masumoto2012exciton} realized a 2D kagome lattice with exciton-polariton condensates and observed a weakly dispersive band. Due to limitations in fabrication techniques, the potential shell was shallow. Therefore, a balance needs to be found between smaller lattice periodicity and the minimum next nearest neighbor hoppings in order to resolve the single band while keeping its flatness. Furthermore, condensation in the FB could not be achieved, because the FB was not the lowest band.

Later, improved fabrication techniques such as micropillar cavity etching made it possible to experimentally achieve ideal FBs. In 2014, Jacqmin et al. \cite{jacqmin2014direct} was the first to report polariton FBs in a honeycomb array of micropillar cavities.
Then, also using micropillar optical cavities, a 1D stub lattice was realized in 2016 by Baboux et al. \cite{baboux2016bosonic}; they were able to achieve polariton condensation into a FB  (see Fig. \ref{fig:ow-lieb-and-ep-stub} (b)). More recently, exciton-polariton condensation into a FB was also demonstrated in a 2D Lieb lattice of micropillars \cite{klembt2017polariton,whittaker2018exciton}. 

\paragraph{Future perspective:} 

Beyond these examples, while a number of new platforms have been theoretically proposed, many of them have yet to be experimentally demonstrated. One of them is FB localization in an opto-mechanical system with controllable photons and phonons, which might bring wide applications in information storage, transfer, and the engineering solid-state quantum devices. In 2017, Wan et al. \cite{wan2017hybrid} reported the appearance of a FB in a general bipartite opto-mechanical lattice, and in the same year, Wan et al. \cite{wan2017controllable} also observed photon and phonon localization in an opto-mechanical Lieb lattice. 

Another platform where FBs have been proposed is cavity QED systems at microwave and optical frequencies~\cite{lehur2016many}. Many novel quantum correlations of the FB states, which are induced by dissipations and interactions, have been reported~\cite{biondi2015incompressible,schmidt2016frustrated,yang2016circuit,casteels2016probing,owen2017dissipation,rota2017on}. Moreover, the effect of multi-photon interactions on FBs can be explored in optical waveguide arrays~\cite{rojas-rojas2017quantum}. 
Finally, FBs in graphene systems have also been reported~\cite{marchenkoeaau2018extremely,lee2009flatband,pierucci2015evidence}.

With the development of fabrication technology, artificial FB systems will become a more interesting topic of study. The previously mentioned approaches will advance with improved precision and quality, as will the application of FB physics into new systems including micro- and nano-scale devices.

\section{Summary}  
\label{section2.6} 

In this chapter, we have introduced the concept of the flatband and how it forms. Starting from the Bloch theorem, we introduced the tight-binding model for discrete translational invariant networks. The origin of FBs has been explained in detail, and the concepts of CLSs and macroscopic degeneracy were introduced. We reviewed existing FB construction methods, applications, and experimental realizations in various systems. 

As discussed, existing methods of constructing FBs are limited to intuitive geometric or symmetry-based procedures. However, the fine-tuned nature of FB networks allows for the freedom to adjust system parameters, such as hopping strength, onsite energy, and wave amplitude, to achieve FBs. Consequently, there are a vast number of uncovered FB lattices, and moreover, the properties of CLSs and their decisive role in the behaviour of FB lattices have yet to be systematically studied. Therefore, it is desirable to establish the systematic construction and complete classification of FB Hamiltonians. In the next chapter, we will present our systematic approach to generating FB lattices based on CLS properties.


\chapter{Classification and construction of flatbands by compact localized states }
\label{chapter3}

\ifpdf
    \graphicspath{{Chapter3/Figs/Raster/}{Chapter3/Figs/PDF/}{Chapter3/Figs/}}
\else
    \graphicspath{{Chapter3/Figs/Vector/}{Chapter3/Figs/}}
\fi

Compact localized states (CLSs) play an important role in the physical characteristics of flatband (FB) lattices. Despite this, little has been known about CLS properties, which we focus on at the beginning of this chapter. First, we introduce the classification of FB lattices by CLSs and discuss their various properties. Then we develop a matrix representation---the main FB generator tool. At the end, we introduce destructive interference conditions and CLS existence conditions, and the core ideas of our FB generator. 


\section{Classification of flatband networks by compact localized states (CLSs)}  
\label{section3.1}


In this section, we discuss how to classify FB lattices by the size and shape of their CLSs. First, we introduce shape vector $\mathbf{U}$ that specifies these characteristics. 

\begin{definition}
    Shape vector $\mathbf{U}$ of a CLS is a vector whose components are integers that give the span of the CLS along each lattice dimension. The CLS size $U$ is the total number of unit cells occupied by a CLS. 
\end{definition} 
For example, in 1D, $\mathbf{U}$ is an integer, i.e. $\mathbf{U}=U$, which is the number of unit cells occupied by the CLS (see Fig. \ref{fig:1d-uclass}). In 2D, we can write $\mathbf{U}=(U_1,U_2)$, which tells that the CLS is occupying the unit cells in a $U_1 \times U_2$ rectangular area, as seen in Fig. \ref{fig:u-classification-2d}.   

\begin{figure}[htb!]
    \centering
    \includegraphics[width=0.7\linewidth]{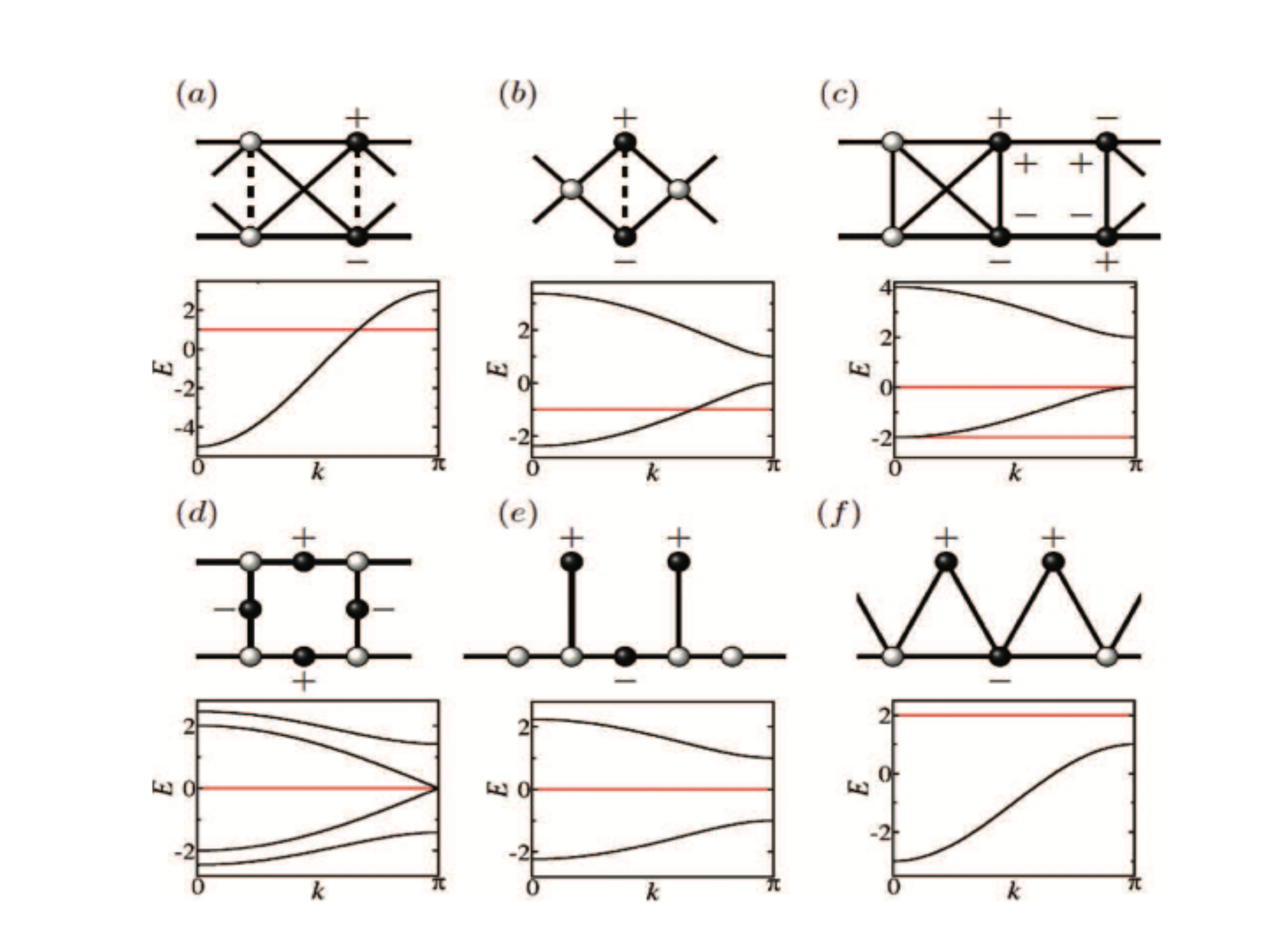}
    \caption[1D examples of $\mathbf{U}$ classification]{Various examples of 1D FB lattices and their $\mathbf{U}$ classification. (a) Cross-stitch, $U=1$, (b) diamond chain, $U=1$, (c) 1D pyrochlore, $U=1$, (d) 1D Lieb, $U=2$, (e) stub, $U=2$, and (f) sawtooth chain, $U=2$. This figure is courtesy of S. Flach et al. \cite{flach2014detangling}. }
    \label{fig:1d-uclass}
\end{figure} 

Given a CLS of a FB, the linear combination of its lattice translations is also an eigenstate of the same FB. It is therefore important to identify the CLS that cannot be decomposed into smaller CLSs. 

\begin{definition}
    An {\bf irreducible CLS} is a flatband eigenstate that occupies the smallest possible number of adjacent unit cells. 
\end{definition} 
From here on, if we do not specify, the notation CLS implies the irreducible one.

An important observation is that all known FB models have a unique irreducible CLS. This leads to the following conjecture. 

\begin{conjecture}
    \label{conj:lin-comb-diff-cls}
    The irreducible CLS of a given flatband is unique. 
\end{conjecture} 

Consequently, all eigenstates of FB energy can be formed by the linear combination of lattice translations of a unique irreducible CLS. Additionally, since the irreducible CLS of a FB is unique, it can be used to identify different FBs. We can therefore classify FB lattices by the shape vector $\mathbf{U}$ of their irreducible CLSs, which we call \emph{$\mathbf{U}$ classification}. 

\begin{figure}[htb!]
    \centering
    \includegraphics[width=0.9\linewidth]{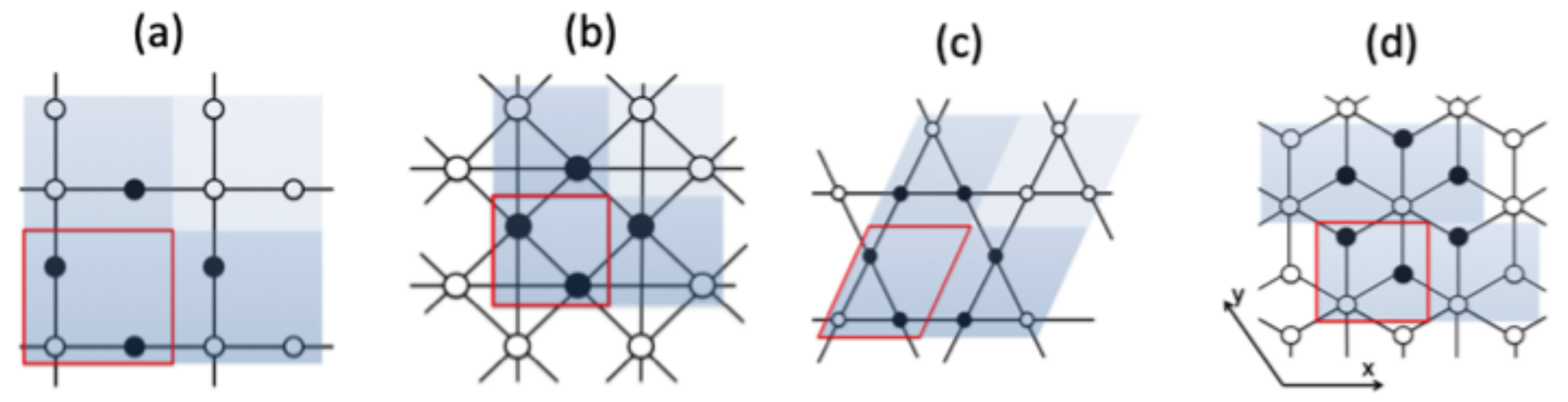}
    \caption[2D examples of $\mathbf{U}$ classification]{Classifications of some known 2D FB lattices. The shaded areas represent the shapes of the CLSs, and the red boxes show the unit cell. In each case, the CLS stretches two unit cells along both $x,y$ directions, and so $\mathbf{U}=(2,2)$. (a) Lieb, (b) checkerboard, (c) kagome, and (d) dice. }
    \label{fig:u-classification-2d}
\end{figure}

\begin{definition}
    A {\bf class $\mathbf{U}$ CLS} is an irreducible CLS with shape vector $\mathbf{U}$. If the irreducible CLS of a flatband lattice is a class $\mathbf{U}$ CLS, then we call this flatband a \emph{class $\mathbf{U}$ flatband}, and call the lattice a \emph{class $\mathbf{U}$ flatband lattice} for lattices with a single flatband.
\end{definition}

In 1D, there is only one shape, and thus FB lattices are classified by their CLS size $U$. For example, a 1D sawtooth chain is a class $U=2$ FB lattice, because its CLS occupies two unit cells. Figure \ref{fig:1d-uclass} illustrates some examples of different class $U$ FB lattices in 1D. Examples of the $\mathbf{U}$ classification of some of the known 2D FB lattices are shown in Fig. \ref{fig:u-classification-2d}. 


Flatband lattices of different $\mathbf{U}$ classes have different properties; this makes $\mathbf{U}$ classification useful to identify these properties. For example, as we discussed in Section \ref{section2.4.2}, the $U=1$ FB eigenstate can be detangled from the dispersive part of the lattice~\cite{flach2014detangling}, but the $U>1$ FB eigenstate cannot be detangled~\cite{flach2014detangling}. 



\section{Properties of CLSs}
\label{section3.2}


In this section, we discuss CLS properties mainly in the context of 1D Hermitian systems with $\nu$ sites per unit cell and nearest neighbor hoppings. Some of the results here hold in higher dimensions, as we will discuss. 

\nomenclature[z-$d$]{$d$}{Lattice dimension}  

\subsection{Relation between CLSs and Bloch wave functions} 


Consider a 1D flatband lattice with $\nu$ sites per unit cell. Suppose the $\nu$ component vector $\vec{\psi}_j$ is the wave function of the $j$th unit cell, and the components give wave functions at the lattice sites in the unit cell. Then, the CLS is given by $\vec{\Psi}_l=\left(\dots0,0,\vec{\psi}_1,\vec{\psi}_2,\dots,\vec{\psi}_U,0,0,\dots\right)$, and $\vec{\psi}_1$ is located in the $l$th unit cell. Since all lattice translations of the CLS share the same eigenenergy, we can construct a Bloch eigenstate (up to normalization, see Section \ref{section2.2}) with
\begin{equation}
    \vec{u}(k)\sim\sum_{l=-\infty}^{\infty}\mathrm{e}^{ikl}\vec{\Psi}_l \sim \sum_{l=1}^U\vec{\psi}_l e^{i(U-1)k}\;. 
    \label{eq:bloch-eig-st-construction}
\end{equation}
Note that, in above equation, the infinite sum in $l$  is truncated due to the compactness of the CLS wave function, $\vec{\psi}_{1\leq l\leq U}=0$. The inverse of \eqref{eq:bloch-eig-st-construction} is also true if the Bloch function for a certain band can be expressed as
\begin{equation}
    \vec{u}(k)=N(k)\sum_{l=1}^U\vec{\psi}_l e^{i(U-1)k},
    \label{eq:1d-cls-and-bloch-relation}
\end{equation}
where $N(k)$ is a common prefactor.

A similar but more involved decomposition can be written for CLSs in higher dimensions, where size alone does not fix the shape of the CLS.

\subsection{Irreducibility condition of CLSs} 

As previously mentioned, a CLS of a given FB could be the linear combination of the lattice translations of the irreducible CLS. So how do we know if a given CLS is irreducible? The answer is to check for linear dependence of the CLS components. 

\begin{figure}[htb!]
    \centering
    \includegraphics[width=0.8\linewidth]{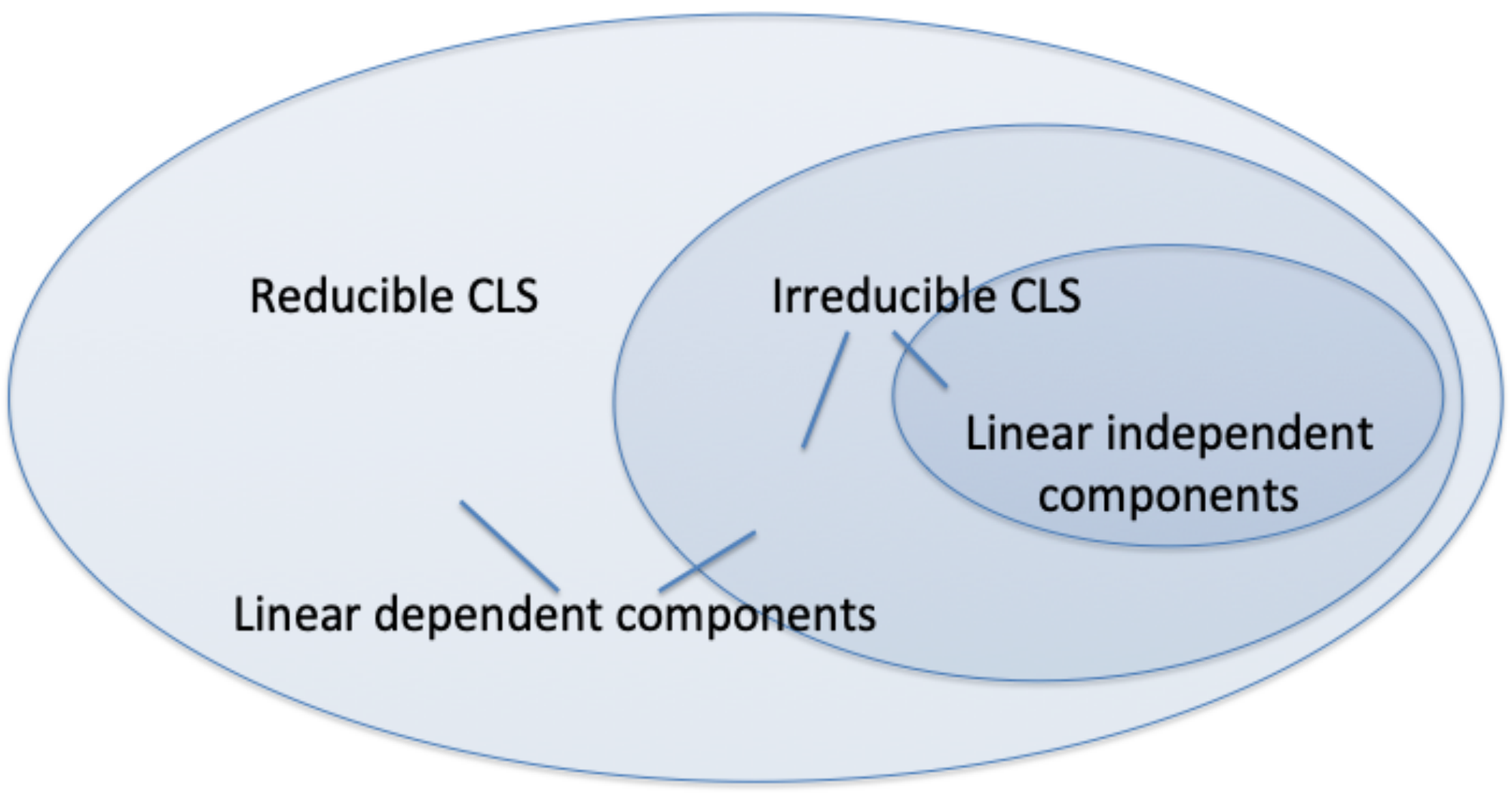}
    \caption[Relationship between reducible and irreducible CLSs, and linear dependence.]{Relationship between reducible and irreducible CLSs, with linear dependence of CLS components.}
    \label{fig:reducible-irreducible-cls}
\end{figure}

\begin{theorem}
    \label{theo:irreducibility-cond}
    Given CLS $\Psi_{U}=(\vec{\psi}_1,\vec{\psi}_2,\cdots ,\vec{\psi}_{U})$ occupying $U$ unit cells, if CLS components $\vec{\psi}_{i=1,\dots,U}$ are linearly independent, then $\Psi_{U}$ is irreducible.
\end{theorem} 

Theorem \ref{theo:irreducibility-cond} can also be stated as follows: \emph{If the matrix $(\vec{\psi}_1 \vec{\psi}_2 \cdots  \vec{\psi}_{U})$ is full rank, then CLS $\Psi_{U}=(\vec{\psi}_1,\vec{\psi}_2,\cdots ,\vec{\psi}_{U})$ is irreducible}. 


Below we prove the following equivalent statement of Theorem \ref{theo:irreducibility-cond}: \emph{If a CLS is reducible, then its components are linearly dependent}.
According to Conjecture \ref{conj:lin-comb-diff-cls}, we assume that if a CLS is reducible then it should be the linear combination of the unique irreducible CLS (i.e., it cannot be written as the linear combination of different class irreducible CLSs). Here, we prove a 1D case and sketch how it can be extended to higher dimensions.

\begin{proof}
Consider a 1D CLS $\Psi_{U}$ which is composed of the lattice translations of irreducible CLS $\Phi_{U-1}=(\vec{\phi}_1,\vec{\phi}_2,\cdots ,\vec{\phi}_{U-1})$, whose components are linearly independent. Then, 
\begin{equation}
    \Psi_{U} = \begin{pmatrix} \vec{\psi}_1 \\ \vec{\psi}_2 \\ \cdots \\ \vec{\psi}_{U} \end{pmatrix} = \alpha \begin{pmatrix} \vec{\phi}_1 \\ \vec{\phi}_2 \\ \cdots \\ \vec{\phi}_{U-1} \\ 0 \end{pmatrix} + \beta \begin{pmatrix} 0 \\ \vec{\phi}_1 \\ \vec{\phi}_2 \\ \cdots \\ \vec{\phi}_{U-1}  \end{pmatrix} , \quad \alpha,\beta \ne 0.
    \label{eq:1d-cls-lin-comb}
\end{equation} 
If $\Psi_U$ can be decomposed into an even smaller-sized irreducible CLS, then there will be more terms in Eq. \eqref{eq:1d-cls-lin-comb}; the following procedure can apply to this case as well. 

From Eq. \eqref{eq:1d-cls-lin-comb}, we can see that 
\begin{equation}
    \begin{cases}
      \vec{\psi}_1 = \alpha \vec{\phi}_1 \\
      \vec{\psi}_{U} = \beta \vec{\phi}_{U-1} \\
      \vec{\psi}_l = \alpha \vec{\phi}_l + \beta \vec{\phi}_{l-1}, \quad 2 \le l \le U-2.
    \end{cases}
\end{equation} 
If Theorem \ref{theo:irreducibility-cond} is true, then the components of $\Psi_U$ must be linearly dependent. More precisely, there should exist a non-zero set $\{c_j \}$ which satisfies
\begin{equation}
    \sum_{j=1}^U c_j \vec{\psi}_j = \sum_{i=1}^{U-1} (c_i \alpha + c_{i+1} \beta)\vec{\phi}_i = 0.
    \label{eq:cls-comp-lin-comb}
\end{equation} 
In the above equation, we paired the non-zero components in Eq. \eqref{eq:1d-cls-lin-comb}. The index $i$ above runs from 1 to $U-1$, because there are $U-1$ pairs in the sum. 

Since the components of $\Phi_{U-1}$ are linearly independent, the only solution of Eq. \eqref{eq:cls-comp-lin-comb} is
\begin{equation}
    c_i \alpha + c_{i+1} \beta = 0,\quad \text{for all}\ 1\le i \le U-1,
\end{equation} 
in which the $U$ coefficients and $U-1$ equations give an under-determined problem, so we can always find a set $\{c_j \}$ that satisfies the linear dependence condition (Eq. \eqref{eq:cls-comp-lin-comb}). This shows that \emph{if a CLS is the linear combination of smaller CLSs, then its components are linearly dependent}. This proves Theorem \ref{theo:irreducibility-cond}. 
\end{proof}

This proof can be extended to higher dimensions once the more complicated shapes of the CLSs are taken into account.


Note that the inverse statement of Theorem \ref{theo:irreducibility-cond}, which is: "\emph{If CLS components are linearly dependent, then the CLS is reducible}", does not hold (see Fig. \ref{fig:reducible-irreducible-cls}). For example, the irreducible CLS of a Lieb lattice has linearly dependent components. 

This result, together with the observation that CLS size completely fixes the shape in 1D , leads to the following conjecture.

\begin{conjecture}
    In 1D, linear dependence of CLS components implies the reducibility of the CLS.
    \label{conj:lin-dep-redblty-1d}
\end{conjecture}

This conjecture can be rigorously proved for the U = 2 and U = 3 cases, and it holds for all the other known $d=1$  examples. We present here the proof for the $U=2$ case, and put the proof for the $U=3$ case in Appendix \ref{app:reducibility-1d}, which uses the block matrix representation that we introduce later in Section \ref{section3.3}. 

\begin{proof}
Consider the $U=2$ class CLS $\vec{\Psi}=\left(\dots,0,\vec{\psi}_{1},\vec{\psi}_{2},0,\dots \right)$,
with nearest neighbor unit cell hoppings. Suppose two components $\vec{\psi_{1}},\vec{\psi_{2}}$
are linearly dependent, such that $\vec{\psi}_{1}=\alpha\vec{\psi}_{2}$.
From the eigenvalue problem in Eq. \eqref{eq:TB-eig-prob-1}, we have  
\begin{equation}
\begin{split}
    \epsilon_j \phi_{1,j} + \sum_{j^\prime} t_{j,j^\prime} \phi_{2,j^\prime} = E \phi_{1,j} , \quad j,j^\prime=1,\dots,\nu \\,
    \epsilon_j \phi_{2,j} + \sum_{j^\prime} t_{j,j^\prime} \phi_{1,j^\prime} = E \phi_{2,j} , \quad j,j^\prime=1,\dots,\nu.
\end{split}
\label{eq:u2-lin-dep-eig-pb}
\end{equation} 
From the linear combination $\vec{\psi}_{1}=\alpha\vec{\psi}_{2}$, we have $\phi_{1,j}=\alpha \phi_{2,j}$, and plugging this
into Eq. \eqref{eq:u2-lin-dep-eig-pb} we get 
\begin{equation}
\begin{aligned}
    \epsilon_j \phi_{1,j} + \alpha \sum_{j^\prime} t_{j,j^\prime} \phi_{1,j^\prime} = E \phi_{1,j} , \quad j,j^\prime=1,\dots,\nu \\,
    \alpha \epsilon_j \phi_{1,j} + \sum_{j^\prime} t_{j,j^\prime} \phi_{1,j^\prime} = \alpha E \phi_{1,j} , \quad j,j^\prime=1,\dots,\nu.
\end{aligned}
\end{equation} 
The above equations only involve $\phi_{i,j}$. For a non-zero $\alpha$, the above equations hold only when 
\begin{equation}
    \sum_{j^\prime} t_{j,j^\prime} \phi_{1,j^\prime} = \sum_{j} t_{j,j^\prime} \phi_{1,j}=0 \; ,
\end{equation} 
which means that there are no hoppings to neighboring unit cells. As a result, we have a $U=1$ CLS. 
\end{proof}


Assuming Conjecture \ref{conj:lin-dep-redblty-1d} is true, it leads to the following conjecture. 

\begin{conjecture}
    In 1D, the CLS class is $U \le \nu$. 
\end{conjecture}

\begin{proof}
Suppose there is a CLS of size $U>\nu$; then there are $U>\nu$ CLS components that are $\nu$-dimensional vectors. In other words, there are $U>\nu$ vectors in $\nu$-dimensional vector space, which always leads to the linear dependence of these vectors. As the result, according to Conjecture \ref{conj:lin-dep-redblty-1d}, the CLS is reducible as long as $U>\nu$. Therefore, the maximum size of the irreducible CLS is $U=\nu$, or in other words, the CLS class in 1D is $U \le \nu$.
\end{proof}

As discussed in Section \ref{section2.4.2}, unitary transformation does not change the band structure while modifying CLS components. Moreover, depending on the choice of unit cell, CLS size can differ. In some cases, one can reduce CLS size by unitary transformation and redefining the unit cell. 

\begin{theorem}
    In 1D, a CLS occupying $U$ unit cells with $\vec{\psi}_1 \perp \vec{\psi}_U$ can be reduced to the $U-1$ class. 
    \label{theo:u-reduce-u-1}
\end{theorem} 

\begin{proof}
Suppose we have CLS $\vpsi_{cls} = (\vpsi_1,\vpsi_2,\dots,\vpsi_U)^T$ of size $U$ with $\vpsi_1 \perp \vpsi_U$. Then, we can apply unitary transformation $R$ to the CLS such that $\vec{\phi}_{i=1,\dots,U}=R \vpsi_{i=1,\dots,U}$ and
\begin{equation}
  \vec{\phi}_1= \begin{bmatrix}
           1 \\
           0\\
           \vdots \\
           0
         \end{bmatrix}, \quad 
  \vec{\phi}_2= \begin{bmatrix}
           \phi_{2,1} \\
           \phi_{2,2}\\
           \vdots \\
           \phi_{2,\nu}
         \end{bmatrix}, \dots,\quad
  \vec{\phi}_U=\begin{bmatrix}
           0 \\
           \phi_{U,2}\\
           \vdots \\
           \phi_{U,\nu}
         \end{bmatrix} , \;
\end{equation}
where $\nu$ is the number of sites per unit cell. Due to the unitary transformation $R$, the eigenvalue problem (Eq. \eqref{eq:TB-eig-prob-1}) does not change. Next, we redefine the unit cell as illustrated in Fig. \ref{fig:u-reduce-u-1} as
\begin{gather*}
    \vec{\varphi}_1= \begin{bmatrix}
        1 \\
        \phi_{2,2} \\
        \vdots \\
        \phi_{2,\nu}
        \end{bmatrix}, \ 
    \vec{\varphi}_2= \begin{bmatrix}
        \phi_{2,1} \\
        \phi_{3,2}\\
        \vdots \\
        \phi_{3,\nu}
        \end{bmatrix}, \dots,\ 
    \vec{\varphi}_{U-1}=\begin{bmatrix}
        \phi_{U-1,1} \\
        \phi_{U,2}\\
        \vdots \\
        \phi_{U,\nu}
    \end{bmatrix},
\end{gather*} 
and $\vec{\varphi}_U = 0$. Therefore, after the unitary transformation $R$ and redefinition of the unit cell, the class of the CLS reduces to $U-1$. Illustration of this procedure are shown in Fig.~\ref{fig:u-reduce-u-1}. 
\end{proof}
\begin{figure}[htb!]
    \centering
    \includegraphics[width=0.6\linewidth]{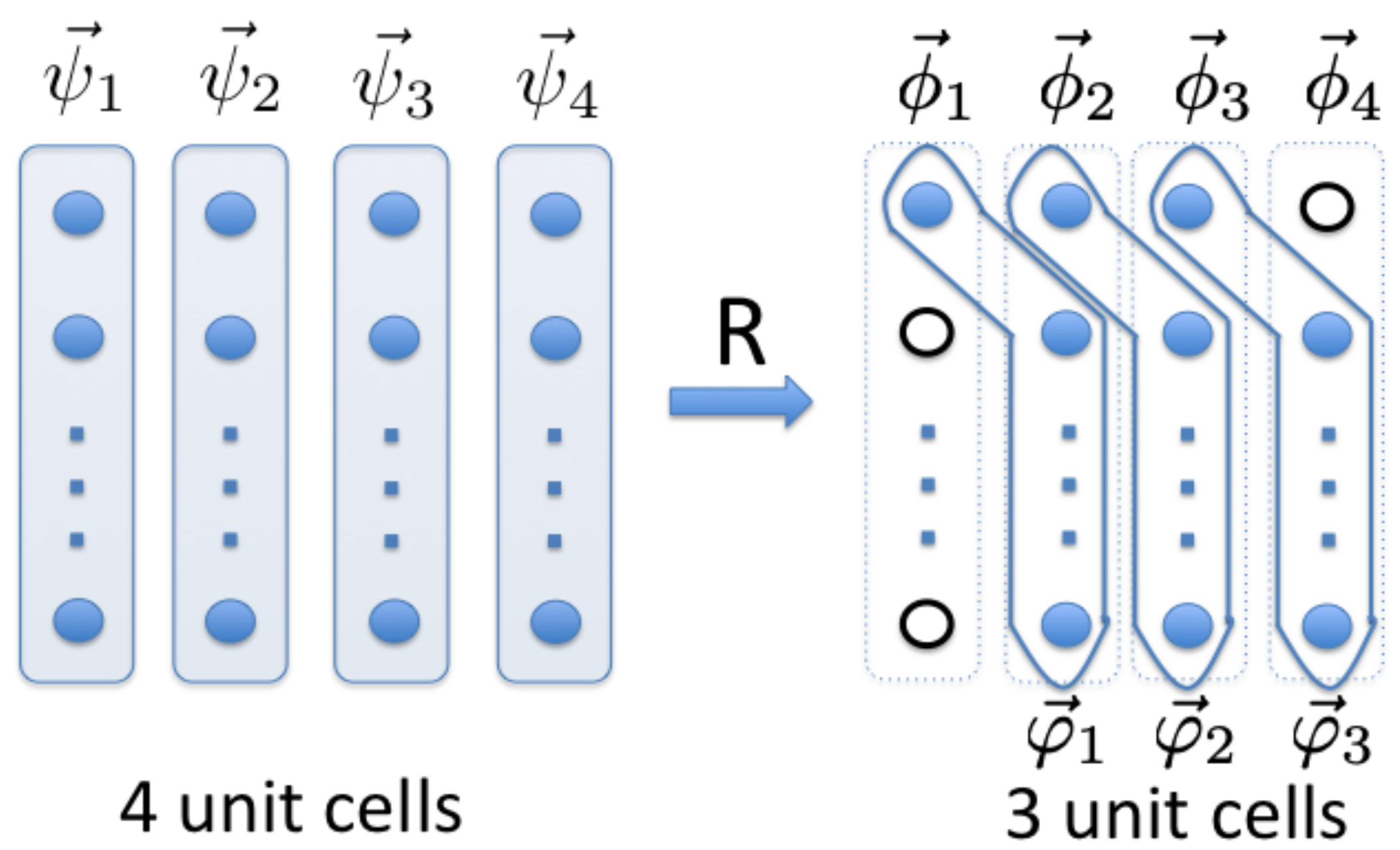}
    \caption[Reduction of $U$ class CLS to $U-1$]{(Color online) Schematics showing how a CLS of class $U=4$ reduces to $U=3$, when $\vpsi_1 \perp \vpsi_4$. Each rectangle stands for one unit cell. The filled circles stand for non-zero wave function components, and the open circles for the zero wave function components.}
    \label{fig:u-reduce-u-1}
\end{figure}

\subsection{Completeness of CLSs}

The Bloch eigenstate of a FB can always form a complete basis for the Hilbert space of the FB. Can CLSs also form such a complete basis? In this section, we discuss this question.

Given a finite-size FB lattice with $N$ sites, lattice translations of its CLS form $(N-1)$ different copies. If these $N$ copies are linearly independent, they form a complete basis for the Hilbert space of the FB. Then, what is the condition for the linear independence of the set of $N$ CLSs?

Suppose the CLS of a given FB lattice occupies $U$ unit cells, and we write the CLS whose first component is located in the $l$th unit cell as $\vec{\Psi}_l=\left(\dots,0,\vec{\psi}_{l},\vec{\psi}_{l+1},\dots,\vec{\psi}_{l+U-1},0,\dots\right)$. The CLS component $\vec{\psi}_{j=l,l+1,\dots,l+U-1}$ is a $\nu$-dimensional vector, with each component giving the wave function of each site in the $j$th unit cell. In this case, all lattice translations form a set $\{ \vec{\Psi}_{l=1,\dots,N} \}$, which leads to the following theorem.

\begin{theorem}
    \label{theo:lin-indep-set-cls}
   If the set of CLS components $\{\vec{\psi}_{j=l,\dots,l+U-1}\}$ is linearly independent, then the set of CLS translations $\{ \vec{\Psi}_{l=1,\dots,N} \}$ is also linearly independent, and forms a complete basis to span the Hilbert space of the flatband.
\end{theorem} 
For convenience, in this thesis, if the set of lattice translations of a CLS forms a complete basis, then we refer to the CLS as \emph{complete}.

Theorem \ref{theo:lin-indep-set-cls}, which applies in any dimension, can be equivalently stated as: \emph{If the set of CLS translations $\{ \vec{\Psi}_{l=1,\dots,N} \}$ is linearly dependent, then the set of CLS components $\{\vec{\psi}_{j=l,\dots,l+U-1}\}$ is also linearly dependent}. Here, we prove this equivalent statement for a 1D case, and sketch an extension for higher dimensions.

\begin{proof}
Consider a linearly dependent set of CLS translations $\{ \vec{\Psi}_{l=1,\dots,N} \}$ of a 1D FB lattice, such that
\begin{equation}
    \label{eq:lin-dep-cls}
    \sum_{l=1}^N \alpha_l \vec{\Psi}_l = 0,\ \ l\in\mathbb{Z},
\end{equation}
where $\vec{\Psi}_l=\left(\dots,0,\vec{\psi}_{l},\vec{\psi}_{l+1},\dots,\vec{\psi}_{l+U-1},0,\dots\right)$, and $\vec{\psi}_{j=l,l+1,\dots,l+U-1}$ are $\nu$-component vectors whose components are the wave functions on all lattice sites in $j=l,l+1,\dots,l+U-1$th unit cells. A necessary condition for Eq. \eqref{eq:lin-dep-cls} is 
\begin{equation}
    \sum_{j=l}^{l+U-1}\alpha_{j+l}\vec{\psi}_j = 0\;,
\end{equation}
and thus set $\{\vec{\psi}_{j=l,\dots,l+U-1}\}$ has to be linearly dependent. Therefore, \emph{if CLS component set $\{\vec{\psi}_{j=l,\dots,l+U-1}\}$ is linearly independent, then CLS translation set $\{ \vec{\Psi}_{l=1,\dots,N} \}$ is linearly independent as well}. This proves Theorem \ref{theo:lin-indep-set-cls} in 1D.
\end{proof}
The proof for higher dimensions can be performed by a similar procedure, but CLS shape has to be first taken into account. 

 
When a CLS is not complete, i.e. the number of linearly independent CLS translations is less than the dimension of the Hilbert space of its FB, there must be missing states. These can be explained by the band touching or crossing properties of the FB, which we discuss in the following section.  

\subsection{Relation between CLS completeness and band touching}
\label{section3.2.3} 

In this section, we discuss how FBs interact with dispersive bands, which in turn is related to CLS properties. When FBs are gapped from dispersive bands, their CLSs are complete \cite{rhim2018classification}; however, when FBs touch or cross dispersive bands, their CLSs may not be complete.


Consider the simplest CLS class $U=1$. As we discussed in Section \ref{section2.4.3}, we know that such a CLS and its dispersive states can be decoupled \cite{flach2014detangling}. Then, band touching or crossing can be removed by shifting the FB without disturbing its flatness, and thus the CLSs are complete. This type of band touching or crossing is called \emph{removable band touching}. Therefore, we may write the following theorem for $U=1$ CLSs.

\begin{theorem}
    \label{theo:u1-non-singular}
    The $U=1$ class CLSs in any dimension are always complete as the band touchings or crossings between $U=1$ flatbands and dispersive bands are removable.
\end{theorem} 

In the $U=2$ case, our analytic results for 1D FB lattices with two bands~\cite{maimaiti2017compact} show that such FBs are always gapped from dispersive bands (see Chapter \ref{chapter4}). Likewise, in our later study on 1D FB generators~\cite{maimaiti2019universal}, all $U>1$ examples show gapped FBs (see Chapter \ref{chapter4}). This leads to the following conjecture.

\begin{conjecture}
    \label{conj:1d-u-gapped-fb}
    In 1D, $U>1$ flatbands are gapped from dispersive bands, and their CLSs are complete.
\end{conjecture} 
Combining Theorem \ref{theo:u1-non-singular} and Conjecture \ref{conj:1d-u-gapped-fb}, we conclude the following: \emph{Band touchings or crossings of 1D FBs are always removable or gapped, and thus their CLSs are always complete}.

For higher dimensions $d>1$, there are some types of band touchings that cannot be removed without destroying the FB. We call such band touchings as \emph{irremovable}, and these can be identified through the linear dependence of the CLS components. The CLS of a FB with irremovable band touchings is not complete~\cite{rhim2018classification}. Then, Theorem \ref{theo:lin-indep-set-cls} leads to the following corollary.

\begin{corollary}
    \label{lem:lin-indep-non-singular}
    In $d>1$, if the components of a CLS are linearly independent, then the CLS is complete and the flatband is either gapped or has removable band touching. 
\end{corollary} 

To the best of our knowledge, if the components of a CLS are linearly dependent, then the set of the CLS translations is linearly dependent (there is no counter example for this), which leads to the following conjuncture.
\begin{conjecture}
A CLS with linearly dependent components is not complete, and so the corresponding flatband has irremovable band touching.
\label{conj:lin-dep-band-touch}
\end{conjecture}
For example, Lieb and kagome lattices have CLSs with linearly dependent components and their CLSs are not complete, while flatbands in these lattices have irremovable band touchings.


\section{CLS existence conditions and the flatband generator}
\label{section3.3}

In FB lattices, there should be at least two sites per unit cell in order to achieve destructive interference, because a FB in a single-band system only exists when the lattice sites are isolated. Therefore, we will only deal with multiple sites per unit cell, where it is convenient to represent the single-particle quantum states of the system in terms of single unit cell state vectors (Eq. \eqref{eq:single-cell-state-vec}) and single unit cell wave functions (Eq. \eqref{eq:single-uc-wf}). This single unit cell representation leads to the decomposition of the tight-binding Hamiltonian into a block matrix form, that we introduce in this section. The FB generator, the main concern of this thesis, is also based on matrix formalism and will be covered below.

\subsection{Block matrix representation of the tight-binding model}

For simplicity, we first consider a 1D translational invariant tight-binding network with $\nu$ sites per unit cell. This is followed by extension to higher dimensions, for which one has to consider hoppings in different directions.

We introduce $\nu \times \nu$ matrix $H_0$ to describe onsite energies and hoppings inside a unit cell, as well as $\nu \times \nu$ matrix $H_{m}$ to describe hoppings between $m$th neighboring unit cells. 
Then, we can write the eigenvalue problem \eqref{eq:TB-eig-prob-1} of the lattice in matrix form as
\begin{equation}
	\label{eq:eig-prob-mat-rep-1d}
     H_0 \vec{\psi}_{n} + \sum_{m=-\infty}^\infty H_{m} \vec{\psi}_{n+m}=E\vec{\psi}_n,\ \ m\ne 0 \; ,
\end{equation}
where the single unit cell wave function $\vec{\psi}_n$ of the $n$th unit cell is given by Eq. \eqref{eq:single-uc-wf}. For Hermitian systems, $H_0$ is Hermitian and $H_{\chi,-m}=H_{\chi,m}^\dagger$. For non-Hermitian systems, there are no restrictions for either $H_0$ or $H_m$. In the case of finite-range hoppings, $-m_c \le m \le m_c$, where $m_c$ is the maximum hopping range. Note that, in matrix representation, \emph{neighbor refers only to neighboring unit cells}, which is different from conventional definitions; we use this notation throughout this thesis unless stated otherwise.

Band structure can be achieved using $k$-space (reciprocal space) representation. Putting the Bloch wave from Eq. \eqref{eq:phi_n-k-representation} into Eq. \eqref{eq:eig-prob-mat-rep-1d} we get
\begin{equation}
  \sum_{k} \left( H_0 + \sum_{m=1}^{m_c} H_m^\dagger e^{i m k} + \sum_{m=1}^{m_c} H_m e^{-i m k }\right) \vec{u}(k) e^{-i n k} = E \sum_{k} \vec{u}(k) e^{-i n k} \; ,
  \label{eq:mat-rep-FT-1d}
\end{equation} 
where we take the lattice constant to be 1. Canceling out $e^{-i n k}$ in Eq. \eqref{eq:mat-rep-FT-1d} we get the following, for $\forall k$, 
\begin{equation}
   \left( H_0 + \sum_{m=1}^{m_c} H_m^\dagger e^{i m k} + \sum_{m=1}^{m_c} H_m e^{-i m k }\right) \vec{u}(k)  = E \vec{u}(k)  \; . 
  \label{eq:tb-k-rep-1d}
\end{equation} 
Therefore, we can write the $k$-space Hamiltonian as 
\begin{equation}
    H(k) = H_0 + \sum_{m=1}^{m_c} H_m^\dagger e^{i m k} + \sum_{m=1}^{m_c} H_m e^{-i m k}.
    \label{eq:k-space-Ham-1d}
\end{equation} 
The eigenvalues of $H(k)$ give the band structure.

Most of the time we consider nearest neighbor hoppings $m_c=1$, in which case $H_1$ describes the hoppings between nearest neighboring unit cells. Then Eq. \eqref{eq:eig-prob-mat-rep-1d} becomes 
\begin{equation}
    \label{eq:1d-nn-eig-prob-mat-rep}
     H_0 \vec{\psi}_{n} + H_1^\dagger \vec{\psi}_{n-1} +  H_1 \vec{\psi}_{n+1}=E\vec{\psi}_n \;.
\end{equation} 
Therefore, the Hamiltonian matrix for 1D nearest neighbor unit cell hoppings is a tri-diagonal block matrix,
\begin{equation}
    H = \begin{pmatrix} \ddots & \ddots & 0 & \dots &    \\
        \ddots & H_0 & H_1   & 0 & \dots \\
        0 & H_1^\dagger & H_0 & H_1 & 0   \\
        \dots & 0 & H_1^\dagger & H_0 & \ddots  \\
          & \dots & 0 & \ddots & \ddots \end{pmatrix},
    \label{eq:tri-diag-block}    
\end{equation}
where $H_1$ and $H_0$ are $\nu \times \nu$ matrices. Then the $k$-space Hamiltonian (Eq.\eqref{eq:k-space-Ham-1d}) for nearest neighbor hoppings becomes 
\begin{equation}
    H(k) = H_0 + H_1^\dagger e^{i k} + H_1 e^{-i k}
    \label{eq:1d-k-Hamiltonian}.
\end{equation}

\begin{figure}[htb!]
    \centering
    \includegraphics[width=0.6\linewidth]{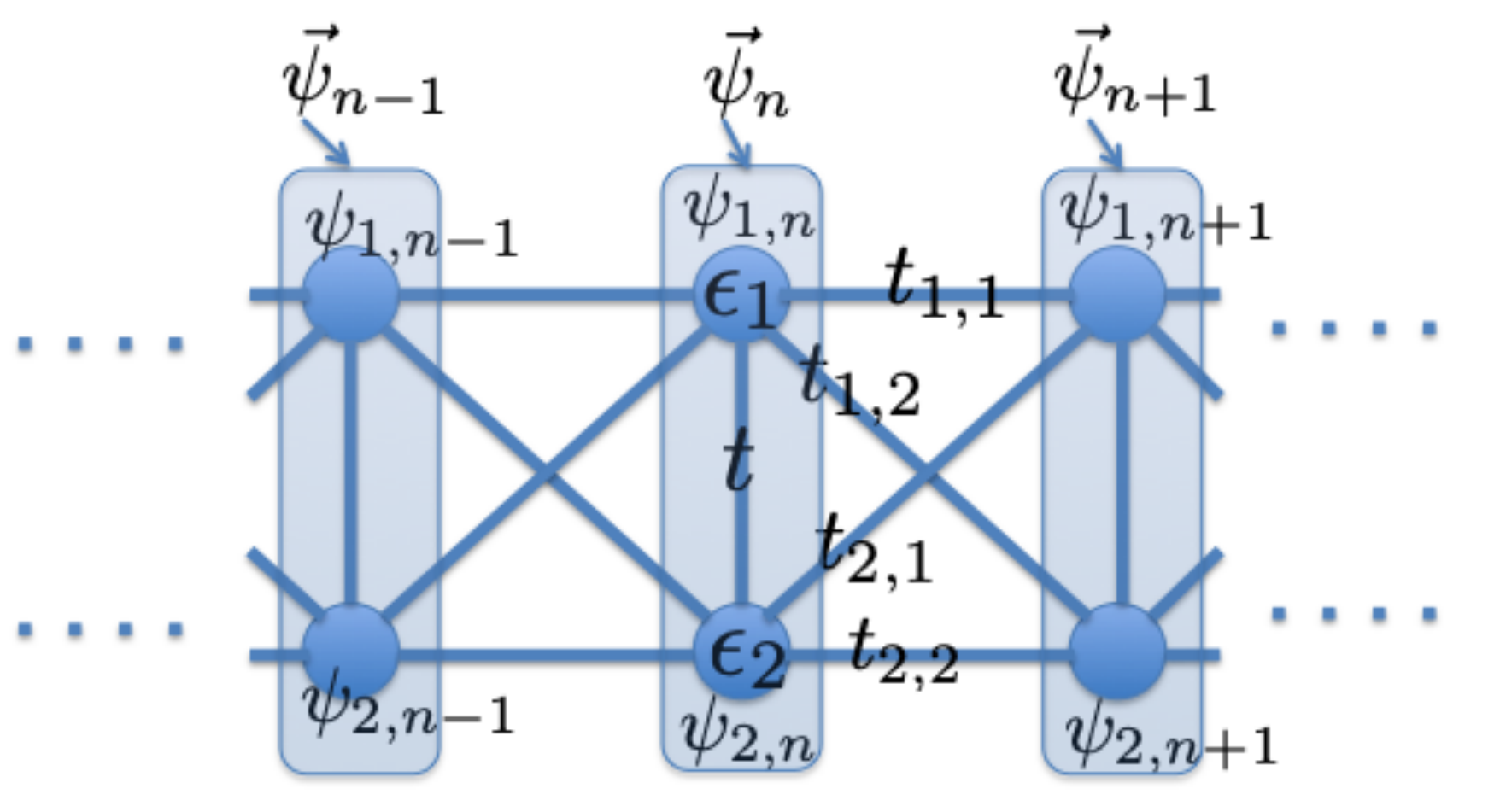}
    \caption[A cross-stitch lattice]{Schematic of a cross-stitch lattice. }
    \label{fig:cross-stich}
\end{figure}

Take a simple 1D lattice with two sites per unit cell and nearest neighbor hoppings, as shown in Fig. \ref{fig:cross-stich}. Then the wave function of the $n$th unit cell and hopping matrices are 
\begin{equation}
    \vec{\psi}_{n}=\begin{pmatrix} \psi_{1,n}\\ \psi_{2,n} \end{pmatrix},\quad H_0=\begin{pmatrix} \epsilon_1 & t\\ t & \epsilon_2 \end{pmatrix},\quad  H_1=\begin{pmatrix} t_{1,1} & t_{1,2}\\ t_{2,1} & t_{2,2} \end{pmatrix}.
\end{equation} 
The eigenvalue problem in Eq. \eqref{eq:1d-nn-eig-prob-mat-rep} now becomes 
\begin{equation}
\begin{aligned}
    \epsilon_1 \psi_{1,n} + t \psi_{2,n} + t_{1,1} \psi_{1,n+1} + t_{1,1} \psi_{1,n-1} + t_{1,2} \psi_{2,n+1} + t_{1,2} \psi_{2,n-1} &= E \psi_{1,n}\\ 
    \epsilon_2 \psi_{2,n} + t \psi_{1,n} + t_{2,1} \psi_{1,n+1} + t_{2,1} \psi_{1,n-1} + t_{2,2} \psi_{2,n+1} + t_{2,2} \psi_{2,n-1} &= E \psi_{1,n},
\end{aligned}    
\end{equation}
and the $k$-space Hamiltonian becomes
\begin{equation}
    H(k) = H_0 + H_1^\dagger e^{i k} + H_1 e^{-i k} = 
    \begin{pmatrix}
     \epsilon _1+2 \cos (k) t_{1,1} & t+e^{-i k} t_{1,2}+e^{i k} t_{2,1} \\
     t+e^{i k} t_{1,2}+e^{-i k} t_{2,1} & \epsilon _2+2 \cos (k) t_{2,2} \\
    \end{pmatrix}.
\end{equation} 
Then, the energy bands are given by the eigenvalues of $H(k)$. 

The block matrix representation naturally extends to higher dimensions, which contain multiple hopping directions, either along the primitive lattice translation vectors or any other lattice translation vectors (as linear combinations of primitive lattice translation vectors). If we consider nearest neighbor hoppings along each hopping direction, we can introduce a hopping matrix for each hopping direction. Suppose $H_{\chi}$ is the nearest neighbor hopping matrix for the $\chi$th direction, and then the tight-binding eigenvalue problem from Eq. \eqref{eq:TB-eig-prob-1} reads
\begin{equation}
     H_0 \vec{\psi}_{n} + \sum_{\chi} H_{\chi}^\dagger \vec{\psi}_{l_{\chi}^\prime} + \sum_{\chi} H_{\chi} \vec{\psi}_{l_{\chi}}=E\vec{\psi}_n \; ,
     \label{eq:nn-mat-rep}
\end{equation} 
where $l_{\chi}$ and $l_{\chi}^\prime$ are the indices of the nearest neighboring unit cells along the $\chi$th direction. Then, the $k$-space Hamiltonian can be written as 
\begin{equation}
    H(k) = H_0 + \sum_{\chi} H_{\chi}^\dagger e^{i \vec{R}_{\chi} \cdot \vec{k}} + \sum_{\chi} H_{\chi} e^{-i \vec{R}_{\chi} \cdot \vec{k}}\; ,
    \label{eq:nn-k-space-ham}
\end{equation} 
where $\vec{R}_{\chi}$ is the lattice vector of the nearest neighbor along the $\chi$th hopping direction.

\subsection{CLS existence conditions and flatband tester} 
\label{section3.4}

As learned in Section \ref{section2.3}, FBs are the consequence of destructive interference. Accordingly, achieving CLSs via destructive interference is the key to constructing FB lattices. In this section, we cover CLS existence conditions (or destructive interference conditions) and a FB tester before introducing our FB generator in the next subsection.

\begin{figure}[htb!]
    \centering
    \includegraphics[clip,width=0.5\columnwidth]{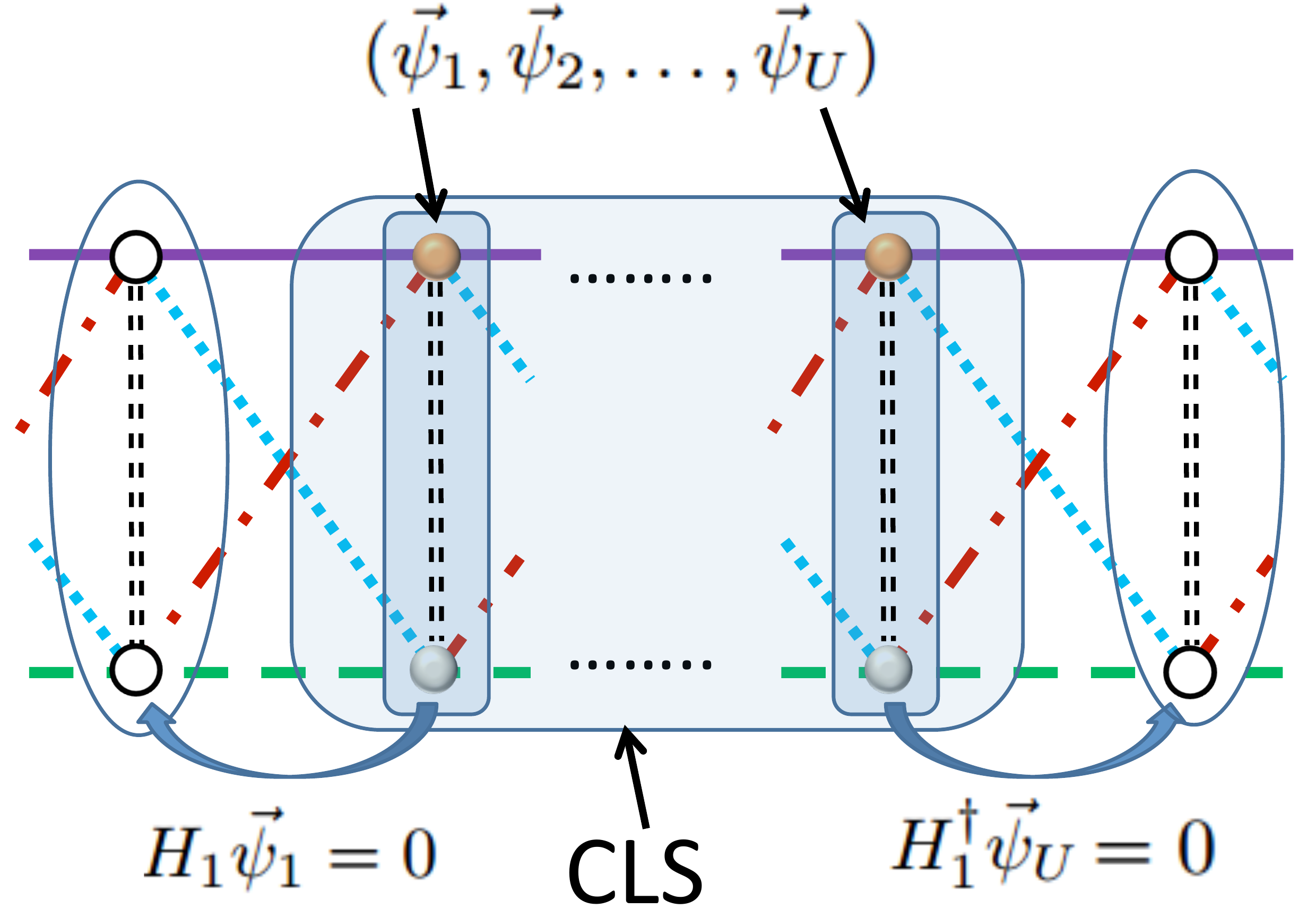}
    \protect\caption[Schematic of a CLS in a two-band lattice.]{(color online) Schematics of a compact localized state with nearest neighbor hopping ($m_c=1$) and $\nu=2$.}
    \label{fig2}
\end{figure}

For simplicity, we consider a 1D FB lattice with nearest neighbor hoppings. Suppose the lattice has a CLS occupying $U$ unit cells, as shown in Fig. \ref{fig2}. Here, the hopping from the first unit cell to the left and hoppings from the $U$th unit cell to the right must be zero, which gives the following \emph{destructive interference conditions} (see Fig. \ref{fig2})

\begin{equation}
    H_1 \vec{\psi}_1 = 0 , \quad H_1^\dagger \vec{\psi}_U = 0 \;.
	\label{eq:dest-int-cond-1d}
\end{equation} 
This tells us that, in 1D, the hopping matrix $H_1$ must have eigenvalues equaling zero. We can then write the following theorem.
\begin{theorem}
    In a 1D lattice with nearest neighbor hoppings, if the lattice has a CLS, then the nearest neighbor hopping matrix must be singular:
    \begin{equation}
        \det H_1 = 0.
    \end{equation}
  \label{theo:nes-cond-cls-1d}
\end{theorem} 

The destructive interference conditon in Eq. \eqref{eq:dest-int-cond-1d} and Theorem \ref{theo:nes-cond-cls-1d} imply that the eigenvalue problem (Eq. \eqref{eq:1d-nn-eig-prob-mat-rep}) of a class $U$ FB lattice must have a solution in the following form
\begin{equation}
	\vec{\Phi}_U = (\dots,0,\vec{\psi_1},\dots,\vec{\psi}_U,0,\dots) \; .
\end{equation}
Then, the eigenvalue problem can be written as
\begin{equation}
	H_U \vec{\Psi}_U = E_{FB} \vec{\Psi}_U \; ,
	\label{eq:1d-eig-prob-HU}
\end{equation}
where $E_{FB}$ is FB energy, and $H_U$ is a $U\times U$ tri-diagonal block matrix
\begin{equation}
    H_U = \begin{pmatrix} H_0 & H_1 & 0 & \dots &  0 & 0   \\
       H_1^\dagger & H_0 & H_1   & 0 & \dots  & 0\\
        0 & \ddots & \ddots & \ddots  & 0 & \vdots  \\
	\vdots & \dots &\ddots & \ddots & \ddots & 0  \\
        0 & \dots & 0 & H_1^\dagger & H_0 & H_1   \\
         0 & 0 & \dots & 0 & H_1^\dagger & H_0 \end{pmatrix}, 
    \label{eq:H-U-1d}    
\end{equation}
and 
\begin{equation}
	\vec{\psi}_U= (\vec{\psi_1},\dots,\vec{\psi}_U) \;.
	\label{eq:psi-U-1d}
\end{equation}

Given a FB Hamiltonian, we can test the class of the FB using the following procedure. For simplicity, we consider a 1D case. 

\begin{definition}
    \emph{Flatband tester}: Given a 1D FB Hamiltonian with nearest neighboring unit cell hoppings, we can find $\vec{\psi}_1$ and $\vec{\psi}_U$ from $H_1 \vec{\psi}_1 = H_1^\dagger \vec{\psi}_U=0$, and the eigenvalue problem \eqref{eq:1d-eig-prob-HU} has a solution in the form \eqref{eq:psi-U-1d} corresponding to eigenvalue $E_{EB}$. Starting from $U=1$, we solve the eigenvalue problem. If there is no solution for $U=1$, we successively increase to $U=2,3,\dots$ until we obtain a solution. In this way, we get the smallest $U$ that gives an irreducible CLS; this $U$ is the class of the FB lattice.
    \label{def:fb-tester}
\end{definition}
A similar procedure can be employed for higher dimensions, with more complexity. 


From this test procedure, we can infer the following necessary and sufficient conditions for the existence of a CLS in a given 1D lattice Hamiltonian, as stated in the theorem below.

\begin{theorem}
    Given a 1D lattice with intracell hoppings $H_0$ and nearest neighbor hoppings $H_1$, if the tri-diagonal block matrix (Eq. \eqref{eq:H-U-1d}) has an eigenvector (Eq. \eqref{eq:psi-U-1d}) that corresponds to eigenvalue $E_{FB}$ and satisfies the destructive interference conditions (Eq. \eqref{eq:dest-int-cond-1d}), then a CLS of size $U$ exists in this lattice.
\label{theo:sufficient-cond-cls-1d}
\end{theorem} 
Note that, in this theorem, $U$ is not necessarily the smallest one; if needed, the tester \ref{def:fb-tester} can be used to find the smallest $U$.

 Theorem \ref{theo:sufficient-cond-cls-1d} can be extended to higher dimensions using Eq. \eqref{eq:nn-mat-rep} and by taking into account the shapes of the CLSs, which complicates the destructive interference conditions. This is stated in the following theorem.



\begin{theorem}
    Given a $d>1$ dimensional lattice with intracell hoppings $H_0$ and nearest neighbor hoppings $H_{\chi}$ in the $\chi$th direction, the necessary and sufficient conditions for the existence of a CLS of size $U$ are (i) the eigenvalue problem 
    \begin{equation}
        H_0 \vec{\psi}_{l} + \sum_{\chi} H_{\chi}^\dagger \vec{\psi}_{l_{\chi}^\prime} + \sum_{\chi} H_{\chi} \vec{\psi}_{l_{\chi}} = E\vec{\psi}_l \quad l=1,\dots,U \; ,
        \label{eq:nn-suf-cond-eig-prob}
\end{equation}
having an eigenvector $\Psi_{U} = \left( \vec{\psi}_1, \vec{\psi}_2, \dots, \vec{\psi}_{U} \right)^T$ with eigenvalue $E_{FB}$, and (ii) an eigenvector satisfying the destructive interference conditions 
\begin{equation}
    \sum_{\chi,m} H_{\chi} \psi_{m} = 0, \quad \sum_{\chi,m^\prime} H_{\chi}^\dagger \psi_{m^\prime} = 0 
    \label{eq:suf-cond-destruct-cond}
\end{equation} 
at the boundaries, i.e. where $m,m^\prime$ runs over the boundary unit cells. 
\label{theo:sufficient-cond-cls}
\end{theorem} 

The destructive interference conditions (Eq. \eqref{eq:suf-cond-destruct-cond}) vary depending on the shape of the CLS, as we discuss later in Chapter \ref{chapter5}.

These necessary and sufficient conditions are formulated in real space. Other work studying the necessary and sufficient conditions for the existence of FBs was based on $k$-space representation \cite{toikka2018necessary}.

Now, by inverting the procedure in the FB tester \ref{def:fb-tester}, assuming the lattice has a CLS of class $U$ and asking what is the Hamiltonian that satisfies Eq. \eqref{eq:1d-eig-prob-HU}, we arrive at our core idea of the FB generator.

\subsection{Flatband Generator}
\label{sec:fb-tester-and-generator}

Consider a 1D translational invariant tight-binding lattice (or network) with $\nu$ sites per unit cell, with intracell hoppings $H_0$ and nearest neighbor hoppings $H_1$. We want to get a FB of class $U\le\nu$. We therefore assume the lattice has an irreducible CLS $\Psi_{U}=(\vec{\psi}_1,\vec{\psi}_2,\cdots ,\vec{\psi}_{U})$, and look for the matrices $H_0$, $H_1$, and FB energy $E_{FB}$ that satisfy the necessary and sufficient conditions for CLS existence (Eqs. \eqref{eq:dest-int-cond-1d} and \eqref{eq:1d-eig-prob-HU}):
\begin{equation}
    \begin{aligned}
        H_0 \vec{\psi}_1 + H_1 \vec{\psi}_2 &= E \vec{\psi}_1, \\
        H_0 \vec{\psi}_{l} + H_1^\dagger \vec{\psi}_{l-1} +  H_1 \vec{\psi}_{l+1} &= E\vec{\psi}_l \quad l=2,\dots,U-1 , \\
        H_0 \vec{\psi}_U + H_1^\dagger \vec{\psi}_{U-1} &= E \vec{\psi}_U \\
        H_1 \vec{\psi}_1 &= 0\\
        H_1^\dagger \vec{\psi}_U &=0 \;.
    \end{aligned}
    \label{eq:nn-eig-plus-cls-cond}
\end{equation}
This is a set of non-linear equations that are in general difficult but solvable, as we demonstrate in later chapters. We may now formally define our FB generator as follows.
\begin{definition}
    A \emph{flatband generator} is a scheme that generates a set of all possible hopping matrices $\{H_0,H_1\}$, flatband energy $E_{FB}$, and irreducible CLSs $\Psi_{U}$ that satisfies the eigenvalue problem and destructive interference conditions (Eq.\eqref{eq:nn-eig-plus-cls-cond}). 
    \label{def:FB-gen}
\end{definition}
This scheme systematically generates all possible FB Hamiltonians that possess CLSs of class $U$. 

Extension of the above FB generator to higher dimensions is straightforward, where one can generate FB Hamiltonians by solving Eqs. \eqref{eq:nn-suf-cond-eig-prob} and \eqref{eq:suf-cond-destruct-cond}. Depending on the dimension and shapes of the CLSs, the solutions vary (see later chapters). In general, the control parameters of the FB generator are as follows:
\begin{itemize}
	\item Lattice dimension $d$: In this thesis, $d=1,2$.
	\item Hopping range $m_c$: In this thesis, we focus on nearest neighboring unit cell hoppings, i.e. $m_c=1$.
	\item The number of sites per unit cell = the number of bands $\nu$.
	\item For $d=1$: CLS size = the number unit cells $U$ occupied by the irreducible CLS. 
	\item For $d\ge2$ and $U\ge2$: size $U$ and CLS shape.
\end{itemize} 
We aim to generate FB Hamiltonians for given lattice dimension $d$, number of bands $\nu$, CLS size $U$, and CLS shape (for $d\ge2$).


\section{Summary}  
\label{section3.5} 

In this chapter, we first introduced $\mathbf{U}$-classification of FB lattices according to their CLS class. Next, we discussed the properties of CLSs, including the conditions for irreducibility, band touching or crossing properties of FBs of different classes, and completeness of CLSs. We then introduced block matrix representation of the tight-binding Hamiltonian. At the end, we formulated CLS existence conditions, a FB tester, and our FB generator. In the following chapters, we apply the FB generator to different lattice dimensions and different systems, starting from a 1D system.



\chapter{Flatband generator in one dimension} 
\label{chapter4}

\ifpdf
    \graphicspath{{Chapter4/Figs/Raster/}{Chapter4/Figs/PDF/}{Chapter4/Figs/}}
\else
    \graphicspath{{Chapter4/Figs/Vector/}{Chapter4/Figs/}}
\fi 

In the previous chapters, we introduced $U$ classification of CLSs and FB lattices, a FB generator for $U=1$, and FBs with $U=\infty$ as well as chiral FBs. In this chapter, we expand the idea presented in the last chapter and discuss the $U>1$ FB generator in 1D. Starting from a two-band problem, we use CLS existence conditions and the FB generator to achieve full parameterization of FB Hamiltonians in 1D two-band networks. Extending the approach used in the two-band case to an arbitrary number of bands, we introduce the inverse eigenvalue method, and obtain analytical and numerical solutions for the FB Hamiltonians.

\section[Introduction]{1D tight-binding Hamiltonian}
\label{section4.1} 

\nomenclature[gST]{gST}{Generalized sawtooth}   



We consider a 1D translational invariant lattice with $\nu>1$ lattice sites per unit cell. The time-independent Schr\"odinger equation on such a network is given by
\begin{equation}
	\label{fncls:hopping}
    \sum_{m=-\infty}^\infty H_m\vec{\psi}_{n+m}=E\vec{\psi}_n \; ,
\end{equation}
where the $\nu\times\nu$ matrices $H_m=H_{-m}^\dagger$ describe the hopping (tunneling) between sites from unit cells at distance $m$. As mentioned in the previous chapter, $H_0$ is Hermitian while $H_m$ with $m\neq0$ are not in general. We further classify networks according to the largest hopping range $m_c$: $H_m\equiv0$ for $|m|>m_c\geq1$. Note that $H_0$ describes intracell connectivities and $H_{m \neq 0}$ intercell links.


As introduced in Section \ref{section2.3}, an irreducible CLS is a solution of Eq. \eqref{fncls:hopping} with $\vec{\psi}_n\neq0$ only on the smallest possible finite number $U$ of adjacent unit cells, and zero everywhere else~\cite{flach2014detangling}. The corresponding eigenenergy is denoted as $\efb$. If such an eigenstate exists, then its translations along the lattice are also eigenstates, leading to a macroscopic degeneracy of $\efb$. The resulting band is flat, i.e. $E_\mu(k)=\efb$ is independent of $k$. As discussed in Section \ref{sec:fb-tester-and-generator}, for 1D, the control parameters that classify FB networks are the hopping range $m_c$, the number of bands $\nu$, and the CLS class $U$. 

The existence of CLSs in a FB lattice can now be used to design a simple local test routine as to whether a given network has a FB of class $U$ or not. Consider the $U\times U$ block matrix
\begin{gather}
    H_U=\left(\begin{array}{cccccc}
        H_0 & H_1 & H_2 & H_3 & \dots & H_U\\
        H_1^\dagger & H_0 & H_1 & H_2 & \dots & H_{U-1}\\
        \vdots & \ddots & \ddots & \ddots & \ddots & \vdots\\
        \vdots & ~ & ~ & ~ & ~ & \vdots\\
        H_{U-1}^\dagger & \dots & H_2^\dagger & H_1^\dagger & H_0 & H_1\\
        H_U^\dagger & \dots & H_3^\dagger & H_2^\dagger & H_1^\dagger & H_0
    \end{array}\right),
    \label{tcls:block}
\end{gather}
and an eigenvector $(\vec{\psi}_1,\vec{\psi}_2,\dots,\vec{\psi}_U)$ with eigenvalue $\efb$ such that
\begin{equation}
    \sum_{m=-m_c}^{m_c} H_m \vec{\psi}_{p+m} = 0 \;,\; \vec{\psi}_{l \leq 0} = \vec{\psi}_{l > U} = 0,
    \label{test}
\end{equation}
for all integers $p$ with $-m_c+1 \leq p \leq 0$ and $U+1 \leq p \leq U+m_c$. Similar equations hold for $H_m^\dagger$. These two sets of equations ensure $\vec{\psi}_{l \leq 0} = \vec{\psi}_{l > U} = 0$. Then the Hamiltonian has a FB of class $U$. As an example, consider $m_c=1$, see Fig.~\ref{fig2}. The corresponding condition simplifies to $H_1^\dagger \vec{\psi}_U = H_1 \vec{\psi}_1 = 0$. 

With that, we arrive at our core result---a novel systematic local FB generator based on CLS properties. For convenience, we set $m_c=1$ which corresponds to nearest neighbor hopping and is one of the most typical cases considered both experimentally and theoretically. Then we have to find those $\nu\times\nu$ matrices $H_0,H_1$ that solve the following set of equations for $1 \leq l \leq U$:
\begin{eqnarray}
	\label{mc1-1}
    H_1^\dagger\vec{\psi}_{l-1} + H_0\vec{\psi}_l + H_1\vec{\psi}_{l+1}  =  \efb\vec{\psi}_l\;,\\
    \label{mc1-2}
    H_1^\dagger\vec{\psi}_1 = H_1\vec{\psi}_U=0\;,\; \vec{\psi}_0 = \vec{\psi}_{U+1} = 0\;.
\end{eqnarray}
Choosing a set of $H_0,H_1$, we need to solve the eigenvalue problem Eq. \eqref{mc1-1} under the constraint of Eq. \eqref{mc1-2}, which makes $H_1$ singular and $\vec{\psi}_1$ and $\vec{\psi}_U$ the left and right eigenvectors of the zero mode(s) of $H_1$. We do so by considering increasing values of $\nu$ and $U$ in the following sections.

\section{Two-band problem}
\label{section4.3} 

The lattice and band structure of a cross-stitch lattice with $U=1$, $\nu=2$, $m_c=1$ was reported in Ref.~\cite{flach2014detangling} and corresponds to 
\begin{equation}
    H_0=\left( \begin{array}{cc} 0 & 0 \\ 0 & 0 \end{array} \right),\quad H_1=-\left( \begin{array}{cc} 1 & 1 \\ 1 & 1 \end{array} \right).
\end{equation}
In Fig.~\ref{fig1} (b) the sawtooth lattice (ST1) with $U=2$, $\nu=2$, and $m_c=1$ is shown~\cite{flach2014detangling} together with its CLS and band structure, with 
\begin{equation}
    H_0= -\left( \begin{array}{cc} 0 & \sqrt{2} \\ \sqrt{2} & 0 \end{array} \right),\quad H_1= -\left( \begin{array}{cc} 0 & \sqrt{2} \\ 0 & 1 \end{array} \right)\; .
\end{equation}
ST1 and its FB were recently experimentally probed in photonic waveguide lattices~\cite{weimann2016transport}. 

\begin{figure}[htb!]
    \centering
    \includegraphics[clip,width=0.8\columnwidth]{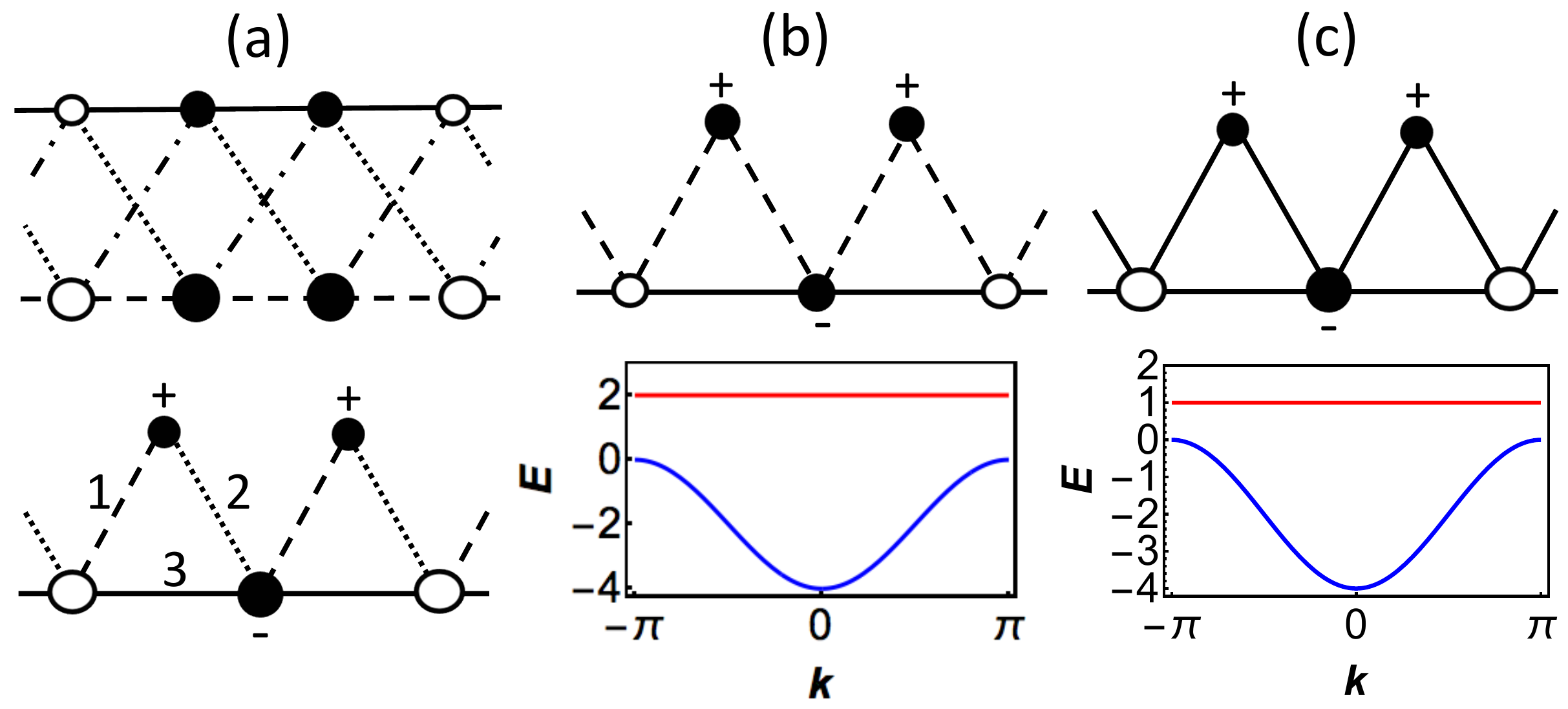}
    \protect\caption[Canonical $\nu=2$ chain, generalized sawtooth, ST1, and ST2 chains]{ (a) Top: canonical $\nu=2$ chain for $U=2$. Circles denote lattices sites (different sizes correspond to different onsite energies), lines denote hopping connections (different lines correspond to different hopping strengths), and filled circles denote the locations of a CLS. Bottom: generalized sawtooth chain after basis rotation (see text for details). Signs indicate the signs of the CLS amplitudes. (b) The known sawtooth ST1 chain. (c) New sawtooth ST2 chain. Top of (b,c): lattice structure, bottom of (b,c): band structure.
}
    \label{fig1}
\end{figure}

Without a loss of generality, we will use a \emph{canonical form} of $H$ for a generic two-band network: a unitary transformation on each unit cell will diagonalize $H_0$ sorting its diagonal elements (eigenvalues) $H_{\mu\mu}$ monotonically increasing with $\mu$. A trivial gauge $H \rightarrow H+\zeta \mathcal{I}$ (with $\mathcal{I}$ the identity matrix) sets $H_{11} = 0$, and a subsequent rescaling of $H \rightarrow \kappa H$ ensures $H_{22} = 1$(the case of a completely degenerate $H_0$ will be treated separately in Appendix \ref{as:generator:H0-degenerate}). In the simplest yet nontrivial case of two bands $\nu=2$, which completely fixes the non-degenerate matrix $H_0$:
\begin{equation}
	\label{nu=2}
    H_0=\left(\begin{array}{cc}
        0 & 0\\
        0 & 1
    \end{array}\right)
    \;,\;
    H_1=\left(\begin{array}{cc}
        a & b\\
        c & d
    \end{array}\right).
\end{equation}
Since $H_1$ is singular (as required by theorem \ref{theo:nes-cond-cls-1d}) and of size 2, it has exactly one zero mode, and can be parameterized by its spectral decomposition $H_1 =\alpha\vert\theta,\delta\rangle\langle\varphi,\gamma\vert$ as follows:
\begin{equation}
	\label{res:H1-angles}
    \begin{aligned}H_1 = \alpha\left(
        \begin{array}{cc}
            \cos\theta\cos\varphi & e^{i\gamma}\cos\theta\sin\varphi\\
            e^{-i\delta}\sin\theta\cos\varphi & e^{-i\left(\delta-\gamma\right)}\sin\theta\sin\varphi
        \end{array}\right) \;.
    \end{aligned}
\end{equation}
The prefactor $\alpha = |\alpha | e^{i\phi_\alpha}$ can be complex, and $\vert\varphi,\gamma\rangle$ and $\vert\theta,\delta\rangle$ are the left and right eigenvectors of the non-zero eigenvalue of $H_1$ (see Appendix~\ref{as:generator}). The upper plot in Fig.~\ref{fig1} (a) illustrates this canonical network structure. A rotation of the unit cell basis by an angle $\omega$ shifts the angles $\theta \rightarrow \theta +\omega$ and $\varphi \rightarrow \varphi+\omega$, and modifies $H_0=\left( \begin{array}{cc} \cos^2 \omega & \cos \omega \sin \omega \\ \cos \omega \sin \omega & \sin^2 \omega \end{array} \right)$. As we previously discussed, FB networks are defined up to unitary transformations. Applying $\theta \rightarrow \theta +\omega$ and $\varphi \rightarrow \varphi+\omega$, we find a \textsl{generalized sawtooth} (gST) chain with three different hoppings $t_{1,2,3}$ per triangle, and an onsite energy detuning (see bottom of Fig.\ref{fig1} (a) and Appendix~\ref{as:gsc} for details).


\subsection{U=1 case}

We start with $U=1$ (Fig.~\ref{fig1} (a) in Ref.~\cite{flach2014detangling}). Equations~(\ref{mc1-1}--\ref{mc1-2}) reduce to 
\begin{equation}
H_0 \vec{\psi}_1= \efb \vec{\psi}_1, \quad H_1\vec{\psi}_1=H_1^{\dagger}\vec{\psi}_1=0 \;.
\end{equation} 
Then the FB energy is $\efb=0$ or $\efb=1$. For $\efb=0$ it follows $\theta=\pi/2$ or $3\pi/2$ and $\varphi=\pi/2$ or $3\pi/2$. Respectively for $\efb=1$ we find $\theta=0$ or $\pi$ and $\varphi=0$ or $\pi$. The canonical form of $H_1$ has exactly one nonzero element on the diagonal, e.g.  for $\efb=0$ it  is $H_1=\left(\begin{array}{cc} 0 & 0\\ 0 & |\alpha|e^{i\phi_{\alpha}} \end{array}\right)$. We therefore obtain the detangled structure of the cross-stitch lattice (see Fig. \ref{fig:fano-ent-cross-stich} (a)). The dispersive band energy is given by $E(k)=C+2|\alpha|\cos\left(k+\phi_{\alpha}\right)$ where $C=0$ for $\efb=1$ and $C=1$ for $\efb=0$. The case of degenerate $H_0\equiv 0$ does not change the structure of $H_1$ and leads to $\efb=0$ and $C=0$. Interestingly, the cross-stitch lattice family in Ref.~\cite{flach2014detangling} was characterized by three parameters: the location of the FB energy, the width of the dispersive band, and an overall gauge. Here we obtain a four-dimensional control parameter space. The first three are the overall gauge $\zeta$, the rescaling $\kappa$, and the band width control $|\alpha|$, which reproduce the findings from Ref.~\cite{flach2014detangling}. The additional fourth control parameter is the phase $\phi_\alpha$. It corresponds to a time-reversal symmetry breaking effective magnetic field in 1D, and completes the class of $m_c=1,\nu=2,U=1$ FB lattices. Remarkably, there is another hidden $U=1$ case with two FBs for which $H_1$ has precisely one nonzero element on one of the two off-diagonals: $\theta=0, \varphi=\pi/2$ or $\theta=\pi/2, \varphi=0$. To observe this, one must also redefine the unit cell (see Appendix~\ref{as:generator} for details).

\subsection{U=2 case}

We proceed to the nontrivial $U=2$ case. Here, the Hamiltonian $H_{U=2}$ is a $2\times2$ block matrix
\begin{equation}
	\label{u2:mh_u2}
    H_2 = \left(\begin{array}{cc}
        H_0 & H_1\\
        H_1^\dagger & H_0,
    \end{array}\right)
\end{equation}
where $H_0$ is given by Eq. \eqref{nu=2} and $H_1$ is given by Eq. \eqref{res:H1-angles}.

The equations Eqs. (\ref{mc1-1}--\ref{mc1-2}) read
\begin{align}
	\label{eq:eigenvalue_eq2}
    H_0\vec{\psi}_1 + H_1 \vec{\psi}_2 & =  \efb\vec{\psi}_1 \;, \\
    H_1^\dagger\vec{\psi}_1 + H_0\vec{\psi}_2 & =  \efb\vec{\psi}_2 \;, \\
    H_1\vec{\psi}_1 & =  0 \;,\\
    \label{eq:cls_con2}
    H_1^\dagger\vec{\psi}_2 & = 0 \;.
\end{align}
The details of solving the above equations are given in Appendix~\ref{as:generator}. The final result reads:
\begin{equation}
    \delta  =  \gamma,\,\,
    |\alpha|  =  \frac{\sqrt{-\sin(2\theta)\sin(2\varphi)}}{|\sin(2(\theta-\varphi))|}.
    \label{u2:sol}
\end{equation}
The solutions in Eq. \eqref{u2:sol} and the Hamiltonian $H$ are invariant under the transformation $\{\varphi \rightarrow \varphi+p \pi\;,\; \theta \rightarrow \theta + q \pi\;, \phi_{\alpha} \rightarrow \phi_\alpha+(p+q)\pi\}$ with $p,q$ being integers. The irreducible angle parameter space therefore reduces to $0 \leq \varphi,\theta \leq \pi$. Since $|\alpha|$ is real, the solutions only exist for $0\le\theta\le\frac{\pi}{2}\cap\frac{\pi}{2}\le\varphi\le\pi$ or $\frac{\pi}{2}\le\theta\le\pi\cap0\le\varphi\le\frac{\pi}{2}$, i.e. two disjointed regions are shown for the FB energy $\efb$ in Fig.~\ref{fig3}. The corresponding band structure is given by (see Appendix~\ref{as:generator} for details)
\begin{align}
	\label{efbu2}
    \efb & =  \frac{\cos(\theta)\cos(\varphi)}{\cos(\theta-\varphi)},\\
    \label{e(k)u2}
    E(k) & =  \frac{\sin(\theta) \sin(\varphi)}{\cos(\theta-\varphi)} + 2 | \alpha | \cos(\theta-\varphi)\cos(k+\phi_\alpha).
\end{align}

\begin{figure}[htb!]
    \centering
    \includegraphics[clip,width=0.8\columnwidth]{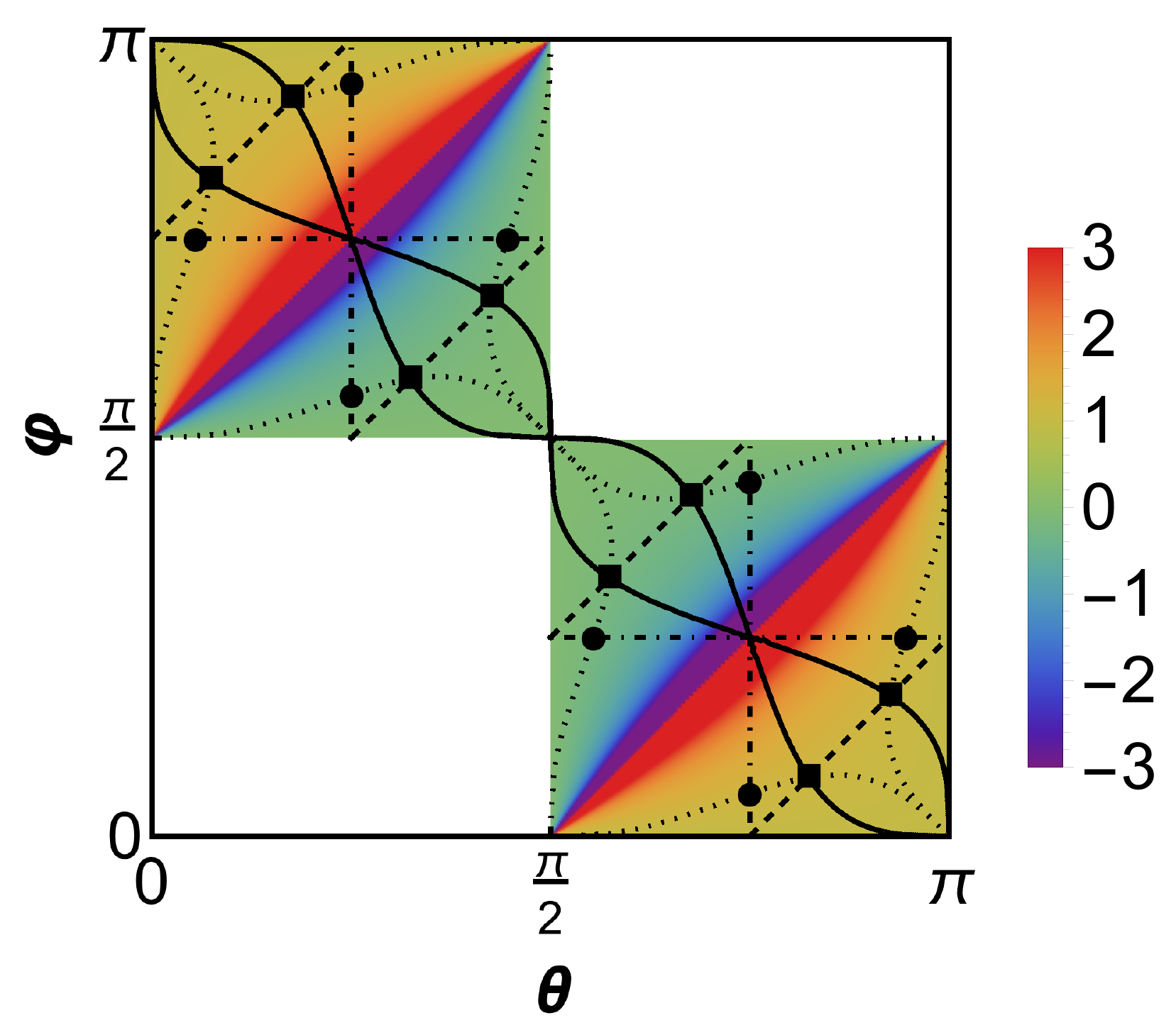}
    \protect\caption[Parameter space of a flatband in a 1D two-band network]{ Flatband energy $\efb \left(\theta,\varphi\right)$ for $m_c=1,\nu=2,U=2$. The colored squares host FB networks, while the white ones do not. The color code shows the energy of the FB. Dashed-dotted lines denote same onsite energies, dotted lines denote $t_1=t_2$, dashed lines denote $t_2=t_3$, and solid lines denote $t_1=t_3$, all in the gST chain. Filled circles represent the ST1 chain in Fig.~\ref{fig1} (b) and filled squares the ST2 chain in Fig.\ref{fig1} (c).
}
    \label{fig3}
\end{figure}

The bandwidth $\Delta_w$ of the dispersive band is given by
\begin{equation}
	\label{bw}
    \Delta_w = 2 \frac{\sqrt{-\sin(2\theta)\sin(2\varphi)}}{| \sin(\theta - \varphi) |},
\end{equation}
and is always bounded by $|\Delta_w| \leq 2$.

The FB energy is always gapped away from the dispersive band by a gap $\Delta_g= \Delta E - \frac{\Delta_w}{2} $ with $\Delta E$ being the distance between the FB energy and the dispersive band center (except for a few isolated points discussed below). The ratio $\Delta_w/ \Delta E$ is shown in Fig.~\ref{fig4}. This ratio is zero for $\theta = \pi/2+\varphi$. There, the FB energy is gapped infinitely far away from the dispersive band. Using a proper rescaling parameter $\kappa$ and a gauge $\zeta$, we can always renormalize the band gap to a finite number, at the expense of flattening the dispersive band. This special line corresponds to the case of degenerate $H_0$ and two FBs of class $U=1$  (see Appendix~\ref{as:generator:H0-degenerate} for details). On the boundary lines $\varphi,\theta = 0,\pi/2,\pi$ the band width $\Delta_w$ strictly vanishes, reducing the problem to a trivial $H_1=0$ case with two FBs of class $U=1$. One exception is the points $\{ \theta = \pi-\varphi \;,\; \varphi=0,\pi/2,\pi \}$ where the band width $\Delta_w$ stays finite but the gap $\Delta_g$ vanishes. Here, the FB becomes class $U=1$ and touches the dispersive band of finite width (see Fig.~\ref{fig4}).

\begin{figure}[!tbp]
    \centering
    \includegraphics[clip,width=0.8\columnwidth]{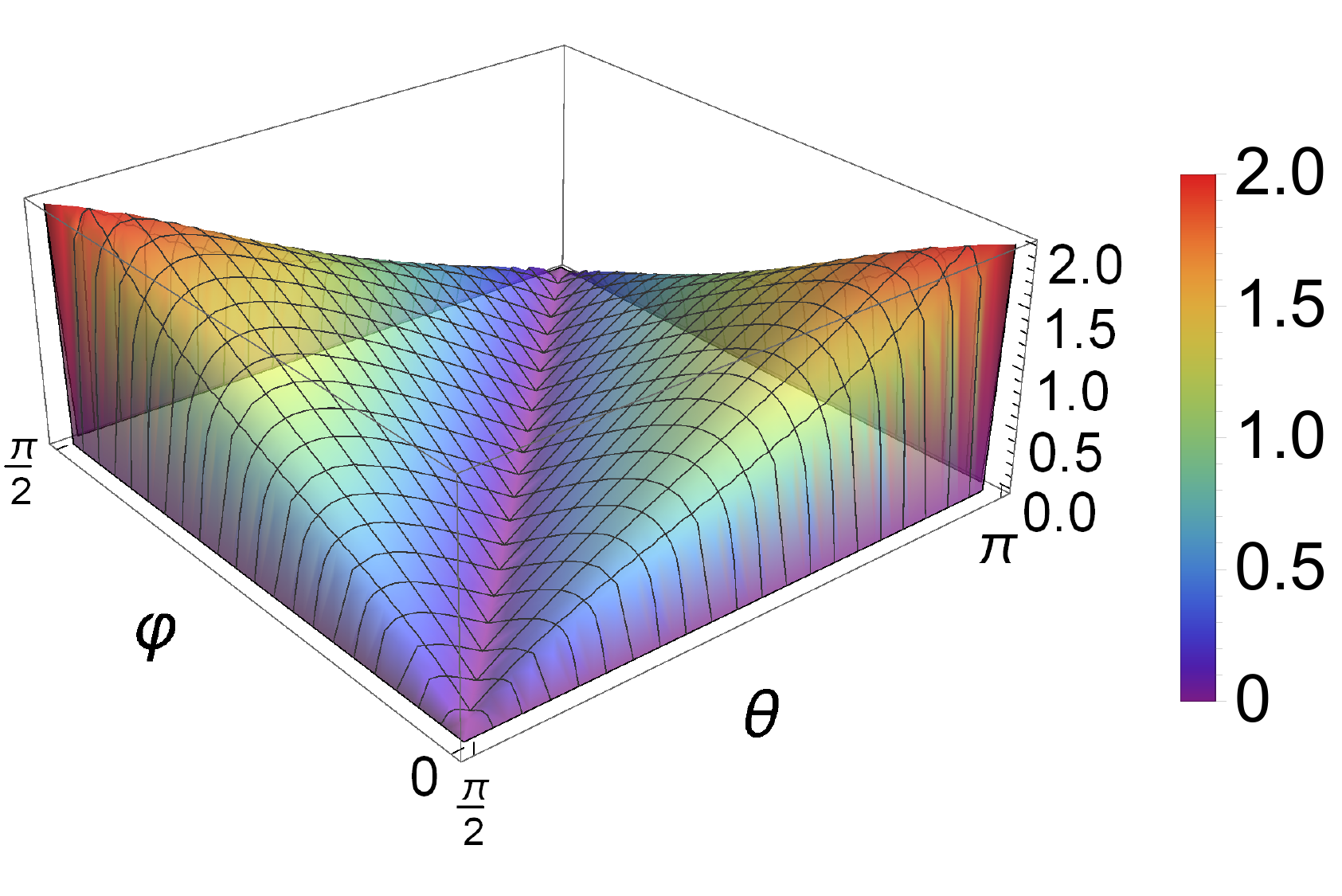}
    \protect\caption[Band width vs. band gap in a 1D two-band FB network]{Ratio of the dispersive band width to the distance between the FB and the dispersive band center $\Delta_w/\Delta E$ versus  $\left(\theta,\varphi\right)$ for $m_c=1,\nu=2,U=2$ in one irreducible quadrant. The color code shows the value of the ratio.}
    \label{fig4}
\end{figure}

The class $U=2$ is the largest possible irreducible CLS for $\nu=2,\ m_c=1$ networks, as can be straightforwardly checked using the above generator construction (see Appendix~\ref{app:appendix-A1}). In particular, we find all one-parameter families of gST chains for which either the onsite energies are equal (dashed-dotted lines in Fig.~\ref{fig3}), or a pair of the hoppings $t_{1,2,3}$ is equal (solid, dashed, and dotted lines in Fig.~\ref{fig3}). The known ST1 chain is given by the intersection of dotted and dashed-dotted lines, in which the two hoppings are equal $t_1=t_2 \neq t_3$ and the onsite energies are equal. We discover a novel intersection point (square symbols in Fig.~\ref{fig3}) where all three hoppings are equal $t_1=t_2=t_3$ (but the onsite energies differ). This is a new ST2 chain (Fig.~\ref{fig1}(c)), which should be easily realized experimentally: simple geometry allows all hoppings to be made equal, an external DC bias can fine-tune the onsite energy differences, and the CLS is easily addressed having identical absolute amplitude values on occupied sites.

In this section, we introduced a novel FB generator for 1D translational invariant tight-binding networks with two bands. We construct the whole FB family of two-band networks with nearest neighbor hoppings, and we found that the largest possible CLS size is $U=2$. To solve the eigenvalue problem, we considered $H_0$ is given and parameterized $H_1$, and then identified the parameter space that gives FBs. However, in the case of higher numbers of bands, $H_0$ cannot be considered as a given variable and the parameterization of $H_1$ involves more parameters, resulting in an overcomplete problem that is hard to solve. If we fix the CLS and $H_0$, solving the eigenvalue problem for $H_1$ becomes easier. In the following section, we use the latter approach to generate FB Hamiltonians for an arbitrary number bands and arbitrary CLS size $U$ in 1D.

\section{Arbitrary number of bands}
\label{section4.4} 

In the case of more than two bands, i.e. $\nu>2$, we have to find $\nu \times \nu$ matrices $H_0, H_1$ and the CLS that satisfy the eigenvalue problem (Eq. \eqref{mc1-1}) and destructive interference conditions (Eq. \eqref{mc1-2}). It is convenient to write the eigenvalue problem as follows: 
\begin{eqnarray}
    H_1\vpsi_2 &=& (\EFB - H_0)\vpsi_1 \label{eig-1} \\
    H_1^\dagger\vpsi_{l-1} + H_1\vpsi_{l+1} &=& (\EFB - H_0)\vpsi_l , 2 \le l \le U-1
    \label{eig-2}\\
    \label{eig-3}
    H_1^\dagger\vpsi_{U-1} &=& (\EFB - H_0)\vpsi_U\\
    H_1\vpsi_1 &=& H_1^\dagger\vpsi_U = 0
    \label{eig-4}\\
    \vpsi_l &=& 0 \;,\;  l<0,\,l>U \;.
    \label{eig-5}
\end{eqnarray}
This set of equations is the starting point of our FB generator. Our goal is to generate all possible matrices $H_1$ which allow for the existence of a FB, given a particular choice of $H_0$. Note that $H_0$ can be diagonal (canonical form), but any non-diagonal  Hermitian choice of $H_0$ is allowed as well. 

One way to look for solutions is to parameterize $H_1$ and to compute FB energy $\EFB$ and the CLS $\PCLS$ for a given set of $U$ and $\nu$. In order to satisfy Eq. \eqref{eig-4}, we choose $H_1$ from the space $\mathcal{Z}$ of $\nu \times \nu$ matrices with one zero eigenvalue. Then the directions of the vectors $\vpsi_1,\vpsi_U$ are fixed by the choice of $H_1$, leaving their two norms as free variables. Together with the remaining unknown CLS components and the FB energy, we arrive at $V=(U-2)\nu+3$ variables. The total number of equations in Eqs. (\ref{eig-1}--\ref{eig-3}) is $E=U\nu$. Since $\nu \geq 2$, it follows that the set of equations is overdetermined. We need $2\nu-3$ additional constraints which will lead us to the proper codimension$(2\nu-3)$ manifold in the space $\mathcal{Z}$. For $\nu=2$, the codimension(1) manifold was computed explicitly and a closed form of the functional dependence of the CLS and FB energy on $H_1$ was obtained in the previous section. For larger values of $\nu$ (and $U$), the constraint computation becomes difficult, and therefore we will simply invert the approach: we will define the CLS (thereby setting $U$) and $\EFB$ and generate the proper $H_1$ matrix manifold. This will turn an overcomplete set of equations into an undercomplete one, which is easier to be analyzed. One of the difficulties of Eqs. (\ref{eig-2}--\ref{eig-5}) is that $\EFB$ is also unknown. We will now show that it can be easily expressed in terms of the CLS and the Hamiltonian.

Let us assume that $\psi_1$ is not orthogonal to $\psi_U$. Multiplying $\langle \psi_U |$ from the left with Eq. \eqref{eig-1}, the FB energy $\EFB$ follows as \footnote{For $m_c>1$, one has to assume $H_m, m<m_c$ are also input parameters.}  
\begin{equation}
    \label{eq:efb-def}
    \EFB = \frac{\langle \psi_U \vert H_0 \vert \psi_1 \rangle}{\langle \psi_1\vert \psi_U \rangle} \;.
\end{equation}
For practical purposes, we can choose the CLS normalization condition $\langle \psi_1\vert \psi_U \rangle = 1$. Note that if $\psi_1$ is orthogonal to $\psi_U$, the CLS class is reduced to a $U-1$ class by an appropriate unitary transformation including a redefinition of the unit cell (see Theorem ~\ref{theo:u-reduce-u-1}).

Knowing $\EFB$, we can treat the problem of FB generation (Eqs. \ref{eig-1}--\ref{eig-5}) as an inverse eigenvalue problem:~\cite{boley1987survey}  given $\EFB$ and $\PCLS$---as well as part of the Hamiltonian, i.e. $H_0$---we reconstruct the Hamiltonian matrix $H$ in Eq.~\eqref{eq:tri-diag-block}. The idea of finding hopping matrices for a fixed CLS was first introduced by Nishino, Goda, and Kusakabe~\cite{nishino2003flat,nishino2005three}. Our results, even if limited to 1D, are much more systematic; compared to the work of Nishino, Goda, and Kusakabe, we classify a CLS by its size $U$, introduce the constraints on $\PCLS$ ensuring that it is a $U$-class CLS, and show how to resolve these constraints.

\subsection{The generator}
\label{sec:fb-gen}

We consider a 1D translational invariant tight-binding network with nearest neighbor hoppings $m_c=1$, arbitrary number $\nu$ of bands, and the CLS class $U \le \nu$. We propose the following algorithm to construct a Hamiltonian of a FB with the above parameters:
\begin{enumerate}
    \item Fix the number of bands $\nu$ and the size of the CLS $U$.
    \item Choose $H_0$ either as a diagonal (canonical form) or as any Hermitian matrix.
    \item Choose a real $\EFB$.
    \item Choose $\vpsi_1$ (or $\vpsi_U$).
    \item Exclude $H_1$ from Eqs. (\ref{eig-1}--\ref{eig-5}), arrive at a set of two linear and further non-linear constraints, and solve them for the remaining CLS components $\vpsi_l$.
    \item Solve the linear system Eqs. (\ref{eig-1}--\ref{eig-5}) to find $H_1$. 
\end{enumerate}
The system of equations (\ref{eig-1}--\ref{eig-5}) is linear, and therefore it is easy to solve or to show that it has no solution. Typically, if this system has a solution, it will be undercomplete and show up with multiple solutions compatible with the input CLS. It is therefore enough to find a \emph{particular} solution $\bar{H}_1$ to Eqs.~(\ref{eig-1}--\ref{eig-5}). A generic solution is $H_1 = \bar{H}_1 + \dH$, where $\dH$ follows from the following homogeneous system of equations:
\begin{gather}
    \dH\vpsi_2 = 0\notag\\
    \dH^\dagger\vpsi_{l-1} + \dH\vpsi_{l+1} = 0,\ \ \  2 \le l \le U-1\notag\\
    \label{eq:cls-def-ieig-homogeneous}
    \delta H_{U-1}^\dagger\vpsi_{U-1} = 0\\
    \dH\vpsi_1 = \dH^\dagger\vpsi_U = 0\notag\\
    \vpsi_l=0\quad l<0,\,l>U.\notag
\end{gather}
The perturbation $\dH$ is a deformation of the Hamiltonian $H$ that preserves the CLS and the FB energy, and only affects the dispersive part of the spectrum. 

It is also possible to further constrain the network connectivity by choosing specific elements of $H_0$ and/or $H_1$ to be zero. This is easily accounted for in $H_0$, which is an input parameter. The case of $H_1$ is more involved, as discussed in Section~\ref{sec:fb-U-gt-3}. 

\subsection{Solutions}
\label{sec:solutions}

We proceed to classify FBs in the order of increasing $U$. The $U=1$ case has already been completed in Ref.~\cite{flach2014detangling}; therefore we start our classification with $U=2$.

\subsubsection{U=2 case}
\label{sec:fb-U2}

We fix the number of bands to $\nu$, and choose $H_0$, $\EFB$, and $\kpsi{1}$. The inverse eigenvalue problem Eq.~(\ref{eig-1}--\ref{eig-5}) now reads
\begin{align}
    H_1\kpsi{2} & = \left(\EFB - H_0\right)\kpsi{1}\notag\\
    \bpsi{1}H_1 & = \bpsi{2}\left(\EFB - H_0\right)\notag\\
    \label{eq:cls-ieig-U2-H1}
    H_1\kpsi{1} & = 0\\
    \bpsi{2}H_1 & = 0.\notag
\end{align}
The eigenfunction $\PCLS=(\vpsi_1, \vpsi_2)$ cannot be chosen arbitrarily; its second part $\kpsi{2}$ has to satisfy the following set of linear and non-linear compatibility constraints:
\begin{align}
    \langle \psi_1\vert\psi_2\rangle & = 1\notag\\
    \label{eq:cls-ieig-U2-constraints}
    \evpsi{1}{H_0}{2} & = \EFB\\
    \expval{\EFB - H_0}{\psi_1} & = \expval{\EFB - H_0}{\psi_2}.\notag
\end{align}
The first constraint is simply a choice of $\PCLS$ normalization. The second constraint follows from Eq.~\eqref{eq:efb-def} and uses $\EFB$ as an input variable. The last identity results from multiplying the first equation in Eq.~\eqref{eq:cls-ieig-U2-H1} by $\bpsi{2}$ from the left, and multiplying the second equation in Eq.~\eqref{eq:cls-ieig-U2-H1} by $\kpsi{1}$ from the right. It is not possible to solve the third constraint analytically in general, but we present in Appendix~\ref{app:U2-constraints} a numerical algorithm that allows to resolve these constraints and enumerate all the solutions, if existing. If existing, the solution to $\kpsi{2}$ has $\nu-3$ free parameters. For the special case of two bands $\nu=2$, FB energy $\EFB$ can not be chosen arbitrarily and needs to be included into the procedure as a to-be-defined variable. Note that this particular case can be solved in closed analytical form following a different solution strategy (see the previous section).

Once $\PCLS=(\vpsi_1, \vpsi_2)$ is known, we can solve Eq.~\eqref{eq:cls-ieig-U2-H1} for $H_1$. First we note that the last two equations (the destructive interference conditions) can be taken into account with the following ansatz for $H_1$: 
\begin{gather}
    \label{eq:cls-ieig-U2-anzats}
    H_1 = Q_2\,M\,Q_1,\quad Q_i = \mI - \frac{\kpsi{i}\bpsi{i}}{\bpsi{i}\kpsi{i}}.
\end{gather}
Then Eq.~\eqref{eq:cls-ieig-U2-H1} becomes an inverse eigenvalue problem; the details of the derivation are presented in Appendix~\ref{app:ieig-toy} and the solution is
\begin{align}
    H_1 & = G_1 + \delta H_1,\notag\\
    \label{eq:cls-ieig-U2-sol} 
    G_1 & = \frac{\left(\EFB - H_0\right)\kpsi{1}\bpsi{2}\left(\EFB - H_0\right)}{\expval{\EFB - H_0}{\psi_1}} ,\\
    \delta H_1 & = Q_{12} \,K\, Q_{12},\notag
\end{align}

where $K$ is an arbitrary $\nu\times\nu$ matrix and $Q_{12}$ is a joint transverse projector on $\kpsi{1},\kpsi{2}$: $Q_{12}\kpsi{i}=0,\,i=1,2$. If the denominator $\evpsi{1}{\EFB - H_0}{1}\equiv 0$, the above solution is replaced with a more complicated expression involving two different projectors (see Appendix~\ref{app:ieig-toy} for details).

It is instructive to count the number $F$ of free parameters in the above solution, given a fixed $H_0$, $\EFB$, and $\kpsi{1}$ for $\nu \geq 3$. It is the sum of two contributions: the number of free parameters in $\delta H_1$  and in the particular solution $G_1$, which are $(\nu-2)^2$ and $(\nu-3)$, respectively. The final result is $F=\nu^2-3\nu+1$. It then follows that the FB Hamiltonians form a codimension$(2\nu-2)$ subspace, since $H_0$ is arbitrary, $\dim(H_1)=\nu^2$, and the total number of free parameters at fixed $H_0$ is $F_t=F+1+\nu=\nu^2-2(\nu+1)$. This is a remarkable result, since it shows that FB Hamiltonians are only weakly fine-tuned, e.g. for $\nu=3$ we find five free parameters when choosing the nine elements of  $H_1$ for an arbitrarily chosen $H_0$. Note that the above counting does not apply to the $\nu=2$ case, which is studied in the previous section and amounts to two free parameters when choosing the four elements of $H_1$.

Equations \eqref{eq:cls-ieig-U2-constraints} and \eqref{eq:cls-ieig-U2-sol} provide the complete solution to the problem of finding all the 1D nearest neighbor Hamiltonians with one FB and a CLS of class $U=2$. Figure~\ref{fig:u2_examples} shows some examples of $U=2$ and $\nu=3$ Hamiltonians constructed using the above scheme. 
 
\begin{figure}
    \centering
   \subfloat[]{
        \label{fig:u2_nu3_can}
        \includegraphics[scale=0.35]{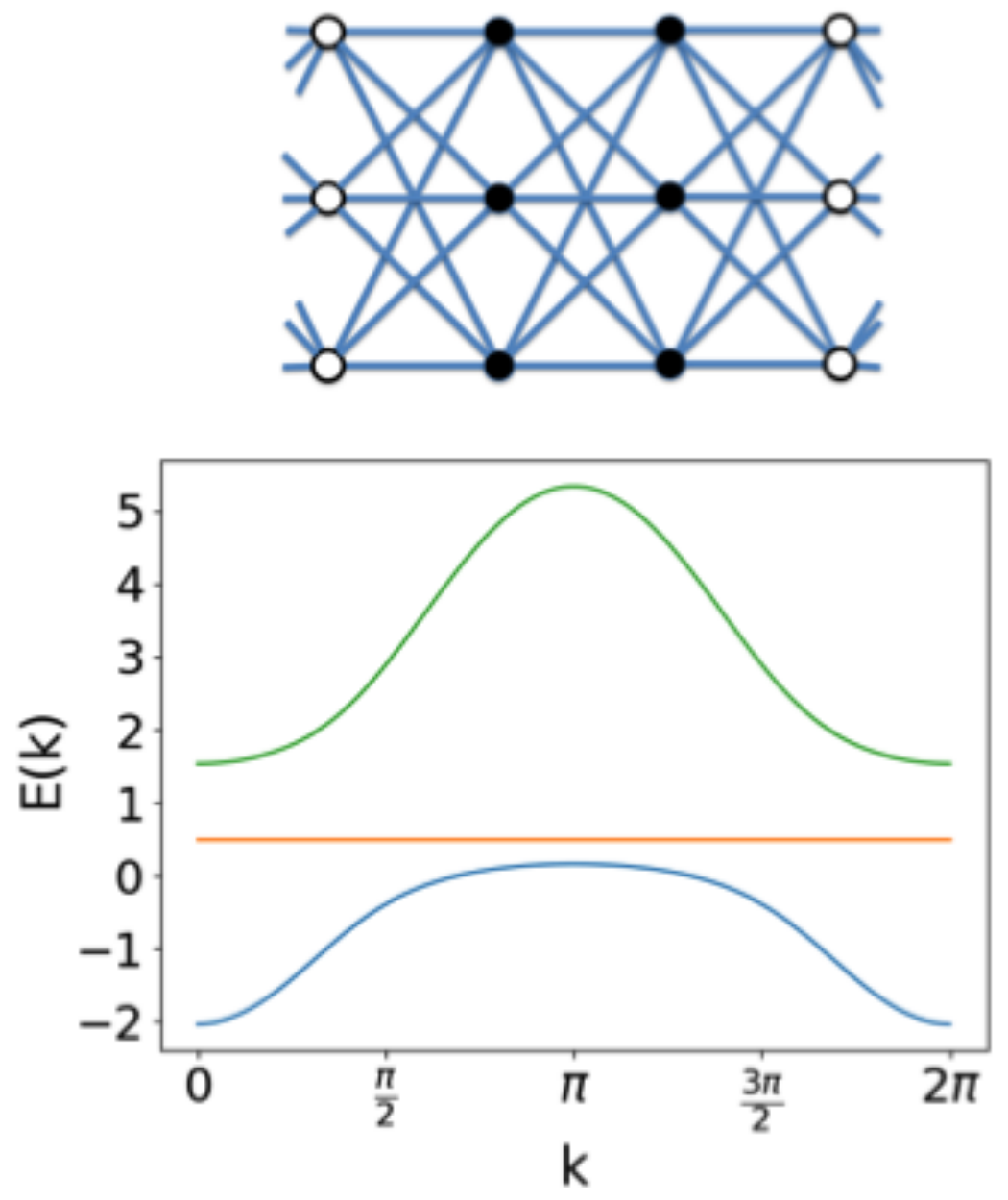}}
   \subfloat[]{
        \includegraphics[scale=0.35]{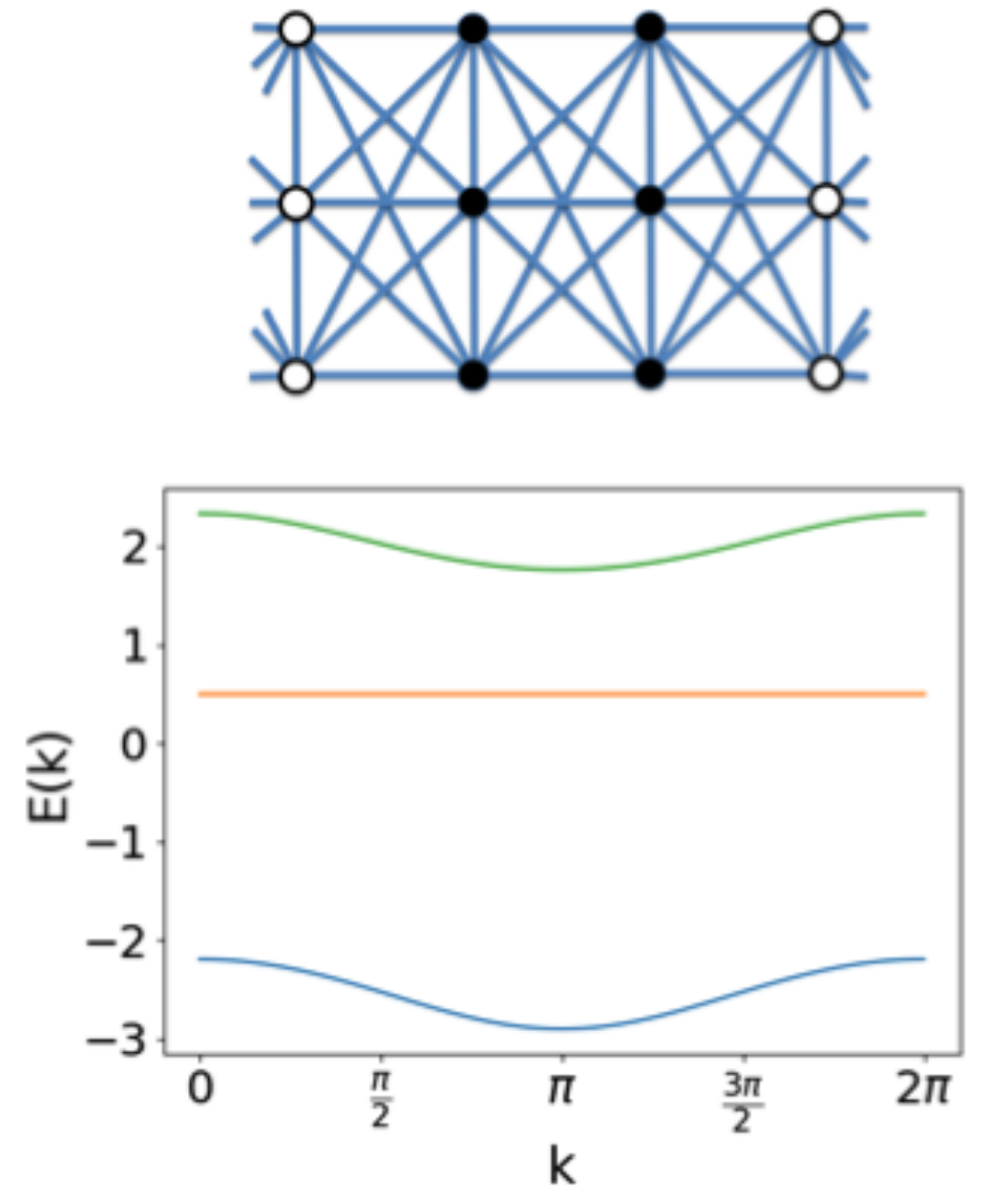}
        \label{fig:u2_nu3_noncan}}
   \subfloat[]{
        \includegraphics[scale=0.35]{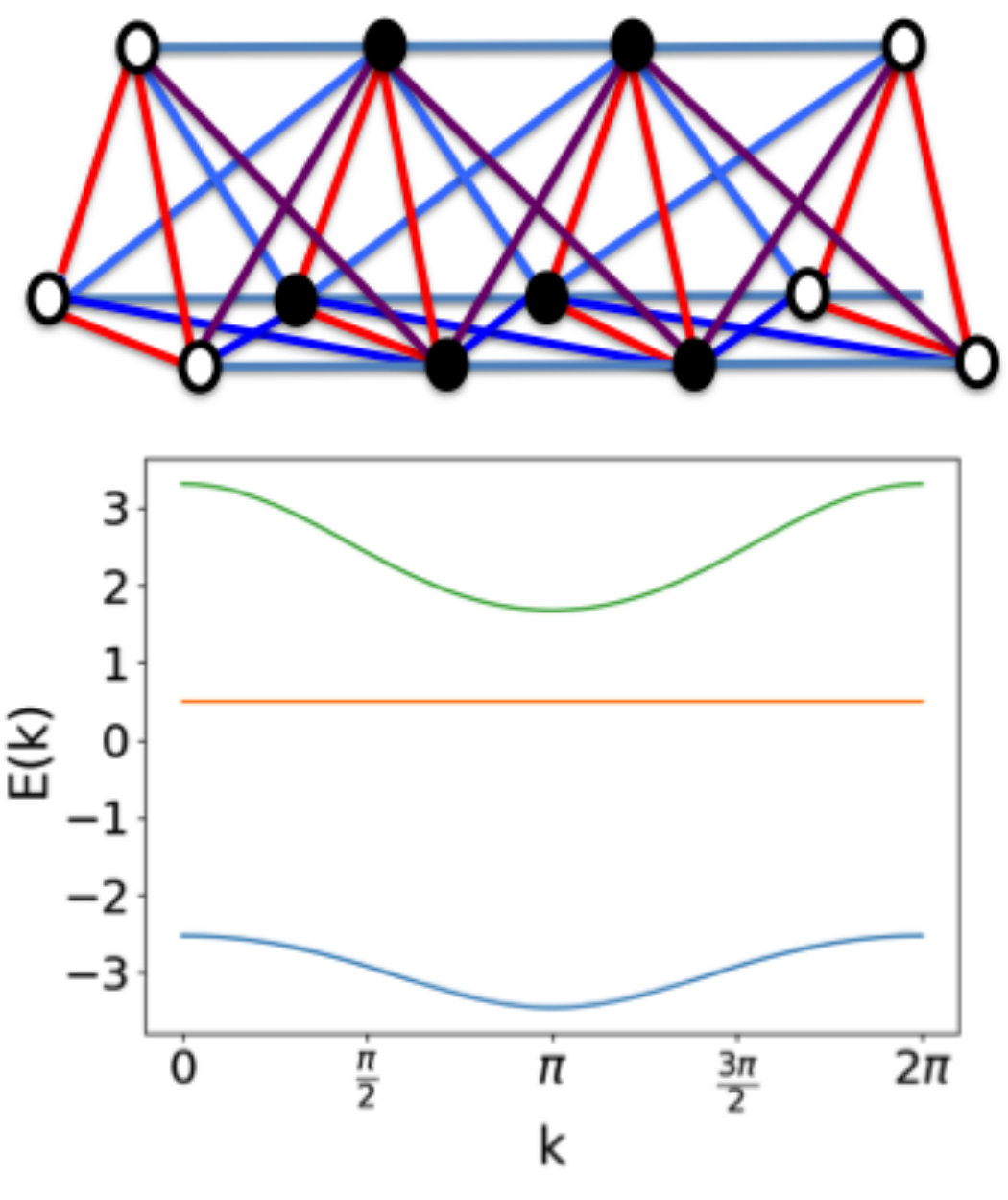}
        \label{fig:u2_nu3_gen}}
    \caption[Examples of flatband Hamiltonians with CLSs of class $U=2$, $\nu=3$]{Examples of flatband Hamiltonians with CLSs of class $U=2$, $\nu=3$. The sites occupied by a CLS are indicated by filled black circles. Each subfigure contains a visualization of the lattice (top) and the band structure (bottom). The FB is colored in orange. (a) Diagonal $H_0$; (b) non-diagonal $H_0$, and (c) non-diagonal and fully connected $H_0$. Appendix~\ref{app:u2-examples} contains a detailed description of the Hamiltonians.}
    \label{fig:u2_examples}
\end{figure}

\subsubsection{$U\ge3$ case}
\label{sec:fb-U-gt-3}

Let us consider larger $U$ values. For simplicity, we use $U=3$ in the examples. First, we fix the number of bands to $\nu$, and choose $H_0$, $\EFB$, and $\kpsi{1}$. Then we have the following inverse eigenvalue problem with $U+2$ equations ($U$ for each CLS-occupied unit cell, and two for the 
destructive interference conditions):
\begin{align}
    H_1 \kpsi{2} & = \left(\EFB - H_0\right) \kpsi{1}\notag\\
    H_1^\dagger \kpsi{1} + H_1 \kpsi{3} & =\left(\EFB - H_0\right) \kpsi{2}\notag\\
    H_1^\dagger \kpsi{2} & = \left(\EFB - H_0\right) \kpsi{3}\notag\\
    \label{eq:cls-ieig-U3-H1}
    H_1 \kpsi{1} & = 0\\
    H_1^\dagger \kpsi{3} & = 0.\notag
\end{align}

The set of constraints for $\PCLS$ reads
\begin{gather}
    \langle \psi_{1}\kpsi{3} = 1\notag\\
    \evpsi{1}{H_0}{3} = \EFB\notag\\
    \label{eq:cls-ieig-U3-constraints}
    \evpsi{1}{\EFB - H_0}{2} = \evpsi{2}{\EFB - H_0}{3}\\
    \evpsi{1}{\EFB - H_0}{1} + \evpsi{3}{\EFB - H_0}{3} =\notag\\
     = \evpsi{2}{\EFB - H_0}{2}.\notag
\end{gather}
Again these identities are derived from Eq.~\eqref{eq:cls-ieig-U3-H1} by multiplying them with $\bpsi{1}$ and  $\bpsi{U}$ and rearranging the terms in order to eliminate $H_1$. Notice that the set of compatibility constraints for $\PCLS$ amounts to $U+1$ equations. Note also that precisely two of those $U+1$ equations, with $\bpsi{1}$ given, are linear, and the remaining are $U-1$ nonlinear equations for the remaining
CLS amplitudes. It is not possible to solve the nonlinear equations analytically in general, but we present in Appendix~\ref{app:U3-constraints} a numerical algorithm that allows to resolve these constraints and enumerate all the solutions, if existing, for the case $U=3$.

Instead of using the ansatz~\eqref{eq:cls-ieig-U2-anzats} for $H_1$, we take a more suitable approach to generate FB Hamiltonians (i.e. matrices $H_1$) for $U \geq 3$. With a given $\PCLS$, which satisfies the constraints from Eq. \eqref{eq:cls-ieig-U3-constraints}, the set of equations in Eq. \eqref{eq:cls-ieig-U3-H1} is a linear system with respect to $H_1$:
\begin{gather}
    \label{eq:cls-ieig-U3-H1-linear}
    T\, h_1 = \Lambda.
\end{gather}
Here, $h_1$ is a $\nu^2$-dimensional vector resulting from the vectorization of the matrix $H_1$, $T$ is a rectangular $\nu(U+2)\times\nu^2$ matrix whose elements are composed of the elements of the CLS, such that the product $T \, h_1$ is the left hand side of Eq.~\eqref{eq:cls-ieig-U3-H1}, and $\Lambda$ is a $\nu(U + 2)$ vector originating from the right hand side of Eq.~\eqref{eq:cls-ieig-U3-H1}:
\begin{gather}
    \label{eq:cls-ieig-U3-lambda}
    \Lambda = (\EFB - H_0)
    \begin{pmatrix}
        \vpsi_1\\
        \vpsi_2\\
        \dots\\
        \vpsi_U\\
        \vec{0}\\
        \vec{0}
    \end{pmatrix}.
\end{gather}
The zero-vector components $\vec{0}$ result from the destructive interference. The linear system of Eq. \eqref{eq:cls-ieig-U3-H1-linear} can then be solved, for example by using a least squares solver. Figure~\ref{fig:u3_examples} shows some examples of $U=3$ FBs, which we generated by resolving the constraints in Eq. \eqref{eq:cls-ieig-U3-constraints} and solving Eq.~\eqref{eq:cls-ieig-U3-H1-linear}.

\begin{figure}[htb!]
    \centering
    \subfloat[]{\includegraphics[width=0.4\columnwidth]{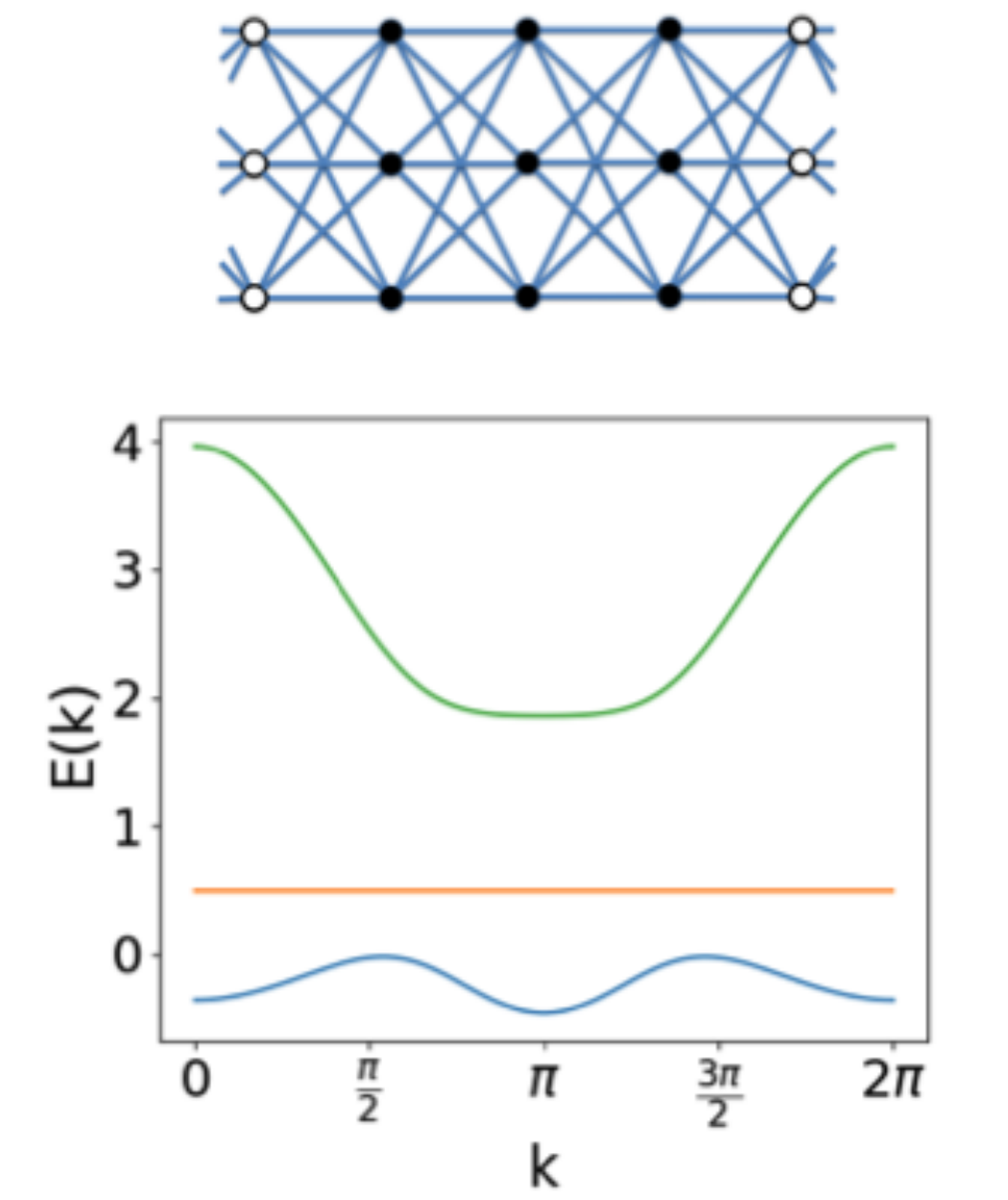}
                        \label{fig:u3_nu3_can}}
    \subfloat[]{\includegraphics[width=0.4\columnwidth]{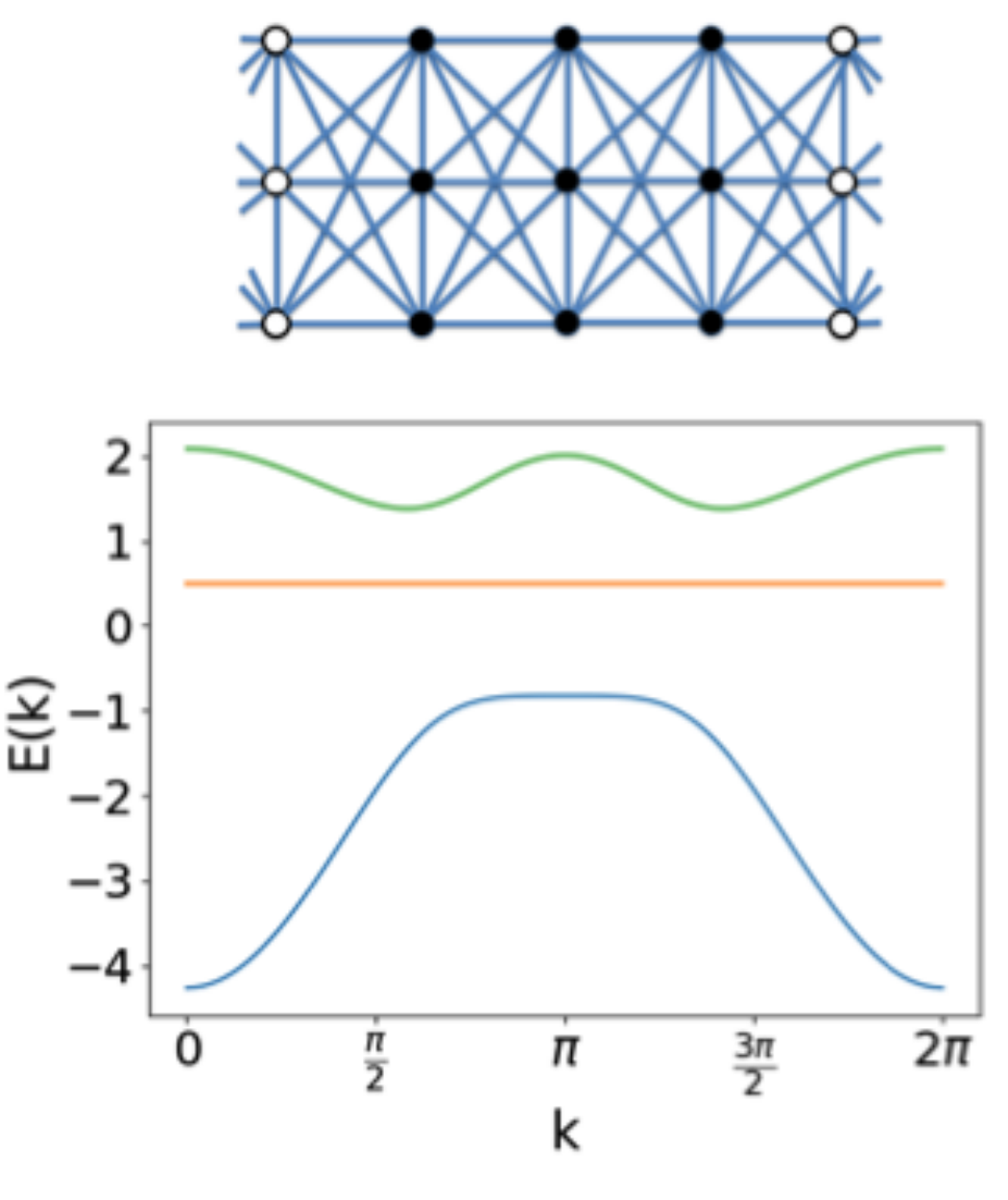}
                        \label{fig:u3_nu3_noncan}}\\
    \subfloat[]{\includegraphics[width=0.4\columnwidth]{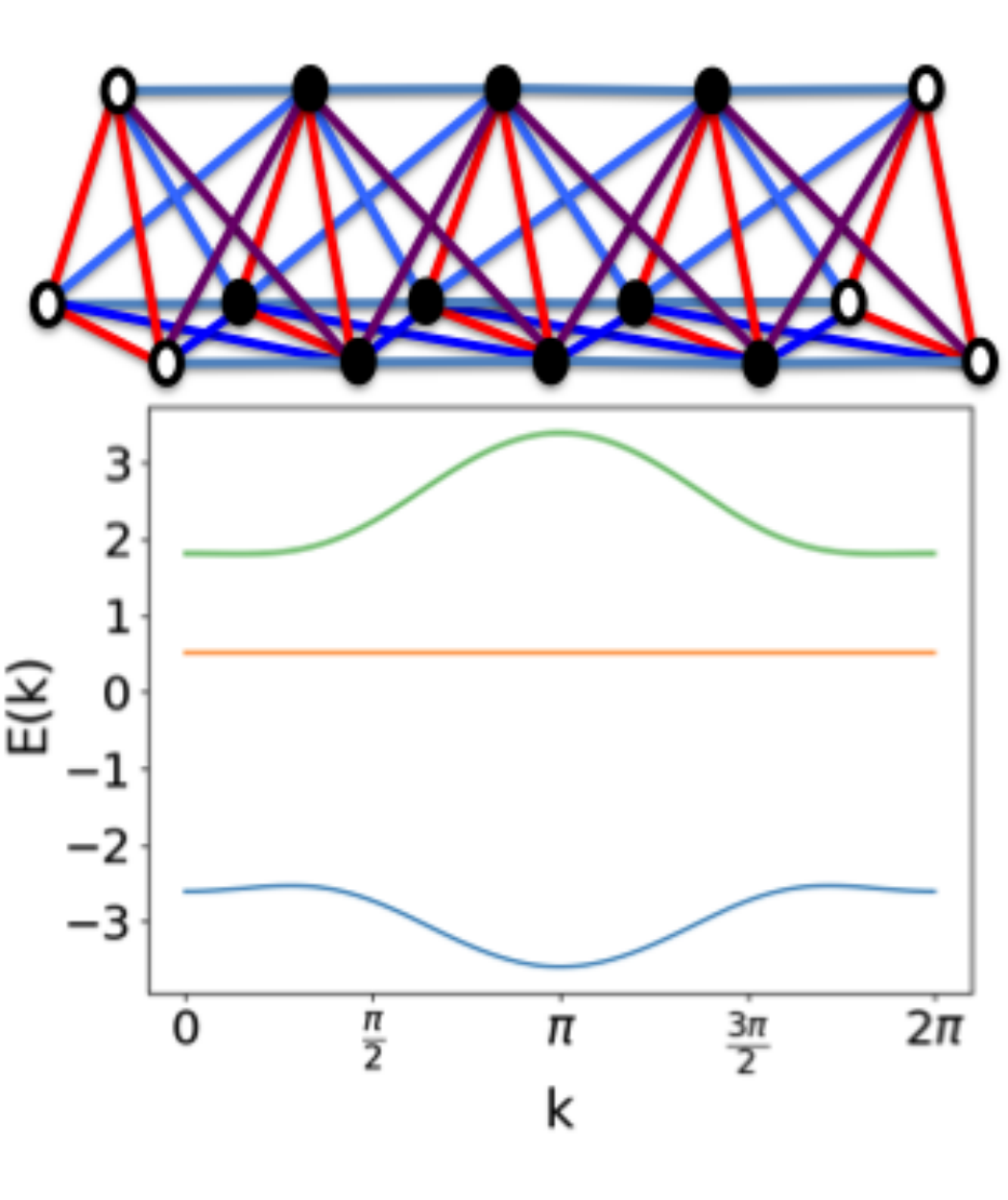}
                        \label{fig:u3_nu3_gen}}
    \subfloat[]{\includegraphics[width=0.4\columnwidth]{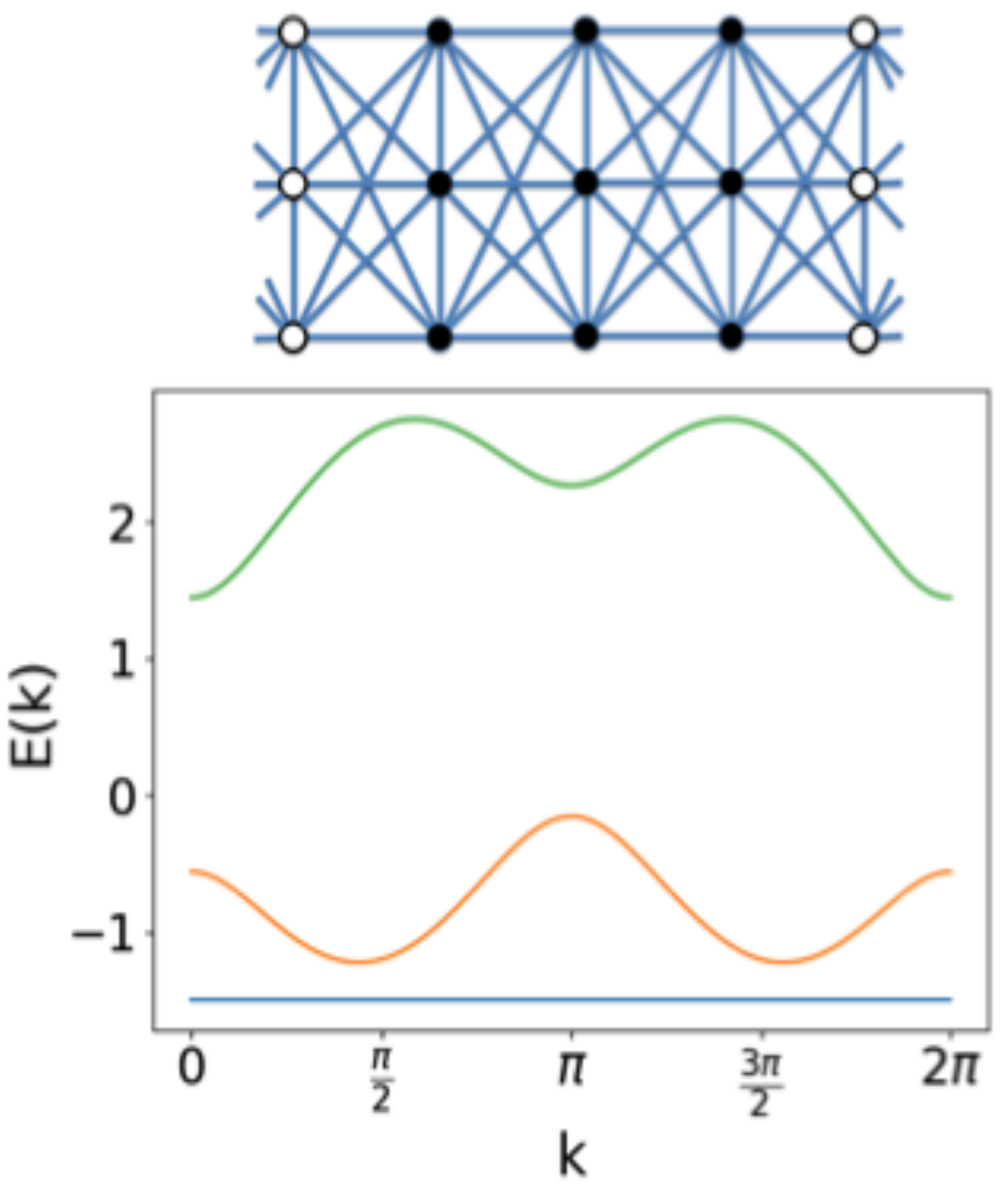}
                        \label{fig:u3_nu3_ground}}
    \caption[Examples of 1D $U=3$ flatband lattices]{ Examples of three-band FB lattices with a CLS of class $U=3$. The sites occupied by a CLS are indicated by black filled circles. (a) Diagonal choice for $H_0$. (b) Chain-like structure for $H_0$. (c) Generic choice for $H_0$. (d) $E_{FB}$ is set negative enough to become the ground state. Details of these examples are presented in Appendix~\ref{app_sec:u3_examples}.}
    \label{fig:u3_examples}
\end{figure} 

\subsection{Chiral symmetry}
\label{sec:fb-U2-cs}

An important subclass of FB networks is those with chiral symmetry (see Section \ref{section2.4.3}), and our generator simplifies in the case of chiral symmetry. Chiral lattices are bipartite networks with minority and majority sublattices. This imposes a specific structure of the hopping integrals and the CLS amplitudes $\vpsi_l$. For that, we split the lattice sites from each unit cell into two subsets, each belonging to one of the two sublattices. This leads to a splitting of each $\vpsi_l$ into two sublattice vectors, as well as to a corresponding block structure of the matrices $H_0,H_1$. As a result, the CLS of a chiral FB will always reside exclusively on the majority sublattice ~\cite{ramachandran2017chiral}:
\begin{gather}
    H_0 = \left(\begin{array}{cc}
        0 & A^\dagger\\
        A & 0
    \end{array}\right),\quad
    H_1 = \left(\begin{array}{cc}
        0 & T^\dagger\\
        S & 0
    \end{array}\right)\notag\\
    \label{eq:H01-psi-cs-def-b}
    \vpsi_l = \left(\begin{array}{c}
        \vphi_l\\
        0
    \end{array}\right),\quad l=1,\dots,U \;.
\end{gather}
Here $A$, $S$, and $T$ are $(\nu-\mu)\times \mu$ matrices, $\mu$ is the number of sites on the majority sublattice in the unit cell, and $\vphi_l$ is a $\mu$-component vector residing on the majority sublattice sites in a unit cell. By definition $\nu - \mu\leq\mu < \nu$. The spectrum of the system enjoys particle-hole symmetry around $E=0$. A chiral flatband has energy $\EFB = 0$ and is symmetry protected. For $\nu < 2\mu$, there are $\mu-\lfloor\nu/2\rfloor$ FBs at $\EFB=0$~\cite{ramachandran2017chiral}. Increasing the range of hopping $m_c > 1$ while preserving chiral symmetry will keep the chiral FBs in place. Moreover, one can preserve the chiral FBs by partially destroying the chiral and sublattice symmetry; this is achieved by adding hopping terms on the minority sublattice only, since the chiral FB CLS is occupying majority sublattice sites only:
\begin{gather}
    H_0 = \left(\begin{array}{cc}
        0 & A^\dagger\\
        A & B
    \end{array}\right),\quad
    H_1 = \left(\begin{array}{cc}
        0 & T^\dagger\\
        S & W
    \end{array}\right)\notag\\
    \label{eq:H01-psi-cs-def}
    \vpsi_l = \left(\begin{array}{c}
        \vphi_l\\
        0
    \end{array}\right),\quad l=1,\dots,U \;,
\end{gather}
where $B$ and $W$ are $(\nu-\mu)\times(\nu-\mu)$ matrices. Note that the overall particle-hole symmetry of the system is lost, but the original chiral FBs are still present at $\EFB = 0$. 

For a bipartite network, the hopping matrix $H_1$ has a specific structure given by Eq.~\eqref{eq:H01-psi-cs-def}, that simplifies Eq.~\eqref{eq:cls-ieig-U2-H1} to
\begin{align}
    \label{eq:bipartite_U2_eq}
    S\kphi{2}  & = -A\kphi{1}\\
    S\kphi{1} & = 0\\
    T\kphi{1} & = -A\kphi{2}\\
    T\kphi{2} & = 0,
\end{align}
where $\EFB = 0$. The minority sublattice hopping matrices $B, W$ dropped out as expected. The above equations are considerably simpler than the generic $U=2$ Eq.~\eqref{eq:cls-ieig-U2-H1}: the above system splits into two independent inverse eigenvalue problems for $S$ and $T$. Details of the solution are presented in Appendix~\ref{app:chiral-inv-eig-prob}, the final answer of which is
\begin{gather}
    S = -\frac{A\kphi{1}\bphi{2} Q_1}{\mel{\varphi_2}{Q_1}{\varphi_2}} + K_S Q_{12}\notag\\
    \label{eq:cls-ieig-U2-cs-sol}
    T = -\frac{A\kphi{2}\bphi{1} Q_2}{\mel{\varphi_1}{Q_2}{\varphi_1}} + K_T Q_{12},
\end{gather}
where $K_T$ and $K_S$ are arbitrary matrices of size $(\nu-\mu)\times\mu$, and $Q_{12}$ is a joint transverse projector on $\kphi{1,2}$. There are no restrictions on the entries of $A, B, W$ and $\kphi{1,2}$---they are all free parameters---in contrast to the generic $U=2$ FB construction. Therefore, the number of free parameters is $(\nu-\mu)(2\nu+\mu-2) - 1$ (see Appendix~\ref{app:chiral-inv-eig-prob} for details). The above solution fails for $\mel{\varphi_2}{Q_1}{\varphi_2} = \mel{\varphi_1}{Q_2}{\varphi_1}\equiv 0$; as a result, $\kphi{2}\propto\kphi{1}$, and the CLS and the FB are of class $U=1$. 

Figure~\ref{fig:u2_bipartite} shows an example of a bipartite lattice with $\nu=4$. There are two sites in the unit cell of each sublattice, and $B\ne0,\ W\ne0$. In this example, the parameters $\vphi_2,\vphi_2,\ A, B, W$ are arbitrarily chosen, and $K_T=0,\ K_S=0$ (see details in Appendix~\ref{app:u2-examples}).

\begin{figure}[htb!]
    \centering
    \includegraphics[width=0.5\columnwidth]{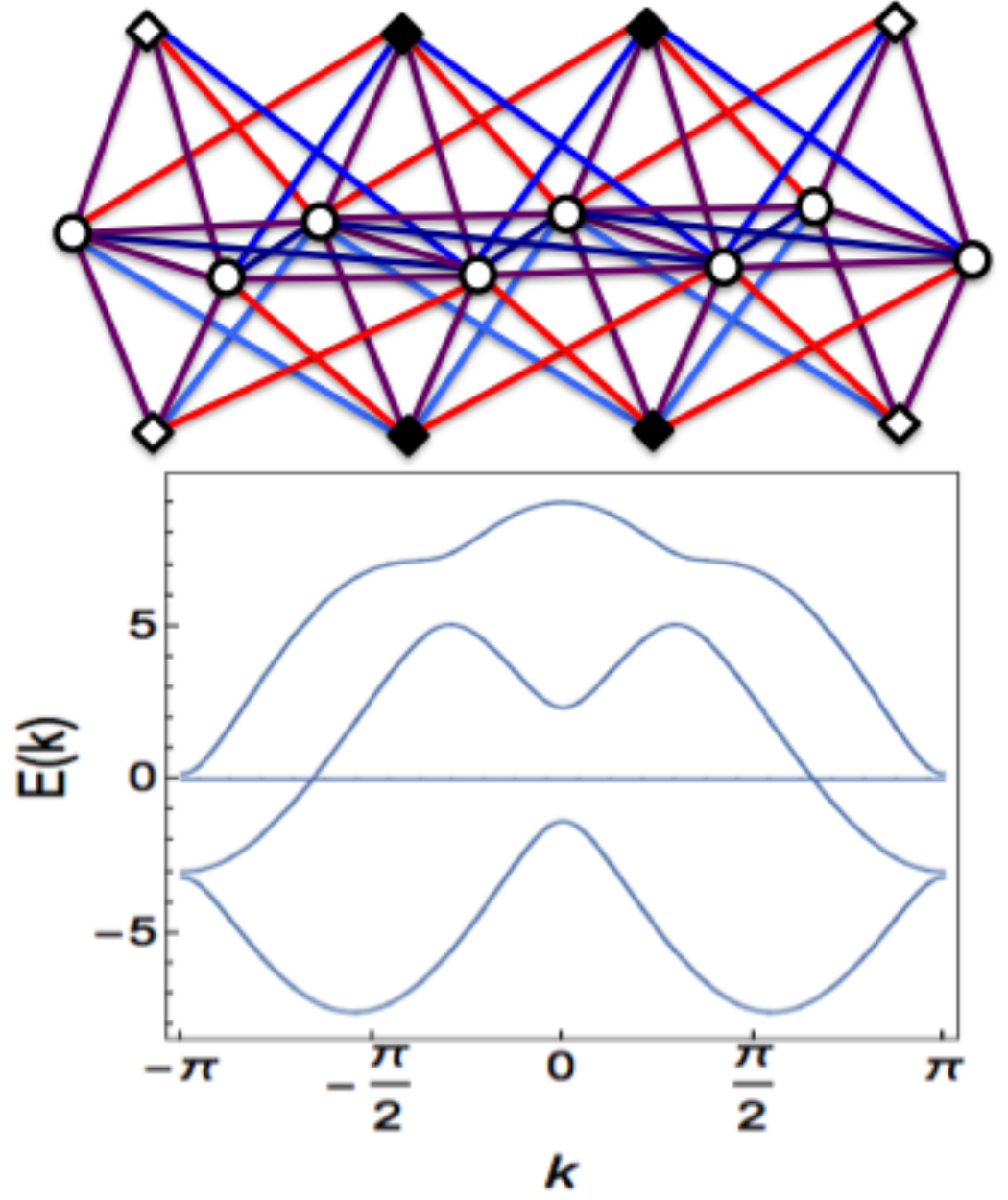}
    \caption[Example of 1D bipartite flatband lattice]{ Example of a bipartite FB Hamiltonian with $U=2$, $\nu=4$. The sites of the CLS are indicated with filled black squares. Links are colored differently for the convenience of visualisation of the chain. In this example, chiral symmetry is broken on the minority sublattice due to the presence of $B\ne0,\ W\ne0$ in Eq.~\eqref{eq:H01-psi-cs-def}. Nevertheless, the chiral FB is preserved. Details are given in Appendix~\ref{app:u2-examples}.}
    \label{fig:u2_bipartite}
\end{figure}

\subsection{Network constraints}
\label{sec:lattices}

For practical purposes,FB fine-tuning of a Hamiltonian network can involve additional network constraints, e.g. the strict vanishing of certain hopping terms between specific sites of the network~\cite{poli2017partial}. This typically happens when arranging network sites in a plane. Let us consider the typical problem of finding a nearest neighbor FB Hamiltonian with specific network constraints. These network constraints dictate the locations of the zero entries in $H_0$ and $H_1$. They can be incorporated into the matrix $T$ of Eq.~\eqref{eq:cls-ieig-U3-H1-linear} as a mask $M$: $T \to TM$ that ensures the zero entries in $H_1$ are in the right positions. The solution of the resulting system is then searched for similarities to the non-constrained case.

Especially when  $H_0$ and $H_1$ are sparse, e.g. the number of variables in $H_1$ is equal to or greater than the number of equations, it is possible to solve Eqs. (\ref{eig-1}--\ref{eig-5}) analytically (see Appendix~\ref{app:imposing_lat_str}). Figure~\ref{fig:imposing_lat_str_examples} shows examples of networks with FBs generated for a 1D kagome chain and for chains with hoppings allowed only inside network plaquettes.

\begin{figure}
    \centering
    \subfloat[]{
        \includegraphics[scale=0.37]{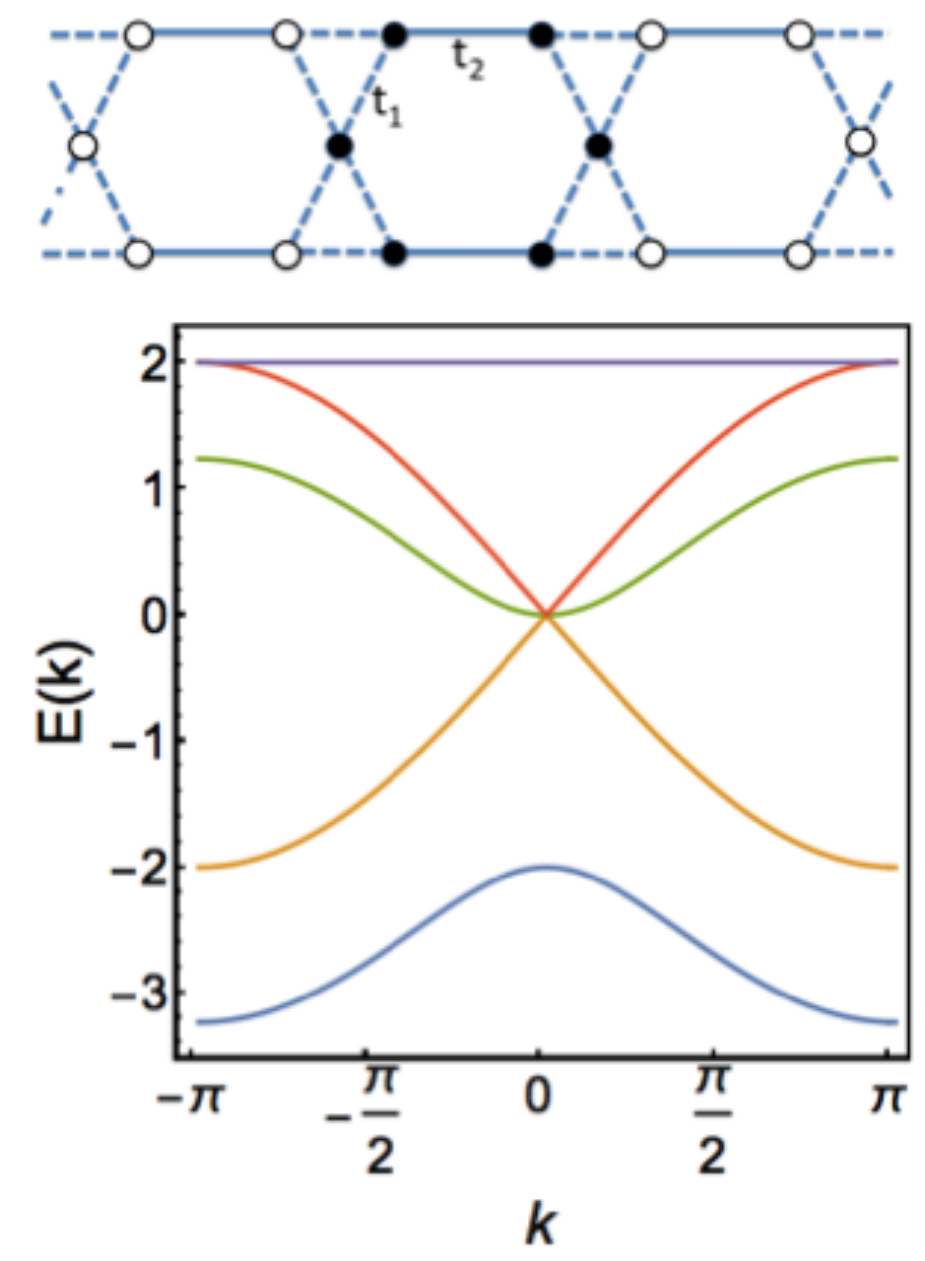}
        \label{fig:1d_kagome}}
    \subfloat[]{
        \includegraphics[scale=0.42]{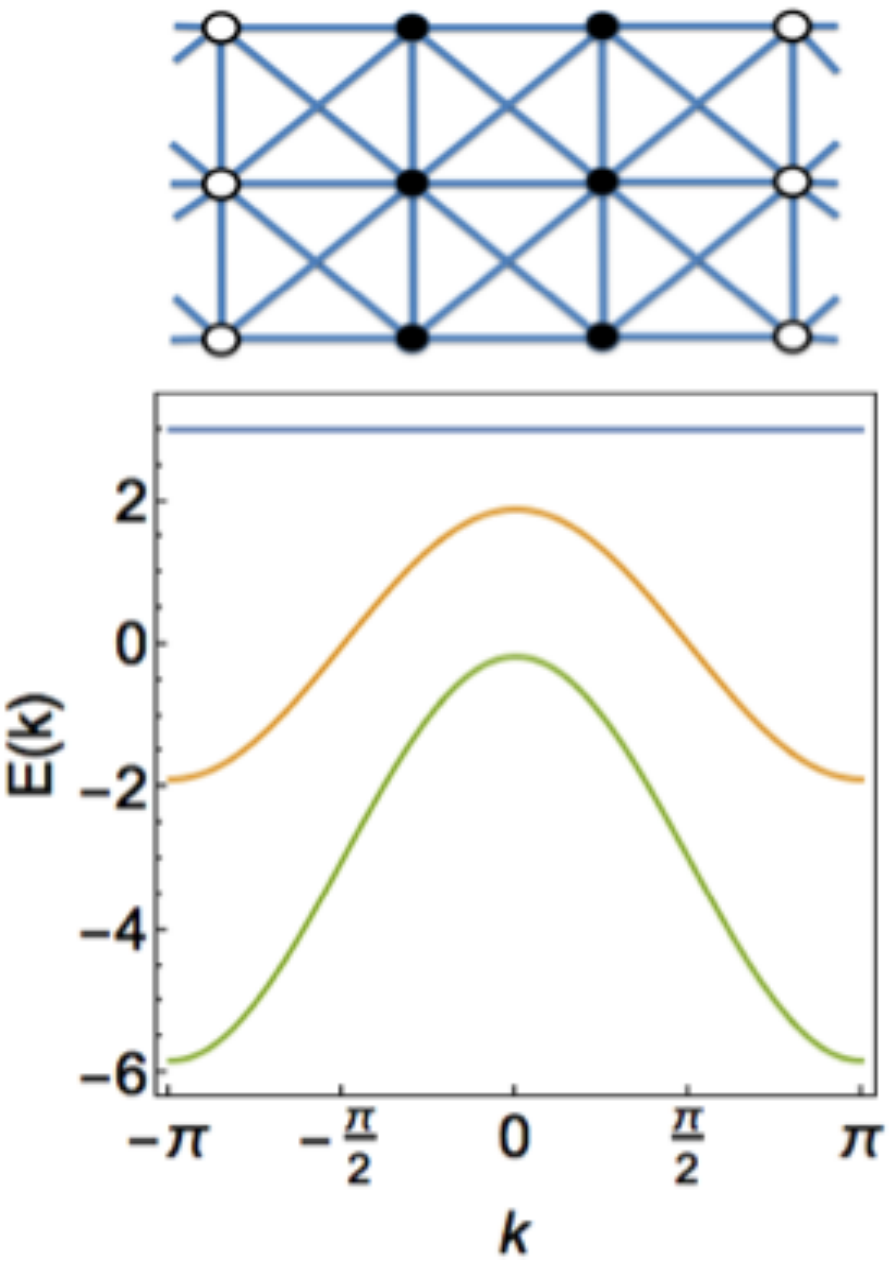}
        \label{fig:u2_nu3_latt_str_example}}
    \subfloat[]{
        \includegraphics[scale=0.42]{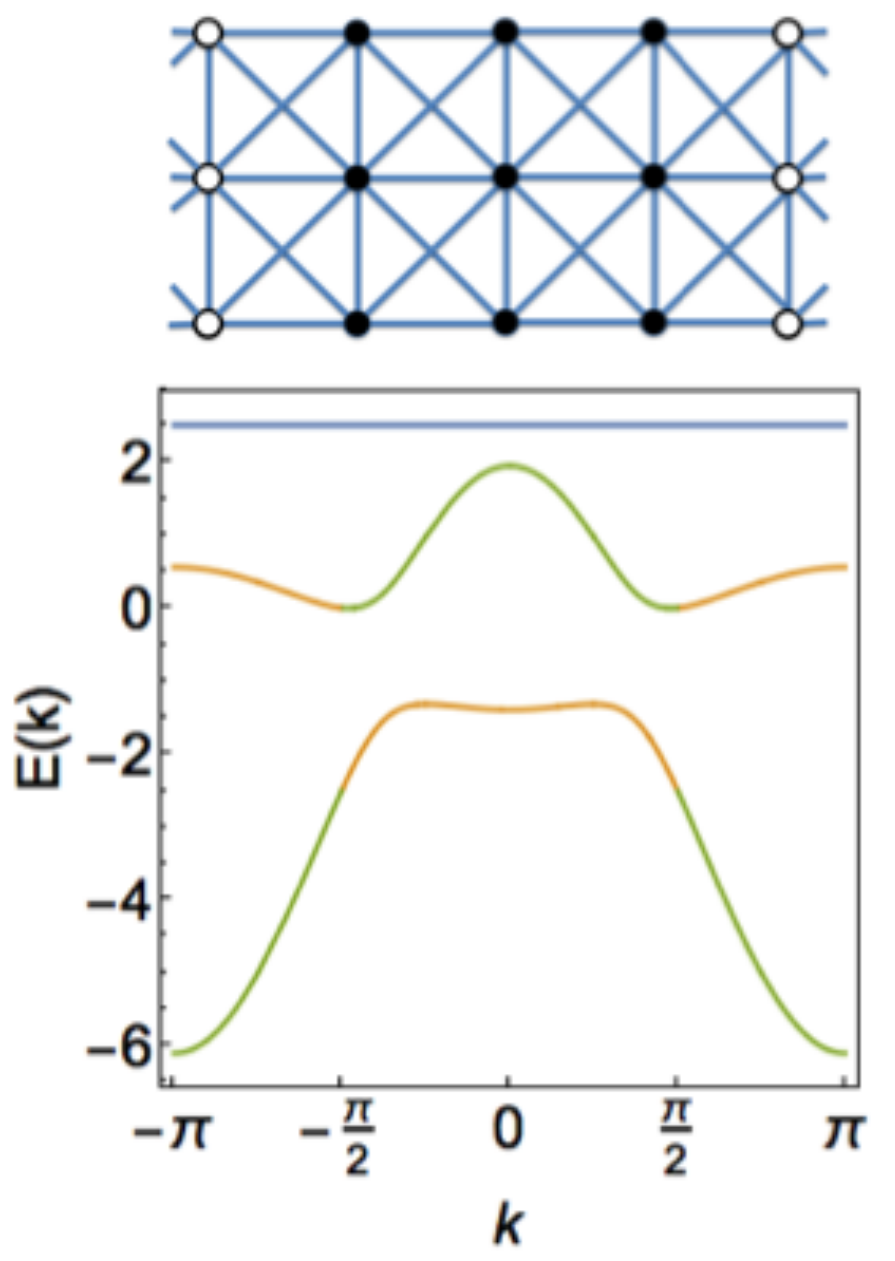}
        \label{fig:u3_nu3_latt_str_example}}
    \caption[Examples of flatband construction under network constraints]{ Examples of FB Hamiltonians constructed in specific networks. The sites occupied by a CLS are marked by black filled circles. (a) 1D kagome with $\nu=5$ and $U=2$ CLSs. The crossing of the three bands indicates that the Hamiltonian can be detangled into two independent sub-Hamiltonians. (b) and (c) Examples of $\nu=3$ Hamiltonians with $U=2$ and $U=3$ CLSs, respectively. The details of all these Hamiltonians are provided in Appendix~\ref{app:imposing_lat_str}.}
    \label{fig:imposing_lat_str_examples}
\end{figure}

\section{Summary}
\label{section4.5}

In this chapter, we introduced a novel flatband generator for 1D translational invariant tight-binding networks. First, we solved a two-band problem to obtain a family of FB Hamiltonians parameterized by a two-parameter family, where we found that the largest possible size of the irreducible CLS for this case is $U=2$. Particularly, we showed that by using unitary transformations, all FB networks with two bands can be mapped into generalized sawtooth chains, and we found a highly symmetric ST2 chain. 

Extending the above methods to an arbitrary number of bands, we presented a systematic construction of 1D Hamiltonians with $\nu>2$ bands including one FB for an arbitrary size CLS $U \leq \nu$, and illustrated the method with several examples. The task of finding FB Hamiltonians is reduced to solving specific inverse eigenvalue problems subject to certain non-linear constraints. Flatband energy enters as a parameter and can be tuned. For the $U=2$ case, we derived analytical solutions to the inverse eigenvalue problem supplemented with a numerical algorithm to resolve the constraints. For $U \geq 3$, analytical solutions are not accessible, yet numerical algorithms can be applied to generate FB Hamiltonians. We illustrated the method by generating several $U=3$ FB Hamiltonians. The same construction allows us to incorporate various network geometry constraints into the search algorithm. Our results show that FB Hamiltonians, while being the result of fine-tuning in the space of all tight-binding Hamiltonian networks, allow for a surprisingly large number of free parameters, that change the network but leave the flatness of the FB untouched. In the next chapter, we extend our 1D FB generator to two dimensions.
\chapter{Flatband generator in two dimensions}
\label{chapter5}

In the previous chapter, we introduced a flatband (FB) generator in 1D. There are more interesting phenomena though in 2D FB lattices that have been reported theoretically and observed experimentally. Therefore, it is highly desirable to extend the 1D FB generator to 2D. The methods for 1D can be implemented in 2D, but with some added complexity. First, identifying the CLS shape is more difficult, and second, depending on CLS size and shape, the destructive interference conditions are more complicated. Last but not least, the expanded hopping range also adds more complexity to the eigenvalue problem as well as the destructive interference. Regarding these complications, we mostly restrict our attention to identifying all possible FB Hamiltonians with nearest and next nearest neighbor unit cell hoppings, with CLSs occupying a maximum of four unit cells in a $2\times2$ plaquette. 

The outline of this chapter is as follows. We start in Section \ref{section5.1} by introducing the eigenvalue problem before discussing the classification of CLSs in 2D and destructive interference conditions in Section \ref{sec:2d-cls-class}. We present the FB generator scheme for this case in Section \ref{sec:2d-fb-gen}. Results for nearest neighbor hoppings and next nearest neighbor hoppings are given in Sections \ref{section5.3} and \ref{section5.4}, respectively.

\section{The eigenvalue problem} 
\label{section5.1} 

We consider a 2D translational invariant tight-binding network with nearest and next nearest neighbor hoppings and $\nu$ sites per unit cell. We use block matrix representation with the same notations and conventions from Section \ref{section3.3}. The eigenvalue problem here reads
\begin{equation}
     H_0 \vec{\psi}_{n} + \sum_{\chi} H_{\chi}^\dagger \vec{\psi}_{l_{\chi}^\prime} + \sum_{\chi} H_{\chi} \vec{\psi}_{l_{\chi}}=E\vec{\psi}_n, \quad n\in \mathcal{Z} \;,
     \label{eq:2d-nn-mat-rep}
\end{equation} 
where $\nu$-component vector $\vec{\psi}_n$ is the wave function of the $n$th unit cell, $H_{\chi}$ is the nearest neighbor hopping matrix for the $\chi$th direction, and $l_{\chi}$ and $l_{\chi}^\prime$ are the indices of the nearest neighboring unit cells along the $\chi$th direction. 
Note that we refer to neighboring unit cells as neighbors, which is different from conventional notation. In this notation, the nearest neighbors are the nearest neighboring unit cells along primitive lattice translation vectors, and the next nearest neighbors are the nearest neighbors along the diagonal (non-primitive) directions. More precisely, when we consider only nearest neighbors, $\chi =1,2$ in Eq. \eqref{eq:2d-nn-mat-rep} to represent two directions along the two primitive lattice translation vectors. In the case of next nearest neighbor hoppings, $\chi=1,2,3$ in Eq. \eqref{eq:2d-nn-mat-rep}, where $\chi=3$ is the third hopping direction, and the neighbors along this direction become the next nearest neighbors (see the example in Fig. \ref{fig:sqare-vs-triang}). 
\begin{figure}[htb!]
    \centering
    \includegraphics[width=0.5\linewidth]{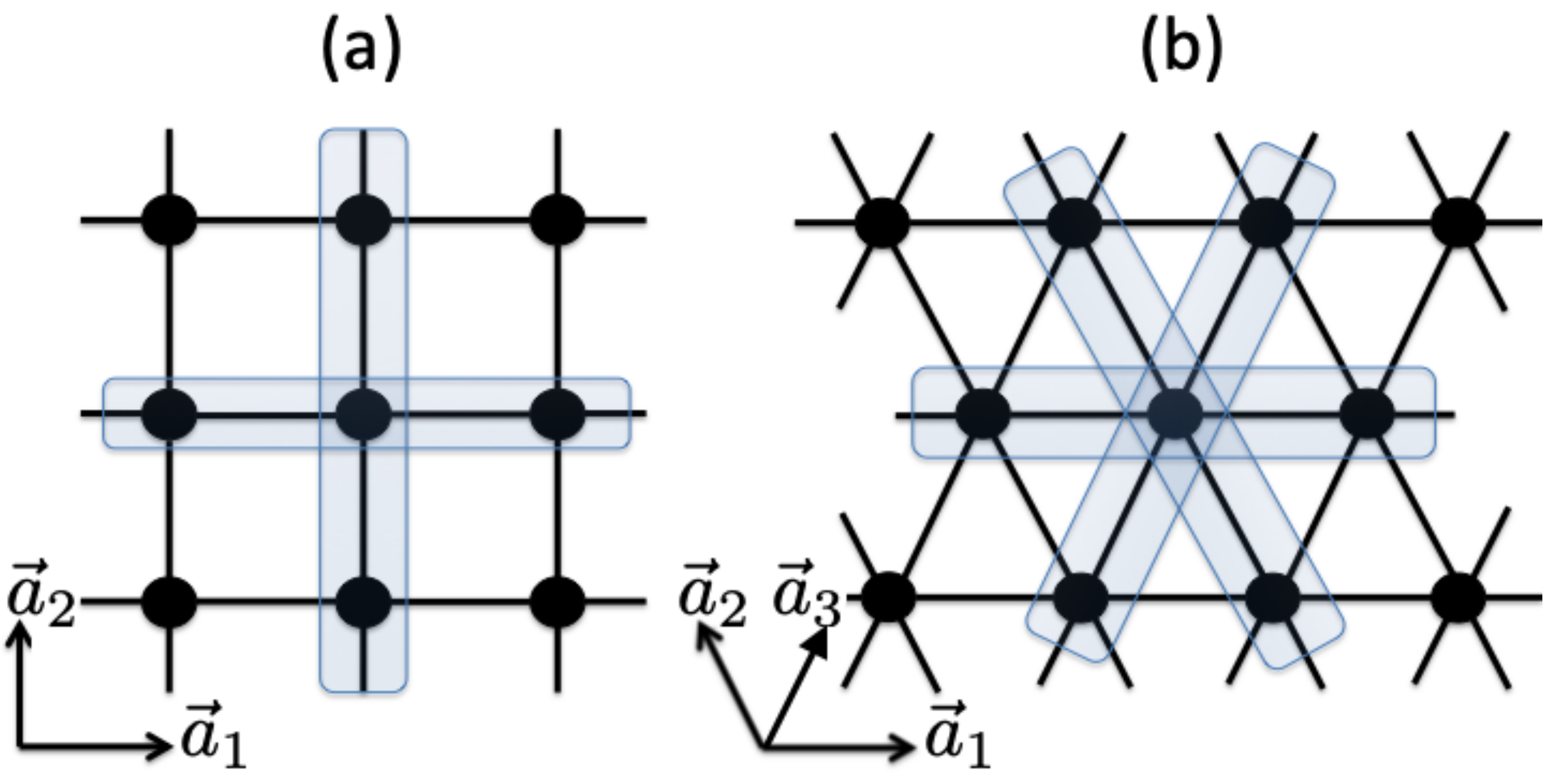}
    \caption[Simple example of nearest neighbors and next nearest neighbors in 2D]{Example of nearest neighbors and next nearest neighbors in simple 2D lattices with one site per unit cell, where $\vec{a}_1,\ \vec{a}_2$ are primitive translation vectors and $\vec{a}_3 = \vec{a}_1 + \vec{a}_2$. (a) Square lattice highlighting only nearest neighbors along $\vec{a}_1,\ \vec{a}_2$. (b) Hexagonal lattice highlighting two next nearest neighbors along $\vec{a}_3$.}
    \label{fig:sqare-vs-triang}
\end{figure}

\section{Classification of CLSs in 2D}
\label{sec:2d-cls-class}

As discussed in Section \ref{section3.1}, we classify a 2D CLS by shape vector $\mathbf{U}=(U_1, U_2)$, where $U_1,\ U_2$ are the range of a CLS along two primitive translation vectors $\vec{a}_1,\  \vec{a}_2$. However, as shown in Fig. \ref{fig:2d-known-eg-u-class}, not all of the $U_1  U_2$ unit cells in the $U_1 \times U_2$ plaquette are occupied. Especially, when $U_1,\ U_2$ are large, the CLS can take complicated shapes. In such cases, we need more parameters to specify which sites are occupied in the $U_1 \times U_2$ plaquette. Then, $\mathbf{U}$ could be a $U_1 \times U_2$ matrix composed of $1$ or $0$, representing occupied or empty unit cells in such a plaquette. However, when $U_1,\ U_2$ are small, all possible shapes can be identified by just  introducing extra parameter $s$ to count the number of unoccupied unit cells. In this case, the shape vector reads $\mathbf{U}=(U_1, U_2, s)$. Then the CLS size $U$, i.e. the number of unit cells occupied by the irreducible CLS, is given by $U=U_1 U_2 -s$. To the best of our knowledge, most known examples fall into the $\mathbf{U}=(U_1=2,1\le U_2 \le 2,1\le s \le 2)$ classes (see Fig. \ref{fig:2d-known-eg-u-class}), and thus we limit our focus to these classes in this thesis.

\begin{figure}[htb!]
    \centering
    \includegraphics[width=0.6\linewidth]{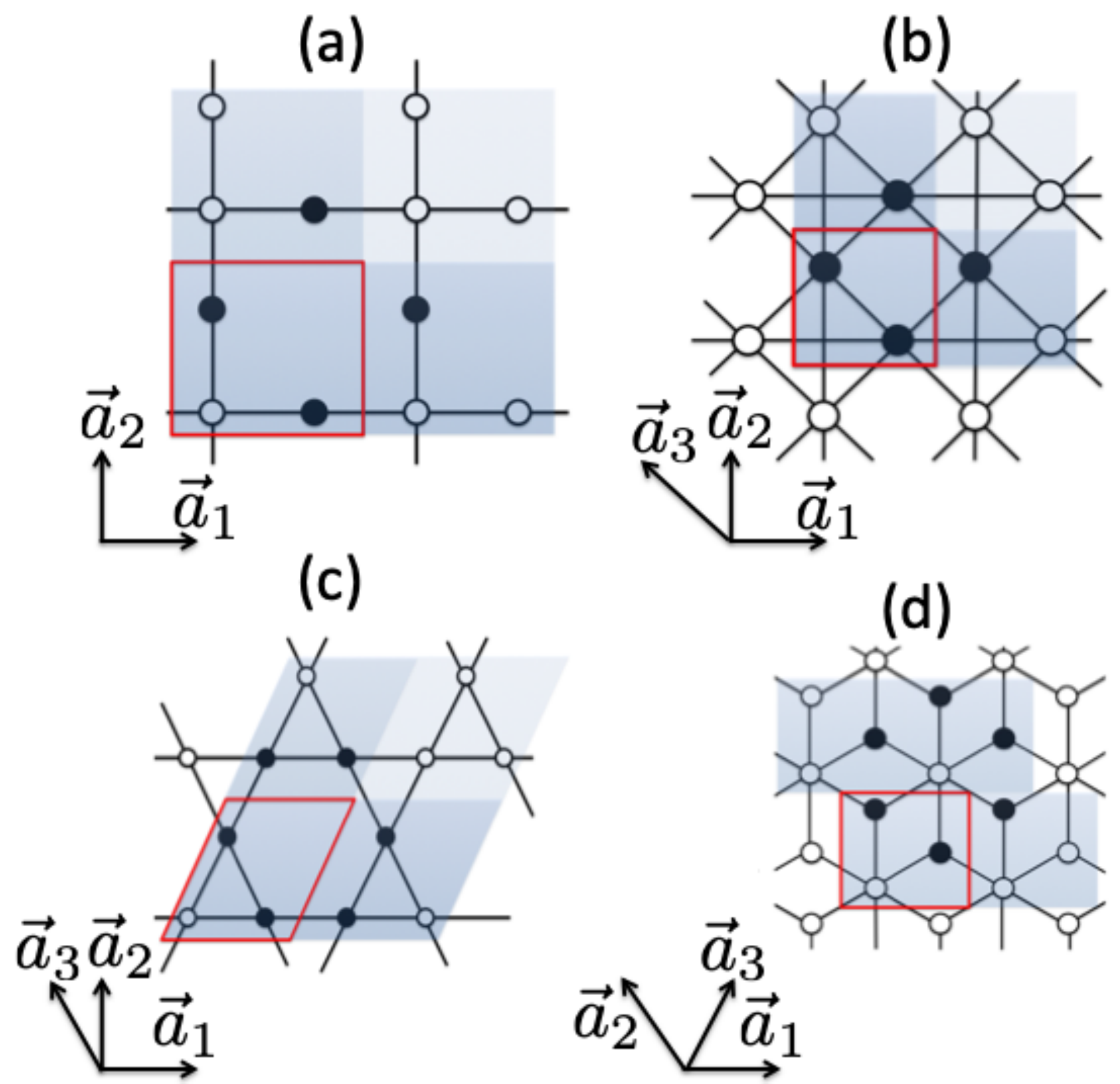}
    \caption[The $\mathbf{U}$-classification of various 2D FB lattices]{The $\mathbf{U}$-classification of four 2D FB lattices. Shaded areas give the range of the CLS (i.e. the plaquette), and the darker shaded regions show the occupied unit cells (black circles). Red boxes denote the unit cell, and $\vec{a}_1,\ \vec{a}_2$ are primitive translation vectors, and $\vec{a}_{3}=\vec{a}_1 \pm \vec{a}_2$.  (a) Lieb: $\mathbf{U}=(2,2,1)$, (b) checkerboard: $\mathbf{U}=(2,2,1)$, (c) kagome: $\mathbf{U}=(2,2,1)$, and (d) dice: $\mathbf{U}=(2,2,0)$. }
    \label{fig:2d-known-eg-u-class}
\end{figure}

Under this limitation, the wave functions of a maximum four unit cells are relevant in Eq. \eqref{eq:2d-nn-mat-rep} (i.e. all other wave functions are zero), and we label the CLS components as $\vec{\psi}_{i=1,\dots,4}$.  We also use bra-ket notations \emph{$\vert \psi_i \rangle$ and $\vec{\psi}_i$} interchangeably throughout this chapter. Then, with these notations, we can write the eigenvalue problem and destructive interference conditions for this limited case of $\mathbf{U}=(U_1=2,1\le U_2 \le 2,1\le s \le 2)$. 

In the case of nearest neighbor hoppings, hopping matrices $H_1,\ H_2$ describe the hoppings along primitive lattice translation vectors $\vec{a}_1,\ \vec{a}_2$, respectively. There are several possible shapes along with hoppings, as illustrated in Fig. \ref{fig:u22-cls-configs}. We can write the eigenvalue problem and destructive interference conditions as follows:

\begin{equation}
    \begin{aligned}
        H_1 \vec{\psi}_2 + H_2 \vec{\psi_3} \delta_{U_2,2} &= ( E_{FB} - H_0 ) \vec{\psi}_1 , \\
        H_1^\dagger \vec{\psi}_1 + H_2 \vec{\psi_4} \delta_{U_2,2} \delta_{s,0} &= ( E_{FB} - H_0 ) \vec{\psi}_2 , \\ 
        \left( H_1 \vec{\psi}_4 \delta_{s,0} + H_2^\dagger \vec{\psi}_1 \right) \delta_{U_2,2}  &= ( E_{FB} - H_0 ) \vec{\psi}_3 \delta_{U_2,2} , \\ 
        \left( H_1^\dagger \vec{\psi}_3 \delta_{s,0} + H_2^\dagger \delta_{s,0} \right) \delta_{U_2,2}  &= ( E_{FB} - H_0 ) \vec{\psi}_4 \delta_{U_2,2} \delta_{s,0}, \\
        H_1 \vec{\psi}_1 = H_1^\dagger \vec{\psi}_2 &= 0, \\
        H_2 \vec{\psi}_1 = H_2 \vec{\psi}_2 &= 0, \\
        H_1 \vec{\psi}_3 \delta_{U_2,2} = H_2^\dagger \vec{\psi}_3 \delta_{U_2,2}  &= 0 , \\
        H_1 \vec{\psi}_3 \delta_{U_2,2} \delta_{s,1} + H_2^\dagger \vec{\psi}_2 &= 0 , \\
        H_1^\dagger \vec{\psi}_4 \delta_{U_2,2} \delta_{s,0} = H_2^\dagger \vec{\psi}_4 \delta_{U_2,2} \delta_{s,0} &= 0 \;.
    \end{aligned}
    \label{eq:u22-gen-eig-prob}
\end{equation} 
 Putting the corresponding values of $U_1,\ U_2,\ s$, we get the eigenvalue problem and destructive interference conditions corresponding to CLS class $\mathbf{U}=(U_1,U_2,s)$. 

\begin{figure}[htb!]
    \centering
    \includegraphics[width=0.6\linewidth]{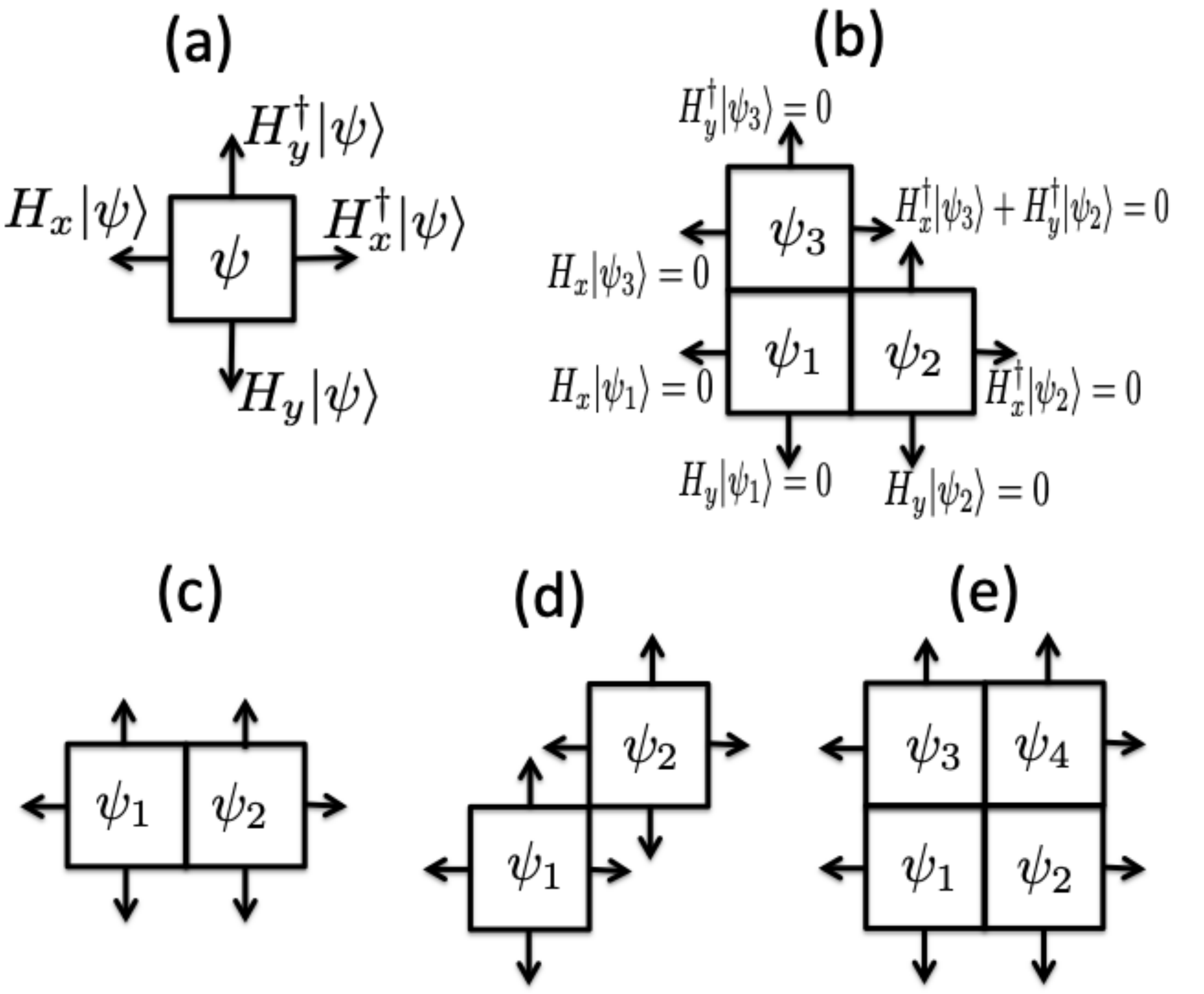}
    \caption[Classification of different CLSs and destructive interference in 2D]{Classification of $U=(2,U_2,s)$ CLSs and destructive interference. Each square represents a unit cell. (a) Single unit cell and hopping directions represented by arrows. (b) U=(2,2,1) case. (c)  U=(2,1,0) case. (d) U=(2,2,2) case. (e) U=(2,2,0) case. In (b)--(e), the arrows represent destructive interference. When two arrows cross, waves from two different unit cells interfere destructively. }
    \label{fig:u22-cls-configs} 
\end{figure}  

Many known FB lattices, such as checkerboard, kagome, and dice lattices, involve diagonal hoppings, i.e. next nearest neighbor hoppings, which are important to consider. In this case, there are two different situations: one when hoppings along both diagonal directions are non-zero, as shown in Fig. \ref{fig:u22-nnn-config} (a), and one when hoppings along only one diagonal direction are non-zero, as shown in Fig. \ref{fig:u22-nnn-config} (b). For simplicity, we only consider the latter case, as the most of the known examples fall into this category.

\begin{figure}[htb!]
    \centering
    \includegraphics[width=0.6\columnwidth]{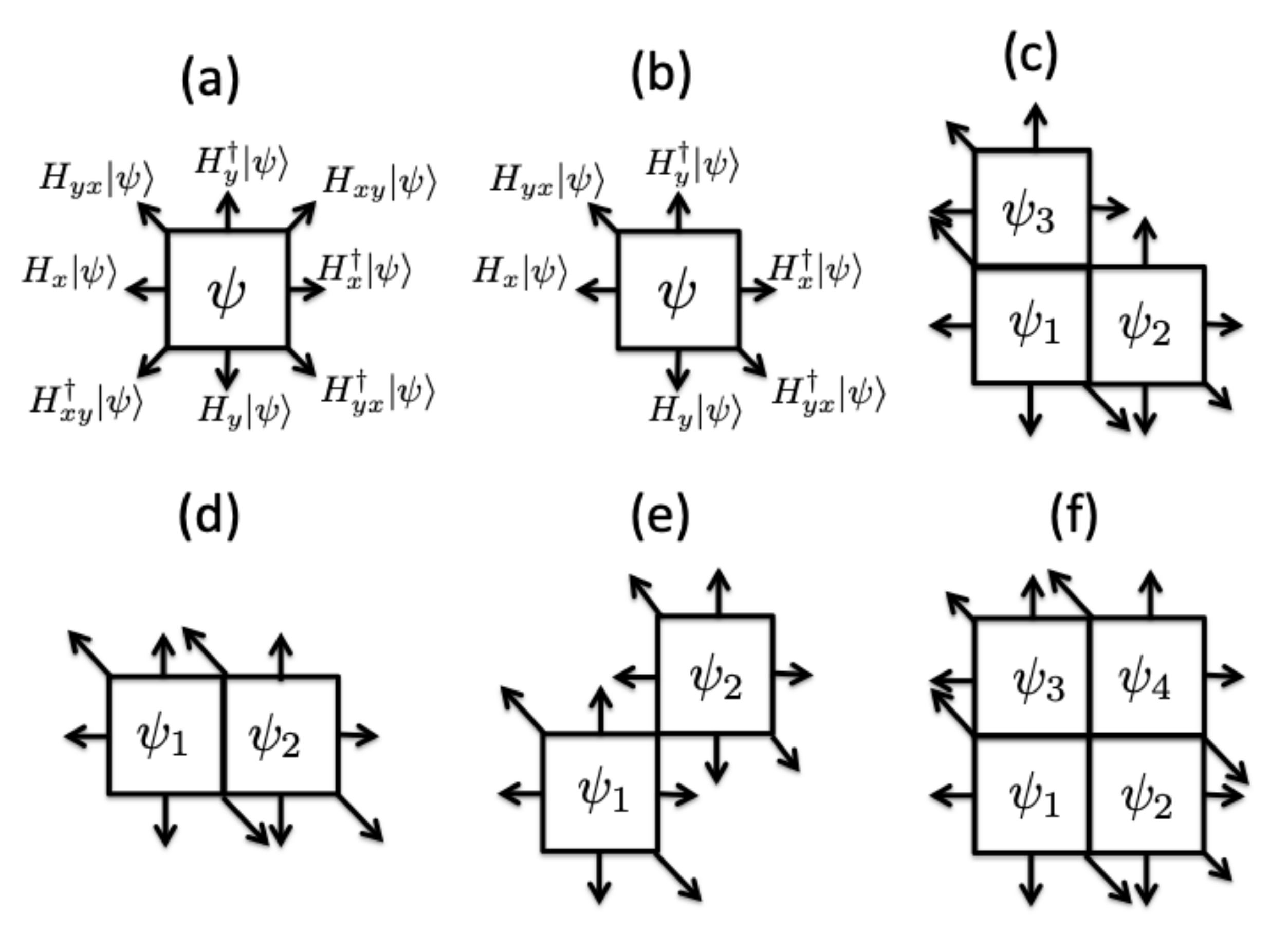}
    \caption[Classification of CLSs in 2D lattices with next nearest neighbor hoppings]{Different configurations of $U=(2,U_2,s)$ CLS classes for different $s$ in a 2D lattice with next nearest neighbor hoppings. (a) Single unit cell with hoppings along both diagonal directions. (b) Single unit cell with hoppings along one diagonal direction. In (a) and (b), arrows indicate the the hopping directions. (c)--(f) Different cases with hoppings along one diagonal direction, where the arrows represent destructive interference, i.e. the waves hopping along the arrow direction become zero. When two arrows cross, waves from two different unit cells interfere destructively. (c) U=(2,2,1),  (d) U=(2,1,0), (e) U=(2,2,2), and (f) U=(2,2,0).}
    \label{fig:u22-nnn-config}
\end{figure} 

We still consider $\mathbf{U}=(U_1=2,1 \le U_2 \le 2, 0 \le s \le 1)$, but now introduce a new hopping matrix $H_3$ to describe next nearest hoppings, i.e. the hoppings between nearest neighbors along the diagonal direction, see Fig. \ref{fig:u22-nnn-config} (b). Then the eigenvalue problem and destruction interference conditions become: 
\begin{equation}
    \begin{aligned}
        H_1 \vec{\psi}_2 + H_2 \vec{\psi}_3 \delta_{U_2,2} &= (E_{FB} - H_0) \vec{\psi}_1 , \\
        H_1^\dagger \vec{\psi}_1 + H_2 \vec{\psi}_4 \delta_{U_2,2} \delta_{s,0} + H_3^\dagger \vec{\psi}_3 \delta_{U_2,2} &= (E_{FB} - H_0) \vec{\psi}_2, \\
        \left( H_1 \vec{\psi}_4 \delta_{s,0} + H_2^\dagger \vec{\psi}_1 + H_3 \vec{\psi}_2 \right) \delta_{U_2,2} &= (E_{FB} - H_0) \vec{\psi}_3 \delta_{U_2,2}, \\
        \left( H_1^\dagger \vec{\psi}_3 + H_2^\dagger \vec{\psi}_2 \right) \delta_{U_2,2} \delta_{s,0} &=(E_{FB} - H_0) \vec{\psi}_4 \delta_{U_2,2} \delta_{s,0}, \\
        H_1 \vec{\psi}_1 = H_2 \vec{\psi}_1 &= 0, \\
        H_1^\dagger \vec{\psi}_4 \delta_{U_2,2} \delta_{s,0} = H_2 \vec{\psi}_4  \delta_{U_2,2} \delta_{s,0} &= 0,\\
        H_3^\dagger \vec{\psi}_2 = H_3 \vec{\psi}_3 \delta_{U_2,2} &= 0, \\
        H_2 \vec{\psi}_2 + H_3^\dagger \vec{\psi}_1 &= 0, \\
        H_1 \vec{\psi}_3 \delta_{U_2,2} + H_3 \vec{\psi}_1 \delta_{U_2,2} &= 0, \\
        H_2^\dagger \vec{\psi}_3 \delta_{U_2,2} + H_3 \vec{\psi}_4 \delta_{U_2,2} \delta_{s,0} &= 0, \\
        H_1^\dagger \vec{\psi}_2 + H_3 \vec{\psi}_4 \delta_{U_2,2} \delta_{s,0} &= 0 \; .
    \end{aligned}
    \label{eq:u22-nnn-eig-prob}
\end{equation} 
Putting corresponding values of $U_1,\ U_2,\ s$ into the above equation, we get the eigenvalue problem and destructive interference conditions corresponding to CLS class $\mathbf{U}=(U_1,U_2,s)$ with next nearest neighbor hoppings.

\section{The flatband generator} 
\label{sec:2d-fb-gen}

According to Definition \ref{def:FB-gen} in Section \ref{section3.4}, a 2D FB generator for a class $\mathbf{U}$ FB is a scheme to generate the set of all possible $\vec{\psi}_{i=1,\dots,U}$ and matrices $H_0,\ H_1,\ H_2$ that satisfy the eigenvalue problem and destructive interference conditions as from Eq. \eqref{eq:u22-gen-eig-prob} or Eq. \eqref{eq:u22-nnn-eig-prob}. 

As we did in 1D, here we want to generate all possible $H_1,\ H_2$ (and $H_3$ in case of next nearest neighbor hoppings) that satisfy Eq. \eqref{eq:u22-gen-eig-prob} (Eq. \eqref{eq:u22-nnn-eig-prob}) for a given $H_0$ (canonical form or any Hermitian matrix) and particular choice of CLS class $\mathbf{U}$.
For a given $H_0$ and CLS, these equations become an inverse eigenvalue problem of finding $H_1,\ H_2$ (and $H_3$ for next nearest neighbor hoppings). Since these equations depend on the shape of the CLS, i.e. $U_2,\ s$, we solve them case by case. In general, our generator works as follows.
\begin{enumerate}
    \item Choose a shape vector $\mathbf{U}=(2,1 \le U_2\le 2,0\le s \le1)$.
    \item Choose a hopping range, either nearest neighbor or next nearest neighbor.
    \item Write down corresponding form of Eq. \eqref{eq:u22-gen-eig-prob} or Eq. \eqref{eq:u22-nnn-eig-prob} for the chosen $\mathbf{U}$. 
    \item Choose an arbitrary (or a particular) $H_0$ that is either canonical or any Hermitian matrix.
    \item Choose a real $E_{FB}$. 
    \item Choose an arbitrary $\vec{\psi}_1$ or $\vec{\psi}_U$.
    \item Exclude $H_1,\ H_2$ from the equations obtained in step 3 to get non-linear constraints on remaining CLS components $\vec{\psi}_i$, and solve these constraints to find the remaining CLS components.
    \item With the chosen $H_0$ and CLS obtained from the previous step, solve the equations obtained in step 3. 
\end{enumerate} 
Below we present our results case by case. First, we note the $\mathbf{U}=(2,2,2)$ (see Fig. \ref{fig:u22-cls-configs} (d)). The eigenvalue problem for this case reduces to $H_0 \vec{\psi}_{i=1,2}=E_{FB} \vec{\psi}_{i=1,2}$, because the nearest neighbors of $\vec{\psi}_{i=1,2}$ are all zero. This leads to one of two possibilities: either a $U=1$ case or two degenerate FB with the same $E_{FB}$. A generic way to construct $U=1$ FBs is presented in Section \ref{section2.4.2}, and the case of a degenerate FB is not our focus in this thesis. Therefore, we do not tackle this $\mathbf{U}=(2,2,2)$ case further in this chapter.

\section{Nearest neighbor hoppings}
\label{section5.3}

Consider the two-band case, where $H_0,\ H_1,\ H_2$ are $2\times2$ matrices, and at least one of them has to have two zero modes. This results in one of three possibilities: a $U=1$ solution, isolated 1D chains, or isolated unit cells (see Appendix \ref{app:two-band-prob-2d}). This leads to the following theorem. 

\begin{theorem}
    In a 2D two-band translational invariant tight-binding network with nearest neighbor hoppings, the only possible CLS class is $U=1$.
\end{theorem}

Therefore, from now on we consider more than two bands, i.e. $\nu \ge 3$ cases.

\subsubsection{U=(2,1,0) case} 

A schematic of the CLS and destructive interference conditions for this case is shown in Fig. \ref{fig:u22-cls-configs} (c). Here, $U_2=1,\ s=0$, and Eq. \eqref{eq:u22-gen-eig-prob} reduces to a 1D form that yields the following solution (see Appendix \ref{app:nn-u21}):
\begin{equation}
    \begin{aligned}
       H_1 &= \frac{(E_{FB} - H_0)\vert \psi_1 \rangle \langle \psi_2 \vert (E_{FB} - H_0)}{\langle \psi_1 \vert (E_{FB} - H_0)\vert \psi_1 \rangle}, \\
       H_2 &= Q_{12} M Q_{12},
    \end{aligned}
    \label{eq:u210-sol}
\end{equation}
where $Q_{1}$ and $Q_2$ are transverse projectors of $\vert \psi_1 \rangle , \vert \psi_2 \rangle$, i.e. $Q_{1,2} \vert \psi_{i=1,2} \rangle =0$ (see Section \ref{section4.4}, and Appendix \ref{app:inv-egv-u2}).  
For given $H_0,\ E_{FB}$, the CLS components $\vec{\psi}_1,\vec{\psi}_2$ have to be chosen respecting the following constraints,
\begin{equation}
    \begin{aligned}
        \langle \psi_2 \vert (E_{FB} - H_0) \vert \psi_1 \rangle &= 0, \\
        \langle \psi_1 \vert (E_{FB} - H_0) \vert \psi_1 \rangle &= \langle \psi_2 \vert (E_{FB} - H_0) \vert \psi_2 \rangle \;.
    \end{aligned}
    \label{eq:u21-cls-cons}
\end{equation} 

For given $H_0,\ E_{FB}, \vert \psi_1 \rangle$ and $\nu \ge 3$ there are $\nu-3$ free variables in $H_1$ (see Section \ref{section4.4}), and $(\nu-2)^2$ free variables in $H_2$. Therefore, there are in total $(\nu-2)^2 + \nu -3=\nu^2 -3\nu + 1$ variables in the solution to Eq. \eqref{eq:u210-sol}. If we only fix $H_0$, then $E_{FB}, \vert \psi_1 \rangle$ contribute other $\nu + 1$ free parameters, which makes the total number of free parameters $\nu^2 -3\nu + 1 + \nu +1 = \nu^2 - 2(\nu-1)$. 

\subsubsection{U=(2,2,1) case} 

A schematic of the CLS and destructive interference conditions for this case is shown in Fig. \ref{fig:u22-cls-configs} (b). In this case $U_2=2,\ s=1$, and our generator gives the following solution for the $\nu>3$ case (see details in Appendix \ref{app:nn-u221}):

\begin{equation}
     \begin{aligned}
         H_{1}&=\frac{1}{\langle\psi_{3}\vert Q_{1,2}\vert a\rangle}Q_{1,2}\vert a\rangle\langle z\vert Q_{1,2,3} +\frac{\left(E_{FB}-H_{0}\right)\vert\psi_{1}\rangle \langle\psi_{2}\vert\left(E_{FB}-H_{0}\right)}{\langle\psi_{2}\vert\left(E_{FB}-H_{0}\right)\vert\psi_{2}\rangle}\\ 
         & -\frac{\langle\psi_{3}\vert\left(E_{FB}-H_{0}\right)\vert\psi_{3}\rangle}{\langle\psi_{2}\vert\left(E_{FB}-H_{0}\right)\vert\psi_{2}\rangle \langle\psi_{1}\vert Q_{2,3}\vert y\rangle}Q_{2,3}\vert y\rangle \langle\psi_{2}\vert\left(E_{FB}-H_{0}\right) \\
         H_{2}&=\frac{1}{\langle\psi_{1}\vert Q_{2,3}\vert y\rangle}Q_{2,3}\vert y\rangle\langle\psi_{3}\vert(E_{FB}-H_{0}) \\ & +\frac{1}{\langle\psi_{2}\vert Q_{3}\vert u\rangle}\left(\frac{\langle\psi_{1}\vert Q_{3}\vert u\rangle}{\langle\psi_{1}\vert Q_{2,3}\vert y\rangle}Q_{2,3}\vert y\rangle-Q_{3}\vert u\rangle\right)\langle z\vert Q_{1,2,3},
     \end{aligned}
     \label{eq:sol-u221-hxy-1}
 \end{equation} 
 where $Q_i$ are transverse projectors on $\vert \psi_i \rangle$, $Q_{2,3}$ is a transverse projector on $\vert \psi_{i=2,3} \rangle$, and $Q_{1,2,3}$ is a transverse projector on $\vert \psi_{i=1,2,3} \rangle$ (see Section \ref{section4.4} and Appendix \ref{app:inv-egv-u2}). The $\nu$ component vectors $\vert y \rangle, \vert z \rangle, \vert a \rangle, \vert u \rangle$ are free parameters, which are defined as 
\begin{equation}
     Q_{2,3} \vert y \rangle = H_2 \vert \psi_3 \rangle = (E_{FB} -H_0) \vert \psi_1 \rangle - H_1 \vert \psi_2 \rangle, \quad \langle z \vert Q_{1,2,3} = \langle \psi_3 \vert H_1 = - \langle \psi_2 \vert H_2.  
    \label{eq:y-z-def}
\end{equation} 
The CLS components $\vert \psi_{i=1,2,3} \rangle$ must satisfy the following constraints:
\begin{equation}
  \begin{aligned}
    & \langle \vec{\psi}_2 \vert H_0 \vert \vec{\psi}_1 \rangle = E_{FB} \langle \vec{\psi}_2 \vert \vec{\psi}_1 \rangle \\
    & \langle \vec{\psi}_3 \vert H_0 \vert \vec{\psi}_1 \rangle = E_{FB} \langle \vec{\psi}_3 \vert \vec{\psi}_1 \rangle \\
    &  \langle \vec{\psi}_3 \vert H_0 \vert \vec{\psi}_2 \rangle = E_{FB} \langle \vec{\psi}_3 \vert \vec{\psi}_2 \rangle \\
    & \langle \vec{\psi}_2 \vert \left( E_{FB} - H_0 \right) \vert \vec{\psi}_2 \rangle + \langle \vec{\psi}_3 \vert \left( E_{FB} - H_0 \right) \vert \vec{\psi}_3 \rangle = \langle \vec{\psi}_1 \vert \left( E_{FB} - H_0 \right) \vert \vec{\psi}_1 \rangle.
  \end{aligned}
  \label{eq:non-lin-const-L-shape}
\end{equation} 


 When $Q_{1,2,3}\vert z \rangle = 0$, it corresponds to the case that the hoppings from $\vec{\psi}_2$ and $\vec{\psi}_3$ are zero individually, instead of the sum of the hoppings from $\vec{\psi}_2$ and $\vec{\psi}_3$ being zero. The solution of this case is given in Appendix \ref{app:u221-special}. Note that, when $Q_{1,2,3}\vert z \rangle = 0$ and $Q_{2,3} \vert y \rangle = (E_{FB} - H_0 ) \vert \psi_1 \rangle$, the solution in Eq. \eqref{eq:sol-u221-hxy-1} also reduces to the solution in Appendix \ref{app:u221-special}. 
 
 \paragraph{$\nu=3$ case:} In this case, transverse projector $Q_{1,2,3}$ on $\vert \psi_{i=1,2,3} \rangle$ is a $3 \times 3$ matrix. If $\vert \psi_{i=1,2,3} \rangle$ are linearly independent, then $Q_{1,2,3}$ becomes zero, and so $Q_{1,2,3}\vert z \rangle = 0$. The solution for this case is given in Appendix \ref{app:u221-special}, as discussed above. If we require $Q_{1,2,3}\ne 0$, and so $Q_{1,2,3} \vert z \rangle \ne 0$, then $\vert \psi_{i=1,2,3} \rangle$ must be linearly dependent as
\begin{equation}
c_{1}\vert\psi_{1}\rangle+c_{2}\vert\psi_{2}\rangle+c_{3}\vert\psi_{3}\rangle=0,
\end{equation}
which gives the following solution (see details in Appendix \ref{app:nn-u221}): 
\begin{equation}
    \begin{aligned}
        H_{1} &=\frac{Q_{2}\vert a\rangle\langle\psi_{2}\vert\left(E_{FB}-H_{0}\right)}{\langle\psi_{3}\vert Q_{2}\vert a\rangle} \\
        H_{2} &= \frac{Q_{3}\vert c\rangle\langle\psi_{3}\vert\left(E_{FB}-H_{0}\right)}{\langle\psi_{2}\vert Q_{3}\vert c\rangle},
    \end{aligned}
    \label{eq:u221-hxy-sol-3band}
\end{equation} 
where $\vert a \rangle,  \vert c \rangle$ are free parameters, and $E_{FB}, H_0, \vert \psi_2 \rangle, \vert \psi_3 \rangle $ are chosen respecting the constraints \begin{equation}
    \begin{aligned}
        \langle\psi_{2}\vert\left(E_{FB}-H_{0}\right)\vert\psi_{2}\rangle&=0\\\langle\psi_{3}\vert\left(E_{FB}-H_{0}\right)\vert\psi_{3}\rangle&=0\\\langle\psi_{3}\vert\left(E_{FB}-H_{0}\right)\vert\psi_{2}\rangle&=0\\\left(E_{FB}-H_{0}\right)\left(\alpha\vert\psi_{2}\rangle+\beta\vert\psi_{3}\rangle\right)&=0, 
    \end{aligned}
    \label{eq:u221-cls-cons-3Band}
\end{equation} 
where $\alpha, \beta$ are proportionality factors such that 
\begin{equation}
    \vert \psi_1 \rangle = \alpha \vert \psi_2 \rangle + \beta \vert \psi_3 \rangle. 
    \label{eq:nu3-lin-dep-cls-comp}
\end{equation} 
According to Conjecture \ref{conj:lin-dep-band-touch}, the CLS given by Eq. \eqref{eq:nu3-lin-dep-cls-comp} always yields band touching.


\subsubsection{U=(2,2,0) case} 

In this case, we only solve the three-band problem using a different method than above. The CLS and destructive interference conditions are illustrated in Fig. \ref{fig:u22-cls-configs} (e). Here, we use a direct parameterization to solve Eq. \eqref{eq:u22-gen-eig-prob} and obtain a lengthy analytic solution, which is given in Appendix \ref{app:nn-u220-nu3}. There are three free parameters in this solution.

\section{Next nearest neighbor hoppings} 
\label{section5.4}

Following the same procedure as for nearest neighbor hoppings except for steps 3 and 7, here we put corresponding values of $U_2,s$ into Eq. \eqref{eq:u22-nnn-eig-prob} to obtain solutions for $\mathbf{U}=(2,U_2,s)$ cases.  

\subsubsection{U=(2,1,0) case} 

The CLS and destructive interference for this case are illustrated in Fig. \ref{fig:u22-nnn-config} (d). Here, $U_2=1,s=0$, and Eq. \eqref{eq:u22-nnn-eig-prob} gives the following solution (see Appendix \ref{app:nnn-u21}):
\begin{equation}
    \begin{aligned}
        H_1 &= \frac{(E_{FB} - H_0) \vert \psi_1 \rangle \langle \psi_2 \vert (E_{FB} - H_0) }{\langle \psi_1 \vert (E_{FB} - H_0) \vert \psi_1 \rangle}, \\
       H_2 &= Q_{12} \vert c \rangle \langle d \vert Q_{12} , \\
       H_3 &= Q_{12} \vert e \rangle \langle f \vert Q_{12} \; ,
    \end{aligned}
    \label{eq:u21-diag-sol}
\end{equation} 
where $Q_{12}$ is a transverse projector of $\vec{\psi}_1, \vec{\psi}_2$, and $H_0,\ E_{FB},\ \vec{\psi}_1,\vec{\psi}_2$ are chosen respecting the following constraints 
\begin{equation}
    \begin{aligned}
       \langle \psi_2 \vert (E_{FB} - H_0) \vert \psi_1 \rangle &= 0 , \\
       \langle \psi_1 \vert (E_{FB} - H_0) \vert \psi_1 \rangle &= \langle \psi_2 \vert (E_{FB} - H_0) \vert \psi_2 \rangle \; .
    \end{aligned}
    \label{eq:u210-nnn-cls-cons}
\end{equation} 

From Eq. \eqref{eq:u21-diag-sol}, we can see that, for the two-band case, the solution gives decoupled 1D chains. Because, for the two-band case, the transverse projector becomes zero, which makes $H_2=H_3=0$. Therefore, \emph{no two-band solution exists for the $U=(2,1,0)$ case}. 

\subsubsection{U=(2,2,1) case} 

 The CLS and hoppings for this case are illustrated in Fig. \ref{fig:u22-nnn-config} (c). Here, $U_2=2,s=1$, and Eq. \eqref{eq:u22-nnn-eig-prob} yields the following solution (see details in Appendix \ref{app:nnn-u221}):
{\footnotesize
\begin{equation}
    \begin{aligned}
        H_{1}&=\frac{Q_{2}\vert z\rangle\langle x\vert Q_{1}}{\langle\psi_{3}\vert Q_{2}\vert z\rangle}+\frac{\left(\langle\psi_{3}\vert Q_{2}\vert z\rangle(E_{FB}-H_{0})\vert\psi_{1}\rangle-\langle\psi_{3}\vert Q_{2}\vert z\rangle Q_{3}\vert u\rangle-\langle x\vert Q_{1}\vert\psi_{2}\rangle Q_{2}\vert z\rangle\right)}{\langle\psi_{3}\vert Q_{2}\vert z\rangle\left( \left(\langle\psi_{1}\vert(E_{FB}-H_{0})\vert\psi_{1}\rangle-\langle\psi_{1}\vert Q_{3}\vert u\rangle\right)\langle\psi_{3}\vert Q_{2}\vert z\rangle-\langle\psi_{1}\vert Q_{2}\vert z\rangle\langle x\vert Q_{1}\vert\psi_{2}\rangle\right)} \\
        & \times \left(\langle\psi_{3}\vert Q_{2}\vert z\rangle\langle\psi_{2}\vert(E_{FB}-H_{0})-\langle\psi_{3}\vert Q_{2}\vert z\rangle\langle w\vert Q_{3}-\langle\psi_{1}\vert Q_{2}\vert z\rangle\langle x\vert Q_{1}\right), \\
        H_{2}&=-\frac{Q_{3}\vert y\rangle\langle x\vert Q_{1}}{\langle\psi_{2}\vert Q_{3}\vert y\rangle}+\frac{\left(\langle\psi_{2}\vert Q_{3}\vert y\rangle Q_{3}\vert u\rangle+\langle x\vert Q_{1}\vert\psi_{3}\rangle Q_{3}\vert y\rangle\right)}{\langle\psi_{2}\vert Q_{3}\vert y\rangle\left(\langle\psi_{2}\vert Q_{3}\vert y\rangle\langle\psi_{1}\vert Q_{3}\vert u\rangle+\langle x\vert Q_{1}\vert\psi_{3}\rangle\langle\psi_{1}\vert Q_{3}\vert y\rangle\right)} \\
        & \times \left(\langle\psi_{2}\vert Q_{3}\vert y\rangle\langle\psi_{3}\vert(E_{FB}-H_{0})-\langle\psi_{2}\vert Q_{3}\vert y\rangle\langle v\vert Q_{2}+\langle\psi_{1}\vert Q_{3}\vert y\rangle\langle x\vert Q_{1}\right), \\
        H_3&=-\frac{Q_{2}\vert z\rangle\langle y\vert Q_{3}}{\langle\psi_{1}\vert Q_{2}\vert z\rangle}+\frac{\left(\langle\psi_{1}\vert Q_{2}\vert z\rangle Q_{2}\vert v\rangle+\langle y\vert Q_{3}\vert\psi_{2}\rangle Q_{2}\vert z\rangle\right)\left(\langle\psi_{1}\vert Q_{2}\vert z\rangle\langle w\vert Q_{3}+\langle\psi_{3}\vert Q_{2}\vert z\rangle\langle y\vert Q_{3}\right)}{\langle\psi_{1}\vert Q_{2}\vert z\rangle\left(\langle\psi_{1}\vert Q_{2}\vert z\rangle\langle\psi_{3}\vert Q_{2}\vert v\rangle+\langle y\vert Q_{3}\vert\psi_{2}\rangle\langle\psi_{3}\vert Q_{2}\vert z\rangle\right)} \; .
    \end{aligned}
\end{equation}
} 
In this solution, $\vert x \rangle, \vert y \rangle, \vert z \rangle, \vert u \rangle, \vert v \rangle, \vert w \rangle$ are introduced to decouple the equations for $H_1,H_2,H_3$ (see Appendix \ref{app:nnn-u221}), and they are defined as
\begin{equation}
    \begin{aligned}
       H_1^\dagger \vert \psi_3 \rangle &= -H_2^\dagger \vert \psi_2 \rangle = Q_1 \vert x \rangle, \\
       H_2 \vert \psi_2 \rangle &= - H_3^\dagger \vert \psi_1 \rangle = Q_3 \vert y \rangle , \\
       H_1 \vert \psi_3 \rangle &= - H_3 \vert \psi_1 \rangle = Q_2 \vert z \rangle , \\
       H_2 \vert \psi_3 \rangle &= Q_3 \vert u \rangle , \\
       H_3 \vert \psi_2 \rangle &= Q_2 \vert v \rangle , \\
       H_3^\dagger \vert \psi_3 \rangle &= Q_3 \vert w \rangle \; .
    \end{aligned}
    \label{eq:u221-nnn-xyz-vectors}
\end{equation} 
These vectors ($\vert x \rangle, \vert y \rangle, \vert z \rangle, \vert u \rangle, \vert v \rangle, \vert w \rangle$) and $\vert \psi_{i=1,2,3} \rangle$ must satisfy the following constraints (see details in Appendix \ref{app:nnn-u221}):
\begin{equation}
    \begin{aligned}
        \langle\psi_{1}\vert(E_{FB}-H_{0})\vert\psi_{1}\rangle-\langle\psi_{1}\vert Q_{3}\vert u\rangle&=\langle\psi_{2}\vert(E_{FB}-H_{0})\vert\psi_{2}\rangle-\langle w\vert Q_{3}\vert\psi_{2}\rangle=\langle\psi_{1}\vert H_{x}\vert\psi_{2}\rangle,\\\langle\psi_{2}\vert(E_{FB}-H_{0})\vert\psi_{3}\rangle&=\langle\psi_{1}\vert Q_{2}\vert z\rangle=\langle\psi_{1}\vert H_{x}\vert\psi_{3}\rangle,\\\langle\psi_{2}\vert(E_{FB}-H_{0})\vert\psi_{1}\rangle&=\langle w\vert Q_{3}\vert\psi_{1}\rangle,\\\langle\psi_{2}\vert(E_{FB}-H_{0})\vert\psi_{1}\rangle&=\langle\psi_{2}\vert Q_{3}\vert u\rangle,\\\langle\psi_{3}\vert(E_{FB}-H_{0})\vert\psi_{1}\rangle&=\langle x\vert Q_{1}\vert\psi_{2}\rangle=\langle\psi_{3}\vert H_{x}\vert\psi_{2}\rangle,\\\langle x\vert Q_{1}\vert\psi_{3}\rangle&=\langle\psi_{3}\vert Q_{2}\vert z\rangle=\langle\psi_{3}\vert H_{x}\vert\psi_{3}\rangle,\\\langle\psi_{3}\vert(E_{FB}-H_{0})\vert\psi_{2}\rangle&=\langle\psi_{1}\vert Q_{3}\vert y\rangle=\langle\psi_{1}\vert H_{y}\vert\psi_{2}\rangle,\\\langle\psi_{3}\vert(E_{FB}-H_{0})\vert\psi_{3}\rangle-\langle v\vert Q_{2}\vert\psi_{3}\rangle&=\langle\psi_{1}\vert Q_{3}\vert u\rangle=\langle\psi_{1}\vert H_{y}\vert\psi_{3}\rangle,\\\langle\psi_{3}\vert(E_{FB}-H_{0})\vert\psi_{1}\rangle&=\langle v\vert Q_{2}\vert\psi_{1}\rangle,\\\langle\psi_{2}\vert Q_{3}\vert y\rangle&=-\langle x\vert Q_{1}\vert\psi_{2}\rangle=\langle\psi_{2}\vert H_{y}\vert\psi_{2}\rangle,\\\langle\psi_{2}\vert Q_{3}\vert u\rangle&=-\langle x\vert Q_{1}\vert\psi_{3}\rangle=\langle\psi_{2}\vert H_{y}\vert\psi_{3}\rangle,\\\langle y\vert Q_{3}\vert\psi_{1}\rangle&=\langle\psi_{1}\vert Q_{2}\vert z\rangle=-\langle\psi_{1}\vert H_{yx}\vert\psi_{1}\rangle,\\\langle\psi_{1}\vert Q_{2}\vert v\rangle&=-\langle y\vert Q_{3}\vert\psi_{2}\rangle=\langle\psi_{1}\vert H_{yx}\vert\psi_{2}\rangle,\\\langle w\vert Q_{3}\vert\psi_{1}\rangle&=-\langle\psi_{3}\vert Q_{2}\vert z\rangle=\langle\psi_{3}\vert H_{yx}\vert\psi_{1}\rangle,\\\langle\psi_{3}\vert Q_{2}\vert v\rangle&=\langle w\vert Q_{3}\vert\psi_{2}\rangle=\langle\psi_{3}\vert H_{yx}\vert\psi_{2}\rangle \; .
    \end{aligned}
\end{equation}

\section{Summary} 
\label{section5.5}

In this chapter, we extended the flatband generator in one dimension to two dimensions. We classified all CLSs that occupy a maximum of 4 unit cells in a $2\times2$ plaquette. Then we setup eigenvalue problems, and identified destructive interference conditions for different CLS classes. By solving these sets of equations, we obtained analytic solutions for the Hamiltonians of given flatband classes. 

Our method not only covers all known examples so far, it also gives a large number of extra free parameters. These results can be extended to larger CLSs in 2D, and higher lattice dimensions as well. In the next chapter, we extend our flatband generator to the non-Hermitian regime.

\chapter{Non-Hermitian flatband generator}
\label{chapter6}

\ifpdf
    \graphicspath{{Chapter6/Figs/Raster/}{Chapter6/Figs/PDF/}{Chapter6/Figs/}}
\else
    \graphicspath{{Chapter6/Figs/Vector/}{Chapter6/Figs/}}
\fi

Recently, systems described by non-Hermitian Hamiltonians have been drawing more attention, with an increasing number of studies being carried out to understand the fate of flatbands (FBs) under non-Hermitian perturbations or non-Hermitian settings. Various schemes have been proposed to construct FB lattices in the non-Hermitian regime. In this chapter, we introduce a generator---a systematic classification and construction scheme for non-Hermitian FB Hamiltonians---for 1D a non-Hermitian tight-binding network with two bands. 

We start this chapter by introducing background and our motivation for the non-Hermitian FB generator in Section \ref{section6.1}, before moving on to main definitions in Section \ref{section6.2}. Then we introduce our FB generator scheme in Section \ref{section6.3}, present our results in Section \ref{section6.4}, and conclude in Section \ref{section6.5}.

\section{\label{sec:introduction} Overview of non-Hermitian physics and flatbands in non-Hermitian systems}
\label{section6.1}

Non-Hermitian systems~\cite{moiseyev2011nonhermitian,Bagarello2015NonSelfadjoint} exhibit extraordinary properties such as complex spectra, non-orthogonal eigenstates~\cite{curtright2007biorghogonal,Brody2013biorthogonal}, exceptional points~\cite{Heiss2004exceptionalpoint,Heiss2001chirality,Berry2004nonhermitian,Hern2006nonhermitian,gao2015observation,jinhui2014nhdegeneracy}, etc. Moreover, in open quantum systems, non-Hermitian Hamiltonians can describe coupling with the environment, which simplifies analysis and reduces the large number of degrees of freedom~\cite{eleuch2018lossgain,rotter2009nonhermitian,rotter2018equilibrium}. After the discovery that parity-time ($\mathcal{PT}$)-symmetric non-Hermitian Hamiltonians can have real eigenvalues~\cite{bender1998realspectra,Bender2007makingsense,znojil1999nhharmonic}, $\mathcal{PT}$ symmetry was demonstrated in optical systems with gain and loss~\cite{makris2008beamdynamics}, which in turn led to an explosion of research outcomes in $\mathcal{PT}$-symmetric non-Hermitian photonics~\cite{Ganainy2007theory,musslimani2008optical,klaiman2008visualization,guo2009observation,ruter2010observation,kottos2010broken,wimmer2015opticalsoliton,feng2017nhphotonics,teimourpour2017rubustness,Zhang2018nhoptics,ramy2019dawn}. Non-Hermitian topology has also been drawing more interest~\cite{daniel2017edgemodes,gong2018topological,yuce2015topological,zeuner2015topological}. 

In terms of macroscopically degenerate FBs, i.e. dispersionless energy bands of translational invariant tight-binding networks, most have have been studied in Hermitian systems~\cite{mielke1991ferromagnetism,mielke1992exact,mielke1993ferromagnetism,tasaki1992ferromagnetism,tasaki2008hubbard,maimaiti2017compact,maimaiti2019universal}, and experimentally realized in photonic systems\cite{guzman2014experimental,vicencio2015observation,mukherjee2015observation}, ultra-cold atoms~\cite{jo2012ultracold}, and exciton-polariton condensates~\cite{masumoto2012exciton}. Such achievements though cannot be applied to the above-mentioned non-Hermitian systems, such as coupled laser cavities, where gain and loss is unavoidable. Therefore, designing models that support FBs in the presence of non-Hermiticity, or in other words gain and loss, is highly desirable.

Accordingly, research concerning FBs in non-Hermitian systems has been initiated~\cite{chern2015pt,Ramezani2017nhinduced,ge2015parity}. Indeed, FBs have been found in various non-Hermitian systems~\cite{leykam2017flat,qi2018defect,zyuzin2018flat,ge2018non}. Most of these studies have been carried out in the context of symmetry, such as $\mathcal{PT}$ symmetry, chiral symmetry, and non-Hermitian particle-hole symmetry~\cite{ge2018non,qi2018defect}. While FBs can exist in non-Hermitian systems without any symmetry, little is known about FBs in the non-symmetric regime, and further the eigenstates of non-Hermitian FBs have been largely untouched. It has become important then to propose a systematic classification and construction scheme for non-Hermitian FB Hamiltonians, as has been done for 1D Hermitian systems~\cite{maimaiti2017compact,maimaiti2019universal}. 

We provide here such a scheme, which we term a generator, for non-Hermitian FB Hamiltonians in 1D two-band networks. Non-Hermiticity in our model comes from the asymmetry in the hoppings to the right and left, as well as onsite gain and loss. We obtain FB Hamiltonians simply by requiring eigenvalues of the non-Hermitian $k$-space Hamiltonian to be $k$-independent (i.e. flat). Therefore, symmetry is not required in our model. Since non-Hermitian Hamiltonians give complex band structures, non-Hermitian FBs naturally fall into three different categories: one in which both real and imaginary parts are flat (the completely flat case), one in which either real or imaginary parts are flat (the partially flat case), and one in which the modulus is flat (the flat modulus case). Using a simple band calculation approach along with some approaches from Refs.~\cite{maimaiti2017compact,maimaiti2019universal}, we identify all possible non-Hermitian FB Hamiltonians in the above-mentioned three categories.

\section{\label{section6.2} Non-Hermitian Hamiltonian}

We consider a one-dimensional ($d=1$) translational invariant non-Hermitian lattice with $\nu>1$ sites per unit cell. As we are extending the work in Chapter \ref{chapter4} (\cite{maimaiti2017compact}) to the non-Hermitian regime, we adopt the same definitions and conventions. However, in this case $H_0$ is no longer Hermitian and hoppings to the right and left directions might be different. Therefore, we define right and left hopping matrices $H_r,\ H_l$, and $H_r\ne H_l^{\dagger}$. Then the eigenvalue problem reads
\begin{equation}
   H_0 \vec{\psi}_n + H_r \vec{\psi}_{n-1} + H_l \vec{\psi}_{n+1}  = E \vec{\psi}_n,\quad n\in \mathcal{Z}.
   \label{sch-eq-nn}
\end{equation}
Discrete translational invariance allows us to use the Floquet--Bloch theorem to write the eigenvector of this equation as $\vec{\psi}_n = \sum_k e^{-ikn} u_k $, where the lattice constant is taken to be 1. Using the Bloch eigenfunction we can then write Eq. \eqref{sch-eq-nn} as
\begin{equation}
    \mathcal{H}(k) u_k = E_k u_k,
\end{equation} 
where the $k$-space Hamiltonian is
\begin{equation}
    \mathcal{H}(k) = H_0 + H_l e^{-ik} + H_r e^{ik}.
    \label{eq:H_k}
\end{equation} 

The FB generator for the Hermitian version of the Hamiltonian in Eq. \eqref{eq:H_k}, where $H_l=H_r^{\dagger}$, was introduced in Chapter \ref{chapter4} (\cite{maimaiti2017compact}). In this work, we  identify the $\mathcal{H}(k)$ that gives a FB.
 
\section{\label{section6.3} The generator}

There are two possible ways to construct non-Hermitian FB Hamiltonians. One way is the direct extension of the FB generator in Chapter \ref{chapter4} (\cite{maimaiti2017compact}), which is based on compact localized states (CLSs), to the non-Hermitian regime (see Appendix \ref{app:cls-gen}). However the CLS based generator only gives completely flat FBs. Here though we are mainly concerned with the second approach: band calculation, which can generate non-Hermitian FBs of all three above mentioned categories.  

In this method, we require some eigenvalues of $\mathcal{H}(k)$ in Eq. \eqref{eq:H_k} to be completely or partially flat, and identify the Hamiltonians, i.e. the matrices $H_0,\ H_r,\ H_l$, that satisfy such requirement. 

We take the simplest case with $\nu=2$ sites per unit cell. In this case, hoppings within the unit cell are described by intracell hopping matrix $H_{0}$ (which can be transformed to more generic form by unitary transformation similar to the one described in Section \ref{section4.3}):
\begin{equation}
H_{0}(\nu,\mu)=\begin{pmatrix}0 & \nu \\
                        0 & \mu
          \end{pmatrix}\;, 
\label{eq:def-H0}
\end{equation} 
where $\mu=0,1$, $\nu=0,1$, and they cannot be 1 at the same time. Therefore, $H_0$ has three different forms: 
\begin{equation}
 \begin{aligned}
    H_{0} &=\begin{pmatrix}0 & 0\\ 0 & 0 \end{pmatrix}, \quad degenerate ,\\ 
    H_0 &= \begin{pmatrix}0 & 0\\ 0 & 1 \end{pmatrix} \quad non-degenerate,\\ 
    H_0 &= \begin{pmatrix}0 & 1\\ 0 & 0 \end{pmatrix}, \quad abnormal.
\end{aligned}  
\label{eq:def-H0-froms}
\end{equation}
The intercell hoppings are described by left and right hopping matrices
\begin{equation}
H_{l}=\begin{pmatrix}f & g\\
h & l
\end{pmatrix},\quad 
H_{r}=\begin{pmatrix}a & b\\
c & d
\end{pmatrix}.
\label{eq:def-Hlr}
\end{equation}
Notice that the form of $H_{0}$ is completely fixed by the symmetries in this case. Then, the eigenvalues $x_{k},\ y_{k}$ of the $k$-space Hamiltonian in Eq. \eqref{eq:H_k} satisfy
\begin{equation}
\begin{aligned}
      x_k + y_k & =\mu+e^{ik}(a+d)+e^{-ik}(f+l)\\
      x_k y_k & =e^{2ik}\det H_{r}+e^{-2ik}\det H_{l} \\ & + (\nu f-            
      \mu h)e^{ik} + (\nu a-\mu c)e^{-ik} \\ & +df-cg-bh+al.
\end{aligned}
\label{eq:band-eqn} 
\end{equation} 
By solving Eq. \eqref{eq:band-eqn} under the condition that at least one of $x_k,\ y_k$ has to be fully or partially (real, imaginary parts, or modulus) $k$ independent, we can get the hopping matrices $H_r,\ H_l$ having FBs. 

\section{\label{section6.4} Results} 

We present the results for each category of non-Hermitian FB: completely flat, partially flat, and modulus flat.

\subsection{\label{sec:completely FB} Completely flat case}

Flatbands that are completely flat have $k$ independent real and imaginary parts. This leads to two different cases: one in which both bands are completely flat, and one in which only one band is completely flat. 

\subsubsection{Both bands are completely flat}

In this case, assuming both $x_k=x,\ y_k=y$ are $k$ independent in Eq. \eqref{eq:band-eqn} and requiring all $k$-dependent terms to vanish, we get 
\begin{equation}
\begin{aligned}
&a+d=0,\\
&f+l=0,\\
&\det\,H_{r}=ad-bc=0,\\
&\det\,H_{l}=fl-hg=0,\\
&\nu f-\mu h=0,\\
&\nu a-\mu c=0,\\
&xy=df-cg-bh+al,\\
&x+y= \mu .
\end{aligned}
\label{eq:bands-eqn-allf-xy-eqs}
\end{equation}
From these it follows $d=-a$, $l=-f$, and either $f=a=0$, or $h=c=0$,
or none (for details see Appendix \ref{app:both-fb}).

Solving Eq. \eqref{eq:bands-eqn-allf-xy-eqs} for different forms of $H_0$ in Eq. \eqref{eq:def-H0-froms}, we can get $H_l,\ H_r$ that gives two FBs at FB energies $E_{FB}=\pm x$. Here, we only show the solution for the degenerate $\mu=\nu=0$ case (see Appendix \ref{app:both-fb} for other cases):
\begin{equation}
    \begin{aligned}
        H_{r}&=\begin{pmatrix}a & b\\ -\frac{a^{2}}{b} & -a \end{pmatrix}, \\ 
        H_{l}&=\begin{pmatrix}f & g\\-\frac{f^{2}}{g} & -f \end{pmatrix}, \\
        x^{2}&=2af-\frac{a^{2}g}{b}-\frac{bf^{2}}{g} \; ,
    \end{aligned}
    \label{eq:nh-all-fb-sol}
\end{equation} 
which has two FBs at $E_{FB}=\pm x$. Note that the parameters in Eq. \eqref{eq:nh-all-fb-sol} can be complex.

Interestingly, the solution to Eq. \eqref{eq:nh-all-fb-sol} can give real FBs in the absence of Hermiticity and $\mathcal{PT}$ symmetry. For example, suppose $a,b,f\in\mathcal{R},g=-b$ in Eq. \eqref{eq:nh-all-fb-sol}, then the Hamiltonian in Eq. \eqref{eq:H_k} becomes 
\begin{equation}
    H(k) = \left(
\begin{array}{cc}
 e^{-i k} \left(e^{2 i k} a+f\right) & b e^{-i k} \left(-1+e^{2 i k}\right) \\
 \frac{e^{-i k} \left(f^2-a^2 e^{2 i k}\right)}{b} & -e^{-i k} \left(e^{2 i k} a+f\right), \\
\end{array}
\right)
\label{eq:nh-all-fb-ham}
\end{equation} 
which gives two real FBs at $E_{FB} = \pm (a+f)$; however, the Hamiltonian in this equation is neither Hermitian nor $\mathcal{PT}$ symmetric. 


\subsubsection{One band is completely flat} 

Requiring either $x$ or $y$ to be $k$ independent in Eq. \eqref{eq:band-eqn} yields $\det H_r =0, \det H_l = 0$. Therefore, with one band that is completely flat, we can parameterize $H_r,\ H_l$ as
\begin{equation}
    H_r = \begin{pmatrix} a & b\\
                          c & \frac{b c}{a} \end{pmatrix}, 
    H_r = \begin{pmatrix} f & g\\
                          h & \frac{g h}{f} \end{pmatrix}.
   \label{eq:Hr-Hl-diff-para}
\end{equation} 
Putting $d=\frac{b c}{a}$ into Eq. \eqref{eq:band-eqn} and requiring $x_k=x$ to be $k$ independent gives the following solution (see Appendix \ref{app:one-band-fb}):
\begin{equation}
 \begin{aligned}
 b & =\frac{\left(-x\pm \sqrt{x^2-4 a f}\right) (f \nu +g x)+2 a f g}{2 f^2},\\
 c &=\frac{(x-\mu ) \left(\left(x^2\pm x \sqrt{x^2-4 a f}\right)-2 a f\right)}{2 (f \nu +g x)},\\
  h & =\frac{f^2 (\mu -x)}{f \nu +g x}.
    \end{aligned} 
    \label{eq:single-flat-gen-sol}
\end{equation}  
This solution gives the following band structure: 
\begin{equation}
\begin{aligned}
    E_{FB}&= x,\\ 
    E_k &= \frac{2 (f \nu +g \mu ) \left(-(x-\mu )+\left(a e^{i k}+f e^{-i k}\right)\right)}{2 (f \nu +g x)},\\ 
    &\mp \frac{e^{i k} \nu  (x-\mu ) \left(\sqrt{x^2-4 a f}-x\right)}{2 (f \nu +g x)}.
\end{aligned}
\label{eq:single-flat-gen-band-1}
\end{equation}

Putting corresponding values of $\mu,\ \nu$ into Eqs. \eqref{eq:single-flat-gen-sol} and \eqref{eq:single-flat-gen-band-1}, we can get solutions for the degenerate, non-degenerate, and abnormal cases with corresponding band structures. From Eq. \eqref{eq:single-flat-gen-sol}, we can see that the parameters $a,f,g$ can be complex.

Using CLS method we can get the same solution as Eq. \eqref{eq:single-flat-gen-sol}, which gives a $U=3$ CLS. This is interesting compared to the Hermitian case, in which the maximum CLS size is $U=2$ (see Chapter \ref{chapter4} or Ref. \cite{maimaiti2017compact}). It can be shown that, when $f=-a-x$, the solution in Eq. \eqref{eq:single-flat-gen-band-1} and the CLS reduce to $U=2$ (see Appendices \ref{app:cls-gen-u2} and \ref{app:cls-gen-u3} for details). It can also be shown that the CLS size of $U=3$ is the maximum attainable. 

We can alternatively use the inverse eigenvalue method; in Appendix \ref{app:nh-inv-eig-method} we demonstrate this for the $U=2$ case.

\subsection{\label{sec:partially FB} Partially flat case}

The case of FBs that are partially flat is a special case that only available in non-Hermitian lattices. As previously mentioned, partially flat means that either the real or the imaginary part of a band is $k$ independent. The eigenstates in this case are not compactly localized, and therefore the CLS approach does not apply. We stick to the band calculation approach. 

We assume that $x_k=x_1 + i x_2, y=y_1 + i y_2$, where $x_i,y_i\in\mathcal{R},i=1,2$. By separating the real and imaginary parts of Eq. \eqref{eq:band-eqn}, we have
\begin{equation}
    \begin{aligned}
        x_{1}+y_{1}&=\mu+(a+d+f+l)\cos(k)\\
        x_{2}+y_{2}&=-(a+d-f-l)\sin(k)\\ x_{1}y_{1}-x_{2}y_{2} =&al-bh-cg+df \\ &+(a\mu-c\nu+f\mu- h\nu)\cos(k)\\ 
        &+(\det H_{l}+\det H_{r})\cos(2k)\\
        x_{2}y_{1}+x_{1}y_{2}&=\left(-a\mu+c\nu+f\mu-
        h\nu\right)\sin(k) \\ 
        & +(\det H_{l}-\det H_{r})\sin(2k).\;
    \end{aligned}
    \label{eq:partial-flat-gen}
\end{equation}
Requiring some of $x_1,\ x_2,\ y_1,\ y_2$ to be $k$ independent and solving Eq. \eqref{eq:partial-flat-gen}, we can get $H_r,\ H_r$ that gives partially flat FBs. There are many possible cases for such FBs, namely where the real (imaginary) parts of both bands are flat, the real (imaginary) part of one band is flat, and the real (imaginary) part of one band and the imaginary (real) part of the other band is flat.
Here we show the results for the case in which the real parts of both bands are flat; details of this case and other cases are given in Appendix \ref{app:partial-fb}.

In this example, $x_1,y_1$ in Eq. \eqref{eq:partial-flat-gen} are $k$ independent, and we assume an abnormal $H_0$, i.e. $\nu=1,\mu=0$. Then Eq. \eqref{eq:partial-flat-gen} gives
\begin{equation}
  \begin{aligned}
     H_{0}&=\begin{pmatrix}0 & 1\\
                        0 & 0
                \end{pmatrix},\\
     H_{l}&=\begin{pmatrix} a & \frac{(a-d)^2}{4 h} 
             -\frac{h}{4 x_1^2} \\ 
             -h & d \end{pmatrix},\\
     H_{r}&=\begin{pmatrix} -a-x_1 & \frac{h^2-x_1^2   
            \left(a-d+2 x_1\right){}^2}{4 h x_1^2} \\
            h & x_1-d \end{pmatrix}
  \end{aligned}
  \label{eq:both-real-flat-abnormal}
\end{equation}
which yields the following band structure
\begin{equation}
    \begin{aligned}
        x_k & = i (a+d) \sin (k)-\frac{i h \sin (k)}{x_1}+x_1,\\
        y_k & = i (a+d) \sin (k)+\frac{i h \sin (k)}{x_1}-x_1 \; .
    \end{aligned}
    \label{eq:both-rl-fb-abn-band}
\end{equation}
 Using the same method, we can solve the case in which the imaginary parts of both bands are flat.

\subsection{Modulus-flat case}  

In addition to the two previous cases, the third category of non-Hermitian FBs is the special flat-modulus case, where the band has a $k$-dependent phase with a non-trivial winding number. In this case, we assume that $x=r e^{i\theta_{k}}, r\in \mathcal{R}$ in Eq. \eqref{eq:band-eqn}. If we assume the phase is an integer multiple of $k$, i.e. $\theta_k=m k$, then $m$ can take two possible values $m=\pm 1$. For each value of $m$, we solve the Eq. \eqref{eq:band-eqn} to get the hopping matrices $H_r,\ H_l$ that gives a FB at $E_{FB}=r e^{i\theta_{k}}, r\in \mathcal{R}$. As an example, we present the results for $m=1$ and abnormal $H_0$, i.e. $\mu=0,\nu=1$ (see details in Appendix \ref{app:mod-fb}).

For $m=1$ and abnormal $H_0$ we have 
\begin{equation}
    \begin{aligned}H_{r} & =\left(\begin{array}{cc}
              r & b\\
             0 & d
            \end{array}\right),\\
         H_{l} & =\left(\begin{array}{cc}
             0 & g\\
            0 & l
            \end{array}\right),
     \end{aligned}
     \label{eq:mod-fl-result}
\end{equation} 
 which gives the following band structure
\begin{equation}
  \begin{aligned}
    E_{1} & =e^{ik}r,\\
    E_{2} & =de^{ik}+e^{-ik}l+1.
  \end{aligned}  
  \label{eq:mod-fl-band}
\end{equation} 
Using a similar method, other cases of $H_0$ and the $m=-1$ case can be solved. Notably, the FB at $E_1$ in Eq. \eqref{eq:mod-fl-band} has a non-trivial winding number.

\section{\label{section6.5} Summary}

We considered a 1D translational invariant two-band network with non-Hermitian hoppings, and introduced a systematic classification of non-Hermitian FB Hamiltonians. This classification provides a systematic way to construct non-Hermitian FB lattices with two bands. The non-Hermitian FB in this case were constructed by fine-tuning the hoppings, and the construction does not require $\mathcal{PT}$ or chiral symmetry. Completely flat FBs host CLSs, on which the classification of the Hamiltonian is based. We found that the maximum class of CLSs for the non-Hermitian two-band network is $U=3$, which differs from the Hermitian case where the maximum is $U=2$. For partially flat FBs, the CLSs are not necessary, so we used band calculation methods to identify the partially flat FB Hamiltonians. Interestingly, for the flat-modulus case, we found that the band has a non-trivial winding number.

\chapter{Flatband in a microwave photonic crystal}
\label{chapter7}

\ifpdf
    \graphicspath{{Chapter7/Figs/Raster/}{Chapter7/Figs/PDF/}{Chapter7/Figs/}}
\else
    \graphicspath{{Chapter7/Figs/Vector/}{Chapter7/Figs/}}
\fi  

Previous chapters introduced our flatband (FB) generators in 1D, 2D, and non-Hermitian systems. As an application, in this chapter we present a tight-binding model for a microwave photonic crystal. The spectrum of the photonic crystal consists of Dirac points at two different energies; while conventional tight-binding models can explain the spectrum around the lower Dirac points, they cannot explain the spectrum around the upper ones. Our tight-binding model, comprising honeycomb and kagome sublattices, fits the density of states (DOS) of the microwave photonic crystal nicely, including the spectrum around the upper Dirac points. The methods and ideas behind our FB generator, such as matrix representation, play an important role in setting up this model and simplifying the calculations.

\nomenclature[DOS]{DOS}{Density of states}

\section{\label{intr}Photonic crystals and Dirac billiards}

The peculiar band structure of graphene and the linear dispersion relation around the touch points of its conduction and valence bands originate from the symmetry properties of its honeycomb structure~\cite{Slonczewski1958}, which is formed by two interpenetrating triangular lattices with threefold symmetry. Consequently, photons (bosons) or waves propagating in a spatially periodic potential with a honeycomb structure may comprise energy spectra regions where they are effectively described by the Dirac equation for spin-1/2 fermions. Actually, there exist numerous realizations~\cite{Polini2013} of artificial graphene using 2D electron gases exposed to a honeycomb potential lattice~\cite{Singha2011,Nadvornik2012}, molecular assemblies arranged on a copper surface~\cite{Gomes2012}, ultracold atoms in optical lattices~\cite{Tarruell2012,Uehlinger2013}, and photonic crystals~\cite{Parimi2004,Joannopoulos2008,Bittner2010,Kuhl2010,Sadurni2010,Bittner2012,Bellec2013,Bellec2013a,Rechtsman2013,Rechtsman2013a,Khanikaev2013}.

It was shown in~\cite{Raghu2008} that under certain conditions, photonic crystals with a triangular lattice geometry exhibit a linear Dirac dispersion relation. In Ref.~\cite{Dietz2015} a microwave Dirac billiard was constructed from a photonic crystal, consisting of a rectangular basin containing about 900 metallic cylinders (as illustrated in~\reffig{figbilliard}) arranged into a triangular lattice with a top plate in order to obtain a microwave resonator~\cite{Dietz2013,Dietz2015,Dietz2019}. 
Below a certain frequency of microwaves sent into the resonator, which is inversely proportional to its height, the electric field is perpendicular to the top and bottom plates, that is, the associated Helmholtz equation is two-dimensional and mathematically equivalent to the Schr\"odinger equation of the corresponding quantum billiard~\cite{Stoeckmann1990,Sridhar1991,Graf1992}. 
\begin{figure}[h]
\centering
\includegraphics[width=0.6\linewidth]{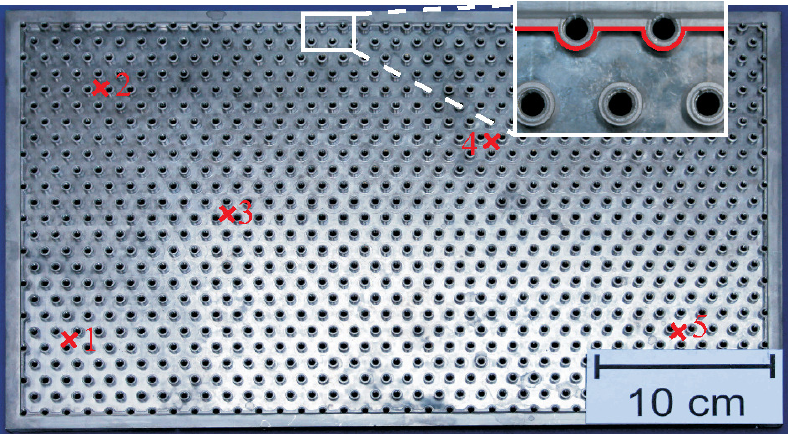}
\caption[Photograph of the basin plate of a microwave Dirac billiard]{Photograph of the basin plate of a microwave Dirac billiard containing 888 metal cylinders. It was constructed from brass and coated with lead to achieve superconductivity at liquid helium temperature. The red crosses mark the positions of the antennas. The inset shows a magnification of the lattice structure from a long side. Taken from~\cite{Dietz2015}.}
\label{figbilliard}
\end{figure}
Due to the periodic lattice structure formed by the cylinders, the frequencies of wave propagation as a function of the two quasimomentum components exhibit a band structure similar to that of graphene. Indeed, the resonance spectra of microwave Dirac billiards exhibit Dirac points where they are governed by the relativistic Dirac equation~\cite{Bittner2010}. The honeycomb structure arises from the voids formed by three neighboring cylinders yielding a triangular cell~\cite{Gaspard1989}, which can be considered as an open resonator. There, the electric field intensity of the propagating modes is localized, as marked by red and blue circles in~\reffig{schemebilliard}. 
\begin{figure}[h]
\centering
\includegraphics[width=0.6\linewidth]{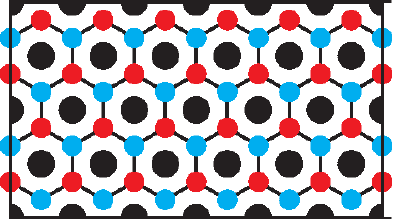}
\caption[Schematic view of the triangular lattice structure of the billiard]{Schematic view of the triangular lattice structure of the billiard. The red and blue circles mark the voids between the cylinders, corresponding to the atoms in graphene arranged on the two independent triangular lattices that form the hexagonal lattice. It is terminated by armchair and zigzag edges along the short and long sides, respectively, and translationally invariant with respect to all sides. Taken from Ref. ~\cite{Dietz2015}.}
\label{schemebilliard}
\end{figure}
Thus, the cells host quasibound states which, depending on the eigenfrequency, might partially overlap with neighboring ones. The size of the electric field intensity at the voids provides the on-site excitations. The void structure terminates with a zigzag edge along the longer edges of the rectangle and with an armchair edge along the shorter ones.


Figure~\ref{figdos} shows a comparison of the DOS, $\rho (f)=\frac{\pi^2}{N}\sum_n\delta(f-f_n)$ with $N$ denoting the number of sites of the honeycomb lattice, i.e., voids formed by the metal cylinders and $f_n$ the eigenfrequencies (left panel), with the calculated band-structure function $f(\vec q)$ (right panel) along the path $\Gamma{\rm M}{\rm K}\Gamma$ inside the first Brillouin zone (inset). Here, K denotes the positions of the Dirac points located at the corners of the Brillouin zone, $\Gamma$ which is the center of the BZ where the band terminates, and M the saddle points~\cite{Castro2009}. The positions of the experimental band gaps and Dirac points agree well with those in the calculated band structure. The latter are bordered by sharp peaks, the van Hove singularities~\cite{VanHove1953}, which correspond to the saddle points in the band structure. Generally, $\rho(f)$ exhibits maxima at frequencies where the band barely changes with the quasimomentum vector $\vec q$, i.e., regions of low group velocity, $\vert\vec\nabla f(\vec q)\vert\simeq 0$.
\begin{figure}[htb!]
\centering
\includegraphics[width=0.6\linewidth]{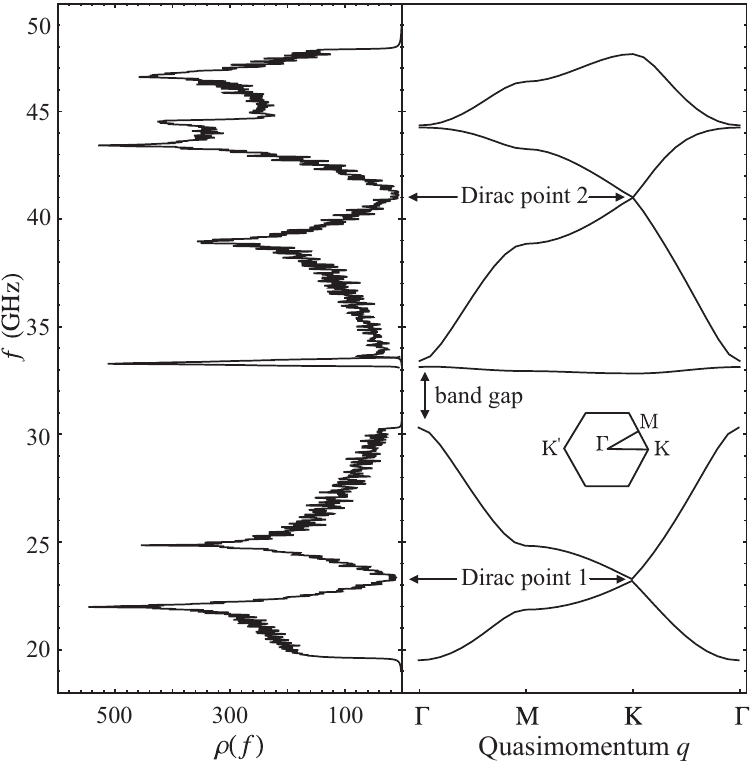}
\caption[Comparison of the experimental DOS with the computed band structure]{Comparison of the experimental DOS (left panel) with the computed band structure of an infinitely extended photonic crystal with the same lattice constant (right panel). The frequencies of the Dirac points, the band gaps, the van Hove singularities, and the other peaks of $\rho(f)$ agree well with the computed ones as well as the locations of the saddle points and the FB. Taken from~\cite{Dietz2015}.}
\label{figdos}
\end{figure}
The narrow region of exceptionally high resonance density at the upper edge of the first band gap is associated with a FB separating the two Dirac points, with the bands framing them. 

\section{\label{TBM} Tight-binding model for a microwave photonic crystal}

In Ref.~\cite{Dietz2015}, a tight-binding model (TBM) was used to describe the experimental DOS. In distinction to the first TBM description of graphene~\cite{Wallace1947}, where only nearest- and second-nearest neighbor interactions of the $p_z$ orbitals were considered, also third-nearest-neighbour couplings and overlaps between the wavefunctions centered at the associated atoms had to be included in order to attain good agreement between the experimental and computed DOS. The latter takes into account the partial overlap of the quasibound states, i.e., of the electric field mode components localized at the voids formed by the metal cylinders with neighboring ones~\cite{Reich2002}. This indicates that the quasibound states localized at the voids extend into the regions of neighboring ones. Accordingly, the band-structure function $f(\vec q)$ was obtained by solving the following generalized eigenvalue problem
\begin{equation}
\mathcal{H}_{TB}\vert\Psi_{\vec q}(\vec r)\rangle=f(\vec q)\mathcal{S}_{WO}\vert\Psi_{\vec q}(\vec r),\rangle\label{eigenvalueeq}
\end{equation}
with the TBM Hamiltonian
\begin{equation}
\mathcal{H}_{TB}=\left({\begin{array}{cc}
                \gamma_0+\gamma_2h_2(\vec q)\, &\gamma_1h_1(\vec q)+\gamma_3h_3(\vec q)\\
                \gamma_1h_1(\vec q)+\gamma_3h_3(\vec q)\, &h_0+\gamma_2h_2(\vec q)
                \end{array}}\right)\label{htb}
\end{equation}
and the wavefunction overlap matrix
\begin{equation}
\mathcal{S}_{WO}=\left({\begin{array}{cc}
                1+s_2h_2(\vec q)\, &s_1h_1(\vec q)+s_3h_3(\vec q)\\
                s_1h_1(\vec q)+s_3h_3(\vec q)\, &1+s_2h_2(\vec q)
                \end{array}}\right).\, ,\label{stb}
\end{equation}
These incorporate the nearest-neighbor coupling $\gamma_1$ with the second- and third-nearest neighbor couplings $\gamma_2$ and $\gamma_3$, and the corresponding overlap parameters $s_1$, $s_2$, and $s_3$. The functions $h_n(\vec q),\, n=1,2,3$ associated with the different couplings were obtained from Ref.~\cite{Reich2002}. 

The parameters were determined by fitting the DOS deduced from the band-structure function $f(\vec q)$ to the experimental one. In order to achieve a good agreement of the DOSs deduced from the TBM and the experimental eigenfrequencies, respectively, the nearest neighbor, second and third nearest neighbor couplings between the quasibound states were taken into account, i.e., the electric field mode components localized at the voids between the metallic cylinders and also the corresponding overlaps.
The spectrum around the lower Dirac point was well described by the TBM. However, for the upper Dirac point, the TBM reproduced only the positions of the peaks exhibited by the DOS, the band edges, and the Dirac point, whereas the overall shape exhibited clear deviations. That is, the TBM does not seem to be suitable to describe the band structure including both Dirac points.

The fact that the TBM yielded a good description only after including wavefunction overlaps, along with the occurrence of a FB together with two Dirac points and the bands framing them, led us to the idea that the microwave Dirac billiard may actually provide for the experimental realization of a honeycomb-kagome lattice. In this picture, the non-vanishing electric field intensity between two neighboring metal cylinders corresponds to atoms of the kagome lattice. The fact that the lower Dirac point and the bands framing it are well described by the tight-binding model for a finite-sized honeycomb lattice implies that the coupling between the honeycomb and kagome lattice is weak. 



\section{The honome lattice}
\label{sec:model}

The DOS of the tight-binding model in the previous section fits well with the experimental spectrum around the lower Dirac point. However, the spectrum around the upper Dirac point, including the FB, was not explained theoretically. We want to propose a tight-binding model that fits the experimental spectrum around both Dirac points and the FB as well. 

 By observing the nodes of the wave functions obtained from the Helmholtz equation, we came up with a tight-binding model with two sublattices, honeycomb and kagome. We dubbed this new lattice as honome lattice, see Fig. \ref{fig:honome}. 

\begin{figure}[hbt!]
    \centering
    \includegraphics[width=0.8\columnwidth]{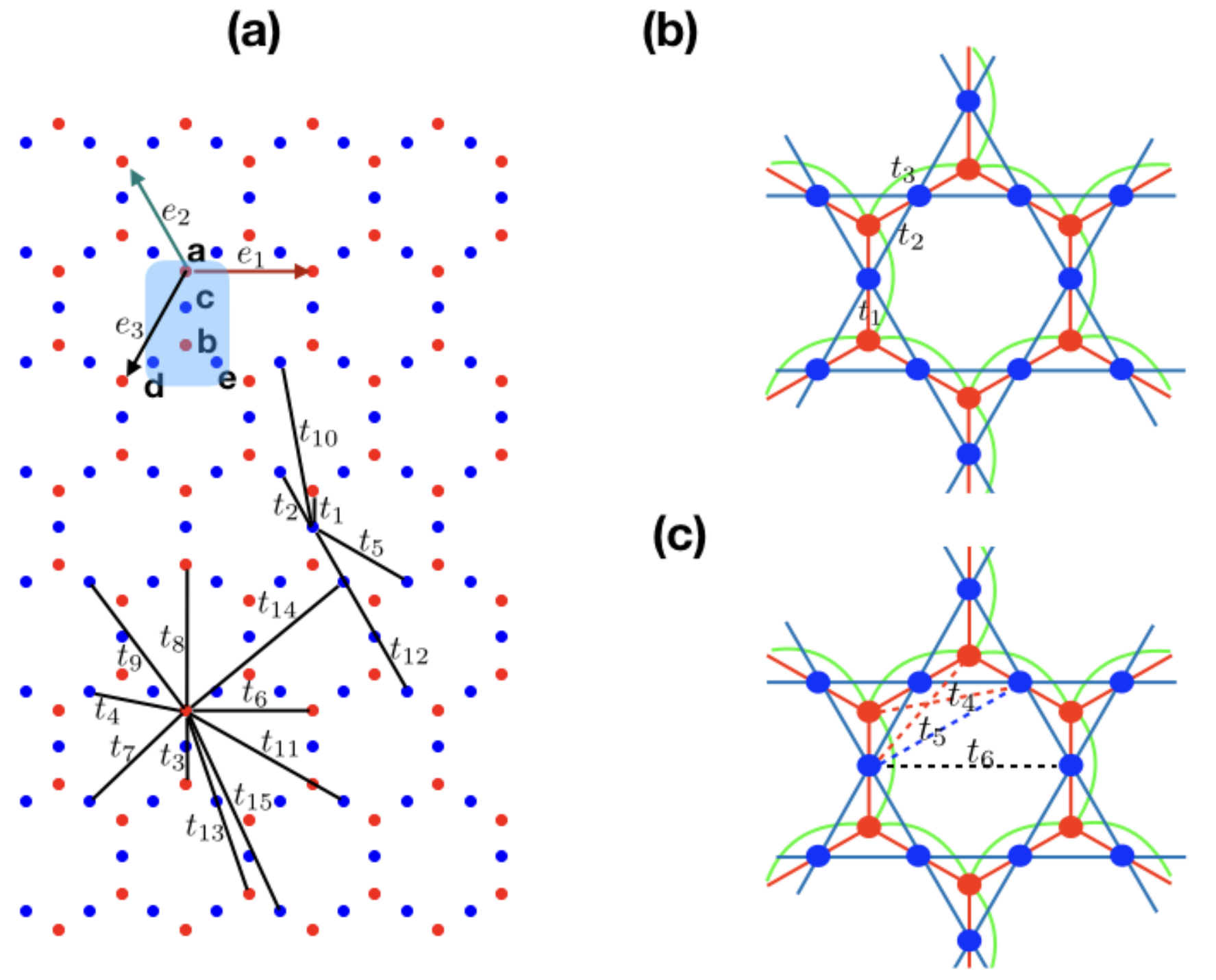}
    \caption[The honome lattice]{The honome lattice consists of two sublattices: honeycomb (red) and kagome (blue). (a) Structure of the honome lattice. The shaded area shows the unit cell. All hoppings within nearest neighboring unit cells are shown. The hopping indices are ordered according to increasing distance, with each representing all hoppings of the same distance. (b) Honome lattice up to third nearest neighbor hoppings. (c) Honome lattice showing 4th to 6th nearest neighbor hoppings.}
    \label{fig:honome}
\end{figure}
 
We write the wave function of a unit cell as 
\begin{equation}
\vec{\psi}_{m,n} = \left( a_{m,n}, b_{m,n}, c_{m,n}, d_{m,n}, e_{m,n}  \right)^T,
\end{equation} 
where $a_{m,n},\ b_{m,n}$ are wave functions of honeycomb sublattice sites, $c_{m,n},\ d_{m,n},\ e_{m,n}$ are wave functions of kagome sublattice sites, and $m,n$ are unit cell indices giving the position of the unit cell with respect to the origin, i.e. the lattice vector of the ($m,n$) unit cell is $\vec{R}_{m,n}= m \vec{e}_1 + n \vec{e}_2$. 

 Within nearest neighboring unit cells, the maximum hopping range is 15th neighboring hopping in conventional terms. Then we can express the hoppings in matrix form:
{\footnotesize 
\begin{equation}
\begin{aligned}H_{0}=\left(\begin{array}{ccccc}
\epsilon & t_{3} & t_{1} & t_{4} & t_{4}\\
t_{3} & \epsilon & t_{1} & t_{1} & t_{1}\\
t_{1} & t_{1} & 0 & t_{2} & t_{2}\\
t_{4} & t_{1} & t_{2} & 0 & t_{2}\\
t_{4} & t_{1} & t_{2} & t_{2} & 0
\end{array}\right),\ \  & H_{1}=\left(\begin{array}{ccccc}
t_{6} & t_{8} & t_{7} & t_{7} & t_{11}\\
t_{8} & t_{6} & t_{4} & t_{4} & t_{9}\\
t_{7} & t_{7} & t_{6} & t_{5} & t_{10}\\
t_{11} & t_{9} & t_{10} & t_{6} & t_{12}\\
t_{7} & t_{4} & t_{5} & t_{2} & t_{6}
\end{array}\right)\\
H_{2}=\left(\begin{array}{ccccc}
t_{6} & t_{3} & t_{4} & t_{4} & t_{1}\\
t_{13} & t_{6} & t_{9} & t_{7} & t_{4}\\
t_{9} & t_{4} & t_{6} & t_{5} & t_{2}\\
t_{14} & t_{7} & t_{10} & t_{6} & t_{5}\\
t_{15} & t_{9} & t_{12} & t_{10} & t_{6}
\end{array}\right),\ \  & H_{3}=\left(\begin{array}{ccccc}
t_{6} & t_{13} & t_{9} & t_{15} & t_{14}\\
t_{3} & t_{6} & t_{4} & t_{9} & t_{7}\\
t_{4} & t_{9} & t_{6} & t_{12} & t_{10}\\
t_{1} & t_{4} & t_{2} & t_{6} & t_{5}\\
t_{4} & t_{7} & t_{5} & t_{10} & t_{6}
\end{array}\right),
\end{aligned}
\label{eq:hop-mats}
\end{equation} 
}
where $t_j$ is the $j$th neighbouring hopping. $H_0$ is the intra unit cell hopping matrix, $H_i$ describes hoppings between nearest neighboring unit cells along $\vec{e}_i$, where $\vec{e}_1,\vec{e}_2$ are primitive lattice translation vectors, and $\vec{e}_3=-\vec{e}_1 - \vec{e}_2$. Then the eigenvalue problem for this lattice is 
\begin{equation}
\begin{aligned}
    & H_0 \vec{\psi}_{m,n} + H_1 \vec{\psi}_{m+1,n} + H_1^\dagger \vec{\psi}_{m-1,n} + H_2 \vec{\psi}_{m,n+1} \\ &+ H_2^\dagger \vec{\psi}_{m,n-1} + H_3 \vec{\psi}_{m+1,n+1} + H_3^\dagger \vec{\psi}_{m-1,n-1} \\
    & = E \vec{\psi}_{m,n}.
\end{aligned}
\label{eq:eig-prob}
\end{equation} 

According to Bloch theorem, one can write $\vec{\psi}_{m,n} = \vec{\phi}_{\vec{k}} e^{- i \vec{k} \cdot \vec{R}_{m,n}}$, and by putting this into Eq. \eqref{eq:eig-prob} we get the $k$-space Hamiltonian, 
\begin{equation}
H(k)=H_{0}+\sum_{j=1}^{3}\left(H_{j}e^{-ik_{j}}+H_{j}^{\dagger}e^{ik_{j}}\right),
\label{eq:hk}
\end{equation}
 where $k_{i}=\vec{k}\cdot\vec{e}_{i},\  \vec{k}=\left(k_{x},k_{y}\right),\  \vec{e}_{1}=\left(1,0\right),\ \vec{e}_{2}=\left(-\frac{1}{2},\frac{\sqrt{3}}{2}\right)$, and $\vec{e}_{3}=\left(-\frac{1}{2},-\frac{\sqrt{3}}{2}\right)$. Eigenvalues of $H(k)$ give the band structure. 
 
For given hoppings, here we consider a maximum 5th neighboring hoppings in conventional terms. The DOS can be computed in two different ways. One is to compute the eigenvalues of the direct (real) space Hamiltonian, and the other one is to take reciprocal lattice sites corresponding to the direct lattice sites and, using band structure, calculate the energy corresponding to each reciprocal lattice site. 
 
 In order to achieve the best fit with experimental results, we have to find the best parameters, i.e. hopping values and onsite energies. In the next section we introduce an algorithm to find these optimal parameters. 
 

\section{Methods}

The experimental data contains 4000 resonant frequencies, and fitting the DOS of the experiment and the honome for all these 4000 points would be very resource consuming and inefficient. Therefore, we only fit some special points; one option is to fit the relative positions of the van Hove singularities (peaks and Dirac points of the DOS) only.

\subsection{Positioning van Hove singularities}

The position of a van Hove singularity is its horizontal coordinate, which is resonant frequency (energy), see Fig. \ref{fig:exp_dos_bandstr}. Then the relative position of van Hove singularity $p$ is defined as
\begin{equation}
   p = \frac{\omega_{vh} - \omega_0}{w}, 
   \label{eq:rel-pos}
\end{equation} 
where $\omega_{vh}$ is the resonant frequency (energy) corresponding to this van Hove singularity, $\omega_0$ is the lower band edge of the lowest band, and $w$ is the width of the spectrum. Figure \ref{fig:exp_dos_bandstr} shows the relative positions of all van Hove singularities.

We extract from the experimental data the relative positions of $i$th peaks $p_{i}^{e}$ and relative positions of $i$th Dirac point $d_{i}^{e}$ of the experiment. For example, the band edge of the experiment is given in Ref. \cite{Dietz2015}, which is 19.64. We can extract the position of the 5th peak, which is $p_5^e = 43.421$. Then the width of the experimental spectrum is 
\begin{equation}
    w_{exp} = 43.421 - 19.64 = 23.781.
\end{equation} 
Similarly, we extract the positions (resonant frequencies) for all van Hove singularities and calculate the relative positions using Eq. \eqref{eq:rel-pos}; see Table \ref{tab:exp-rel-pos}. 
\begin{table}[htb!]
    \centering
    \begin{tabular}{|c|c| c| c| c| c| c|}
    \hline
       VHS  & $p_1^e$ & $p_2^e$ & $p_3^e$ & $FB_e$ & $d_1^e$ & $d_2^e$ \\ 
       \hline
        RP & 0.098524  & 0.218872 & 0.809259 & 0.573672 & 0.15586 & 0.902927 \\
        \hline
    \end{tabular}
    \caption{Relative positions (RP) of van Hove singularities (VHS) extracted from experimental data.}
    \label{tab:exp-rel-pos}
\end{table} 

When we calculate the relative positions of the van Hove singularities of the honome lattice, we assume that the peaks of the energy (frequency) spectrum are located at the $M$ point, the Dirac points are located at the $K$ point, and the band edges are located at the $\Gamma$ point.
This allows us to obtain the positions of the van Hove singularities by calculating the eigenvalues of $H(k)$ at these $M,\ K,\ \Gamma$ points. Their positions in first Brillouin zone are shown in Fig. \ref{fig:exp_dos_bandstr}. In some cases, we have analytical solutions, e.g. in the case of 3rd neighboring hopping range, while in other cases we have no analytic solution and thus calculate the van Hove singularities numerically. The relative positions are then calculated according to Eq. \eqref{eq:rel-pos}, where we take the difference between the 5th peak and the lowest band edge as the width of the spectrum $w$ 
\begin{equation}
w=Max\left(m_{1},m_{2},m_{3},m_{4},m_{5}\right)-\gamma,\label{eq:span_of_spec}
\end{equation}
where $\gamma$ is the lowest band edge, which is the lowest eigenvalue of $H(k)$ at the $\Gamma$ point, i.e. at $k_x=k_y=0$. The reason for choosing the spectral width in this way is that the highest band, which is beyond the Helmholtz regime, in experiment cannot be described by the honome model.

\subsection{Deviations from experiment}

For a given set of parameters (hopping values and onsite energy $\epsilon$), we can calculate the position of the van Hove singularities.

With these positions, we can find the difference between the calculated values and the experimental results by
\begin{equation}
\begin{aligned}
\Delta d_{i} &= d_{i}-d_{i}^{e},\\ 
\Delta p_{i} &= p_{i}-p_{i}^{e}, \\ 
\Delta_{fb} &= E_{fb} - E_{fb}^e \; ,
\end{aligned}
\label{eq:diff_with_exp}
\end{equation}
where $d_i, d_i^e$ are the $i$th Dirac point from theory and experiment, respectively, $p_i, p_i^e$ are the $i$th peak from theory and experiment, respectively, and $E_{fb}, E_{fb}^e$ are the FB from theory and experiment, respectively.

Afterwards, we look for the parameters (hoppings, onsite energies) which minimize $\Delta d_{i}$, $\Delta p_{i}$, and $\delta_{fb}$. The next section details this process.


\begin{figure}[htb!]
\centering{}\includegraphics[width=0.6\columnwidth]{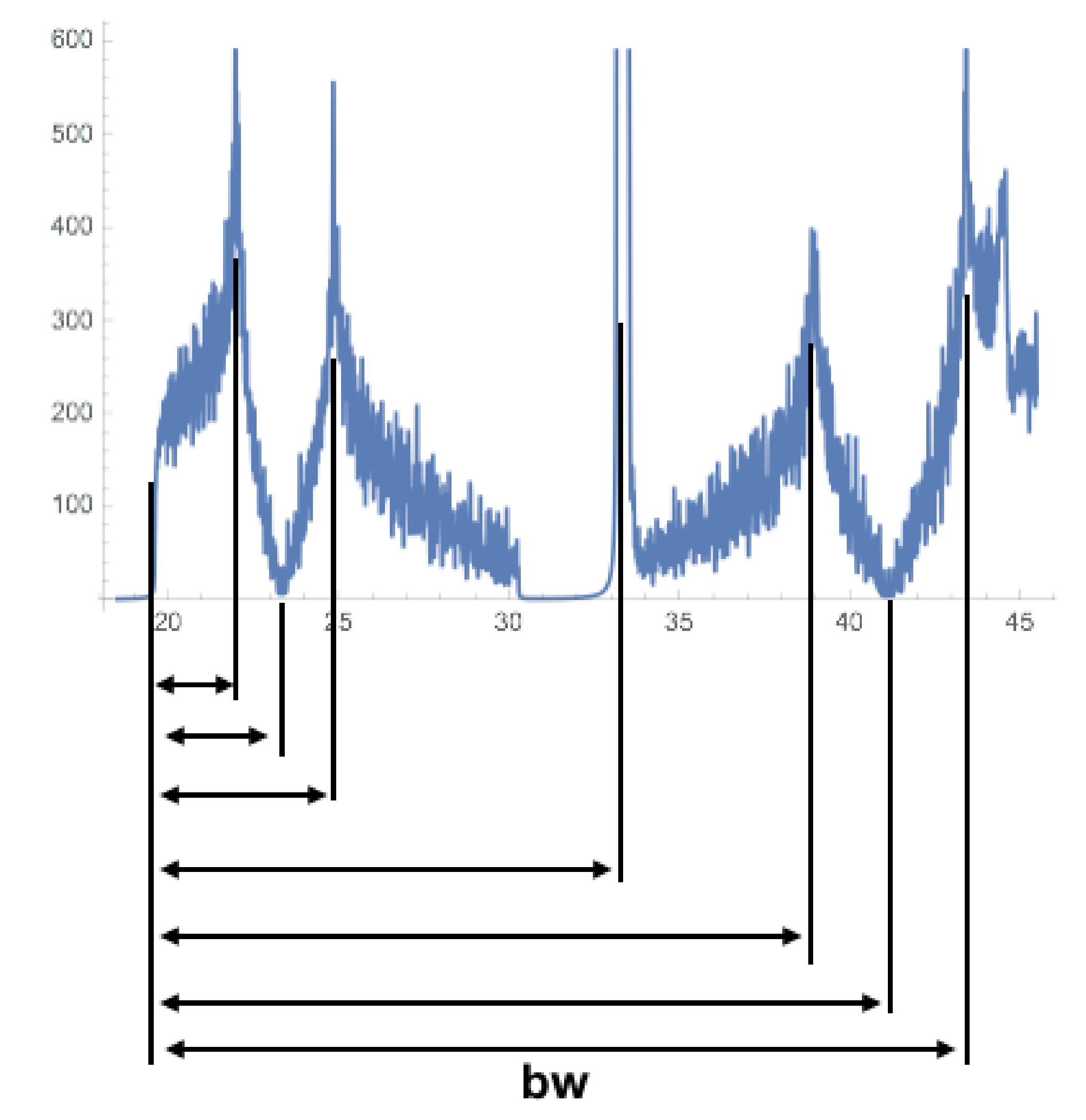}
\caption[Relative positions of peaks and Dirac points]{Relative positions of peaks and Dirac points.}
 \label{fig:exp_dos_bandstr}
\end{figure} 

Alternatively, we can also calculate the standard deviations 
\begin{equation}
\Delta^{2}=\sum_{i}(\Delta d_{i})^{2}+\sum_{j}(\Delta p_{j})^{2} \; ,
\label{eq:standard_diviation}
\end{equation} 
and try to minimize them. 

\subsection{The algorithm}
\label{sec:algorithm}

After computing the relative positions of all van Hove singularities, we look for the parameters which minimize the differences (and standard deviations) between the relative positions of calculated and experimental van Hove singularities. To do so, we use two different algorithms: plain Monte Carlo and importance sampling. Note that both are Monte Carlo methods.

\paragraph{Plain Monte Carlo algorithm:} The plain Monte Carlo algorithm operates as follows: 
\begin{enumerate}
    \item Choose a hopping range, i.e. the number of hoppings. 
    \item Start from an initial hopping value, and calculate the eigenvalues of $H(k)$ at $M,\ K,\ \Gamma$ points. From these values, calculate the relative positions of the peaks and Dirac points, and their deviations from experimental results.
    \item Randomly generate hopping values. 
    \item Repeat the calculations in step 2 for the randomly generated hoppings.
    \item Compare the calculation results from steps 2 and 4.
    \item If the new set of hoppings decreases the deviations, accept the new set. If the new set of hoppings increases or does not change the deviations, reject the new set (i.e. keep the previous set of hoppings).
    \item Repeat steps 3 to 6 until the pre-determined number of cycles is complete.
    \end{enumerate}

\paragraph{Importance sampling:} In this method, instead of randomly generating all hoppings in each cycle, we modify one of the hoppings by a small random step. The algorithm operates as follows:
\begin{enumerate}
    \item Choose a hopping range, i.e. the number of hoppings.
    \item Start from an initial hopping value; one option is to initialize the hoppings with the values obtained from the plain Monte Carlo method. Then calculate the eigenvalues of $H(k)$ at $M,\ K,\ \Gamma$ points. From these values, calculate the relative positions of the peaks and Dirac points, and their deviations from experimental results.
    \item Randomly choose a hopping and add a small random number (random shift) to it.
    \item Repeat the calculations in step 2 for the modified hoppings.
   \item Compare the calculation results from steps 2 and 5.
    \item If the random step decreases the deviations, accept the change. If the random step increases or does not change the deviations, reject the change (i.e. keep the previous set of hoppings).
 \item Repeat steps 3 to 5 until the pre-determined number of cycles is complete.
\end{enumerate} 

We start by considering the hoppings $t_{1},t_{2},t_{3}$ and onsite energy $\epsilon$, and then add more hoppings to improve the fitting. Note that we can always perform a rescaling such that one of the hoppings is 1; therefore, in our later calculations, we set $t_{1}=1$. 

\section{Results}

We consider two different hoppings ranges: up to 3rd and up to 5th nearest neighbor hopping. The results show good agreement with the experiment. The parameters giving the best fitting results are not unique, which reflects the fine-tuned nature of the system. Note that as the parameters we obtained gives a spectrum that is rescaled and shifted from the experimental spectrum, we rescaled and shifted this spectrum to fit the experimental spectrum. 

\subsection{Fitting with three hoppings}

In the case of up to 3rd nearest neighbor hoppings, only $t_1, t_2, t_3$ are non-zero in Eqs. \eqref{eq:hop-mats}, and \eqref{eq:hk}, which gives an analytic solution for the eigenvalues of $H(k)$, at $M,\ K,\ \Gamma$ points. 

Then using Eq. \eqref{eq:rel-pos}, we can find the relative positions of the peaks and Dirac points. We minimize the deviations in Eq. \eqref{eq:diff_with_exp} using the plain Monte Carlo algorithm, as introduced in Section \ref{sec:algorithm}. We obtained a perfect FB in this case and good fitting with the experiment, as shown in Fig. \ref{fitting_3hoppings_1}.

\begin{figure}[htb!]
\centering
\includegraphics[width=0.55\columnwidth]{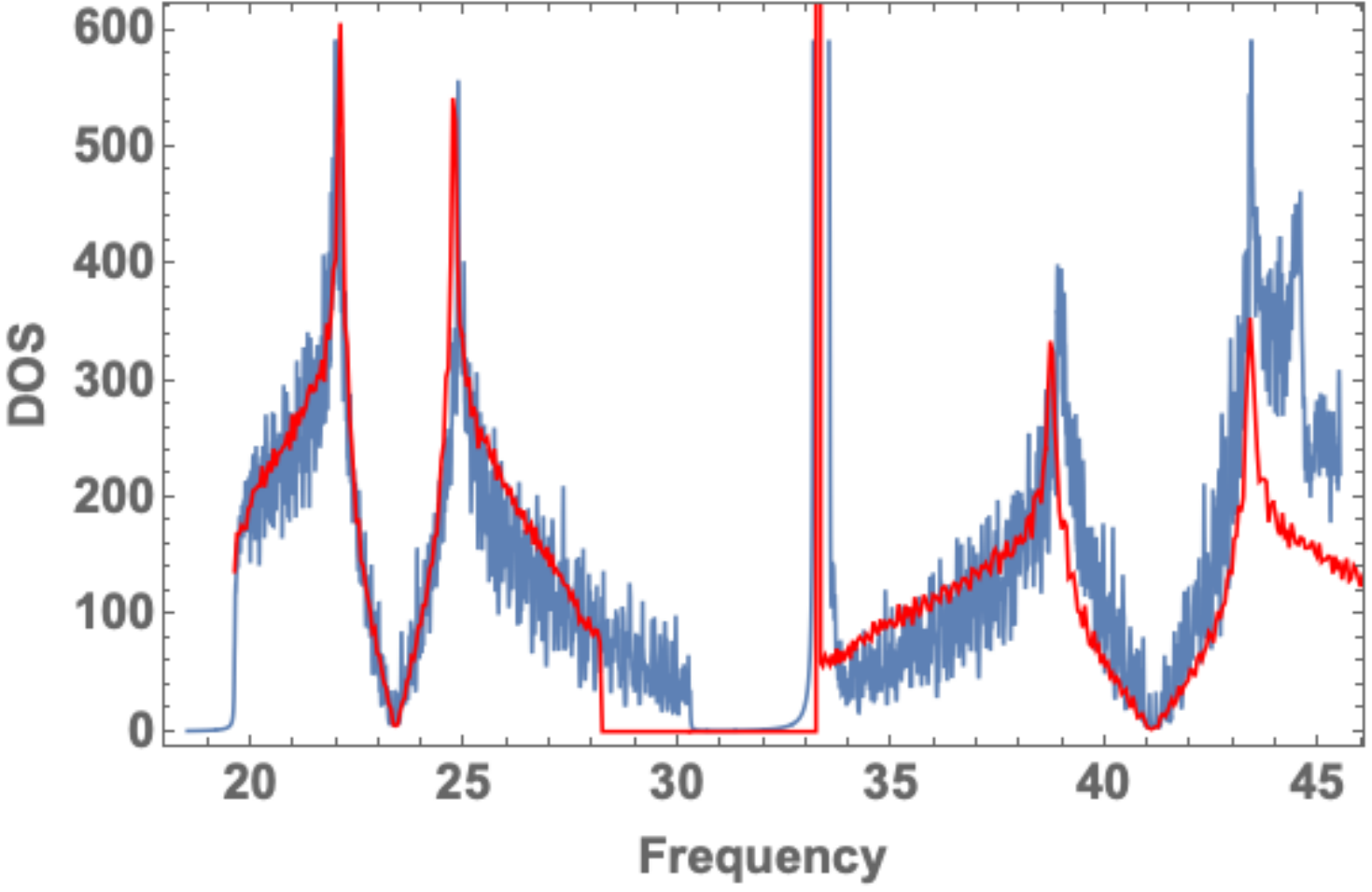}\hspace{4mm}\includegraphics[width=0.36\columnwidth]{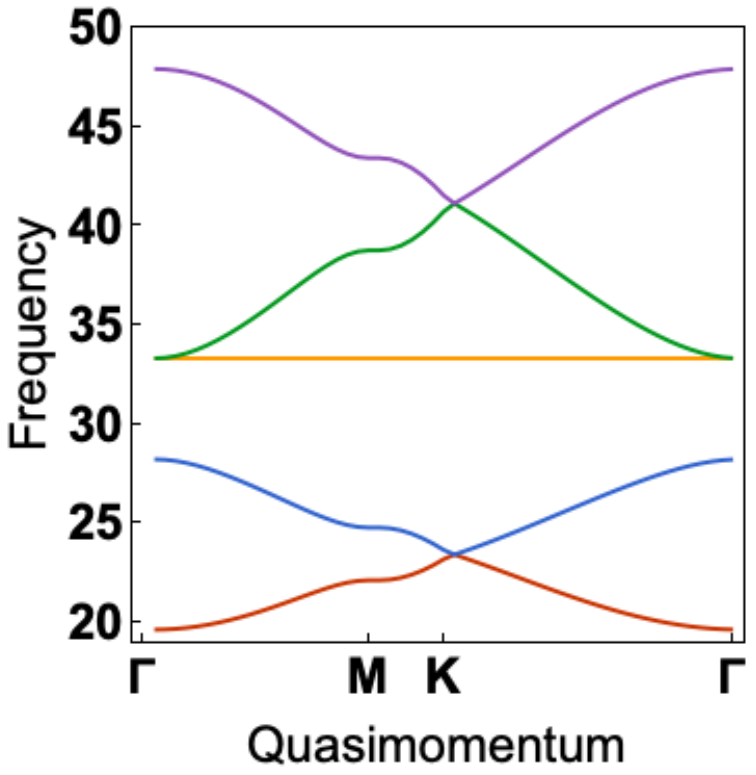}
\caption[Plain Monte Carlo fitting with three hoppings]{Fitting with three hoppings and minimizing the deviations by plain the Monte Carlo method. (Left) Density of states, where the red line is the spectrum of the honome lattice and the blue line is the experimental spectrum. (Right) Band structure of the honome lattice along the path $\Gamma,M,K,\Gamma$ inside the first Brillouin zone (see Fig. \ref{figdos}). The values of the parameters are $t_{1}=1$ and $(\epsilon,t_{2},t_{3})=(-4.79586, 0.803977, -0.431564)$.}
 \label{fitting_3hoppings_1}
\end{figure}

Using the parameters obtained from the plain Monte Carlo method, we then tried the importance sampling method in order to further minimize the deviations in \eqref{eq:diff_with_exp}. Results are shown in Fig. \ref{imp_samp_fitting_3hoppings-1}.
\begin{figure}[htb!]
\centering
\includegraphics[width=0.55\columnwidth]{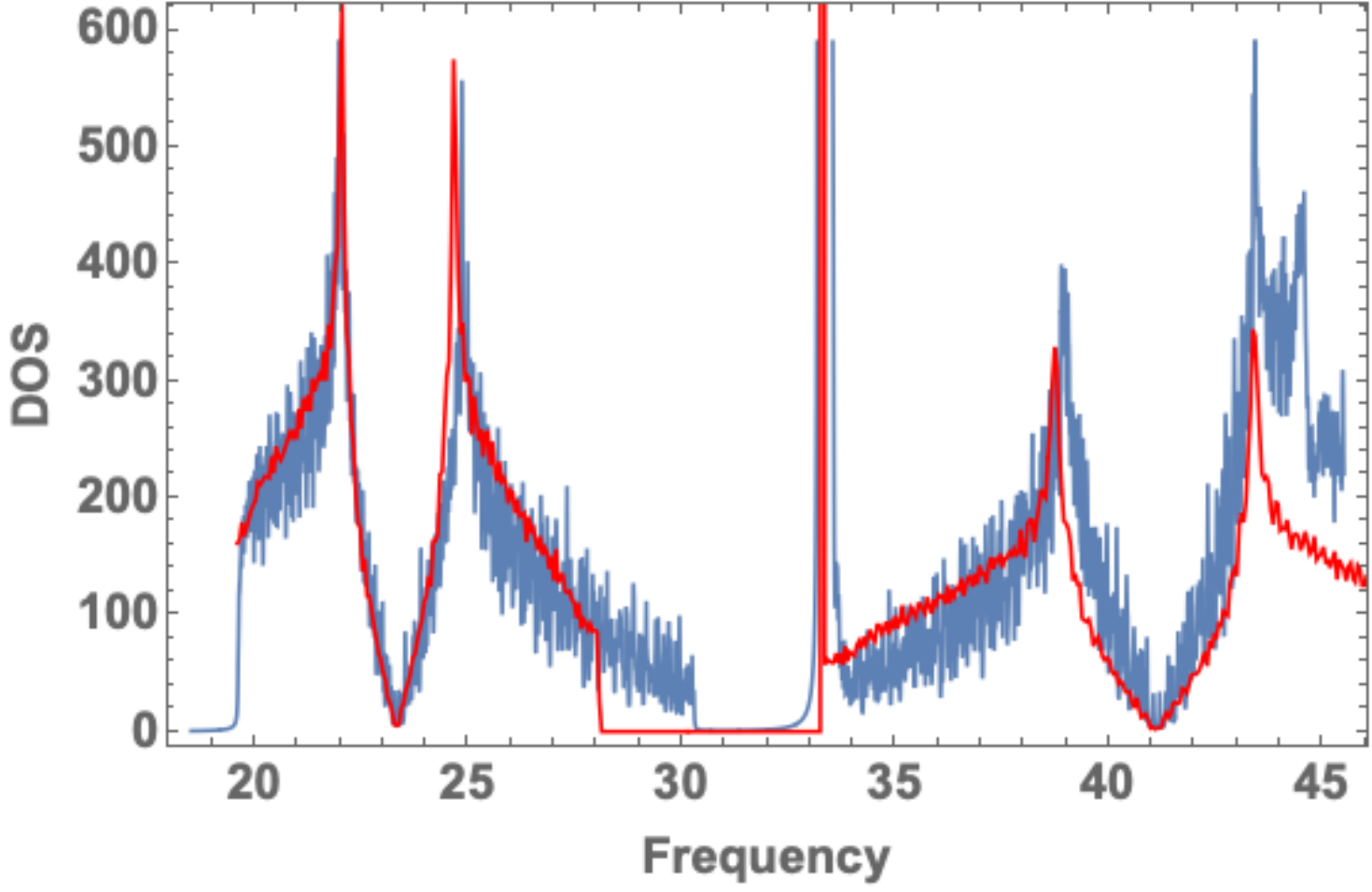}\hspace{4mm}\includegraphics[width=0.36\columnwidth]{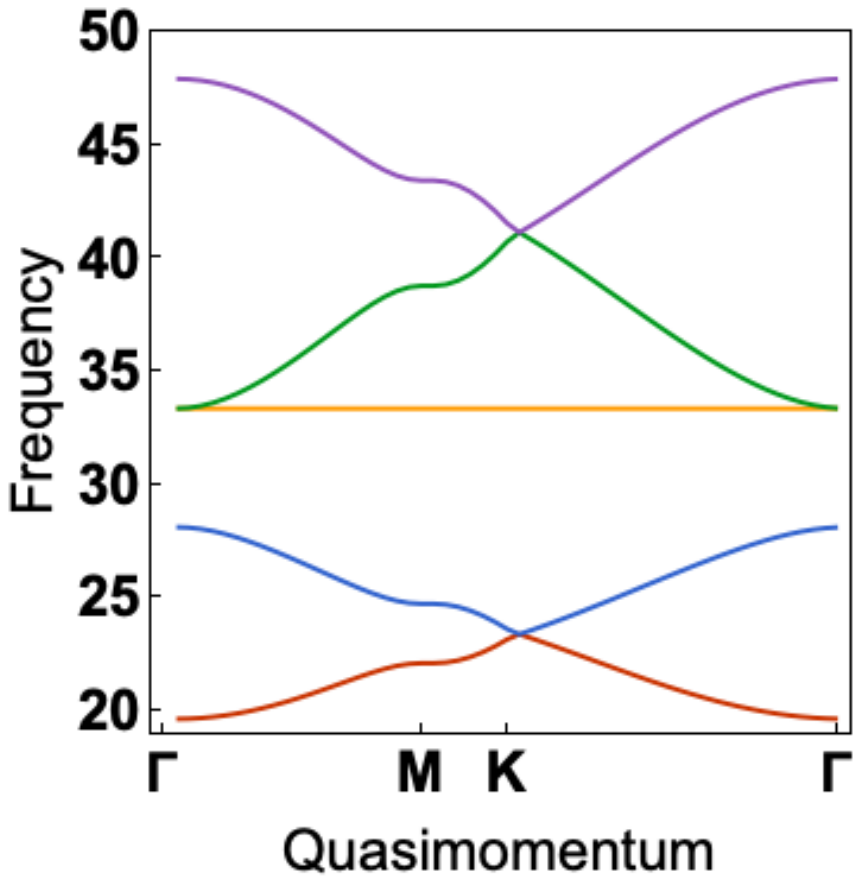}
\caption[Importance sampling fitting with three hoppings]{Fitting with three hoppings and minimizing the deviations by the importance sampling method. (Left) Density of states, where the red line is the spectrum of the honome lattice and the blue line is the experimental spectrum. (Right) Band structure of the honome lattice along the path $\Gamma,M,K,\Gamma$ inside the first Brillouin zone (see Fig. \ref{figdos}). The values of the parameters are $t_{1}=1$ and $(\epsilon,t_{2},t_{3})=(-4.88694, 0.812533, -0.43105)$.}
 \label{imp_samp_fitting_3hoppings-1}
\end{figure} 
We also tried to minimize the standard deviation in \eqref{eq:standard_diviation}, which gives similar results. 

\subsection{Fitting with five hoppings}

In the case of up to 5th nearest neighbor hoppings, only $t_i, i=1,\dots,5$ are non-zero in Eqs. \eqref{eq:hop-mats} and \eqref{eq:hk}, and one can compute analytically the eigenvalues of $H(k)$ at $\Gamma,\ M$ points, but not at the $K$ point, where the eigenvalues are computed numerically. 


\begin{figure}[htb!]
\centering 
\includegraphics[width=0.55\columnwidth]{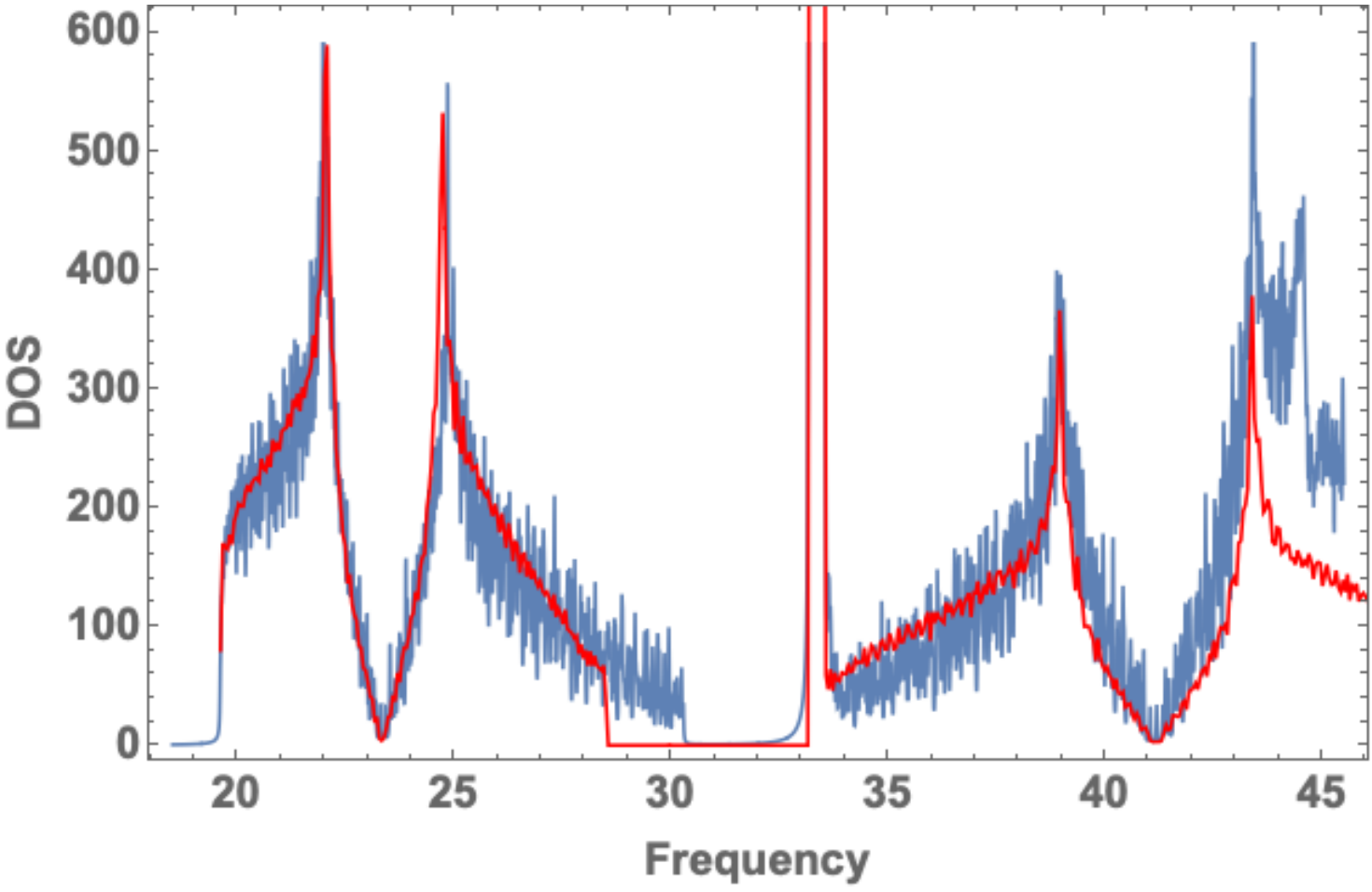}\hspace{4mm}\includegraphics[width=0.36\columnwidth]{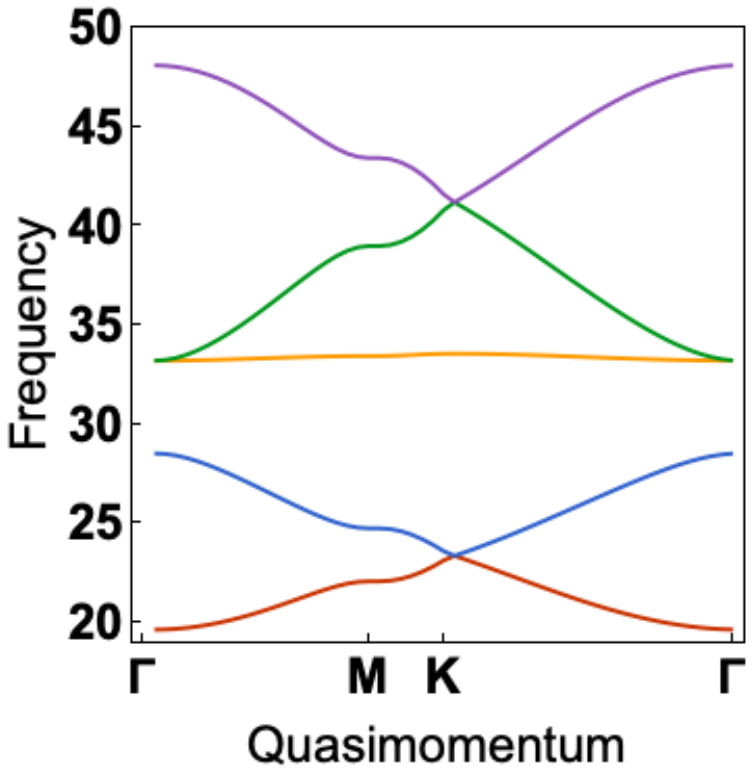} 
\caption[Plain Monte Carlo fitting with five hoppings]{Fitting with five hoppings and minimizing the deviations by the plain Monte Carlo method. (Left) Density of states, where the red line is the spectrum of the honome lattice and the blue line is the experimental spectrum. (Right) Band structure of the honome lattice along the path $\Gamma,M,K,\Gamma$ inside the first Brillouin zone (see Fig. \ref{figdos}). The values of the parameters are $t_{1}=1$ and $(\epsilon, t_{2},t_{3}, t_{4}, t_{5})=(-4.98748, 0.847096, -0.486464, -0.0425766, 0.0225378)$.}
 \label{fitting_5hoppings_1}
\end{figure} 

 Then using Eqs. \eqref{eq:rel-pos} and \eqref{eq:span_of_spec}, we minimize the deviations in Eq. \eqref{eq:diff_with_exp} via the plain Monte Carlo method to get the hopping parameters that best fit the experimental data. The result shows a little broadening of the FB, which is closer to the experimental result, as shown in Fig. \ref{fitting_5hoppings_1}.  Minimizing the standard deviation in Eq. \eqref{eq:standard_diviation} produces a similar result. 

\begin{figure}[htb!]
\centering
\includegraphics[width=0.55\columnwidth]{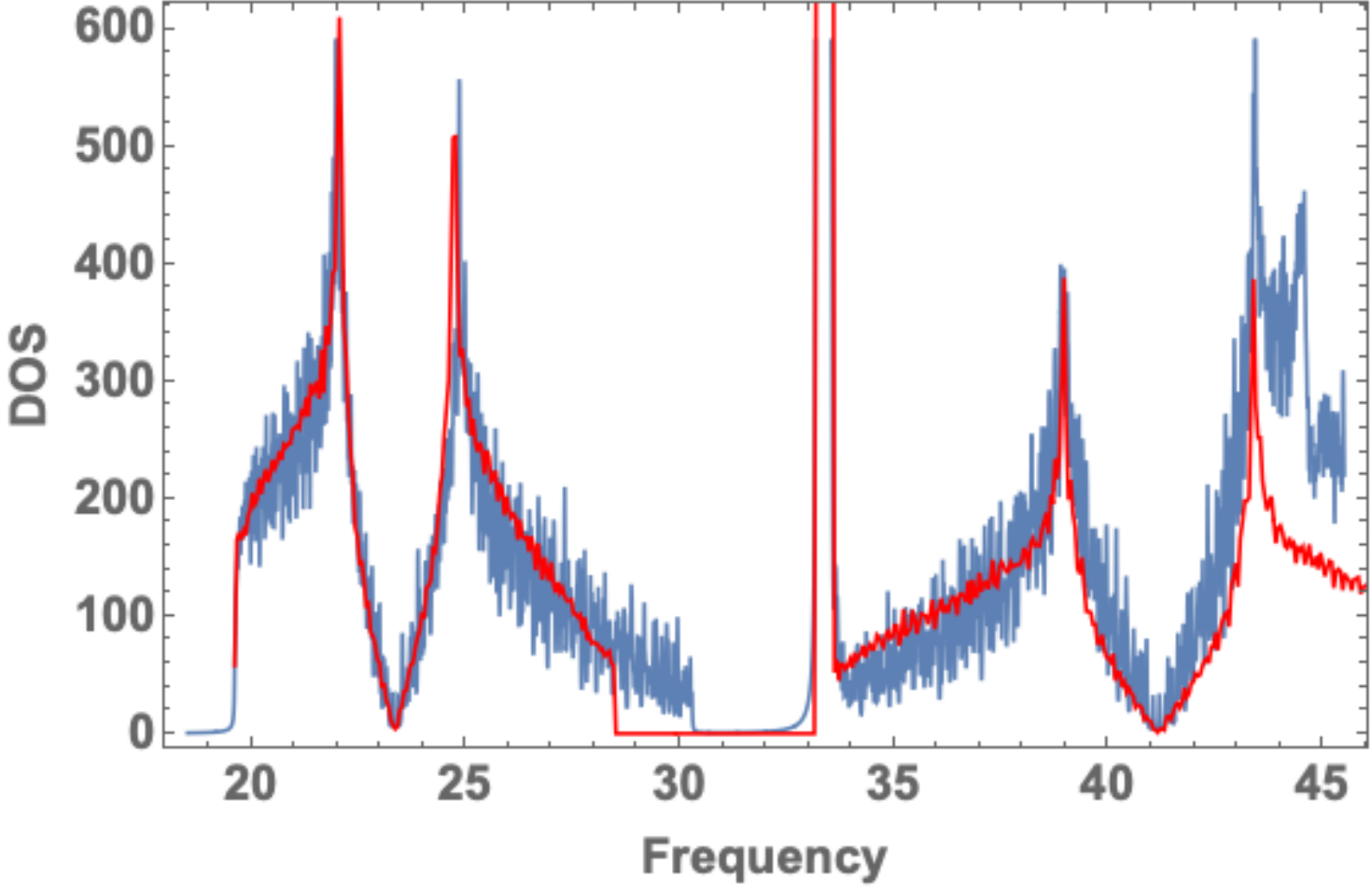}\hspace{4mm}\includegraphics[width=0.36\columnwidth]{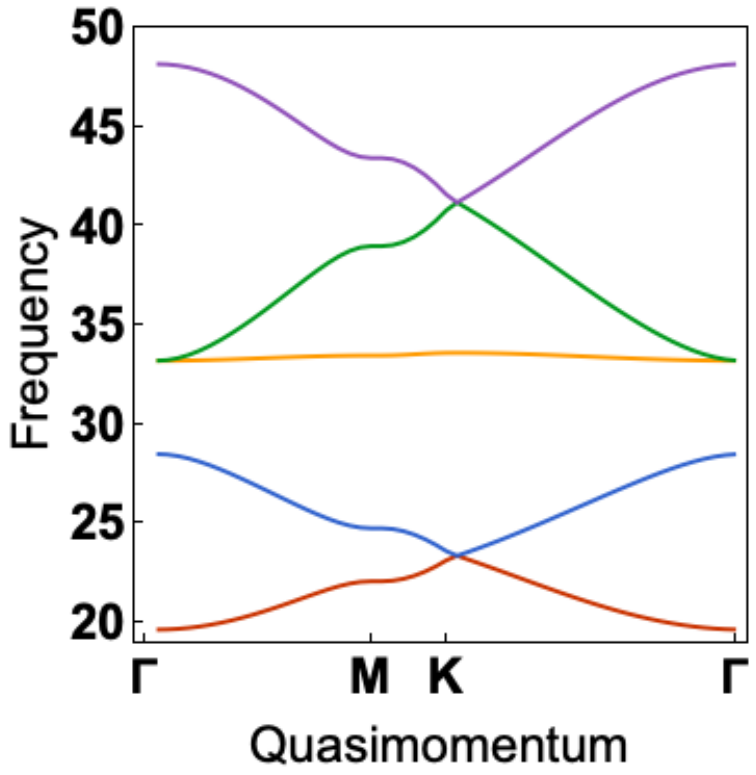}
\caption[Importance sampling fitting with five hoppings]{Fitting with five hoppings and minimizing the deviations by the importance sampling method. (Left) Density of states, where the red line is the spectrum of the honome lattice and the blue line is the experimental spectrum. (Right) Band structure of the honome lattice along the path $\Gamma,M,K,\Gamma$ inside the first Brillouin zone (see Fig. \ref{figdos}). The values of the parameters are $t_{1}=1$ and $(\epsilon,t_{2},t_{3},t_{4},t_{5})=(-5.02261, 0.852182, -0.486464, -0.04, 0.0256164)$.}
 \label{imp_samp_fitting_5hoppings-1}
\end{figure} 

Using the parameters obtained from the plain Monte Carlo method, we tried the importance sampling method in order to further minimize the deviations in Eq. \eqref{eq:diff_with_exp}. Results are shown in Fig, \ref{imp_samp_fitting_5hoppings-1}. 
We again tried to minimize the standard deviation in Eq. \eqref{eq:standard_diviation}, and we got a result similar to the one achieved from minimizing the deviations.

Both methods produce the fits that show good agreement with the experiment, as can be seen in Figs. \ref{fitting_3hoppings_1}--\ref{imp_samp_fitting_5hoppings-1}. In particular, the spectrum of our model, the honome lattice, also displays the Dirac points at two different energies, in agreement with the experiment. The lower Dirac point arises from the honeycomb sublattice, and the upper Dirac point arises from the kagome sublattice. The singularity in the experimental spectrum comes from the FB of the honome lattice. This singularity (i.e. the FB) has a little broadening, as we show in the case of five hoppings, which is resulted from the slight distortion of the FB by long-range hoppings, i.e. $t_4,\ t_5$ in our examples (see Figs. \ref{fitting_5hoppings_1}, \ref{imp_samp_fitting_5hoppings-1}).



\section{Summary}

In this chapter, we introduced a tight-binding model for a 2D microwave photonic crystal. Patterns of eigenfuctions obtained from numerical diagonalization indicates bipartite structure of the nodes of the eigenfunctions. This suggested to search for effective tight-binding description based on chiral lattices, that were discussed before, and lead us to study the tight-binding model on the combination of the honeycomb and kagome lattices, that was dubbed honome lattice. Then the honome lattice was used to describe the spectral properties of the microwave photonic crystal. We developed a numerical algorithm to fit an experimental spectrum with the density of states of the honome lattice, and our fitting result shows good agreement with the experimental observation. Particularly, the singularity in the experimental spectrum is well explained by the singularity of the density of states of the honome lattice around the FB energy. In this work, we used block matrix formalism to help simplify the fitting algorithm. Moreover, we implemented the fine-tuning property of FB Hamiltonians to achieve the best fit. This work proved that the concepts and methods used in our FB generator are useful in practical applications.

\chapter{Conclusions and outlook}

\ifpdf
    \graphicspath{{Chapter8/Figs/Raster/}{Chapter8/Figs/PDF/}{Chapter8/Figs/}}
\else
    \graphicspath{{Chapter8/Figs/Vector/}{Chapter8/Figs/}}
\fi 

Tight-binding models with dispersionless energy bands in their single particle spectrum have proven their importance in various applications and diverse settings ranging from the Hubbard model to Bose--Einstein condensates and photonics. Macroscopic degeneracy of such systems allows physicists to tune flatbands (FBs) to various exotic phases of matter, with the help of different perturbations. The FB models studied thus far are limited to some specific classes that have special lattice geometry and symmetry. However, there are a vast amount of FB Hamiltonians that have not been uncovered so far. The absence of systematic classification and construction methods for FB Hamiltonians led us to work on this thesis. 

We introduced in this thesis a systematic classification of FB lattices and developed novel methods to generate FB Hamiltonians with specific compact localized state (CLS) properties. Our ultimate aim was to identify all possible FB Hamiltonians of different classes in 1D and 2D systems.  

After reviewing the basic concepts of the tight-binding model, FBs, and macroscopic degeneracy, we reviewed existing FB construction methods, and applications and realizations of FBs. Next we proposed a classification approach for FBs according to the size and shape of their irreducible CLSs. Furthermore, we studied the properties of CLSs, and with the help of block matrix representation, we formalized destructive interference conditions and CLS existence conditions in a mathematically rigorous way. Based on these conditions, we developed a test procedure to identify the CLS class of a given FB Hamiltonian, and the inversion of this test procedure led us to the core idea of a FB generator. 

Having the essential tools and the idea of FB generation, we realized the complete classification and systematic generation of FB Hamiltonians in 1D that cover all possible Hamiltonians possessing at least one FB. Particularly, we have achieved the complete parameterization of all FB networks in the two-band case. Remarkably, we found that all two-band FB networks can be mapped into generalized sawtooth lattices using local unitary transformations. Extension of the two-band case to higher numbers of bands in 1D was more challenging, for which we developed a new method that we called the inverse eigenvalue method. This method yielded analytic solutions for $U=2$ FBs and numerical solutions for $U>2$ cases. With this method, in principle, all possible FBs in 1D can be generated. 

Moving on to higher dimensions, we extended our approach to 2D. We started by introducing a classification of 2D FB lattices whose CLSs occupy a maximum of four unit cells in a $2\times2$ plaquette. Implementing the methods used for 1D, we obtained analytic solutions for various CLS classes in the plaquette. We were able to cover most known examples with our analytic solutions. 

Extending the idea of our FB generator to the non-Hermitian regime, we considered a 1D two-band network with non-Hermitian couplings. Using $k$-space representation, we identified the conditions required to have a FB. Solving these conditions, we acquired analytic solutions for non-Hermitian FB Hamiltonians. We were able to generate completely flat FBs, partially flat FBs, and the special case of modulus-flat FBs. Unlike conventional methods, our generator does not rely on a particular symmetry, such as $\mathcal{PT}$ symmetry. 

Finally, turning to applications, we applied our methodologies to explain the spectral properties of a microwave photonic crystal. We designed a novel tight-binding model, which we term a honome lattice, to fit the experimental data obtained from spectral measurements. In particular, we were able to explain the singularity in the experimental spectrum in terms of the singularity in the density of states of the honome coming from the FB. We developed two algorithms that both fit the experimental data with nice accuracy. 

The topics explored in this thesis will impact FB-related research in a number of aspects. First of all, our FB generator provides more flexibility in the experimental design of FB lattices, because this method does not require special lattice geometry or symmetry, and allows for the fine-tuning of parameters. By tuning the control parameters in our generator, it is possible to design FB modes that suit the experimental conditions.     

Moreover, our approaches are based on real-space analysis. The analytic solutions provided here for FB Hamiltonians in 1D and 2D lattices will enable researchers to further understand the properties of FB lattices. In particular, from our 2D analytic solution, band touching properties can be studied more systematically using real-space analysis instead of $k$-space analysis~\cite{rhim2018classification}. This will give more insight into the band topology of 2D FB lattices from the real-space perspective.  

In addition to these effects on ongoing research, our work will open new topics for future studies. First, some of the conjectures introduced here require strict proof. Should these conjectures be proved wrong, finding counter-examples will be interesting and will lead to new questions. Along these lines, following from the theorems proven here for the 1D case, rigorous proof for 2D and higher dimensions is desirable; to do so, extending the FB generator to more generic cases, like arbitrary CLS sizes, can help to carry over these theorems to even higher dimensions. This is also appealing given that in higher dimensions, more interesting phenomena are expected. Further, our approach is based on real-space analysis, and extending the current approach to 3D will provide a theoretical platform for designing real materials with FBs as well as searching for natural compounds having FBs. 

Our FB generator in non-Hermitian systems can be extended to higher numbers of bands in 1D and higher lattice dimensions. Interestingly, we have found a non-trivial winding number of a modulus-flat FB. The consequence of such a non-trivial winding number is an open question that demands investigation. Moreover, extending the current approach is encouraging as there is plenty of room to connect our results in non-Hermitian systems to more generic cases. In this way, there are myriad things to be investigated such as system dynamics, topology, perturbation effects, and so on.

Recently, nearly FBs found in bilayer graphene systems have been drawing attention~\cite{marchenkoeaau2018extremely,chebrolu2019bilayer,wolf2019electrically}. Such systems host a FB when two layers are tilted with respect to each other and a particular tilting angle is reached. The origin of such a FB is not clear and leads to a number of questions: What is the fate of CLSs? If CLSs exist in such systems, what is the CLS class? What are the CLS properties? All of these questions are open for further studies, and our methods in this work may provide a tool to answer them.  

This thesis focused on the single particle (non-interacting) systems. Exploring the effects interactions is an interesting and promising topic. One example is superconductivity in FB, that has been drawing attention recently~\cite{volovik2018graphite,kauppila2016flatband,xu2018topological}. Our results rise a number of interesting questions: How does the Bardeen-Cooper-Schrieffer (BCS) interaction interplays with flatbands in general? Does it always lead to superconductivity? Is the superconductivity always enhanced? How does the enhancement depend on the type of the flatband? 

To sum up, our methods enhance our understanding of the properties of FB systems, provide a flexible tool to design FB lattices suiting experimental conditions, and open new interesting topics for future research.

\bibliographystyle{plain}
\addcontentsline{toc}{chapter}{Bibliography}
\bibliography{references}

\begin{thebibliography}{100}

\bibitem{abilio1999magnetic}
C.~C. Abilio, P.~Butaud, Th. Fournier, B.~Pannetier, J.~Vidal, S.~Tedesco, and
  B.~Dalzotto.
\newblock Magnetic field induced localization in a two-dimensional
  superconducting wire network.
\newblock {\em Phys. Rev. Lett.}, 83:5102--5105, Dec 1999.

\bibitem{aharanov1959significance}
Y.~Aharonov and D.~Bohm.
\newblock Significance of electromagnetic potentials in the quantum theory.
\newblock {\em Phys. Rev.}, 115:485--491, Aug 1959.

\bibitem{an2017flux}
F.~A. An, E.~J. Meier, and B.~Gadway.
\newblock Flux-dependent localization in a disordered flat-band lattice.
\newblock 05 2017.

\bibitem{aoki2004design}
Hideo Aoki.
\newblock Design of electron correlation effects in interfaces and
  nanostructures.
\newblock {\em Applied Surface Science}, 237(1):2 -- 12, 2004.
\newblock Proceedings of the Seventh International Symposium on Atomically
  Controlled Surfaces, Interfaces and Nanostructures.

\bibitem{apaja2010flat}
V.~Apaja, M.~Hyrk\"as, and M.~Manninen.
\newblock Flat bands, dirac cones, and atom dynamics in an optical lattice.
\newblock {\em Phys. Rev. A}, 82:041402, Oct 2010.

\bibitem{arita2002gateinduced}
Ryotaro Arita, Yuji Suwa, Kazuhiko Kuroki, and Hideo Aoki.
\newblock Gate-induced band ferromagnetism in an organic polymer.
\newblock {\em Phys. Rev. Lett.}, 88:127202, Mar 2002.

\bibitem{arita2003flatband}
Ryotaro Arita, Yuji Suwa, Kazuhiko Kuroki, and Hideo Aoki.
\newblock Flat-band ferromagnetism in undoped and doped polyaminotriazole
  crystal.
\newblock {\em Phys. Rev. B}, 68:140403, Oct 2003.

\bibitem{Ashcroft76}
N.~W. Ashcroft and N.~D. Mermin.
\newblock {\em {S}olid {S}tate {P}hysics}.
\newblock Holt-Saunders, 1976.

\bibitem{baba2008slow}
Toshihiko Baba.
\newblock Slow light in photonic crystals.
\newblock {\em Nat. Phot.}, 2:465 EP --, 08 2008.

\bibitem{baboux2016bosonic}
F.~Baboux, L.~Ge, T.~Jacqmin, M.~Biondi, E.~Galopin, A.~Lema\^{\i}tre,
  L.~Le~Gratiet, I.~Sagnes, S.~Schmidt, H.~E. T\"ureci, A.~Amo, and J.~Bloch.
\newblock Bosonic condensation and disorder-induced localization in a flat
  band.
\newblock {\em Phys. Rev. Lett.}, 116:066402, Feb 2016.

\bibitem{Bagarello2015NonSelfadjoint}
F~Bagarello, J.P. Gazeau, F.H. Szafraniec, and Miloslav Znojil.
\newblock {\em Non-Selfadjoint Operators in Quantum Physics: Mathematical
  Aspects}.
\newblock 07 2015.

\bibitem{barrett2017equitable}
Wayne Barrett, Amanda Francis, and Benjamin Webb.
\newblock Equitable decompositions of graphs with symmetries.
\newblock {\em Linear Algebra and its Applications}, 513:409 -- 434, 2017.

\bibitem{belicev2017localized}
P.~P. Beli\ifmmode~\check{c}\else \v{c}\fi{}ev,
  G.~Gligori\ifmmode~\acute{c}\else \'{c}\fi{}, A.~Maluckov,
  M.~Stepi\ifmmode~\acute{c}\else \'{c}\fi{}, and M.~Johansson.
\newblock Localized gap modes in nonlinear dimerized lieb lattices.
\newblock {\em Phys. Rev. A}, 96:063838, Dec 2017.

\bibitem{Slonczewski1958}
Matthieu Bellec, Ulrich Kuhl, Gilles Montambaux, and Fabrice Mortessagne.
\newblock Tight-binding couplings in microwave artificial graphene.
\newblock {\em Phys. Rev. B}, 88:115437, Sep 2013.

\bibitem{Bellec2013a}
Matthieu Bellec, Ulrich Kuhl, Gilles Montambaux, and Fabrice Mortessagne.
\newblock Tight-binding couplings in microwave artificial graphene.
\newblock {\em Phys. Rev. B}, 88:115437, Sep 2013.

\bibitem{Bellec2013}
Matthieu Bellec, Ulrich Kuhl, Gilles Montambaux, and Fabrice Mortessagne.
\newblock Topological transition of dirac points in a microwave experiment.
\newblock {\em Phys. Rev. Lett.}, 110:033902, Jan 2013.

\bibitem{ben2003generalized}
Adi Ben-Israel and Thomas~NE Greville.
\newblock {\em Generalized inverses: theory and applications}, volume~15.
\newblock Springer Science \& Business Media, 2003.

\bibitem{Bender2007makingsense}
Carl~M Bender.
\newblock Making sense of non-hermitian hamiltonians.
\newblock {\em Reports on Progress in Physics}, 70(6):947--1018, may 2007.

\bibitem{bender1998realspectra}
Carl~M. Bender and Stefan Boettcher.
\newblock Real spectra in non-hermitian hamiltonians having
  $\mathcal{P}\mathcal{T}$ symmetry.
\newblock {\em Phys. Rev. Lett.}, 80:5243--5246, Jun 1998.

\bibitem{Berry2004nonhermitian}
M.V. Berry.
\newblock Physics of nonhermitian degeneracies.
\newblock {\em Czechoslovak Journal of Physics}, 54(10):1039--1047, Oct 2004.

\bibitem{biondi2015incompressible}
Matteo Biondi, Evert P.~L. van Nieuwenburg, Gianni Blatter, Sebastian~D. Huber,
  and Sebastian Schmidt.
\newblock Incompressible polaritons in a flat band.
\newblock {\em Phys. Rev. Lett.}, 115:143601, Sep 2015.

\bibitem{Bittner2010}
S.~Bittner, B.~Dietz, M.~Miski-Oglu, P.~Oria~Iriarte, A.~Richter, and
  F.~Sch\"afer.
\newblock Observation of a dirac point in microwave experiments with a photonic
  crystal modeling graphene.
\newblock {\em Phys. Rev. B}, 82:014301, Jul 2010.

\bibitem{Bittner2012}
S.~Bittner, B.~Dietz, M.~Miski-Oglu, and A.~Richter.
\newblock Extremal transmission through a microwave photonic crystal and the
  observation of edge states in a rectangular dirac billiard.
\newblock {\em Phys. Rev. B}, 85:064301, Feb 2012.

\bibitem{bodyfelt2014flatbands}
Joshua~D. Bodyfelt, Daniel Leykam, Carlo Danieli, Xiaoquan Yu, and Sergej
  Flach.
\newblock Flatbands under correlated perturbations.
\newblock {\em Phys. Rev. Lett.}, 113:236403, Dec 2014.

\bibitem{boley1987survey}
D.~Boley and G.~H. Golub.
\newblock A survey of matrix inverse eigenvalue problems.
\newblock {\em Inv. Probl.}, 3(4):595, 1987.

\bibitem{Brody2013biorthogonal}
Dorje~C Brody.
\newblock Biorthogonal quantum mechanics.
\newblock {\em Journal of Physics A: Mathematical and Theoretical},
  47(3):035305, dec 2013.

\bibitem{casteels2016probing}
W.~Casteels, R.~Rota, F.~Storme, and C.~Ciuti.
\newblock Probing photon correlations in the dark sites of geometrically
  frustrated cavity lattices.
\newblock {\em Phys. Rev. A}, 93:043833, Apr 2016.

\bibitem{Castro2009}
A.~H. Castro~Neto, F.~Guinea, N.~M.~R. Peres, K.~S. Novoselov, and A.~K. Geim.
\newblock The electronic properties of graphene.
\newblock {\em Rev. Mod. Phys.}, 81:109--162, Jan 2009.

\bibitem{chalker2010anderson}
J.~T. Chalker, T.~S. Pickles, and Pragya Shukla.
\newblock Anderson localization in tight-binding models with flat bands.
\newblock {\em Phys. Rev. B}, 82:104209, Sep 2010.

\bibitem{chebrolu2019bilayer}
Narasimha~Raju Chebrolu, Bheema~Lingam Chittari, and Jeil Jung.
\newblock Flat bands in twisted double bilayer graphene.
\newblock {\em Phys. Rev. B}, 99:235417, Jun 2019.

\bibitem{chen2017disorder}
Rui Chen, Dong-Hui Xu, and Bin Zhou.
\newblock Disorder-induced topological phase transitions on lieb lattices.
\newblock {\em Phys. Rev. B}, 96:205304, Nov 2017.

\bibitem{chern2015pt}
Gia-Wei Chern and Avadh Saxena.
\newblock Pt-symmetric phase in kagome-based photonic lattices.
\newblock {\em Opt. Lett.}, 40(24):5806--5809, Dec 2015.

\bibitem{curtright2007biorghogonal}
Thomas Curtright and Luca Mezincescu.
\newblock Biorthogonal quantum systems.
\newblock {\em Journal of Mathematical Physics}, 48(9):092106, 2007.

\bibitem{danieli2018breather}
C.~Danieli, A.~Maluckov, and S.~Flach.
\newblock Compact discrete breathers on flat-band networks.
\newblock {\em Low Temp. Phys.}, 44(7):678--687, 2018.

\bibitem{danieli2015flatband}
Carlo Danieli, Joshua~D. Bodyfelt, and Sergej Flach.
\newblock Flat-band engineering of mobility edges.
\newblock {\em Phys. Rev. B}, 91:235134, Jun 2015.

\bibitem{derzhko2015strongly}
Oleg Derzhko, Johannes Richter, and Mykola Maksymenko.
\newblock Strongly correlated flat-band systems: The route from heisenberg
  spins to hubbard electrons.
\newblock {\em Int. J. Mod. Phys. B}, 29(12):1530007, 2015.

\bibitem{dias2015origami}
R.~G. Dias and J.~D. Gouveia.
\newblock Origami rules for the construction of localized eigenstates of the
  hubbard model in decorated lattices.
\newblock {\em Sci. Rep.}, 5:16852 EP --, 11 2015.

\bibitem{Dietz2013}
B.~Dietz, F.~Iachello, M.~Miski-Oglu, N.~Pietralla, A.~Richter, L.~von Smekal,
  and J.~Wambach.
\newblock Lifshitz and excited-state quantum phase transitions in microwave
  dirac billiards.
\newblock {\em Phys. Rev. B}, 88:104101, Sep 2013.

\bibitem{Dietz2015}
B.~Dietz, T.~Klaus, M.~Miski-Oglu, and A.~Richter.
\newblock Spectral properties of superconducting microwave photonic crystals
  modeling dirac billiards.
\newblock {\em Phys. Rev. B}, 91:035411, Jan 2015.

\bibitem{Dietz2019}
B~Dietz and A~Richter.
\newblock From graphene to fullerene: experiments with microwave photonic
  crystals.
\newblock {\em Physica Scripta}, 94(1):014002, nov 2018.

\bibitem{drost2017topological}
Robert Drost, Teemu Ojanen, Ari Harju, and Peter Liljeroth.
\newblock Topological states in engineered atomic lattices.
\newblock {\em Nat. Phys.}, 13:668 EP --, 03 2017.

\bibitem{efremidis2002optinduction}
Nikos~K. Efremidis, Suzanne Sears, Demetrios~N. Christodoulides, Jason~W.
  Fleischer, and Mordechai Segev.
\newblock Discrete solitons in photorefractive optically induced photonic
  lattices.
\newblock {\em Phys. Rev. E}, 66:046602, Oct 2002.

\bibitem{Ganainy2007theory}
R.~El-Ganainy, K.~G. Makris, D.~N. Christodoulides, and Ziad~H. Musslimani.
\newblock Theory of coupled optical pt-symmetric structures.
\newblock {\em Opt. Lett.}, 32(17):2632--2634, Sep 2007.

\bibitem{ramy2019dawn}
Ramy El-Ganainy, Mercedeh Khajavikhan, Demetrios~N. Christodoulides, and
  Sahin~K. Ozdemir.
\newblock The dawn of non-hermitian optics.
\newblock {\em Communications Physics}, 2(1):37, 2019.

\bibitem{eleuch2018lossgain}
Hichem Eleuch and Ingrid Rotter.
\newblock Loss, gain, and singular points in open quantum systems.
\newblock {\em Advances in Mathematical Physics}, 2018:1--9, 09 2018.

\bibitem{endo2010tight}
Shimpei Endo, Takashi Oka, and Hideo Aoki.
\newblock Tight-binding photonic bands in metallophotonic waveguide networks
  and flat bands in kagome lattices.
\newblock {\em Phys. Rev. B}, 81:113104, Mar 2010.

\bibitem{feigenbaum2010resonant}
Eyal Feigenbaum and Harry~A. Atwater.
\newblock Resonant guided wave networks.
\newblock {\em Phys. Rev. Lett.}, 104:147402, Apr 2010.

\bibitem{feng2017nhphotonics}
Liang Feng, Ramy El-Ganainy, and Li~Ge.
\newblock Non-hermitian photonics based on parity--time symmetry.
\newblock {\em Nature Photonics}, 11(12):752--762, 2017.

\bibitem{flach2014detangling}
Sergej Flach, Daniel Leykam, Joshua~D. Bodyfelt, Peter Matthies, and Anton~S.
  Desyatnikov.
\newblock Detangling flat bands into fano lattices.
\newblock {\em EPL (Europhysics Letters)}, 105(3):30001, 2014.

\bibitem{francis2017extensions}
Amanda Francis, Dallas Smith, Derek Sorensen, and Benjamin Webb.
\newblock Extensions and applications of equitable decompositions for graphs
  with symmetries.
\newblock {\em Linear Algebra and its Applications}, 532:432 -- 462, 2017.

\bibitem{fritscher2016exploring}
E.~Fritscher and V.~Trevisan.
\newblock Exploring symmetries to decompose matrices and graphs preserving the
  spectrum.
\newblock {\em SIAM Journal on Matrix Analysis and Applications},
  37(1):260--289, 2016.

\bibitem{fuse1998fabulous}
T.~Fuse.
\newblock {\em Fabulous Origami Boxes}.
\newblock Japan Publications Trading Company, 1998.

\bibitem{fuse1990unit}
Tomoko Fuse.
\newblock {\em Unit Origami: Multidimensional Transformations}.
\newblock Japan Publications, 1990.

\bibitem{gao2015observation}
T.~Gao, E.~Estrecho, K.~Y. Bliokh, T.~C.~H. Liew, M.~D. Fraser, S.~Brodbeck,
  M.~Kamp, C.~Schneider, S.~H{\"o}fling, Y.~Yamamoto, F.~Nori, Y.~S. Kivshar,
  A.~G. Truscott, R.~G. Dall, and E.~A. Ostrovskaya.
\newblock Observation of non-hermitian degeneracies in a chaotic
  exciton-polariton billiard.
\newblock {\em Nature}, 526:554 EP --, 10 2015.

\bibitem{Gaspard1989}
Pierre Gaspard and Stuart~A. Rice.
\newblock Exact quantization of the scattering from a classically chaotic
  repellor.
\newblock {\em The Journal of Chemical Physics}, 90(4):2255--2262, 1989.

\bibitem{ge2015parity}
Li~Ge.
\newblock Parity-time symmetry in a flat-band system.
\newblock {\em Phys. Rev. A}, 92:052103, Nov 2015.

\bibitem{ge2018non}
Li~Ge.
\newblock Non-hermitian lattices with a flat band and polynomial power increase
  [invited].
\newblock {\em Photon. Res.}, 6(4):A10--A17, Apr 2018.

\bibitem{gligoric2016nonlinear}
G.~Gligori\ifmmode~\acute{c}\else \'{c}\fi{}, A.~Maluckov, Lj.
  Had\ifmmode~\check{z}\else \v{z}\fi{}ievski, Sergej Flach, and Boris~A.
  Malomed.
\newblock Nonlinear localized flat-band modes with spin-orbit coupling.
\newblock {\em Phys. Rev. B}, 94:144302, Oct 2016.

\bibitem{gneiting2018lifetime}
Clemens Gneiting, Zhou Li, and Franco Nori.
\newblock Lifetime of flatband states.
\newblock {\em Phys. Rev. B}, 98:134203, Oct 2018.

\bibitem{goda2006inverse}
Masaki Goda, Shinya Nishino, and Hiroki Matsuda.
\newblock Inverse anderson transition caused by flatbands.
\newblock {\em Phys. Rev. Lett.}, 96:126401, Mar 2006.

\bibitem{Gomes2012}
Kenjiro~K. Gomes, Warren Mar, Wonhee Ko, Francisco Guinea, and Hari~C.
  Manoharan.
\newblock Designer dirac fermions and topological phases in molecular graphene.
\newblock {\em Nature}, 483:306 EP --, 03 2012.

\bibitem{gong2018topological}
Zongping Gong, Yuto Ashida, Kohei Kawabata, Kazuaki Takasan, Sho Higashikawa,
  and Masahito Ueda.
\newblock Topological phases of non-hermitian systems.
\newblock {\em Phys. Rev. X}, 8:031079, Sep 2018.

\bibitem{Graf1992}
H.-D. Gr\"af, H.~L. Harney, H.~Lengeler, C.~H. Lewenkopf, C.~Rangacharyulu,
  A.~Richter, P.~Schardt, and H.~A. Weidenm\"uller.
\newblock Distribution of eigenmodes in a superconducting stadium billiard with
  chaotic dynamics.
\newblock {\em Phys. Rev. Lett.}, 69:1296--1299, Aug 1992.

\bibitem{guo2009observation}
A.~Guo, G.~J. Salamo, D.~Duchesne, R.~Morandotti, M.~Volatier-Ravat, V.~Aimez,
  G.~A. Siviloglou, and D.~N. Christodoulides.
\newblock Observation of $\mathcal{P}\mathcal{T}$-symmetry breaking in complex
  optical potentials.
\newblock {\em Phys. Rev. Lett.}, 103:093902, Aug 2009.

\bibitem{guzman2014experimental}
D~Guzm{\'a}n-Silva, C~Mej{\'\i}a-Cort{\'e}s, M~A Bandres, M~C Rechtsman,
  S~Weimann, S~Nolte, M~Segev, A~Szameit, and R~A Vicencio.
\newblock Experimental observation of bulk and edge transport in photonic lieb
  lattices.
\newblock {\em New J. Phys.}, 16(6):063061, 2014.

\bibitem{hase2018possibility}
I.~Hase, T.~Yanagisawa, Y.~Aiura, and K.~Kawashima.
\newblock Possibility of flat-band ferromagnetism in hole-doped pyrochlore
  oxides ${\mathrm{sn}}_{2}{\mathrm{nb}}_{2}{\mathrm{o}}_{7}$ and
  ${\mathrm{sn}}_{2}{\mathrm{ta}}_{2}{\mathrm{o}}_{7}$.
\newblock {\em Phys. Rev. Lett.}, 120:196401, May 2018.

\bibitem{Hatsugai2015flatband}
Y~Hatsugai, K~Shiraishi, and H~Aoki.
\newblock Flat bands in the weaire{\textendash}thorpe model and silicene.
\newblock {\em New Journal of Physics}, 17(2):025009, feb 2015.

\bibitem{Heiss2004exceptionalpoint}
W~D Heiss.
\newblock Exceptional points of non-hermitian operators.
\newblock {\em Journal of Physics A: Mathematical and General},
  37(6):2455--2464, jan 2004.

\bibitem{Heiss2001chirality}
W.D. Heiss and H.L. Harney.
\newblock The chirality of exceptional points.
\newblock {\em The European Physical Journal D - Atomic, Molecular, Optical and
  Plasma Physics}, 17(2):149--151, Nov 2001.

\bibitem{Hern2006nonhermitian}
E~Hern{\'{a}}ndez, A~J{\'{a}}uregui, and A~Mondrag{\'{o}}n.
\newblock Non-hermitian degeneracy of two unbound states.
\newblock {\em Journal of Physics A: Mathematical and General},
  39(32):10087--10105, jul 2006.

\bibitem{lehur2016many}
Karyn~Le Hur, Lo{\"\i}c Henriet, Alexandru Petrescu, Kirill Plekhanov,
  Guillaume Roux, and Marco Schir{\'o}.
\newblock Many-body quantum electrodynamics networks: Non-equilibrium condensed
  matter physics with light.
\newblock {\em Comptes Rendus Physique}, 17(8):808 -- 835, 2016.

\bibitem{ichimura1998deltachain}
M.~Ichimura, K.~Kusakabe, S.~Watanabe, and T.~Onogi.
\newblock Flat-band ferromagnetism in extended \ensuremath{\Delta}-chain
  hubbard models.
\newblock {\em Phys. Rev. B}, 58:9595--9598, Oct 1998.

\bibitem{jacqmin2014direct}
T.~Jacqmin, I.~Carusotto, I.~Sagnes, M.~Abbarchi, D.~D. Solnyshkov,
  G.~Malpuech, E.~Galopin, A.~Lema\^{\i}tre, J.~Bloch, and A.~Amo.
\newblock Direct observation of dirac cones and a flatband in a honeycomb
  lattice for polaritons.
\newblock {\em Phys. Rev. Lett.}, 112:116402, Mar 2014.

\bibitem{jo2012ultracold}
Gyu-Boong Jo, Jennie Guzman, Claire~K. Thomas, Pavan Hosur, Ashvin Vishwanath,
  and Dan~M. Stamper-Kurn.
\newblock Ultracold atoms in a tunable optical kagome lattice.
\newblock {\em Phys. Rev. Lett.}, 108:045305, Jan 2012.

\bibitem{Joannopoulos2008}
John~D. Joannopoulos, Steven~G. Johnson, Joshua~N. Winn, and Robert~D. Meade.
\newblock {\em Photonic Crystals: Molding the Flow of Light - Second Edition}.
\newblock Princeton University Press, 2008.

\bibitem{johansson2015compactification}
Magnus Johansson, Uta Naether, and Rodrigo~A. Vicencio.
\newblock Compactification tuning for nonlinear localized modes in sawtooth
  lattices.
\newblock {\em Phys. Rev. E}, 92:032912, Sep 2015.

\bibitem{julku2016geometric}
Aleksi Julku, Sebastiano Peotta, Tuomas~I. Vanhala, Dong-Hee Kim, and P\"aivi
  T\"orm\"a.
\newblock Geometric origin of superfluidity in the lieb-lattice flat band.
\newblock {\em Phys. Rev. Lett.}, 117:045303, Jul 2016.

\bibitem{kajiwara2016observation}
Sho Kajiwara, Yoshiro Urade, Yosuke Nakata, Toshihiro Nakanishi, and Masao
  Kitano.
\newblock Observation of a nonradiative flat band for spoof surface plasmons in
  a metallic lieb lattice.
\newblock {\em Phys. Rev. B}, 93:075126, Feb 2016.

\bibitem{kantor1968chiral}
P.~B. Kantor and J.~L. Pietenpol.
\newblock Chiral symmetry and the pion mass.
\newblock {\em Phys. Rev. Lett.}, 21:241--242, Jul 1968.

\bibitem{kauppila2016flatband}
V.~J. Kauppila, F.~Aikebaier, and T.~T. Heikkil\"a.
\newblock Flat-band superconductivity in strained dirac materials.
\newblock {\em Phys. Rev. B}, 93:214505, Jun 2016.

\bibitem{Khanikaev2013}
Alexander~B. Khanikaev, S.~Hossein~Mousavi, Wang-Kong Tse, Mehdi Kargarian,
  Allan~H. MacDonald, and Gennady Shvets.
\newblock Photonic topological insulators.
\newblock {\em Nature Materials}, 12:233 EP --, 12 2012.

\bibitem{khomeriki2016landau}
Ramaz Khomeriki and Sergej Flach.
\newblock Landau-zener bloch oscillations with perturbed flat bands.
\newblock {\em Phys. Rev. Lett.}, 116:245301, Jun 2016.

\bibitem{kimura2002magnetic}
Takashi Kimura, Hiroyuki Tamura, Kenji Shiraishi, and Hideaki Takayanagi.
\newblock Magnetic-field effects on a two-dimensional kagom\'e lattice of
  quantum dots.
\newblock {\em Phys. Rev. B}, 65:081307, Feb 2002.

\bibitem{Kittel2004}
Charles Kittel.
\newblock {\em Introduction to Solid State Physics}.
\newblock Wiley, 8 edition, 2004.

\bibitem{klaiman2008visualization}
Shachar Klaiman, Uwe G\"unther, and Nimrod Moiseyev.
\newblock Visualization of branch points in $\mathcal{P}\mathcal{T}$-symmetric
  waveguides.
\newblock {\em Phys. Rev. Lett.}, 101:080402, Aug 2008.

\bibitem{klembt2017polariton}
S.~Klembt, T.~H. Harder, O.~A. Egorov, K.~Winkler, H.~Suchomel, J.~Beierlein,
  M.~Emmerling, C.~Schneider, and S.~H\"ofling.
\newblock Polariton condensation in s- and p-flatbands in a two-dimensional
  lieb lattice.
\newblock {\em Applied Physics Letters}, 111(23):231102, 2017.

\bibitem{kolovsky2018topological}
A.~R. Kolovsky, A.~Ramachandran, and S.~Flach.
\newblock Topological flat wannier-stark bands.
\newblock {\em Phys. Rev. B}, 97:045120, Jan 2018.

\bibitem{kottos2010broken}
Tsampikos Kottos.
\newblock Broken symmetry makes light work.
\newblock {\em Nature Physics}, 6:166 EP --, 03 2010.

\bibitem{Kuhl2010}
U.~Kuhl, S.~Barkhofen, T.~Tudorovskiy, H.-J. St\"ockmann, T.~Hossain,
  L.~de~Forges~de Parny, and F.~Mortessagne.
\newblock Dirac point and edge states in a microwave realization of
  tight-binding graphene-like structures.
\newblock {\em Phys. Rev. B}, 82:094308, Sep 2010.

\bibitem{lee2009flatband}
Yen-Chen Lee and Hsiu-Hau Lin.
\newblock Flat-band ferromagnetism in armchair graphene nanoribbons.
\newblock {\em Journal of Physics: Conference Series}, 150(4):042110, mar 2009.

\bibitem{letartre2001group}
X.~Letartre, C.~Seassal, C.~Grillet, P.~Rojo-Romeo, P.~Viktorovitch,
  M.~Le~Vassor~d'Yerville, D.~Cassagne, and C.~Jouanin.
\newblock Group velocity and propagation losses measurement in a single-line
  photonic-crystal waveguide on inp membranes.
\newblock {\em Applied Physics Letters}, 79(15):2312--2314, 2001.

\bibitem{leykam2018artificial}
Daniel Leykam, Alexei Andreanov, and Sergej Flach.
\newblock Artificial flat band systems: from lattice models to experiments.
\newblock {\em Adv. Phys.: X}, 3(1):1473052, 2018.

\bibitem{daniel2017edgemodes}
Daniel Leykam, Konstantin~Y. Bliokh, Chunli Huang, Y.~D. Chong, and Franco
  Nori.
\newblock Edge modes, degeneracies, and topological numbers in non-hermitian
  systems.
\newblock {\em Phys. Rev. Lett.}, 118:040401, Jan 2017.

\bibitem{leykam2017flat}
Daniel Leykam, Sergej Flach, and Y.~D. Chong.
\newblock Flat bands in lattices with non-hermitian coupling.
\newblock {\em Phys. Rev. B}, 96:064305, Aug 2017.

\bibitem{li2008systematic}
Juntao Li, Thomas~P. White, Liam O'Faolain, Alvaro Gomez-Iglesias, and
  Thomas~F. Krauss.
\newblock Systematic design of flat band slow light in photonic crystal
  waveguides.
\newblock {\em Opt. Exp.}, 16(9):6227--6232, Apr 2008.

\bibitem{lieb1989two}
Elliott~H. Lieb.
\newblock Two theorems on the hubbard model.
\newblock {\em Phys. Rev. Lett.}, 62:1201--1204, Mar 1989.

\bibitem{maimaiti2017compact}
Wulayimu Maimaiti, Alexei Andreanov, Hee~Chul Park, Oleg Gendelman, and Sergej
  Flach.
\newblock Compact localized states and flat-band generators in one dimension.
\newblock {\em Phys. Rev. B}, 95(11):115135, Mar 2017.

\bibitem{maimaiti2019universal}
Wulayimu Maimaiti, Sergej Flach, and Alexei Andreanov.
\newblock Universal $d=1$ flat band generator from compact localized states.
\newblock {\em Phys. Rev. B}, 99:125129, Mar 2019.

\bibitem{makasyuk2006optinduction}
Igor Makasyuk, Zhigang Chen, and Jianke Yang.
\newblock Band-gap guidance in optically induced photonic lattices with a
  negative defect.
\newblock {\em Phys. Rev. Lett.}, 96:223903, Jun 2006.

\bibitem{makris2008beamdynamics}
K.~G. Makris, R.~El-Ganainy, D.~N. Christodoulides, and Z.~H. Musslimani.
\newblock Beam dynamics in $\mathcal{P}\mathcal{T}$ symmetric optical lattices.
\newblock {\em Phys. Rev. Lett.}, 100:103904, Mar 2008.

\bibitem{maksymenko2012flatband}
M.~Maksymenko, A.~Honecker, R.~Moessner, J.~Richter, and O.~Derzhko.
\newblock Flat-band ferromagnetism as a pauli-correlated percolation problem.
\newblock {\em Phys. Rev. Lett.}, 109:096404, Aug 2012.

\bibitem{marchenkoeaau2018extremely}
D.~Marchenko, D.~V. Evtushinsky, E.~Golias, A.~Varykhalov, Th. Seyller, and
  O.~Rader.
\newblock Extremely flat band in bilayer graphene.
\newblock {\em Science Advances}, 4(11), 2018.

\bibitem{masumoto2012exciton}
Naoyuki Masumoto, Na~Young Kim, Tim Byrnes, Kenichiro Kusudo, Andreas
  L{\"o}ffler, Sven H{\"o}fling, Alfred Forchel, and Yoshihisa Yamamoto.
\newblock Exciton--polariton condensates with flat bands in a two-dimensional
  kagome lattice.
\newblock {\em New J. Phys.}, 14(6):065002, 2012.

\bibitem{mielke1991ferromagnetic}
A~Mielke.
\newblock Ferromagnetic ground states for the hubbard model on line graphs.
\newblock {\em J Phys. A: Math. and Gen.}, 24(2):L73, 1991.

\bibitem{mielke1991ferromagnetism}
A~Mielke.
\newblock Ferromagnetism in the hubbard model on line graphs and further
  considerations.
\newblock {\em J. Phys. A: Math. Gen.}, 24(14):3311, 1991.

\bibitem{mielke1992exact}
A~Mielke.
\newblock Exact results for the u= infinity hubbard model.
\newblock {\em J. Phys. A: Math. Gen.}, 25(24):6507, 1992.

\bibitem{mielke1993ferromagnetism}
Andreas Mielke and Hal Tasaki.
\newblock Ferromagnetism in the hubbard model.
\newblock {\em Comm. Math. Phys.}, 158(2):341--371, Nov 1993.

\bibitem{moiseyev2011nonhermitian}
Nimrod Moiseyev.
\newblock {\em Non-Hermitian Quantum Mechanics}.
\newblock Cambridge University Press, 2011.

\bibitem{morales2016simple}
Luis Morales-Inostroza and Rodrigo~A. Vicencio.
\newblock Simple method to construct flat-band lattices.
\newblock {\em Phys. Rev. A}, 94:043831, Oct 2016.

\bibitem{motruk2012bose}
Johannes Motruk and Andreas Mielke.
\newblock Bose{\textendash}hubbard model on two-dimensional line graphs.
\newblock {\em Journal of Physics A: Mathematical and Theoretical},
  45(22):225206, may 2012.

\bibitem{mukherjee2015observation}
Sebabrata Mukherjee, Alexander Spracklen, Debaditya Choudhury, Nathan Goldman,
  Patrik \"Ohberg, Erika Andersson, and Robert~R. Thomson.
\newblock Observation of a localized flat-band state in a photonic lieb
  lattice.
\newblock {\em Phys. Rev. Lett.}, 114:245504, Jun 2015.

\bibitem{musslimani2008optical}
Z.~H. Musslimani, K.~G. Makris, R.~El-Ganainy, and D.~N. Christodoulides.
\newblock Optical solitons in $\mathcal{P}\mathcal{T}$ periodic potentials.
\newblock {\em Phys. Rev. Lett.}, 100:030402, Jan 2008.

\bibitem{Nadvornik2012}
L~N{\'{a}}dvorn{\'{\i}}k, M~Orlita, N~A Goncharuk, L~Smr{\v{c}}ka,
  V~Nov{\'{a}}k, V~Jurka, K~Hru{\v{s}}ka, Z~V{\'{y}}born{\'{y}}, Z~R
  Wasilewski, M~Potemski, and K~V{\'{y}}born{\'{y}}.
\newblock From laterally modulated two-dimensional electron gas towards
  artificial graphene.
\newblock {\em New Journal of Physics}, 14(5):053002, may 2012.

\bibitem{nakata2012observation}
Yosuke Nakata, Takanori Okada, Toshihiro Nakanishi, and Masao Kitano.
\newblock Observation of flat band for terahertz spoof plasmons in a metallic
  kagom\'e lattice.
\newblock {\em Phys. Rev. B}, 85:205128, May 2012.

\bibitem{naud2001aharonov}
C\'ecile Naud, Giancarlo Faini, and Dominique Mailly.
\newblock Aharonov-bohm cages in 2d normal metal networks.
\newblock {\em Phys. Rev. Lett.}, 86:5104--5107, May 2001.

\bibitem{nishino2005three}
Shinya Nishino and Masaki Goda.
\newblock Three-dimensional flat-band models.
\newblock {\em J. Phys. Soc. Jpn}, 74(1):393--400, 2005.

\bibitem{nishino2003flat}
Shinya Nishino, Masaki Goda, and Koichi Kusakabe.
\newblock Flat bands of a tight-binding electronic system with hexagonal
  structure.
\newblock {\em J. Phys. Soc. Jpn}, 72(8):2015--2023, 2003.

\bibitem{nishino2007flat}
Shinya Nishino, Hiroki Matsuda, and Masaki Goda.
\newblock Flat-band localization in weakly disordered system.
\newblock {\em J Phys. Soc. Jap.}, 76(2):024709, 2007.

\bibitem{nixon2013observing}
Micha Nixon, Eitan Ronen, Asher~A. Friesem, and Nir Davidson.
\newblock Observing geometric frustration with thousands of coupled lasers.
\newblock {\em Phys. Rev. Lett.}, 110:184102, May 2013.

\bibitem{notmi2001extremely}
M.~Notomi, K.~Yamada, A.~Shinya, J.~Takahashi, C.~Takahashi, and I.~Yokohama.
\newblock Extremely large group-velocity dispersion of line-defect waveguides
  in photonic crystal slabs.
\newblock {\em Phys. Rev. Lett.}, 87:253902, Nov 2001.

\bibitem{owen2017dissipation}
E.~T. Owen, O.~T. Brown, and M.~J. Hartmann.
\newblock Dissipation-induced mobility and coherence in frustrated lattices.
\newblock {\em Phys. Rev. A}, 95:063851, Jun 2017.

\bibitem{ozawa2017interaction}
Hideki Ozawa, Shintaro Taie, Tomohiro Ichinose, and Yoshiro Takahashi.
\newblock Interaction-driven shift and distortion of a flat band in an optical
  lieb lattice.
\newblock {\em Phys. Rev. Lett.}, 118:175301, Apr 2017.

\bibitem{Parimi2004}
P.~V. Parimi, W.~T. Lu, P.~Vodo, J.~Sokoloff, J.~S. Derov, and S.~Sridhar.
\newblock Negative refraction and left-handed electromagnetism in microwave
  photonic crystals.
\newblock {\em Phys. Rev. Lett.}, 92:127401, Mar 2004.

\bibitem{peotta2015superfluidity}
Sebastiano Peotta and P{\"{a}}ivi T{\"{o}}rm{\"{a}}.
\newblock Superfluidity in topologically nontrivial flat bands.
\newblock {\em Nat. Comm.}, 6:8944, Nov 2015.

\bibitem{perchikov2017flat}
Nathan Perchikov and O.~V. Gendelman.
\newblock Flat bands and compactons in mechanical lattices.
\newblock {\em Phys. Rev. E}, 96:052208, Nov 2017.

\bibitem{pierucci2015evidence}
Debora Pierucci, Haikel Sediri, Mahdi Hajlaoui, Jean-Christophe Girard, Thomas
  Brumme, Matteo Calandra, Emilio Velez-Fort, Gilles Patriarche, Mathieu~G.
  Silly, Gabriel Ferro, Veronique Souliere, Massimiliano Marangolo, Fausto
  Sirotti, Francesco Mauri, and Abdelkarim Ouerghi.
\newblock {Evidence for Flat Bands near the Fermi Level in Epitaxial
  Rhombohedral Multilayer Graphene}.
\newblock {\em {ACS Nano}}, 9(5):5432--5439, 2015.

\bibitem{poli2017partial}
Charles Poli, Henning Schomerus, Matthieu Bellec, Ulrich Kuhl, and Fabrice
  Mortessagne.
\newblock Partial chiral symmetry-breaking as a route to spectrally isolated
  topological defect states in two-dimensional artificial materials.
\newblock {\em 2D Mat.}, 4(2):025008, 2017.

\bibitem{Polini2013}
Marco Polini, Francisco Guinea, Maciej Lewenstein, Hari~C. Manoharan, and
  Vittorio Pellegrini.
\newblock Artificial honeycomb lattices for electrons, atoms and photons.
\newblock {\em Nature Nanotechnology}, 8:625 EP --, 09 2013.

\bibitem{qi2018defect}
Bingkun Qi, Lingxuan Zhang, and Li~Ge.
\newblock Defect states emerging from a non-hermitian flatband of photonic zero
  modes.
\newblock {\em Phys. Rev. Lett.}, 120:093901, Feb 2018.

\bibitem{qiu2016designing}
Wen-Xuan Qiu, Shuai Li, Jin-Hua Gao, Yi~Zhou, and Fu-Chun Zhang.
\newblock Designing an artificial lieb lattice on a metal surface.
\newblock {\em Phys. Rev. B}, 94:241409, Dec 2016.

\bibitem{radosavljevic2017light}
A.~Radosavljevi\ifmmode~\acute{c}\else \'{c}\fi{},
  G.~Gligori\ifmmode~\acute{c}\else \'{c}\fi{}, P.~P.
  Beli\ifmmode~\check{c}\else \v{c}\fi{}ev, A.~Maluckov, and
  M.~Stepi\ifmmode~\acute{c}\else \'{c}\fi{}.
\newblock Light propagation in binary kagome ribbons with evolving disorder.
\newblock {\em Phys. Rev. E}, 96:012225, Jul 2017.

\bibitem{Raghu2008}
S.~Raghu and F.~D.~M. Haldane.
\newblock Analogs of quantum-hall-effect edge states in photonic crystals.
\newblock {\em Phys. Rev. A}, 78:033834, Sep 2008.

\bibitem{ramachandran2017chiral}
Ajith Ramachandran, Alexei Andreanov, and Sergej Flach.
\newblock Chiral flat bands: Existence, engineering, and stability.
\newblock {\em Phys. Rev. B}, 96:161104, Oct 2017.

\bibitem{ramachandran2018fano}
Ajith Ramachandran, Carlo Danieli, and Sergej Flach.
\newblock {\em Fano Resonances in Flat Band Networks}, pages 311--329.
\newblock Springer International Publishing, Cham, 2018.

\bibitem{Ramezani2017nhinduced}
Hamidreza Ramezani.
\newblock Non-hermiticity-induced flat band.
\newblock {\em Phys. Rev. A}, 96:011802, Jul 2017.

\bibitem{Rechtsman2013a}
Mikael~C. Rechtsman, Yonatan Plotnik, Julia~M. Zeuner, Daohong Song, Zhigang
  Chen, Alexander Szameit, and Mordechai Segev.
\newblock Topological creation and destruction of edge states in photonic
  graphene.
\newblock {\em Phys. Rev. Lett.}, 111:103901, Sep 2013.

\bibitem{Rechtsman2013}
Mikael~C. Rechtsman, Julia~M. Zeuner, Andreas T{\"u}nnermann, Stefan Nolte,
  Mordechai Segev, and Alexander Szameit.
\newblock Strain-induced pseudomagnetic field and photonic landau levels in
  dielectric structures.
\newblock {\em Nature Photonics}, 7:153 EP --, 12 2012.

\bibitem{Reich2002}
S.~Reich, J.~Maultzsch, C.~Thomsen, and P.~Ordej\'on.
\newblock Tight-binding description of graphene.
\newblock {\em Phys. Rev. B}, 66:035412, Jul 2002.

\bibitem{rhim2018classification}
Jun-Won Rhim and Bohm-Jung Yang.
\newblock Classification of flat bands according to the band-crossing
  singularity of bloch wave functions.
\newblock {\em Phys. Rev. B}, 99:045107, Jan 2019.

\bibitem{richter2004exact}
J~Richter, J~Schulenburg, A~Honecker, J~Schnack, and H-J Schmidt.
\newblock Exact eigenstates and macroscopic magnetization jumps in strongly
  frustrated spin lattices.
\newblock {\em Journal of Physics: Condensed Matter}, 16(11):S779--S784, mar
  2004.

\bibitem{rojas-rojas2017quantum}
S.~Rojas-Rojas, L.~Morales-Inostroza, R.~A. Vicencio, and A.~Delgado.
\newblock Quantum localized states in photonic flat-band lattices.
\newblock {\em Phys. Rev. A}, 96:043803, Oct 2017.

\bibitem{roentgen2018compact}
M.~R\"ontgen, C.~V. Morfonios, and P.~Schmelcher.
\newblock Compact localized states and flat bands from local symmetry
  partitioning.
\newblock {\em Phys. Rev. B}, 97:035161, Jan 2018.

\bibitem{rota2017on}
Riccardo Rota, Wim Casteels, and Cristiano Ciuti.
\newblock On the robustness of strongly correlated multi-photon states in
  frustrated driven-dissipative cavity lattices.
\newblock {\em Eur. Phys. J Spec. Topics}, 226(12):2805--2814, Jul 2017.

\bibitem{rotter2018equilibrium}
Ingrid Rotter.
\newblock Equilibrium states in open quantum systems.
\newblock {\em Entropy}, 20(6), 2018.

\bibitem{rotter2009nonhermitian}
Ingrid Rotter.
\newblock A non-hermitian hamilton operator and the physics of open quantum
  systems.
\newblock 42(15):51, Apr 2009.

\bibitem{ruter2010observation}
Christian~E. R{\"u}ter, Konstantinos~G. Makris, Ramy El-Ganainy, Demetrios~N.
  Christodoulides, Mordechai Segev, and Detlef Kip.
\newblock Observation of parity--time symmetry in optics.
\newblock {\em Nature Physics}, 6:192 EP --, 01 2010.

\bibitem{Sadurni2010}
E~Sadurn{\'{\i}}, T~H Seligman, and F~Mortessagne.
\newblock Playing relativistic billiards beyond graphene.
\newblock {\em New Journal of Physics}, 12(5):053014, may 2010.

\bibitem{schmidt2016frustrated}
Sebastian Schmidt.
\newblock Frustrated polaritons.
\newblock {\em Phys. Scrip.}, 91(7):073006, 2016.

\bibitem{Schnack2001independent}
J.~Schnack, H.-J. Schmidt, J.~Richter, and J.~Schulenburg.
\newblock Independent magnon states on magnetic polytopes.
\newblock {\em The European Physical Journal B - Condensed Matter and Complex
  Systems}, 24(4):475--481, Dec 2001.

\bibitem{schulenburg2002macroscopic}
J.~Schulenburg, A.~Honecker, J.~Schnack, J.~Richter, and H.-J. Schmidt.
\newblock Macroscopic magnetization jumps due to independent magnons in
  frustrated quantum spin lattices.
\newblock {\em Phys. Rev. Lett.}, 88:167207, Apr 2002.

\bibitem{schulz2017photonic}
Sebastian~A. Schulz, Jeremy Upham, Liam O'Faolain, and Robert~W. Boyd.
\newblock Photonic crystal slow light waveguides in a kagome lattice.
\newblock {\em Opt. Lett.}, 42(16):3243--3246, Aug 2017.

\bibitem{schwartz2007photonic}
Tal Schwartz, Guy Bartal, Shmuel Fishman, and Mordechai Segev.
\newblock Transport and anderson localization in disordered two-dimensional
  photonic lattices.
\newblock {\em Nature}, 446:52 EP --, Mar 2007.

\bibitem{schwinger1967chiral}
J.~Schwinger.
\newblock Chiral dynamics.
\newblock {\em Physics Letters B}, 24(9):473 -- 476, 1967.

\bibitem{schwinger1957theory}
Julian Schwinger.
\newblock A theory of the fundamental interactions.
\newblock {\em Annals of Physics}, 2(5):407 -- 434, 1957.

\bibitem{shen2010single}
R.~Shen, L.~B. Shao, Baigeng Wang, and D.~Y. Xing.
\newblock Single dirac cone with a flat band touching on line-centered-square
  optical lattices.
\newblock {\em Phys. Rev. B}, 81:041410, Jan 2010.

\bibitem{shiraishi2004theoretical}
Kenji Shiraishi, Hiroyuki Tamura, and Hideaki Takayanagi.
\newblock Theoretical design of a semiconductor ferromagnet based on quantum
  dot superlattices.
\newblock {\em Physica E: Low-dimensional Systems and Nanostructures},
  24(1):107 -- 110, 2004.

\bibitem{shukla2017criticality}
Pragya Shukla.
\newblock Disorder perturbed flat bands: Level density and inverse
  participation ratio.
\newblock {\em Phys. Rev. B}, 98:054206, Aug 2018.

\bibitem{Singha2011}
A.~Singha, M.~Gibertini, B.~Karmakar, S.~Yuan, M.~Polini, G.~Vignale, M.~I.
  Katsnelson, A.~Pinczuk, L.~N. Pfeiffer, K.~W. West, and V.~Pellegrini.
\newblock Two-dimensional mott-hubbard electrons in an artificial honeycomb
  lattice.
\newblock {\em Science}, 332(6034):1176--1179, 2011.

\bibitem{singleton2001band}
J.~Singleton.
\newblock {\em Band Theory and Electronic Properties of Solids}.
\newblock Oxford master series in condensed matter physics. Oxford University
  Press, 2001.

\bibitem{slater1954simplified}
J.~C. Slater and G.~F. Koster.
\newblock Simplified lcao method for the periodic potential problem.
\newblock {\em Phys. Rev.}, 94:1498--1524, Jun 1954.

\bibitem{slot2017experimental}
Marlou~R. Slot, Thomas~S. Gardenier, Peter~H. Jacobse, Guido C.~P. van Miert,
  Sander~N. Kempkes, Stephan J.~M. Zevenhuizen, Cristiane~Morais Smith, Daniel
  Vanmaekelbergh, and Ingmar Swart.
\newblock Experimental realization and characterization of an electronic lieb
  lattice.
\newblock {\em Nat. Phys.}, 13:672 EP --, 04 2017.

\bibitem{Sridhar1991}
S.~Sridhar.
\newblock Experimental observation of scarred eigenfunctions of chaotic
  microwave cavities.
\newblock {\em Phys. Rev. Lett.}, 67:785--788, Aug 1991.

\bibitem{Stoeckmann1990}
H.-J. St\"ockmann and J.~Stein.
\newblock ``quantum'' chaos in billiards studied by microwave absorption.
\newblock {\em Phys. Rev. Lett.}, 64:2215--2218, May 1990.

\bibitem{sutherland1986localization}
Bill Sutherland.
\newblock Localization of electronic wave functions due to local topology.
\newblock {\em Phys. Rev. B}, 34:5208--5211, Oct 1986.

\bibitem{suwa2003flatband}
Yuji Suwa, Ryotaro Arita, Kazuhiko Kuroki, and Hideo Aoki.
\newblock Flat-band ferromagnetism in organic polymers designed by a computer
  simulation.
\newblock {\em Phys. Rev. B}, 68:174419, Nov 2003.

\bibitem{szameit2010discrete}
Alexander Szameit and Stefan Nolte.
\newblock Discrete optics in femtosecond-laser-written photonic structures.
\newblock {\em J Phys. B: At. Mol. Opt. Phys.}, 43(16):163001, 2010.

\bibitem{tadjine2016from}
Athmane Tadjine, Guy Allan, and Christophe Delerue.
\newblock From lattice hamiltonians to tunable band structures by lithographic
  design.
\newblock {\em Phys. Rev. B}, 94:075441, Aug 2016.

\bibitem{taie2017spatial}
S.~Taie, T.~Ichinose, H.~Ozawa, and Y.~Takahashi.
\newblock Spatial adiabatic passage of massive quantum particles.
\newblock 08 2017.

\bibitem{taie2015coherent}
Shintaro Taie, Hideki Ozawa, Tomohiro Ichinose, Takuei Nishio, Shuta Nakajima,
  and Yoshiro Takahashi.
\newblock Coherent driving and freezing of bosonic matter wave in an optical
  lieb lattice.
\newblock {\em Sci. Adv.}, 1(10), 2015.

\bibitem{takeda2004flat}
Hiroyuki Takeda, Tetsuya Takashima, and Katsumi Yoshino.
\newblock Flat photonic bands in two-dimensional photonic crystals with kagome
  lattices.
\newblock {\em J Phys.: Cond. Mat.}, 16(34):6317, 2004.

\bibitem{tamura2002ferromagnetism}
Hiroyuki Tamura, Kenji Shiraishi, Takashi Kimura, and Hideaki Takayanagi.
\newblock Flat-band ferromagnetism in quantum dot superlattices.
\newblock {\em Phys. Rev. B}, 65:085324, Feb 2002.

\bibitem{Tarruell2012}
Leticia Tarruell, Daniel Greif, Thomas Uehlinger, Gregor Jotzu, and Tilman
  Esslinger.
\newblock Creating, moving and merging dirac points with a fermi gas in a
  tunable honeycomb lattice.
\newblock {\em Nature}, 483:302 EP --, 03 2012.

\bibitem{tasaki2008hubbard}
H.~Tasaki.
\newblock Hubbard model and the origin of ferromagnetism.
\newblock {\em Eur. Phys. J. B}, 64(3):365--372, Aug 2008.

\bibitem{tasaki1992ferromagnetism}
Hal Tasaki.
\newblock Ferromagnetism in the hubbard models with degenerate single-electron
  ground states.
\newblock {\em Phys. Rev. Lett.}, 69:1608--1611, Sep 1992.

\bibitem{tasaki1994stability}
Hal Tasaki.
\newblock Stability of ferromagnetism in the hubbard model.
\newblock {\em Phys. Rev. Lett.}, 73:1158--1161, Aug 1994.

\bibitem{tasaki1998fbferromagnetism}
Hal Tasaki.
\newblock {From Nagaoka's Ferromagnetism to Flat-Band Ferromagnetism and
  Beyond: An Introduction to Ferromagnetism in the Hubbard Model}.
\newblock {\em Progress of Theoretical Physics}, 99(4):489--548, 04 1998.

\bibitem{teimourpour2017rubustness}
M.~H. Teimourpour, M.~Khajavikhan, D.~N. Christodoulides, and R.~El-Ganainy.
\newblock Robustness and mode selectivity in parity-time (pt) symmetric lasers.
\newblock {\em Scientific Reports}, 7(1):10756, 2017.

\bibitem{toikka2018necessary}
L~A Toikka and A~Andreanov.
\newblock Necessary and sufficient conditions for flat bands in m-dimensional
  n-band lattices with complex-valued nearest-neighbour hopping.
\newblock {\em Journal of Physics A: Mathematical and Theoretical},
  52(2):02LT04, dec 2018.

\bibitem{murad2018performed}
Murad Tovmasyan, Sebastiano Peotta, Long Liang, P\"aivi T\"orm\"a, and
  Sebastian~D. Huber.
\newblock Preformed pairs in flat bloch bands.
\newblock {\em Phys. Rev. B}, 98:134513, Oct 2018.

\bibitem{tovmasyan2016effective}
Murad Tovmasyan, Sebastiano Peotta, P\"aivi T\"orm\"a, and Sebastian~D. Huber.
\newblock Effective theory and emergent $\text{SU}(2)$ symmetry in the flat
  bands of attractive hubbard models.
\newblock {\em Phys. Rev. B}, 94:245149, Dec 2016.

\bibitem{Uehlinger2013}
Thomas Uehlinger, Gregor Jotzu, Michael Messer, Daniel Greif, Walter
  Hofstetter, Ulf Bissbort, and Tilman Esslinger.
\newblock Artificial graphene with tunable interactions.
\newblock {\em Phys. Rev. Lett.}, 111:185307, Oct 2013.

\bibitem{VanHove1953}
L\'eon Van~Hove.
\newblock The occurrence of singularities in the elastic frequency distribution
  of a crystal.
\newblock {\em Phys. Rev.}, 89:1189--1193, Mar 1953.

\bibitem{vicencio2015observation}
Rodrigo~A. Vicencio, Camilo Cantillano, Luis Morales-Inostroza, Basti\'an Real,
  Cristian Mej\'{\i}a-Cort\'es, Steffen Weimann, Alexander Szameit, and
  Mario~I. Molina.
\newblock Observation of localized states in lieb photonic lattices.
\newblock {\em Phys. Rev. Lett.}, 114:245503, Jun 2015.

\bibitem{vidal1998aharonov}
Julien Vidal, R\'emy Mosseri, and Benoit Dou\ifmmode~\mbox{\c{c}}\else
  \c{c}\fi{}ot.
\newblock Aharonov-bohm cages in two-dimensional structures.
\newblock {\em Phys. Rev. Lett.}, 81:5888--5891, Dec 1998.

\bibitem{volovik2018graphite}
G.~E. Volovik.
\newblock Graphite, graphene, and the flat band superconductivity.
\newblock {\em JETP Letters}, 107(8):516--517, Apr 2018.

\bibitem{Wallace1947}
P.~R. Wallace.
\newblock The band theory of graphite.
\newblock {\em Phys. Rev.}, 71:622--634, May 1947.

\bibitem{wan2017controllable}
Liang-Liang Wan, Xin-You L\"{u}, Jin-Hua Gao, and Ying Wu.
\newblock Controllable photon and phonon localization in optomechanical lieb
  lattices.
\newblock {\em Opt. Express}, 25(15):17364--17374, Jul 2017.

\bibitem{wan2017hybrid}
Liang-Liang Wan, Xin-You L{\"u}, Jin-Hua Gao, and Ying Wu.
\newblock Hybrid interference induced flat band localization in bipartite
  optomechanical lattices.
\newblock {\em Scientific Reports}, 7(1):15188, 2017.

\bibitem{weimann2016transport}
Steffen Weimann, Luis Morales-Inostroza, Basti\'{a}n Real, Camilo Cantillano,
  Alexander Szameit, and Rodrigo~A. Vicencio.
\newblock Transport in sawtooth photonic lattices.
\newblock {\em Opt. Lett.}, 41(11):2414--2417, Jun 2016.

\bibitem{whittaker2018exciton}
C.~E. Whittaker, E.~Cancellieri, P.~M. Walker, D.~R. Gulevich, H.~Schomerus,
  D.~Vaitiekus, B.~Royall, D.~M. Whittaker, E.~Clarke, I.~V. Iorsh, I.~A.
  Shelykh, M.~S. Skolnick, and D.~N. Krizhanovskii.
\newblock Exciton polaritons in a two-dimensional lieb lattice with spin-orbit
  coupling.
\newblock {\em Phys. Rev. Lett.}, 120:097401, Mar 2018.

\bibitem{wimmer2015opticalsoliton}
Martin Wimmer, Alois Regensburger, Mohammad-Ali Miri, Christoph Bersch,
  Demetrios~N. Christodoulides, and Ulf Peschel.
\newblock Observation of optical solitons in pt-symmetric lattices.
\newblock {\em Nature Communications}, 6:7782 EP --, 07 2015.

\bibitem{wolf2019electrically}
T.~M.~R. Wolf, J.~L. Lado, G.~Blatter, and O.~Zilberberg.
\newblock Electrically tunable flat bands and magnetism in twisted bilayer
  graphene.
\newblock {\em Phys. Rev. Lett.}, 123:096802, Aug 2019.

\bibitem{jinhui2014nhdegeneracy}
Jin-Hui Wu, M.~Artoni, and G.~C. La~Rocca.
\newblock Non-hermitian degeneracies and unidirectional reflectionless atomic
  lattices.
\newblock {\em Phys. Rev. Lett.}, 113:123004, Sep 2014.

\bibitem{xia2016demonstration}
Shiqiang Xia, Yi~Hu, Daohong Song, Yuanyuan Zong, Liqin Tang, and Zhigang Chen.
\newblock Demonstration of flat-band image transmission in optically induced
  lieb photonic lattices.
\newblock {\em Opt. Lett.}, 41(7):1435--1438, Apr 2016.

\bibitem{xu2018topological}
Cenke Xu and Leon Balents.
\newblock Topological superconductivity in twisted multilayer graphene.
\newblock {\em Phys. Rev. Lett.}, 121:087001, Aug 2018.

\bibitem{xu2015design}
Changqing Xu, Gang Wang, Zhi~Hong Hang, Jie Luo, C~T Chan, and Yun Lai.
\newblock Design of full-k-space flat bands in photonic crystals beyond the
  tight-binding picture.
\newblock {\em Sci. Rep.}, 5:18181, 2015.

\bibitem{yang2016circuit}
Zi-He Yang, Yan-Pu Wang, Zheng-Yuan Xue, Wan-Li Yang, Yong Hu, Jin-Hua Gao, and
  Ying Wu.
\newblock Circuit quantum electrodynamics simulator of flat band physics in a
  lieb lattice.
\newblock {\em Phys. Rev. A}, 93:062319, Jun 2016.

\bibitem{yuce2015topological}
C.~Yuce.
\newblock Topological phase in a non-hermitian pt symmetric system.
\newblock {\em Physics Letters A}, 379(18):1213 -- 1218, 2015.

\bibitem{zeuner2015topological}
Julia~M. Zeuner, Mikael~C. Rechtsman, Yonatan Plotnik, Yaakov Lumer, Stefan
  Nolte, Mark~S. Rudner, Mordechai Segev, and Alexander Szameit.
\newblock Observation of a topological transition in the bulk of a
  non-hermitian system.
\newblock {\em Phys. Rev. Lett.}, 115:040402, Jul 2015.

\bibitem{Zhang2018nhoptics}
Zhaoyang Zhang, Danmeng Ma, Jiteng Sheng, Yiqi Zhang, Yanpeng Zhang, and Min
  Xiao.
\newblock Non-hermitian optics in atomic systems.
\newblock {\em Journal of Physics B: Atomic, Molecular and Optical Physics},
  51(7):072001, mar 2018.

\bibitem{znojil1999nhharmonic}
Miloslav Znojil.
\newblock $\mathcal{P}\mathcal{T}$-symmetric harmonic oscillators.
\newblock {\em Physics Letters A}, 259(3):220 -- 223, 1999.

\bibitem{zyuzin2018flat}
A.~A. Zyuzin and A.~Yu. Zyuzin.
\newblock Flat band in disorder-driven non-hermitian weyl semimetals.
\newblock {\em Phys. Rev. B}, 97:041203, Jan 2018.

\end{thebibliography}

\appendix
\chapter{Supplementary materials for CLS properties} 

\section{Linear dependence and reducibility of CLSs in 1D (U=3 case)}
\label{app:reducibility-1d}

Consider the $m_{c}=1,\ U=3$ case. Suppose $\vec{\psi}=\left(\vec{\psi_{1}},\vec{\psi_{2}},\vec{\psi_{3}}\right)$
is the compact localized eigenstate, where $\vec{\psi_{i}}$ is a
$\nu$-component vector for the $i$th unit cell with $\nu$ sites. Now
we have the following eigenvalue equations with CLS conditions: 
\begin{equation}
\begin{split}H_{0}\vec{\psi}_{1}+H_{1}\vec{\psi}_{2} & =\efb\vec{\psi}_{1},\\
H_{1}^{\dagger}\vec{\psi}_{1}+H_{0}\vec{\psi}_{2}+H_{1}\vec{\psi}_{3} & =\efb\vec{\psi}_{2},\\
H_{1}^{\dagger}\vec{\psi}_{2}+H_{0}\vec{\psi}_{3} & =\efb\vec{\psi}_{3},\\
H_{1}\vec{\psi}_{1} & =0,\\
H_{1}^{\dagger}\vec{\psi}_{3} & =0.
\end{split}
\label{eq:ev-eq-u3}
\end{equation}
Suppose the components of the CLS $\vec{\psi}$ are linearly dependent as
\begin{equation}
\vec{\psi}_{2}=\alpha\vec{\psi}_{1}+\beta\vec{\psi}_{3},\label{eq:ld-u3}
\end{equation}
and then Eq. (\ref{eq:ev-eq-u3}) becomes 
\begin{equation}
\begin{split}H_{0}\vec{\psi}_{1}+\beta H_{1}\vec{\psi}_{3} & =\efb\vec{\psi}_{1},\\
H_{1}^{\dagger}\vec{\psi}_{1}+H_{0}\left(\alpha\vec{\psi}_{1}+\beta\vec{\psi}_{3}\right)+H_{1}\vec{\psi}_{3} & =\efb\left(\alpha\vec{\psi}_{1}+\beta\vec{\psi}_{3}\right),\\
\alpha H_{1}^{\dagger}\vec{\psi}_{1}+H_{0}\vec{\psi}_{3} & =\efb\vec{\psi}_{3},\\
H_{1}\vec{\psi}_{1} & =0,\\
H_{1}^{\dagger}\vec{\psi}_{3} & =0.
\end{split}
\label{eq:ev-eq-u3-ld}
\end{equation}
 Plugging Eq. (\ref{eq:ld-u3}) into the second equation in
(\ref{eq:ev-eq-u3-ld}) above and using the first and third equations gives
\begin{equation}
\begin{aligned}H_{1}^{\dagger}\vec{\psi}_{1}+H_{0}\left(\alpha\vec{\psi}_{1}+\beta\vec{\psi}_{3}\right)+H_{1}\vec{\psi}_{3} & =\efb\left(\alpha\vec{\psi}_{1}+\beta\vec{\psi}_{3}\right),\\
H_{1}^{\dagger}\vec{\psi}_{1}+\alpha H_{0}\vec{\psi}_{1}+\beta H_{0}\vec{\psi}_{3}+H_{1}\vec{\psi}_{3} & =\alpha\efb\vec{\psi}_{1}+\beta\efb\vec{\psi}_{3},\\
H_{1}^{\dagger}\vec{\psi}_{1}+\alpha H_{0}\vec{\psi}_{1}+\beta H_{0}\vec{\psi}_{3}+H_{1}\vec{\psi}_{3} & =\alpha H_{0}\vec{\psi}_{1}+\alpha\beta H_{1}\vec{\psi}_{3}+\beta\alpha H_{1}^{\dagger}\vec{\psi}_{1}+\beta H_{0}\vec{\psi}_{3}.
\end{aligned}
\end{equation}
We can see that when $\alpha\beta=1$, the above equation is satisfied.
Therefore, any linear dependence of the CLS components cannot lead
to a CLS of U=3, and we should have an extra condition $\alpha\beta=1$
or $\alpha=\nicefrac{1}{\beta}$. Now we show that under this constraint
($\alpha\beta=1$ or $\alpha=\nicefrac{1}{\beta}$), the state $\vec{\psi}=\left(\vec{\psi_{1}},\vec{\psi_{2}},\vec{\psi_{3}}\right)$,
with linear dependence as in (\ref{eq:ld-u3}), is actually a $U=2$
class. 

Suppose we have a CLS $\vec{\psi}=\left(\vec{\psi_{1}},\alpha\vec{\psi}_{1}+\beta\vec{\psi}_{3},\vec{\psi_{3}}\right)$,
where $\alpha=\nicefrac{1}{\beta}$, and then the second equation in (\ref{eq:ev-eq-u3-ld})
is always satisfied and equation (\ref{eq:ev-eq-u3-ld}) can be rewritten
as 
\begin{equation}
\begin{split}\alpha H_{0}\vec{\psi}_{1}+H_{1}\vec{\psi}_{3} & =\alpha\efb\vec{\psi}_{1}\\
\alpha H_{1}^{\dagger}\vec{\psi}_{1}+H_{0}\vec{\psi}_{3} & =\efb\vec{\psi}_{3}\\
H_{1}\vec{\psi}_{1} & =0\\
H_{1}^{\dagger}\vec{\psi}_{3} & =0,
\end{split}
\ \ or\ \ \begin{split}H_{0}\vec{\psi}_{1}+\beta H_{1}\vec{\psi}_{3} & =\efb\vec{\psi}_{1}\\
H_{1}^{\dagger}\vec{\psi}_{1}+\beta H_{0}\vec{\psi}_{3} & =\beta\efb\vec{\psi}_{3}\\
H_{1}\vec{\psi}_{1} & =0\\
H_{1}^{\dagger}\vec{\psi}_{3} & =0.
\end{split}
\label{eq:u3-to-u2}
\end{equation}

From (\ref{eq:u3-to-u2}) we can see that $\left(\vec{\psi_{1}},\beta\vec{\psi}_{3}\right)$
and $\left(\alpha\vec{\psi}_{1},\vec{\psi_{3}}\right)$ can be a CLS
of class $U=2$. Now if we write
\[
\vec{\psi}=\left(\vec{\psi_{1}},\alpha\vec{\psi}_{1}+\beta\vec{\psi}_{3},\vec{\psi_{3}}\right)=\left(\vec{\psi_{1}},\beta\vec{\psi}_{3},0\right)+\left(0,\alpha\vec{\psi}_{1},\vec{\psi_{3}}\right),
\]
 then we can see that $\vec{\psi}$ is actually a linear combination
of two $U=2$ CLSs. 

This indicates that there is no $U=2$ class for the $\nu=2$ case. Suppose
we have a $U=3$ CLS in a $\nu=2$ lattice, giving $H_{0},\ H_{1}$ as
$2\times2$ matrices; therefore, being the eigenvectors of the $2\times2$
matrix, the three components of the CLS should be linearly dependent.
According to the analysis above, it reduces to the $U=2$ class.


\chapter{Supplementary materials for the flatband generator in 1D} 

\section{On the linear independence of CLSs }
\label{app:appendix-A1}





\subsection*{Linear dependence for $m_c=1$ and $U=2$}

Consider a CLS of $U=2$ class $\vec{\psi}=\left(\vec{\psi}_1,\vec{\psi}_2\right)$
with $m_c=1$. Assume that the two components $\vec{\psi}_1,\vec{\psi}_2$
are linearly dependent such that $\vec{\psi}_1=a\vec{\psi}_2$.
Since $\left(\vec{\psi}_1,\vec{\psi}_2\right)$ is a CLS, it follows
\begin{equation}
    \begin{split}
        H_1\vec{\psi}_1 & =0\;,\\
        H_1^{\dagger}\vec{\psi}_2 & =0\;.
    \end{split}
    \label{eq:u2-cls-con}
\end{equation}
This yields 
\begin{equation}
aH_1\vec{\psi}_2 = 0\;,\;H_1^\dagger\vec{\psi}_2 = 0\;\;\textrm{or}\;\;H_1\vec{\psi}_1 = 0\;,\;\frac{1}{a}H_1^\dagger\vec{\psi}_1 = 0\;.
\end{equation}
Thus, $\vec{\psi}_1,\vec{\psi}_2$ are left and right eigenvectors of $H_1$ at the same time, and therefore either $\vec{\psi}_1$ or $\vec{\psi}_2$ serves as the only component of a CLS of class $U=1$.

Interestingly, a similar (but more lengthy) proof can be obtained for the $\nu=2,\ m_c=1,\ U=3$ case. Given a CLS $\vec{\psi}=\left(\vec{\psi}_1,\vec{\psi}_2,\vec{\psi}_3\right)$, it can be shown that the linear dependence of $\{\vec{\psi}_1,\vec{\psi}_2,\vec{\psi}_3\}$ implies that the FB is of class $U\leq2$. Then $\nu>2$ linear dependence is only a necessary but not a sufficient condition.

\subsection*{Orthogonality for $\nu=2$ and $U=2$ }

Consider $U=2$ with $\vec{\psi}_1\perp\vec{\psi}_2$. Then a suitable rotation of the basis in each unit cell will result in $\vec{\psi}_1=(1,0)$ and $\vec{\psi}_2=(0,1)$. This seeming $U=2$ case can be reduced to $U=1$ by redefining the unit cell. Indeed, after the above rotation we may denote each site in the unit cell by $a_l$ and $b_l$. Then a CLS is given by $a_l=\delta_{l,l_{0}}$ and $b_l=\delta_{l,l_{0}+1}$ (up to prefactors and renormalization factors). Redefining the unit cell using $\tilde{a}_l=a_l$ and $\tilde{b}_l=b_{l+1}$ turns the above CLS into class $U=1$. We can conclude that for $m_c=1$, $\nu=2$, and $U=2$, $\vec{\psi}_1$ and $\vec{\psi}_2$ must be neither parallel (linearly dependent) nor orthogonal in order for the FB to not be reducible to $U=1$.



\section{Generator and band structure for two-band $U=2$ FB networks}
\label{as:generator}

For the $\nu=2,\ m_c=1,\ U=2$ case we solve the following equations:
\begin{align}
    \label{eq:eigenvalue_eq1}
    H_0\vec{\psi}_1 + H_1\vec{\psi}_2 & =  \efb\vec{\psi}_1\;,\\
    \label{eq:eigenvalue_eq2-1}
    H_1^\dagger\vec{\psi}_1 + H_0\vec{\psi}_2 & =  \efb\vec{\psi}_2\;,\\
    \label{eq:u2-cls-con1}
    H_1\vec{\psi}_1 & =  0\;, \\
    \label{eq:u2-cls-con2}
    H_1^\dagger\vec{\psi}_2 & =  0 \;.
\end{align}
We can always diagonalize $H_0$ and gauge and rescale the full Hamiltonian to obtain
\begin{equation}
    \label{eq:canonical-H0}
    H_0 = \left(\begin{array}{cc}
        0 & 0\\
        0 & 1
    \end{array}\right) \;.
\end{equation}
For non-singular $\Lambda=\efb-H_0$, we find $\vec{\psi}_2$ from~\eqref{eq:eigenvalue_eq2-1} and insert it into~\eqref{eq:eigenvalue_eq1} to get 
\begin{eqnarray}
    \label{eq:new_eigenvalue_eq1}
    \Lambda^{-1}H_1\Lambda^{-1}H_1^\dagger\vec{\psi}_1 & = & \vec{\psi}_1\;, \\
    \Lambda^{-1}H_1^{\dagger}\vec{\psi}_1 & = & \vec{\psi}_2 \;, 
\end{eqnarray}
where $\Lambda^{-1}=\frac{1}{\efb-H_0}$. Similarly, we have 
\begin{eqnarray}
    \label{eq:new_eigenvalue_eq2}
    \Lambda^{-1}H_1^\dagger\Lambda^{-1}H_1\vec{\psi}_2 & = & \vec{\psi}_2\;, \\
    \Lambda^{-1}H_1\vec{\psi}_2 & = & \vec{\psi}_1 \;.
\end{eqnarray}

\subsection{Real $H_1$}
\label{as:generator:H1-real}

Consider all elements of $H_1$ to be real values. Equations~\eqref{eq:u2-cls-con1} and~\eqref{eq:u2-cls-con2} allow us to redefine $H_1$ in terms of unit vectors $\vert\varphi\rangle$ and $\vert\theta\rangle$, which are the left and right eigenvectors of the non-zero eigenvalue of $H_1$, and which are orthogonal to $\vec{\psi}_1$ and $\vec{\psi}_2$:
\begin{equation}
    \begin{aligned}H_1 & =\alpha\vert\theta\rangle\langle\varphi\vert \;, \\
        \vert\varphi\rangle & =\left(\begin{array}{c}
        \cos\varphi\\
        \sin\varphi
    \end{array}\right) \;, \\
    \vert\theta\rangle & =\left(\begin{array}{c}
        \cos\theta\\
        \sin\theta
    \end{array}\right) \;, 
    \end{aligned}
    \label{eq:redef_H1}
\end{equation}
where the scalar products are
\begin{eqnarray}
    \label{eq:redef_cls_con1}
    \langle\vec{\psi}_1\vert\varphi\rangle & = & 0\;, \\
    \label{eq:redef_cls_con2}
    \langle\vec{\psi}_2\vert\theta\rangle & = & 0 \;.
\end{eqnarray}
Using these definitions and solving Eqs.~\eqref{eq:new_eigenvalue_eq1} and~\eqref{eq:u2-cls-con1}, we obtain
\begin{eqnarray}
    \label{eq:eigen_value}
    \efb & = & \frac{\cos(\theta)\cos(\varphi)}{\cos(\theta-\varphi)}\;, \\
    \label{eq:coefficien_of_H1}
    |\alpha| & = & \sqrt{-\frac{\tan(\theta)\tan(\varphi)\csc^{2}(\theta-\varphi)}{(\tan(\theta)\tan(\varphi)+1)^{2}}} = \frac{\sqrt{-\sin(2\theta)\sin(2\varphi)}}{|\sin(2(\theta-\varphi))|} \;.
\end{eqnarray}

\subsection{Complex $H_1$}
\label{as:generator:H1-complex}

A complex $H_1$ can be parameterized as
\begin{equation}
    \begin{aligned}
        H_1 & =\alpha\vert\theta,\delta\rangle\langle\varphi,\gamma| = \alpha\left(\begin{array}{cc}
            \cos\theta\cos\varphi & e^{i\gamma}\cos\theta\sin\varphi\\
            e^{-i\delta}\sin\theta\cos\varphi & e^{-i\left(\delta-\gamma\right)}\sin\theta\sin\varphi
        \end{array}\right) \;, \\
        \vert\varphi,\gamma\rangle & =\left(\begin{array}{c}
            \cos\varphi\\
            e^{i\gamma}\sin\varphi
        \end{array}\right) \;, \\
        \vert\theta,\delta\rangle & =\left(\begin{array}{c}
            \cos\theta\\
            e^{i\delta}\sin\theta
        \end{array}\right) \;,
    \end{aligned}
	\label{eq:complex H1}
\end{equation}
where $\alpha$ is a complex number.

Following the same procedure as for real $H_1$ we obtain
\begin{align}
    \begin{aligned}
        \efb & =\frac{e^{i\delta}\cot(\theta)\cos(\varphi)}{e^{i\gamma}\sin(\varphi)+e^{i\delta}\cot(\theta)\cos(\varphi)} \;, \\
        |\alpha| & =\frac{2e^{2i(\gamma+\delta)}\sin(2\theta)\sin(2\varphi)}{\left(\left(e^{i\gamma}-e^{i\delta}\right)^{2} (-\cos(2(\theta+\varphi)))+\left(e^{i\gamma}+e^{i\delta}\right)^{2}\cos(2(\theta-\varphi))-4e^{i(\gamma+\delta)}\right)} \\ 
        & \times \frac{1}{\left(e^{i\gamma}\cos(\theta)\cos(\varphi)+e^{i\delta}\sin(\theta)\sin(\varphi)\right)^{2}}\;.
    \end{aligned}
\label{sol_for_comp}
\end{align}
Because $\vert \alpha \vert$ is real, which imposes $\delta=\gamma$,  consequently
\begin{eqnarray}
    \label{eq:eigenvalue_complex}
    \efb & = & \frac{\cos(\theta)\cos(\varphi)}{\cos(\theta-\varphi)}\;, \\
    \label{eq:coefficient_complex}
    |\alpha| & = & \frac{\sqrt{- \sin(2\theta)\sin(2\varphi)}}{|\sin(2(\theta-\varphi))|} \;.
\end{eqnarray}

Solutions~\eqref{eq:eigen_value},~\eqref{eq:coefficien_of_H1},~\eqref{eq:eigenvalue_complex}, and~\eqref{eq:coefficient_complex} are identical. Since $|\alpha|$ is real, solutions only exist in the parameter regions $0\le\theta\le\frac{\pi}{2}\cap\frac{\pi}{2}\le\varphi\le\pi$ or $\frac{\pi}{2}\le\theta\le\pi\cap0\le\varphi\le\frac{\pi}{2}$.

In Bloch representation, the Hamiltonian reads
\begin{equation}
    H(k)=H_{1}^{\dagger}e^{ik}+H_{0}+H_{1}e^{-ik} \;.
    \label{eq:bloch eq}
\end{equation}
With the above parameterization (~\eqref{eq:complex H1} and~\eqref{eq:canonical-H0}), the band structure follows as
\begin{equation}
    \begin{aligned}
        E_{FB} & =\frac{\cos\theta\cos\varphi}{\cos(\theta-\varphi)} \;, \\
        E_k & =\frac{\cos\theta\cos\varphi}{\cos(\theta-\varphi)}+2|\alpha|\cos(\theta-\varphi)\cos(k+\phi_{\alpha})\;,
    \end{aligned}
    \label{eq:band-struct}
\end{equation}
where $\phi_{\alpha}$ is the phase of $\alpha=|\alpha|e^{i\phi_{\alpha}}$.

\subsection{Degenerate $H_0$}
\label{as:generator:H0-degenerate}

The solutions $\alpha$ and $\efb$ in~\eqref{eq:eigenvalue_complex} and~\eqref{eq:coefficient_complex} diverge for $\theta-\varphi=\pm\frac{\pi}{2}$ and $\theta=\phi$, and so does $H_1$. We renormalize the Hamiltonian by multiplying it with $\frac{1}{\alpha}$. Then $H_0$ vanishes, and $H_1$ turns finite. 

However, when $\theta-\varphi=\pm\frac{\pi}{2}$, the dispersion bandwidth in~\eqref{eq:band-struct} is finite, and after normalization the dispersive band becomes flat as well. Therefore, we have two coexisting FBs on the lines $\theta-\varphi=\pm\frac{\pi}{2}$. According to~\eqref{eq:redef_cls_con1} and~\eqref{eq:redef_cls_con2}, $\vec{\psi}_1\perp\vec{\psi}_2$. In such a case we can always perform a rotation in each unit cell such that $\psi_1=(1,0)$, $\psi_2=(0,1)$. A subsequent redefinition of the unit cell turns the CLS into class $U=1$ (see the above subsection
on orthogonality for $\nu=2$ and $U=2$).

When $\theta=\varphi$, according to~\eqref{eq:redef_cls_con1} and~\eqref{eq:redef_cls_con2} $\vec{\psi}_1\parallel\vec{\psi}_2$, and are therefore linearly dependent, turning the FB into the $U=1$ class. In this case 
\begin{equation}
    H_1=\vert\theta,\gamma\rangle\langle\theta,\gamma\vert=
    \left(\begin{array}{cc}
        2\cos^{2}\theta & \frac{1}{2}e^{i\gamma}\sin(2\theta)\\
        \frac{1}{2}e^{-i\gamma}\sin(2\theta) & 2\cos^{2}\theta
    \end{array}\right) \;.
\end{equation}
The corresponding Bloch Hamiltonian in momentum space reads
\begin{equation}
    H(k) = H_1^\dagger e^{ik} + H_1 e^{-ik}=
    \cos k\left(\begin{array}{cc}
        2\cos^{2}\theta & e^{i\gamma}\sin(2\theta)\\
        e^{-i\gamma}\sin(2\theta) & 2\cos^{2}\theta
    \end{array}\right) \;,
\end{equation}
which yields one flat and one dispersive band 
\begin{equation}
    E_{FB} = 0 \;,\;  E_k=2\cos k \;.
\end{equation}

\subsection{FB energy equals one of the eigenvalues of $H_0$: Reduction to $U=1$}

When FB energy $\efb$ equals one of the eigenvalues $0,\ 1$ of $H_0$, we have to solve the original equations~(\ref{eq:eigenvalue_eq1}--\ref{eq:u2-cls-con2}). A simple calculation shows that the only remaining FB solutions are again of class $U=1$.

\section{Generalized sawtooth chain}
\label{as:gsc}

In our FB generator, we have considered the canonical form of $H_0$, which was diagonal. We can perform unitary transformations (rotations) of the unit cell basis
that will modify $H_1$ and make $H_0$ non-diagonal, turning the whole model into a non-canonical one.
Using the rotation matrix
\begin{equation}
	R\left(\omega\right) = 
	\left(\begin{array}{cc}
		\cos\omega & -\sin\omega\\
		\sin\omega & \cos\omega
	\end{array}\right),
\end{equation}
we define $\tilde{\psi}_i=R(\omega) \vec{\psi}_i$, $\tilde{H}_m=R(\omega)H_m R^\dagger(\omega)$ where $m=-1,0,1$ and $H_{-m}=H_m^\dagger$, and
\begin{equation}
	H_0 = \left(\begin{array}{cc}
		0 & 0\\
		0 & 1
	\end{array}\right),\ \ 
	H_1 = \alpha\left(\begin{array}{cc}
		\cos\theta\cos\varphi & e^{i\gamma}\cos\theta\sin\varphi\\
		e^{-i\gamma}\sin\theta\cos\varphi & \sin\theta\sin\varphi
	\end{array}\right),
\end{equation}
with $|\alpha|=\frac{\sqrt{-\sin(2\theta)\sin(2\varphi)}}{|\sin(2(\theta-\varphi))|}$. We consider the case $\gamma=\delta=0$, and thus $H_0$ and $H_1$ read
\begin{equation}
	\begin{aligned}
		\tilde{H_0} & = R^\dagger(\omega) H_0 R(\omega) = \left(\begin{array}{cc}
			\sin^2\omega & \cos\omega\sin\omega\\
			\cos\omega\sin\omega & \cos^{2}\omega
		\end{array}\right)\;,\\
		\tilde{H_1} & = \alpha\left(\begin{array}{cc}
			\cos(\theta+\omega)\cos(\varphi+\omega) & \cos(\theta+\omega)\sin(\varphi+\omega)\\
			\cos(\varphi+\omega)\sin(\theta+\omega) & \sin(\theta+\omega)\sin(\varphi+\omega)
		\end{array}\right)\;.
	\end{aligned}
    \label{34}
\end{equation}
We can always find a value of $\omega$ that will zero one row or one column of $H_1$. This simplifies the non-canonical lattice into a {\sl generalized sawtooth} chain.
It has in general  three different hopping strengths and an onsite energy detuning between the two sites in a unit cell.
As an example, we consider the case when the first column of $H_1$ vanishes:
\begin{equation}
	\varphi + \omega = \pm\frac{\pi}{2}\;.
\end{equation}
It follows that
\begin{equation}
	\begin{aligned}
		\tilde{H}_0 &= \left(\begin{array}{cc}
			\cos ^2(\varphi ) & -\cos (\varphi ) \sin (\varphi ) \\
			-\cos (\varphi ) \sin (\varphi ) & \sin ^2(\varphi ) \\
		\end{array}\right) = \left(\begin{array}{cc}
			\epsilon_1 & t_1\\
			t_1 & \epsilon_2
		\end{array}\right)  \;, \\
		\tilde{H}_1 &= \alpha\left(\begin{array}{cc}
			0 & -\sin(\theta-\varphi)\\
			0 & \cos(\theta-\varphi)
		\end{array}\right) = \left(\begin{array}{cc}
			0 & t_2\\
			0 & t_3
		\end{array}\right) \;.
	\end{aligned}
\end{equation}
For the particular case of the ST1 chain (sawtooth chain with two equal hoppings and zero onsite energy detuning, see Section \ref{section4.3} in the main text),
where $\epsilon_1=\epsilon_2$ and $t_1=t_2$, we find
\begin{equation}
	\begin{aligned}
		\varphi &= \frac{3\pi}{4}, \ \ \theta = \arctan(3+2\sqrt{2}) \;,  \\
		\varphi &= \frac{\pi}{4}, \ \ \theta = \pi - \arctan(3-2\sqrt{2}).
	\end{aligned}
\end{equation}
This leads to the following tight-binding equations:
\begin{equation}
Ea_n = -\sqrt{2} b_n - \sqrt{2} b_{n+1} \;,\; E b_n =-\sqrt{2} a_n -\sqrt{2} a_{n-1} -b_{n+1} - b_{n-1} \;.
\end{equation}
The FB is located at $E_{FB}=2$ (Fig. 4.1 (b) in Section \ref{section4.3} in the main text). The CLS has the form $a_0=a_1=1$ and $b_1=-\sqrt{2}$ (up to a normalization factor) with all other amplitudes vanishing.

Detuning the angles $\theta,\varphi$ away from this point, we deform the ST1 model while maintaining one FB.

Let us require $t_1=t_2=t_3$. Then it follows
\begin{equation}
	\begin{aligned}
		\theta &= \frac{\pi }{2} - \frac{1}{2} \tan^{-1}\left(\frac{1}{2}\right), \ \ \varphi = \frac{3 \pi }{4} - \frac{1}{2} \tan^{-1}\left(\frac{1}{2}\right) \;, \\
		\theta &= \frac{\pi }{2} - \frac{1}{2} \tan^{-1}\left(\frac{1}{2}\right), \ \ \varphi = \frac{3 \pi }{4} - \frac{1}{2} \tan^{-1}\left(\frac{1}{2}\right) \;.
	\end{aligned}
\end{equation}
This is a novel, high-symmetry sawtooth chain (ST2 chain, see Section \ref{section4.3} in main text). It can be obtained for example with the following matrices 
\begin{equation}
	\begin{aligned}
		H_0 = - \left(\begin{array}{cc}
			0 & 1\\
			1 & 1
		\end{array}\right),\  &
		H_1 = -\left(\begin{array}{cc}
			0 & 1\\
			0 & 1
		\end{array}\right).
	\end{aligned}
\end{equation}
This leads to the following tight-binding equations:
\begin{equation}
Ea_n = -b_n - b_{n+1} \;,\; E b_n = -b_n -a_n -a_{n-1} -b_{n+1} - b_{n-1} \;.
\end{equation}
The FB is located at $E_{FB}=1$ (Fig. 4.1 (c) in Section \ref{section4.3} of the main text). The CLS has the simple form $a_0=a_1=-b_1=1$ (up to a normalization factor) with all other amplitudes vanishing.
Note that we can also set other columns and rows of $H_1$ in (\ref{34}) to zero and obtain all points in the $\theta, \varphi$ diagram corresponding to ST1 and ST2 chains. All these points are shown by filled squares and filled circles in Fig. 4.3 in Section \ref{section4.3} in the main text.

\section{Inverse eigenvalue problem: A toy example and the solution of the $U=2$ CLS}
\label{app:ieig-toy}

This appendix explains the solution of the inverse eigenvalue problems~\eqref{eq:cls-ieig-U2-H1}. As discussed in the main text, 1D flatband lattices with CLS class $U$ satisfy 
\begin{align}
    H_1\vpsi_2 & = \left(\EFB - H_0\right)\vpsi_1\label{eq:app-gen-eq-1}, \\
    H_1^\dagger\vpsi_{l-1} + H_1\vpsi_{l+1} & = \left(\EFB - H_0\right)\vpsi_l\quad l = 2,\dots,U-1\label{eq:app-gen-eq-2},\\
    H_1^\dagger\vpsi_{U-1} & = \left(\EFB - H_0\right)\vpsi_U \label{eq:app-gen-eq-3},\\
    H_1\vpsi_1 & = 0, \label{eq:app-gen-eq-4} \\
    H_1^\dagger\vpsi_U & = 0.
    \label{eq:app-gen-eq-5}
\end{align}
Assuming that $\EFB,\ H_0,\ \vpsi_{l=1,\dots,U}$ are given, Eqs.~\eqref{eq:app-gen-eq-1}--\eqref{eq:app-gen-eq-5} constitute an inverse eigenvalue problem for a block-tridiagonal matrix, where diagonal blocks are $H_0$ and off-diagonal ones are $H_1$. 

\subsection{Toy example}
\label{app:toy}

As a warmup, we solve a toy inverse eigenvalue problem: reconstruct $\nu\times\nu$ matrix $T$ given its action $\ket{y}$ on some vector $\ket{x}$:
\begin{gather}
    \label{eq:ieig-toy}
    T\,\ket{x} = \ket{y}.
\end{gather}
The solution is not unique, as a generic solution can be represented as $T=T_* + \delta T$, where $T_*$ is any particular solution of~\eqref{eq:ieig-toy} and $\delta T\ket{x}=0$. One possible solution is easily found to be
\begin{gather}
    T_* = \frac{\ket{y}\bra{x}}{\bra{x}\ket{x}},\quad\delta T = Q_x\, K,
\end{gather}
where $Q_x$ is a transverse projector on $x$. This construction is straightforward to generalize to the case of many vectors (we assume here implicitly that the equations are consistent):
\begin{gather}
    T\ket{x_k} = \ket{y_k},\quad k=1..m.
\end{gather}
The generic solution to this problem is given by
\begin{gather}
    T_* = \sum_{ij} A_{ij}\ket{y_i}\bra{x_j},\quad A_{ij}^{-1} = \bra{x_i}\ket{x_j},\\
    \delta T = Q\, K,
\end{gather}
where $Q$ is the orthogonal projector on the subspace spanned by $\{x_k\}$ and $K$ is an arbitrary $\nu\times\nu$ matrix. For later convenience we refer to $T_*$ as a \emph{particular solution} and $\delta T$ as a \emph{free part}.

\subsection{U=2 case}
\label{app:inv-egv-u2}

In this case,Eqs.~\eqref{eq:app-gen-eq-1}--\eqref{eq:app-gen-eq-5} read
\begin{align}
    H_1\kpsi{2} & = \left(\EFB - H_0\right)\kpsi{1},\notag\\
    H_1^\dagger\kpsi{1} & = \left(\EFB - H_0\right)\kpsi{2},\notag\\
    \label{eq:app-cls-ieig-U2-H1}
    H_1\kpsi{1} & = 0,\\
    H_1^{\dagger}\vert\psi_{2}\rangle & = 0.\notag
\end{align}
We know $H_0,\kpsi{1},\kpsi{2}$ and $\EFB=\evpsi{1}{H_0}{2}$, and we need to determine $H_1$. As discussed above for the toy case, a generic solution to this problem can be decomposed into a particular solution and a free part. The last two equations in the above set are satisfied by the following ansatz:
\begin{gather}
    \label{eq:app-cls-ieig-U2-anzats}
    H_1 = Q_2 M Q_1,\quad Q_i = \mI - \frac{\kpsi{i}\bpsi{i}}{\langle\psi_i\vert\psi_i\rangle}.
\end{gather}
Plugging this ansatz back into the system, we find
\begin{gather}
    \label{eq:app-U2-ieig-M}
    Q_2\, M\,Q_1\kpsi{2} = \left(\EFB - H_0\right)\kpsi{1},\\
    \bpsi{1} Q_2\, M\, Q_1 = \bpsi{2}\left(\EFB - H_0\right).\notag
\end{gather}
Note the identity 
\begin{gather}
    \evpsi{1}{H_1}{2} = \evpsi{1}{\EFB - H_0}{1} = \evpsi{2}{\EFB - H_0}{2},
    \label{eq:app-u2-nl-constraint}
\end{gather}
that follows straightforwardly from the first two equations of~\eqref{eq:app-cls-ieig-U2-H1}. Defining the projectors
\begin{gather}
    R_{12} = \mI - \frac{Q_1\kpsi{2}\bpsi{2} Q_1}{\expval{Q_1}{\psi_2}},\\
    R_{21} = \mI - \frac{Q_2\kpsi{1}\bpsi{1} Q_2}{\expval{Q_2}{\psi_1}},\notag
\end{gather}
we can write 
\begin{equation}
  M = T + R_{21}\, K\, R_{12},
\end{equation}
where $T$ is a particular solution of~\eqref{eq:app-U2-ieig-M}. The second term, where $K$ is an arbitrary $\nu\times\nu$ matrix, satisfies~\eqref{eq:app-U2-ieig-M} by construction and is the free part of the solution. Therefore, we only need to find a particular solution to the system to get the generic solution. This is achieved by the same ansatz $T=\ket{x}\bra{y}$ as in the toy case discussed above. The ansatz yields the following equations:
\begin{gather}
    Q_2\, T\, Q_1\kpsi{2} = \mel{y}{Q_1}{\psi_2} Q_2\ket{x} = (\EFB - H_0)\kpsi{1},\\
    \bpsi{1} Q_2\, T\, Q_1 = \mel{\psi_1}{Q_2}{x}\bra{y} Q_1 = \bpsi{2} (\EFB - H_0).
\end{gather}
From these, vectors $x$ and $y$ are fixed (up to unimportant normalization):
\begin{gather*}
    \bra{y} Q_1 = \frac{1}{\mel{\psi_1}{Q_2}{x}}\bpsi{2} (\EFB - H_0),\\
    Q_2\ket{x} = \frac{1}{\mel{y}{Q_1}{\psi_2}} (\EFB - H_0)\kpsi{1},\\
    = \frac{\mel{\psi_1}{Q_2}{x}}{\evpsi{2}{\EFB - H_0}{2}} (\EFB - H_0)\kpsi{1},\\
    = \frac{\mel{\psi_1}{Q_2}{x}}{\evpsi{1}{\EFB - H_0}{1}} (\EFB - H_0) \kpsi{1}.
\end{gather*}
We used the condition~\eqref{eq:app-u2-nl-constraint} to replace the denominator in the fourth line. Also note that the expression for $y$ from the first line was used to simplify the second line (eliminate $y$). The particular solution is then
\begin{gather}
    \label{eq:app-u2-anzats-sol}
    Q_2 T Q_1 = \frac{\left(\EFB - H_0\right)\kpsi{1}\bpsi{2}\left(\EFB - H_0\right)}{\evpsi{1}{\EFB - H_0}{1}}.
\end{gather}
Thanks to~\eqref{eq:app-u2-nl-constraint}, it is symmetric with respect to $\kpsi{1},\kpsi{2}$. This and the above mentioned free part $Q_{21}KQ_{12}$ give the full family of solutions~\eqref{eq:cls-ieig-U2-sol}:
\begin{gather*}
    H_1 = \frac{(\EFB - H_0)\kpsi{1}\bpsi{2}(\EFB - H_0)}{\evpsi{1}{\EFB - H_0}{1}} +  Q_2  R_{21} K R_{12} Q_1.
\end{gather*}
This expression is further simplified by noticing that $R_{12}Q_1$ and $ Q_2 R_{21}$ are the same projector on the subspace spanned by $\kpsi{1},\kpsi{2}$ that we denote $Q_{12}$: $(R_{12}Q_1)^2 = R_{12} Q_1$, idem for $Q_2 R_{21}$, and both vanish when acting on $\kpsi{1,2}$ as can be straightforwardly verified. We can therefore replace these combinations by $Q_{12}$:
\begin{gather}
    \label{eq:app-u2_solution}
    H_1 = \frac{(\EFB - H_0)\kpsi{1}\bpsi{2}(\EFB - H_0)}{\evpsi{1}{\EFB - H_0}{1}} + Q_{12} K Q_{12}.
\end{gather}
This solution is supplemented by the following non-linear constraints
\begin{align}
    \label{eq:app-constraints_u2}
    \langle\psi_2\vert\psi_1\rangle & = 1,\\
    \evpsi{2}{H_0}{1}\notag & = \EFB,\\
    \evpsi{1}{\EFB - H_0}{1} & = \evpsi{2}{\EFB - H_0}{2},\notag
\end{align}
that are obtained by eliminating $H_1$ from~\eqref{eq:app-cls-ieig-U2-H1} using destructive interference conditions, i.e. the last two equations in~\eqref{eq:app-cls-ieig-U2-H1}. 

In case the denominator in~\eqref{eq:app-u2_solution} is zero, the single-projector ansatz fails, and a two-projector ansatz has to be used:
\begin{equation}
\begin{aligned}
    \label{eq:app-u2_proj2}
    H_1 &= \frac{(\EFB - H_0)\kpsi{1}\bpsi{2}Q_1}{\evpsi{2}{Q_1}{2}} +\\
   & + \frac{Q_2\kpsi{1}\bpsi{2}(\EFB - H_0)}{\evpsi{1}{Q_2}{1}} + Q_{12}K Q_{12},
\end{aligned}
\end{equation}
as can be verified by a direct substitution. In this special solution, the denominators only vanish when $\Psi_1\propto\Psi_2$, i.e. in the $U=1$ case.

\subsection{Bipartite lattices and chiral symmetry}
\label{app:chiral-inv-eig-prob}

In this section, we solve the inverse eigenvalue problem for $U=2$ for the special case of bipartite lattices. We consider a bipartite lattice with $\nu$ sites per unit cell that split into majority and minority sublattices with $\mu$ and $\nu-\mu$ sites, respectively. Since the lattice is  bipartite, the sites on one sublattice only have neighbors belonging to the other sublattice. This enforces the following structure on the hopping matrices and the wave functions of the CLS (see Eq.~\eqref{eq:H01-psi-cs-def}):
\begin{align}
    & H_0 = \left(
    \begin{array}{cc}
        0 & A^\dagger\\
        A & B
    \end{array}\right),\quad
    H_1 = \left(
    \begin{array}{cc}
        0 & T^\dagger\\
        S & W
    \end{array}\right),\notag\\
    & \vpsi_1 = \left(
    \begin{array}{c}
        \varphi_1\\
        0
    \end{array}\right),\quad
    \vpsi_2 = \left(
    \begin{array}{c}
        \varphi_2\\
        0
    \end{array}\right).
    \label{eq:app-chiral-anzats}
\end{align}
Here, $\varphi_{1,2}$ are $\mu$-component vectors describing the wave amplitudes of the majority sublattice sites. $A,S,T$ are $(\nu-\mu)\times\mu$ matrices, while $B, W$ are $(\nu-\mu)\times(\nu-\mu)$ matrices. $B, W$ formally break the bipartiteness of the lattice, but do not affect the $\EFB=0$ FB(s). This special structure simplifies Eq.~\eqref{eq:app-cls-ieig-U2-H1}:
\begin{gather}
    S\kphi{2} = -A\kphi{1},\\
    T\kphi{1} = -A\kphi{2},\\
    S\kphi{1} = 0,\\
    T\kphi{2} = 0.
\end{gather}
These equations need to be resolved with respect to $S$ and $T$. The last two equations are satisfied by the anz\"atse $S=S^\prime Q_1$, $T=T^\prime Q_2$, where $Q_i$ is a transverse projector on $\varphi_i$. The remaining two equations are identical to the toy problem discussed above (Appendix~\ref{app:toy}), and their solution is precisely Eq.~\eqref{eq:cls-ieig-U2-cs-sol}:
\begin{gather}
    \label{eq:app-chiral-u2-sol}
    S = -\frac{A\kphi{1}\bphi{2}Q_1}{\evphi{2}{Q_1}{2}} + K_S Q_{12},\\
    T = -\frac{A\kphi{2}\bphi{1}Q_2}{\evphi{1}{Q_2}{1}} + K_T Q_{12},\notag\\
\end{gather}
where $Q_{12}$ is a joint transverse projector on $\kphi{1,2}$.

Now let's count the number of free parameters. $\vert \varphi_1 \rangle , \vert \varphi_2 \rangle$ all are free parameters with each containing $\mu$ free parameters. $A$ contains $(\nu-\mu)\mu$ free variables. $B, W$ each contains $(\nu-\mu)^2$ free parameters. $K_S Q_{12}$ and $Q_{21} K_T$ are $(\nu-\mu)\times\mu$ and $\mu\times(\nu-\mu)$ matrices, and, because of the transverse projectors, they contain $(\nu-\mu)(\mu-2)$ and $\mu(\nu-\mu-2)$ free parameters, respectively. 
Therefore, \eqref{eq:app-chiral-u2-sol} contains a total $2\mu - 1 + (\nu-\mu)\mu + (\nu-\mu)(\mu-2) + \mu(\nu-\mu-2) + 2(\nu - \mu)^2=(\nu-\mu)(2\nu+\mu-2) - 1$ free parameters.  The extra $-1$ corresponds to the overall normalization of the CLS that is not fixed. 

\section{Resolving the non-linear constraints}
\label{app:sol-nl-constraints}

Let us discuss how one can efficiently resolve the set of non-linear constraints that appear in the inverse eigenvalue problem, for example~\eqref{eq:app-constraints_u2}. Since these are a non-linear system of equations, one can always try a numerical solver. However, our experience was not particularly successful: the solver was not converging, and found no solution more often than not. Instead, it is possible to design a numerical algorithm that eliminates the constraints one by one to either find and enumerate all the solutions, or to prove that there are none.

\subsection{U=2 case}
\label{app:U2-constraints}

The non-linear equations that we need to solve are:
\begin{align}
    \label{eq:app-U2-constraint-l1}
    \langle \psi_1 \vert \psi_2 \rangle & = 1,\\
    \label{eq:app-U2-constraint-l2}
    \evpsi{1}{H_0}{2} & = \EFB,\\
    \label{eq:app-U2-constraint-l3}
    \expval{\EFB - H_0}{\psi_1} & = \expval{\EFB - H_0}{\psi_2}.
\end{align}
We assume that $\EFB$, $H_0$, and $\psi_1$ (or $\psi_2$) are given input parameters.

Then we need to solve the above equations for $\psi_2$. The first two equations~(\ref{eq:app-U2-constraint-l1}--\ref{eq:app-U2-constraint-l2}) are linear and are easily satisfied with the following expansion for $\psi_2$ by the choice of the basis vectors $e_1$ and $e_2$:
\begin{gather}
    \label{eq:app-U2-psi2-series}
    \kpsi{2} = \sum_{k=1}^\nu x_k \ket{e_k},\\
    \kev{1} = \frac{1}{\sqrt{\bra{\psi_1}\ket{\psi_1}}}\kpsi{1},\\
    \kev{2} = \frac{1}{\sqrt{\evpsi{1}{H_0 Q_1 H_0}{1}}} Q_1 H_0\kpsi{1},\\
    \bra{e_l}\ket{e_m} = \delta_{lm},\quad l,m=1,2,\dots\nu.
\end{gather}
Here, $Q_1$ is a transverse projector on $\kpsi{1}$. With this choice of basis vectors, the equations~(\ref{eq:app-U2-constraint-l1}--\ref{eq:app-U2-constraint-l2}) imply:
\begin{gather*}
    x_1 = \frac{1}{\sqrt{\bra{\psi_1}\ket{\psi_1}}},\\
    x_2 = \frac{1}{\sqrt{\evpsi{1}{H_0 Q_1 H_0}{1}}}\left[\EFB - \frac{\evpsi{1}{H_0}{1}}{\bra{\psi_1}\ket{\psi_1}}\right]. 
\end{gather*}

The remaining basis vectors are fixed by requiring their orthonormality, for example, by using Gram--Schmidt orthogonalization. 
Next we plug the expansion~\eqref{eq:app-U2-psi2-series} into~\eqref{eq:app-U2-constraint-l3} and separate out the terms with $e_1$, $e_2$:
\begin{gather*}
    \expval{\EFB - H_0}{\psi_1} = \sum_{ij=1}^\nu x_i^* x_j\mel{e_i}{\EFB - H_0}{e_j}\\
    = \sum_{ij=1}^2 x_i^* x_j\mel{e_i}{\EFB - H_0}{e_j}\\
    + \sum_{i=1}^2\sum_{j=3}^\nu\left[x_i^* x_j\mel{e_i}{\EFB - H_0}{e_j} + x_j^* x_i\mel{e_j}{\EFB - H_0}{e_i}\right]\\
    + \sum_{ij=3}^\nu x_i^* x_j\mel{e_i}{\EFB - H_0}{e_j}.
\end{gather*}
This expression can be rewritten as follows:
\begin{gather}
    \label{eq:app-U2-nl3-qf}
    \sum_{ij=1}^{\nu-2} y_i^* M_{ij} y_j + \sum_{i=1}^{\nu-2}\left[ u_i^* y_i + u_i y_i^*\right] = w,\\
    M_{ij} = \mel{e_{i+2}}{\EFB - H_0}{e_{j+2}},\\
    u_i = \sum_{j=1}^2 x_j \mel{e_{i+2}}{\EFB - H_0}{e_j},\\
    w = \sum_{ij=1}^2 x_i^* x_j\mel{e_i}{\EFB - H_0}{e_j} - \expval{\EFB - H_0}{\psi_1},
\end{gather}
where $y_i = x_{i+2}$. 
The equations with $y_i$ are further simplified by the shift $z_i=y_i + M_{ij}^{-1} u_j$ that eliminates the linear term. This gives the following equation in quadratic form
\begin{gather}
    \sum_{ij=1}^{\nu-2} z_i^* M_{ij} z_j = w + \sum_{ij=1}^{\nu-2} u_i^* M_{ij} u_j.
\end{gather}
Notice that the right hand side of the above equation is real. The matrix $M$ is Hermitian, and can be diagonalized: $M_{ij} = \sum_\alpha E_\alpha \ket{r_\alpha}\bra{r_\alpha}$. The above equation is solved with the help of this spectral decomposition:
\begin{gather}
    \sum_{\alpha=1}^{\nu-2} E_\alpha |t_\alpha|^2 = \tw,\\
    \tw = w + \sum_{\alpha=1}^{\nu-2} E_\alpha |s_\alpha|^2,\\
    t_\alpha = \bra{r_\alpha}\ket{z_\alpha}\qquad s_\alpha = \bra{r_\alpha}\ket{u}.
\end{gather}
The presence or absence of a solution is decided by the mutual signs of $\tw$ and $E_\alpha$: if $\tw>0$ and $E_\alpha<0$ $\forall\alpha$, then there is no solution. If one $E_\alpha > 0$, there is a single solution, and for two or more $E_\alpha>0$ there is a multiparametric family of solutions. Knowing $t_\alpha$, it is straightforward to reconstruct the original $\vpsi_2$.

Above, $M$ was assumed non-singular. If it is singular, then $M_{ij}^{-1}$ is the Moore--Penrose pseudoinverse~\cite{ben2003generalized}, and we have $y_i = z_i + g_i - M_{ij}^{-1} u_j$ where $g\in\ker{M}$. For $g_i$ the quadratic terms in~\eqref{eq:app-U2-nl3-qf} vanish (by definition of $g_i$) and only $g_i$  linearly enter the equation, while $z_i$ can be treated as in the non-singular case (for convenience we assume that the first $k$ eigenvalues of $M$ are zero):
\begin{gather}
    \sum_{\alpha=k+1}^{\nu-2} E_\alpha t_\alpha^2 = \tw - \sum_{\alpha=1}^{k} \left[\bra{u}\ket{r_\alpha} + \bra{r_\alpha}\ket{u}\right].
\end{gather}
The presence of zero modes renormalizes $\tw$.

The more refined version of counting relies on the above solution and the counting of $E_\alpha$ with the ``right" sign. This tells us that for $\nu=2, 3$, there is a single solution for fixed $\vpsi_1,\ \EFB,\ H_0$. For larger $\nu$, there could be a single solution or multiparametric families of solutions, from $0$ to $\nu-3$.

\subsection{ U=3 case}
\label{app:U3-constraints}

In this case, the nonlinear constraints \eqref{eq:cls-ieig-U3-constraints} read
\begin{align}
    & \langle\psi_1\vert\psi_3\rangle = 1,\notag\\
    & \evpsi{1}{H_0}{3} = \EFB,\notag\\
    & \evpsi{1}{\EFB - H_0}{2} = \evpsi{2}{\EFB - H_0}{3},
    \label{eq:app-u3-cls-constraint}\\
    & \evpsi{3}{\EFB - H_0}{3} = \evpsi{2}{\EFB - H_0}{2},\notag\\
    & - \evpsi{1}{\EFB - H_0}{1}.\notag
\end{align}
Resolution of this set of constraints is very similar to the $U=2$ case, so here we only outline the main steps. We search to resolve the above equations with respect to $\Psi_3$, taking $\Psi_1,\Psi_2$ as inputs. The first three lines are linear, and we solve them by expanding $\Psi_3$ over a suitable orthonormal basis:

\begin{align*}
    \kpsi{3} & = \sum_k x_k \kev{k},\\
    \kev{1} & = \frac{1}{\sqrt{\bra{\psi_1}\ket{\psi_1}}} \kpsi{1},\\
    \kev{2} & = \frac{1}{\sqrt{\evpsi{1}{H_0 Q_1 H_0}{1}}} Q_1 H_0\kpsi{1},\\
    \kev{3} & = \frac{Q_*(\EFB - H_0)\kpsi{2}}{\sqrt{\evpsi{2}{(\EFB - H_0) Q_* (\EFB - H_0)}{2}}},\\
    & \vdots\\
    \bra{e_l}\ket{e_m} & = \delta_{l,m},\ \ l,m=1,2,\dots,\nu\\
    Q_1 & = \mathbb{I} - \frac{\kpsi{1}\bpsi{1}}{\bpsi{1}\kpsi{1}},
\end{align*}
and $Q_*$ is a joint transverse projector on $\kpsi{1}$ and $Q_1 H_0\kpsi{1}$. Then $x_1$ and $x_2$ are the same as in the $U=2$ case, and $x_3$ is directly expressed through the third equation in Eqs.~\eqref{eq:app-u3-cls-constraint}. The last (fourth) equation in~\eqref{eq:app-u3-cls-constraint} is solved in the same way as that in the $U=2$ case: it is reduced to solving a quadratic form.

\section{Examples for FB generators}
\label{app:examples}

In this section, we present details of the example FB Hamiltonians generated using the algorithm discussed in the main text. In all of these examples, we pick some $H_0,\ \EFB$ and part of $\psi_l$ as inputs. Next, following the algorithm outlined in Appendix~\ref{app:sol-nl-constraints}, we construct a set of $\{\psi_l\}$ consistent with the CLS structure. Then we find the hopping matrix $H_1$ using the algorithm from Section~\ref{sec:fb-gen} in the main text (detailed in Appendix~\ref{app:ieig-toy}). For simplicity we drop the free part $K$ in all the examples below.

\subsection{$\nu=3,\ U=2$ case}
\label{app:u2-examples}

\emph{Example shown in Figure~\ref{fig:u2_nu3_can}}: We start with a three-band case $\nu=3$ and no additional constraints on the form of $H_1$. We assume a canonical (diagonal) form of
$H_0$ and choose $\vpsi_1$ as
\begin{gather}
    H_0 = \left[\begin{array}{ccc}
        0 & 0 & 0\\
        0 & 1 & 0\\
        0 & 0 & 2
    \end{array}\right],
    \ \ \vpsi_1 = \left(1, -1, 1\right).
\end{gather}

Using the FB algorithm, we find the particular solution:
\begin{align}
    & \EFB = 0.5,\quad\vpsi_2 = \left(1.5,\ 1.5,\ 1\right),\\
    & H_1 = \left[\begin{array}{ccc}
        -0.25 &  0.25 &  0.5 \\
        -0.25 &  0.25 & 0.5\\
        0.75 & -0.75 & -1.5
    \end{array}\right].
\end{align}

\emph{Example shown in Figure~\ref{fig:u2_nu3_noncan}}: Taking non-diagonal $H_0$ and $\vpsi_1$ as 
\begin{gather}
    H_0 = \left[\begin{array}{ccc}
        0 & 1 & 0\\
        1 & 0 & 1\\
        0 & 1 & 0
    \end{array}\right],
    \ \ \vpsi_1 = \left( 1, -1,  1\right),
\end{gather}

we construct the following FB Hamiltonian:
\begin{align}
    & H_1 = \left[\begin{array}{ccc}
        0.19926929 & -0.47727273 & -0.67654202\\
        -0.33211549 &  0.79545455 &  1.12757003\\
        0.19926929 & -0.47727273 & -0.67654202
    \end{array}\right],\\ 
    & \vpsi_2 = \left(4.46130814,\ 1.5,\  -1.96130814\right),\ \EFB = 0.5.
\end{align}

\emph{Example shown in Figure~\ref{fig:u2_nu3_gen}}: Taking all the sites in the unit cell connected to each other and the same $\vpsi_1$ and $\EFB$ as in the above example gives

\begin{gather}
    H_0 = \left[\begin{array}{ccc}
        0 & 1 & 1\\
        1 & 0 & 1\\
        1 & 1 & 0
    \end{array}\right],\
    \ \vpsi_1 = \left(1,-1,1\right),
\end{gather}
and we land at the following Hamiltonian:
\begin{align}
    & H_1 = \left[\begin{array}{ccc}
        0.18163216 & -0.16071429 & -0.34234644\\
        -0.90816078 &  0.80357143 &  1.71173221\\
        0.18163216 & -0.16071429 & -0.34234644
    \end{array}\right],\\
    & \vpsi_2 = \left(1.84761673,\ 0.25 ,\ -0.59761673\right),\ \EFB = 0.5.
\end{align}

\emph{Bipartite lattice $U=2$, Figure~\ref{fig:u2_bipartite}}: 
We consider the $\nu=4,\ \mu=2$ case and pick the following input variables:
\begin{align*}
    & A = \frac{1}{4} \left(\begin{array}{cc}
        \sqrt{3} & 1 \\
        3 & \sqrt{3} \\ \end{array} \right), \ \ 
    B=\left(\begin{array}{cc}
        1 & -2 \\
        -2 & 1 \\ \end{array} \right),\\ 
    & \vphi_1 = \frac{1}{\sqrt{2}} \left( 1, 1 \right),
    \ \ \vphi_2 = \frac{1}{2} \left( 1, \sqrt{3} \right), \\
    & W = \left(\begin{array}{cc}
        2 & -1 \\
        -1 & 2 \\
    \end{array}\right).
\end{align*}
Solving Eq. \eqref{eq:app-chiral-u2-sol}, i.e.~\eqref{eq:cls-ieig-U2-cs-sol}, yields the following solution: 
\begin{align*}
    & S = \frac{\left(\sqrt{3}+2\right)
    \left(\begin{array}{cc}
        -1 & 1 \\
        -\sqrt{3} & \sqrt{3} \\
    \end{array}\right)}{2 \sqrt{2}}, \\
    & T = \frac{\left(\sqrt{3}+1\right)
    \left(\begin{array}{cc}
        3 & 3 \sqrt{3} \\
        -\sqrt{3} & -3 \\
    \end{array}\right)}{4 \sqrt{2}}.
\end{align*}
The corresponding hopping matrices $H_0, H_1$ read
\begin{align*}
    H_0 & =  \left(\begin{array}{cccc}
        0 & 0 & \frac{\sqrt{3}}{4} & \frac{3}{4} \\
        0 & 0 & \frac{1}{4} & \frac{\sqrt{3}}{4} \\
        \frac{\sqrt{3}}{4} & \frac{1}{4} & 1 & -2 \\
        \frac{3}{4} & \frac{\sqrt{3}}{4} & -2 & 1 \\
    \end{array}\right),\\
    H_1 & = \left(\begin{array}{cccc}
        0 & T  \\
        S & W \\
    \end{array}\right),\\
    \vpsi_1 & = \left(\frac{1}{\sqrt{2}},\frac{1}{\sqrt{2}},0,0\right), \\
    \vpsi_2 & = \left(\frac{1}{2},\frac{\sqrt{3}}{2},0,0\right).
\end{align*}

\subsection{$\nu=3$, $U=3$ case}
\label{app_sec:u3_examples}

\emph{Example shown in Figure~\ref{fig:u3_nu3_can}}: We pick $H_0$ in canonical form and choose $\vpsi_1$ as follows 
\begin{gather*}
    H_0 = \left[\begin{array}{ccc}
        0 & 0 & 0\\
        0 & 1 & 0\\
        0 & 0 & 2
    \end{array}\right],
    \ \ \vpsi_1 = \left(1, -1, 1\right)
\end{gather*}
Solving the non-linear constraints \eqref{eq:app-u3-cls-constraint}, i.e., Eq. \eqref{eq:cls-ieig-U3-constraints}, we get $\vpsi_2, \vpsi_3$. Then solving the Eq.~\eqref{eq:cls-ieig-U3-H1-linear}, which is equivalent to Eqs. (\ref{eig-1}--\ref{eig-5}), we get
\begin{align*}
    H_1 & = \left[\begin{array}{ccc}
        -0.06548573 & -0.27210532 & -0.2066196 \\
        -0.15130619 & -0.28682832 & -0.13552213\\
        -0.14682469 &  0.75742396 &  0.90424865
    \end{array}\right],\\
    \vpsi_2 & = \left(-0.05144152, -1.53640189, -0.38025523\right),\\
    \vpsi_3 & = \left(0.58333333, -0.33333333,  0.08333333\right),\\
    \EFB & = 0.5.
\end{align*}

\emph{Example shown in Figure~\ref{fig:u3_nu3_noncan}}: We choose $H_0$ and $\vpsi_1$ as 
\begin{gather*}
    H_0 = \left[\begin{array}{ccc}
        0 & -1 & 0\\
        -1 & 0 & 1\\
        0 & 1 & 0
    \end{array}\right],
    \ \ \vpsi_1 = \left(1, -1, 1\right).
\end{gather*}
The corresponding FB $H_1$ is
\begin{align*}
    H_1 & = \left[\begin{array}{ccc}
        0.23624218 &  0.15535892 & -0.08088326\\
        -0.87350793 & -0.69073091 &  0.18277702\\
        1.31303601 &  0.95651792 & -0.35651809
    \end{array}\right],\\
    \vpsi_2 & = \left(3.14189192, -2.05220768, -0.94681365\right),\\
    \vpsi_3 & =\left(1.08333333, -0.33333333, -0.41666667\right),\\
    \EFB & = 0.5.
\end{align*}

\emph{Example shown in Figure~\ref{fig:u3_nu3_gen}}: For the following input
\begin{gather*}
    H_0 = \left[\begin{array}{ccc}
        0 & -1 & 2\\
        -1 & 0 & 1\\
        2 & 1 & 0
    \end{array}\right],
    \ \ \vpsi_1 = \left(1, -1, 1\right),
\end{gather*}
we find the following FB nearest neighbor hopping matrix $H_1$:
\begin{align*}
    H_1 & = \left[\begin{array}{ccc}
        0.06915801 & -0.66620419 & -0.7353622\\
        -0.31644957 & -0.3029663 &   0.01348327\\
        -0.46657738 & -0.38011423 &  0.08646314
    \end{array}\right],\\
    \vpsi_2 & =\left(0.77717503,  2.50899893,  1.05355773\right),\\
    \vpsi_3 & =\left(0.03571429, -0.57142857,  0.39285714\right),\\
    \EFB & =0.5.
\end{align*}

\emph{Example shown in Figure~\ref{fig:u3_nu3_ground}}: The following input data
\begin{gather*}
    H_0 = \left[\begin{array}{ccc}
        0 & 1 & 0\\
        1 & 0 & 1\\
        0 & 1 & 0
    \end{array}\right],
    \ \ \vpsi_1 = \left(1, -1, 1\right),
\end{gather*}
provides an example FB Hamiltonian in which the FB is the ground state:
\begin{align*}
    H_1 & = \left[\begin{array}{ccc}
        -0.52279625 &  0.17024672 &  0.69304298\\
        -0.62702148 & -0.11461122 &  0.51241027\\
        -0.73124671 & -0.39946915 &  0.33177756
    \end{array}\right],\\
    \vpsi_2 & = \left(0.25537008,  0.28652804, -0.59920373\right),\\
    \vpsi_3 & = \left(0.25, -0.5,   0.25\right),\\
    \EFB & =-1.5.
\end{align*}

\section{Network constraints}
\label{app:imposing_lat_str}

We present here details of the examples where network connectivity was provided as an input to the FB generator. In all cases, one can find particular solutions to the resulting non-linear system of equations.

Often network connectivity implies very sparse $H_0$ and $H_1$. Therefore, inserting sparse $H_0$ and $H_1$ into Eqs. (\ref{eig-1}--\ref{eig-5}) gives a set of equations that can be solved analytically. More precisely, as seen in the examples below, when $H_0$ and $H_1$ are so sparse that the number of unknowns (non-zero elements of $H_1, H_0$ and part of the CLS) is less than or equal to the number of equations, we can solve the Eqs. (\ref{eig-1}--\ref{eig-5}) analytically. Note that, instead of inserting $H_1$ and $H_0$ into Eqs. (\ref{eig-1}-\ref{eig-5}), we can get the same set of equations from~\eqref{eq:cls-ieig-U3-H1-linear} by zeroing the elements of $h_1$ corresponding to the zero elements of $H_1$.

\subsection{U=2 case}

\subsubsection{1D kagome}

We consider a $d=1$ version of the 2D kagome lattice. The nearest neighbor Hamiltonian is restricted by lattice connectivity to
\begin{gather*}
    H_0 = \left[\begin{array}{ccccc}
        0 & t_2 & 0 & 0 & 0\\
        t_2 & 0 & t_1 & 0 & 0\\
        0 & t_1 & 0 & t_1 & 0\\
        0 & 0 & t_1 & 0 & t_2\\
        0 & 0 & 0 & t_2 & 0
    \end{array}\right],
    \ \ H_1 = \left[\begin{array}{ccccc}
        0 & t_1 & t_1 & 0 & 0\\
        0 & 0 & 0 & 0 & 0\\
        0 & 0 & 0 & 0 & 0\\
        0 & 0 & 0 & 0 & 0\\
        0 & 0 & t_1 & t_1 & 0
    \end{array}\right].
\end{gather*}
Destructive interference, i.e. the last two equations in \eqref{eq:cls-ieig-U2-H1},
implies that 
\begin{gather*}
    \vpsi_1 = \left(x_1, -x_2, x_2, -x_2, x_3\right),
    \ \ \vpsi_2 = \left(0, a, b, c, 0\right).
\end{gather*}
 If we insert $\vpsi_1,\vpsi_2$ above into the equations in \eqref{eq:cls-ieig-U2-H1}, we find 
\begin{align*}
\left(\begin{array}{c}
-x_{2}t_{2}+\left(y_{2}+y_{3}\right)t_{1}\\
x_{2}t_{1}+x_{1}t_{2}\\
-2x_{2}t_{1}\\
x_{2}t_{1}+x_{3}t_{2}\\
-x_{2}t_{2}+\left(y_{3}+y_{4}\right)t_{1}
\end{array}\right) & =E_{FB}\left(\begin{array}{c}
x_{1}\\
-x_{2}\\
x_{2}\\
-x_{2}\\
x_{3}
\end{array}\right),\\
\left(\begin{array}{c}
at_{2}\\
\left(b+y_{1}\right)t_{1}\\
\left(a+c+y_{1}+y_{5}\right)\\
\left(b+y_{5}\right)t_{1}\\
ct_{2}
\end{array}\right) & =E_{FB}\left(\begin{array}{c}
0\\
a\\
b\\
c\\
0
\end{array}\right).
\end{align*}
One possible solution of this is 
\[
\begin{aligned}a=c & =0,\\
t_{1} & =-\frac{E_{FB}}{2},\\
t_{2} & =\frac{E_{FB}}{2},\\
x_{1} & =-x,\\
x_{2} & =x,\\
x_{3} & =-x,\\
a & =0,\\
b & =x,\\
c & =0.
\end{aligned}
\]
 This solution gives a FB with energy $E_{FB}$. Thus the final
solution is 
\[
\begin{aligned}\vec{\psi}_{1} & =\left(-x,-x,x,-x,-x\right),\\
\vec{\psi}_{2} & =\left(0,0,x,0,0\right),\\
H_{1} & =\left[\begin{array}{ccccc}
0 & -\frac{E_{FB}}{2} & -\frac{E_{FB}}{2} & 0 & 0\\
0 & 0 & 0 & 0 & 0\\
0 & 0 & 0 & 0 & 0\\
0 & 0 & 0 & 0 & 0\\
0 & 0 & -\frac{E_{FB}}{2} & -\frac{E_{FB}}{2} & 0
\end{array}\right],\\
H_{0} & =\left[\begin{array}{ccccc}
0 & \frac{E_{FB}}{2} & 0 & 0 & 0\\
\frac{E_{FB}}{2} & 0 & -\frac{E_{FB}}{2} & 0 & 0\\
0 & -\frac{E_{FB}}{2} & 0 & -\frac{E_{FB}}{2} & 0\\
0 & 0 & -\frac{E_{FB}}{2} & 0 & \frac{E_{FB}}{2}\\
0 & 0 & 0 & \frac{E_{FB}}{2} & 0
\end{array}\right].
\end{aligned}
\]
 This lattice has a FB with flatband energy $E_{FB}$.

\subsubsection{$U=2$, $\nu=3$ example}
\label{sec:app-imposing-lat-str-u2-nu3}

The connectivity of the network shown in Figure~\ref{fig:u2_nu3_latt_str_example} implies the following hopping matrices:
\begin{align*}
    & H_0 = \left(\begin{array}{ccc}
        0 & t_1 & 0 \\
        t_1 & 0 & t_2 \\
        0 & t_2 & 0 \\
    \end{array}\right), \ \
    H_1 = \left(\begin{array}{ccc}
        s_1 & s_2 & 0 \\
        s_4 & s_5 & s_6 \\
        0 & s_7 & s_8 \\
    \end{array}\right).
\end{align*}
We parameterize the CLS amplitudes as follows: $\vpsi_1 = (x,y,z),\ \vpsi_2 = (a,b,c)$. Then Eq. \eqref{eq:cls-ieig-U2-H1} gives: 
\begin{align*}
    & \left(\begin{array}{c}
        a s_1+b s_2 \\
        a s_4+b s_5+c s_6\\
        b s_7+c s_8
    \end{array}\right)=
    \left(\begin{array}{c}
        x \EFB  - t_1 y\\
        -t_1 x-t_2 z+ y \EFB\\
        z \EFB - t_2 y
    \end{array}\right),\\
    & \left(\begin{array}{c}
        s_1 x+s_4 y \\
        s_2 x+s_5 y+s_7 z \\
        s_6 y+s_8 z
    \end{array}\right) =
    \left(\begin{array}{c}
        a \EFB - b t_1 \\
        -a t_1 + b \EFB - c t_2 \\
        c \EFB - b t_2
    \end{array} \right), \\
    & \left(\begin{array}{c}
        s_1 x+s_2 y \\
        s_4 x+s_5 y+s_6 z \\
        s_7 y+s_8 z
    \end{array}\right) = 
    \left(\begin{array}{c}
        0 \\
        0 \\
        0
    \end{array} \right), \\
    & \left(\begin{array}{c}
        a s_1+b s_4 \\
        a s_2+b s_5+c s_7 \\
        b s_6+c s_8
    \end{array}\right)=
    \left(\begin{array}{c}
        0 \\
        0 \\
        0
    \end{array} \right).
\end{align*}
Here, $H_0$, $\EFB$, and $\vpsi_1$ are free parameters. If we fix $x=1,\ y= 2,\ z=1,\ t_1=1,\ t_2= 2,\ \EFB = 3$, then we find one particular solution of above the equations
\begin{align*}
    & s_1=\frac{2 \sqrt{2}}{3}, \ \ s_2 = -\frac{\sqrt{2}}{3}, \\
    & s_4 = \frac{2 \sqrt{2}}{3}, \ \  s_5 = \frac{\sqrt{2}}{3}, \\
    & s_6 = -\frac{1}{3} \left(4 \sqrt{2}\right), \ \  s_7 = -\frac{\sqrt{2}}{3}, \\
    & s_8 = \frac{2 \sqrt{2}}{3}, \ \ a = \frac{1}{\sqrt{2}}, \\
    & b = -\frac{1}{\sqrt{2}}, \ \ c = -\sqrt{2},
\end{align*}
from which follow the hopping matrices and CLS amplitudes
\begin{align*}
    & H_0 = \left(\begin{array}{ccc}
        0 & 1 & 0 \\
        1 & 0 & 2 \\
        0 & 2 & 0 \\
    \end{array}\right), \ \ 
    H_1 = \left(\begin{array}{ccc}
        \frac{2 \sqrt{2}}{3} & -\frac{\sqrt{2}}{3} & 0 \\
        \frac{2 \sqrt{2}}{3} & \frac{\sqrt{2}}{3} & -\frac{1}{3} \left(4 \sqrt{2}\right) \\
        0 & -\frac{\sqrt{2}}{3} & \frac{2 \sqrt{2}}{3} \\
    \end{array}\right), \\
    & \vpsi_1 = (1, 2, 1),\ \ \vpsi_2 = \left(\frac{1}{\sqrt{2}},-\frac{1}{\sqrt{2}},-\sqrt{2} \right).
\end{align*}

\subsection{ U=3 case}

\subsubsection{$U=3$, $\nu=3$ example}
\label{sec:app-imposing-lat-str-u3-nu3}

We consider the network shown in Fig.~\ref{fig:u3_nu3_latt_str_example}. Its connectivity requires the following hopping matrices:
\begin{gather*}
    H_0 = \left(\begin{array}{ccc}
        0 & t_1 & 0 \\
        t_1 & 0 & t_2 \\
        0 & t_2 & 0 \\
    \end{array}\right), \ \ 
    H_1 =\left(\begin{array}{ccc}
        s_1 & s_1 & 0 \\
        -\frac{s_1}{2} & -\frac{s_1}{2}-s_6 & s_6 \\
        0 & 2 s_6 & -2 s_6 \\
    \end{array}\right).
\end{gather*}
According to destructive interference conditions~\eqref{eq:app-gen-eq-4} and \eqref{eq:app-gen-eq-5}, we paramterize $\vpsi_1,\vpsi_2,\vpsi_3$ as 
\begin{gather*}
    \vpsi_1 = (-y, y, y), \ \ \vpsi_2 = (a,b,c), \ \ \vpsi_3 = (d, 2d, d).
\end{gather*}
Then the main equations~\eqref{eq:cls-ieig-U3-H1} become:
\begin{align*}
    & \left( \begin{array}{c}
        (a+b) s_1 \\
        (c-b) s_6-\frac{1}{2} (a+b) s_1 \\
        2 (b-c) s_6 \\
    \end{array} \right) =
    \left(\begin{array}{c}
        -y \left(\EFB + t_1\right) \\
        y \left(\EFB + t_1-t_2\right) \\
        y \left(\EFB - t_2\right), \\
    \end{array} \right) \\
    & \left(\begin{array}{c}
        \frac{3}{2} (2 d-y) s_1 \\
        (y-d) s_6-\frac{3}{2} (d+y) s_1 \\
        (2 d-y) s_6 \\
    \end{array} \right) =
    \left(\begin{array}{c}
        a \EFB - b t_1 \\
        b \EFB - a t_1-c t_2 \\
        c \EFB - b t_2 \\
    \end{array} \right), \\
    & \left(\begin{array}{c}
        \frac{1}{2} (2 a-b) s_1 \\
        \left(a-\frac{b}{2}\right) s_1-(b-2 c) s_6 \\
        (b-2 c) s_6 \\
    \end{array} \right) =
    \left( \begin{array}{c}
        d \left(\EFB - 2 t_1\right) \\
        d \left(2 \EFB - t_1 - t_2\right) \\
        d \left(\EFB - 2 t_2\right). \\
    \end{array} \right)
\end{align*}
Again, the above system admits many solutions. We pick one with $t_1 = 1,\ t_2 = 2,\ b=\frac{1}{2} $ and get
\begin{align*}
    & a= \frac{1}{80} \left(3 \sqrt{21}+23\right),
    c= \frac{1}{80} \left(\sqrt{21}+41\right), \\
    & d = \frac{1}{40} \left(-7 \sqrt{\frac{3}{2}}-\sqrt{\frac{7}{2}}\right),
    y = \frac{1}{40} \left(\sqrt{\frac{3}{2}}+3 \sqrt{\frac{7}{2}}\right), \\
    & \EFB = \frac{5}{2},
    s_1 =-\frac{\sqrt{\frac{7}{2}}}{3},
    s_6 = -\frac{\sqrt{\frac{3}{2}}}{2}.
\end{align*}

Therefore, the CLS amplitudes and hopping matrices are:  
\begin{align*}
    \vpsi_1 & = \left(\begin{array}{c}
        \frac{1}{40} \left(-\sqrt{\frac{3}{2}}-3 \sqrt{\frac{7}{2}}\right) \\
        \frac{1}{40} \left(\sqrt{\frac{3}{2}}+3 \sqrt{\frac{7}{2}}\right) \\
        \frac{1}{40} \left(\sqrt{\frac{3}{2}}+3 \sqrt{\frac{7}{2}}\right) \\
    \end{array}\right),\\
    \vpsi_2 & = \left(\begin{array}{c}
        \frac{1}{80} \left(3 \sqrt{21}+23\right) \\
        \frac{1}{2} \\
        \frac{1}{80} \left(\sqrt{21}+41\right) \\
    \end{array}\right),\\
    \vpsi_3 & = \left(\begin{array}{c}
        \frac{1}{40} \left(-7 \sqrt{\frac{3}{2}}-\sqrt{\frac{7}{2}}\right) \\
        \frac{1}{20} \left(-7 \sqrt{\frac{3}{2}}-\sqrt{\frac{7}{2}}\right) \\
        \frac{1}{40} \left(-7 \sqrt{\frac{3}{2}}-\sqrt{\frac{7}{2}}\right) \\
    \end{array}\right),\\
    H_1 & = \left(\begin{array}{ccc}
        -\frac{\sqrt{\frac{7}{2}}}{3} & -\frac{\sqrt{\frac{7}{2}}}{3} & 0 \\
        \frac{\sqrt{\frac{7}{2}}}{6} & \frac{\sqrt{\frac{3}{2}}}{2}+\frac{\sqrt{\frac{7}{2}}}{6} & -\frac{\sqrt{\frac{3}{2}}}{2}, \\
        0 & -\sqrt{\frac{3}{2}} & \sqrt{\frac{3}{2}} \\
    \end{array}\right)\\
    H_0 & = \left(\begin{array}{ccc}
        0 & 1 & 0 \\
        1 & 0 & 2 \\
        0 & 2 & 0 \\
    \end{array}\right),
\end{align*}
which give a FB with energy $\EFB=5/2$. Schematics and the band structure of this lattice are shown in Figure \ref{fig:u3_nu3_latt_str_example}.


\chapter{Supplementary materials for the flatband generator in 2D}

\section{Two-band problem}
\label{app:two-band-prob-2d} 

First we consider the $U=(2,1,0)$ case in Fig. \ref{fig:u22-cls-configs} (c) in the main text. In this case, the eigenvalue problem and destructive interference conditions in Eqs. \eqref{eq:u22-gen-eig-prob} read 
\begin{equation}
 \begin{aligned}
  & H_0 \vec{\psi}_1 + H_1 \vec{\psi}_2 = E_{FB} \vec{\psi}_1,\\
  & H_0 \vec{\psi}_2 + H_1^{\dagger} \vec{\psi}_1  = E_{FB} \vec{\psi}_3, \\
  & H_1 \vec{\psi}_1 = 0,\ H_1^{\dagger} \vec{\psi}_2 = 0,  \\
  & H_2 \psi_{s} = 0,\ H_2^{\dagger} \psi_{s} = 0, \ \  1 \le s \le 2.
 \end{aligned}
 \label{ev_prob_nu2_u21}
\end{equation}
We can see that in the $y$ direction $U=1$, therefore we have CLS conditions in the $y$ direction as 
\begin{equation}
 H_2 \vec{\psi}_1 = H_2^{\dagger} \vec{\psi}_1 = 0,\ \ H_2 \vec{\psi}_2 = H_2^{\dagger} \vec{\psi}_2 = 0.
\end{equation}
The above equation implies that $H_2$ has two zero eigenvalues, which in turn implies that $H_2=0$ or 
\begin{equation}
 H_2 = \left(\begin{array}{cc} 
              t &  t\\ 
              -t & -t\\ 
              \end{array}\right),
\end{equation}
having only one non-zero eigenvector $(-x,x)$. Thus, the problem reduces to the $U=1$ class. 

Using the same analysis, the $U=(2,2,2)$ case in Fig. \ref{fig:u22-cls-configs} (d) in the main text also reduces to $U=1$. Redefining the $x,y$ directions by tilting them $45$ degrees, we can transform the problem to a $U=(2,1,0)$ problem.

Now if we consider the $U=(2,2,0)$ case in Fig. \ref{fig:u22-cls-configs} (e), the eigenvalue problem and destructive interference conditions in Eqs. \eqref{eq:u22-gen-eig-prob} are:
\begin{equation}
 \begin{aligned}
  & H_0 \vec{\psi}_1 + H_1 \vec{\psi}_2 + H_2 \vec{\psi}_3 = E_{FB} \vec{\psi}_1,\\
  & H_0 \vec{\psi}_2 + H_1^{\dagger} \vec{\psi}_1 + H_2 \vec{\psi}_4  = E_{FB} \vec{\psi}_2, \\
  & H_0 \vec{\psi}_3 + H_1 \vec{\psi}_4 + H_2^{\dagger} \vec{\psi}_1 = E_{FB} \vec{\psi}_3,\\
  & H_0 \vec{\psi}_4 + H_1^{\dagger} \vec{\psi}_3 + H_2^{\dagger} \vec{\psi}_2 = E_{FB} \vec{\psi}_4,\\
  & H_1 \vec{\psi}_{1} = 0,\ \ H_1 \vec{\psi}_{3} = 0, \\
  & H_1^{\dagger} \vec{\psi}_{2} = 0 ,\ \ H_1^{\dagger} \vec{\psi}_{4} = 0,\\
  & H_2 \vec{\psi}_{1} = 0,\ \ H_2 \vec{\psi}_{2} = 0,\\
  & H_2^{\dagger} \vec{\psi}_{3} = 0, \ \ H_2^{\dagger} \vec{\psi}_{4} = 0.
 \end{aligned}
 \label{ev_prob_nu2_u22}
\end{equation}
As we see from the above equation, the $2\times2$ matrices $H_1,H_2$ have two zero modes, either implying $H_1=H_2=0$ or reducing to the
$U=1$ case, according to the previous subsection. 

Similarly, it can also be shown that for the  $U=(2,2,1$) in Fig. \ref{fig:u22-cls-configs} (b), the only possibility is either $H_1=H_2=0$ or a reduction to $U=1$.

\section{More than two bands with nearest neighbor hoppings}

\subsection{U=(2,1,0) case} 
\label{app:nn-u21}

Putting the values $U_2=1,\ s=0$ to Eq. \eqref{eq:u22-gen-eig-prob}, we get the eigenvalue problem and destructive interference conditions in this case as follows 
\begin{equation}
    \begin{aligned}
       H_1 \vert \psi_2 \rangle & = (E_{FB} - H_0) \vert \psi_1 \rangle, \\ 
       \langle \psi_1 \vert H_1 &= \langle \psi_1 \vert (E_{FB} - H_0), \\ 
       H_1 \vert \psi_1 \rangle &= 0 , \\
       \langle \psi_2 \vert H_1 &= 0, \\
       H_2 \vert \psi_1 \rangle &= 0, \\
       H_2 \vert \psi_2 \rangle &= 0, \\
       \langle \psi_1 \vert H_2 &=0, \\
       \langle \psi_2 \vert H_2 &=0. \
    \end{aligned}
    \label{eq:u21-eig-prob-app}
\end{equation} 
Using destructing interference \eqref{eq:u21-eig-prob-app}, we eliminate $H_1, H_2$ from the eigenvalue problem to get the CLS conditions as below 
\begin{equation}
    \begin{aligned}
        \langle \psi_2 \vert (E_{FB} - H_0) \vert \psi_1 \rangle &= 0, \\
        \langle \psi_1 \vert (E_{FB} - H_0) \vert \psi_1 \rangle &= \langle \psi_2 \vert (E_{FB} - H_0) \vert \psi_2 \rangle. \
    \end{aligned}
    \label{eq:u21-cls-cons-app}
\end{equation} 
The destructive interference conditions suggests following form of $H_1, H_2$: 
\begin{equation}
    H_1 = Q_2 \vert u \rangle \langle v \vert Q_1, \quad H_2 = Q_{12} M Q_{12} \; ,
    \label{eq:u21-hx-hy-def-app}
\end{equation}
where $\vert u \rangle, \vert v \rangle$ are arbitrary vectors, $M$ is an arbitrary matrix, $Q_i$ is a transverse projector on $\vert \psi_i \rangle$, and $Q_{12}$ is transverse projector on $\vert \psi_{i=1,2} \rangle$, i.e. $Q_{12} \vert \psi_{i=1,2} \rangle = 0$ (see Appendix \ref{app:inv-egv-u2}).
Then the eigenvalue problem \eqref{eq:u21-eig-prob-app} becomes 
\begin{equation*}
    \begin{aligned}
       Q_2 \vert u \rangle \langle v \vert Q_1 \vert \psi_2 \rangle &= (E_{FB} - H_0) \vert \psi_1 \rangle, \\
        \langle \psi_1 \vert Q_2 \vert u \rangle \langle v \vert Q_1 &= \langle \psi_2 \vert (E_{FB} - H_0) \; ,
    \end{aligned}
\end{equation*} 
which gives 
\begin{equation*}
    Q_2 \vert u \rangle = \frac{1}{\langle v \vert Q_1 \vert \psi_2 \rangle} (E_{FB} - H_0) \vert \psi_1 \rangle , \quad \langle v \vert Q_1 = \frac{\langle v \vert Q_1 \vert \psi_2 \rangle}{\langle \psi_1 \vert (E_{FB} - H_0) \vert \psi_1 \rangle} \langle \psi_2 \vert (E_{FB} - H_0) \; .
\end{equation*} 
Putting this into \eqref{eq:u21-hx-hy-def-app} we get the solution 
\begin{equation}
    \begin{aligned}
       H_1 &= \frac{(E_{FB} - H_0)\vert \psi_1 \rangle \langle \psi_2 \vert (E_{FB} - H_0)}{\langle \psi_1 \vert (E_{FB} - H_0)\vert \psi_1 \rangle}, \\
       H_2 &= Q_{12} M Q_{12}.
    \end{aligned}
\end{equation} 

\subsection{U=(2,2,1) case}
\label{app:nn-u221}

Putting the values $U_2=2,\ s=1$ to Eq. \eqref{eq:u22-gen-eig-prob}, we get the eigenvalue problem and destructive interference conditions in this case as follows
\begin{equation}
 \begin{aligned}
    & H_1 \Vec{\psi}_2 + H_2 \vec{\psi}_3 = (E_{FB} -H_0) \vec{\psi}_2, \\
    & H_1^{\dagger} \vec{\psi}_1 = (E_{FB} - H_0) \Vec{\psi}_2, \\
    & H_2^{\dagger} \vec{\psi}_1 = (E_{FB} - H_0) \vec{\psi}_3, \\
    & H_1 \vec{\psi}_1 = 0, \quad H_1^{\dagger} \Vec{\psi}_2 = 0, \\
    & H_1 \vec{\psi}_3 =0, \quad H_2 \vec{\psi}_1 = 0, \\
    & H_2 \vec{\psi}_2 = 0, \quad H_2^{\dagger} \vec{\psi}_3 = 0, \\
    & H_1^\dagger \vec{\psi}_3 + H_2^{\dagger} \Vec{\psi}_2 = 0.
 \end{aligned}  
 \label{eq:lshape-gen-eig-prob-app}
\end{equation} 

Using the destructive interference conditions in \eqref{eq:lshape-gen-eig-prob-app}, we can eliminate $H_1,H_2$ and obtain the CLS constraints:
\begin{equation}
  \begin{aligned}
    & \langle \vec{\psi}_2 \vert H_0 \vert \vec{\psi}_1 \rangle = E_{FB} \langle \vec{\psi}_2 \vert \vec{\psi}_1 \rangle, \\
    & \langle \vec{\psi}_3 \vert H_0 \vert \vec{\psi}_1 \rangle = E_{FB} \langle \vec{\psi}_3 \vert \vec{\psi}_1 \rangle, \\
    &  \langle \vec{\psi}_3 \vert H_0 \vert \vec{\psi}_2 \rangle = E_{FB} \langle \vec{\psi}_3 \vert \vec{\psi}_2 \rangle, \\
    & \langle \vec{\psi}_2 \vert \left( E_{FB} - H_0 \right) \vert \vec{\psi}_2 \rangle + \langle \vec{\psi}_3 \vert \left( E_{FB} - H_0 \right) \vert \vec{\psi}_3 \rangle = \langle \vec{\psi}_1 \vert \left( E_{FB} - H_0 \right) \vert \vec{\psi}_1 \rangle.
  \end{aligned}
  \label{eq:non-lin-const-L-shape-app}
\end{equation} 
Here, bra/ket notations are equivalent to the corresponding vector, i.e. $\vec{\psi}_{i}=\vert \psi_{i} \rangle$.

We can now define
\begin{equation}
    H_2 \vert \psi_3 \rangle = \vert y \rangle = (E_{FB} -H_0) \vert \psi_1 \rangle - H_1 \vert \psi_2 \rangle, \quad \langle \psi_3 \vert H_1 = - \langle \psi_2 \vert H_2 = \langle z \vert. 
    \label{eq:y-z-def-app}
\end{equation}
Using destructive interference conditions (and by multiplying $\vert \psi_i \rangle, i=1,2,3$ to the above equations from left or right) we can get 
\begin{equation}
    \vert y \rangle \rightarrow Q_{2,3} \vert y \rangle, \langle z \vert \rightarrow \langle z \vert Q_{1,2,3}.
    \label{eq:y_z_vectors-app}
\end{equation} 

For $\nu=3$, the only possibility to satisfy the second equation above for $\langle z\vert Q_{1,2,3}\ne0$ is that $\vert\psi_{1}\rangle,\vert\psi_{2}\rangle,\vert\psi_{3}\rangle$ are linearly dependent, 
\begin{equation}
    c_{1}\vert\psi_{1}\rangle+c_{2}\vert\psi_{2}\rangle+c_{3}\vert\psi_{3}\rangle=0.
\end{equation} 
As long as all the vectors are not proportional, the CLS is not reducible to the U=1 class. For example, the Lieb lattice has $\psi_{1}=\psi_{2}+\psi_{3}$ and it is a $U=(2,2,1)$ class. 

\subsubsection{$\nu>3$ case}

\paragraph{{\bf Resolving $H_1$:}}

Suppose $H_1$ can be written as the sum of two projectors:
\begin{equation}
    H_1 = \vert a \rangle \langle b \vert + \vert c \rangle \langle d \vert. 
\end{equation} 
Then using the destructive interference conditions in \eqref{eq:lshape-gen-eig-prob-app}, we have 
\begin{equation}
    H_1 = Q_2 ( \vert a \rangle \langle b \vert + \vert c \rangle \langle d \vert ) Q_{1,3}.
    \label{eq:Hx_lshape_proj}
\end{equation} 
Putting \eqref{eq:Hx_lshape_proj} and \eqref{eq:y_z_vectors-app} into \eqref{eq:lshape-gen-eig-prob-app}, we have 
\begin{align}
    Q_2 \vert a \rangle \langle b \vert Q_{1,3} \vert \psi_2 \rangle + Q_2 \vert c \rangle \langle d \vert Q_{1,3} \vert \psi_2 \rangle &= (E_{FB} - H_0) \vert \psi_1 \rangle - Q_{2,3} \vert y \rangle, \label{eq:hx-eq-1} \\
    \langle \psi_1 \vert Q_2 \vert a \rangle \langle b \vert Q_{1,3} + \langle \psi_1 \vert Q_2 \vert c \rangle \langle d \vert  Q_{1,3} &= \langle \psi_2 \vert (E_{FB} - H_0), \label{eq:hx-eq-2} \\ 
    \langle \psi_3 \vert Q_2 \vert a \rangle \langle b \vert Q_{1,3} + \langle \psi_3 \vert Q_2 \vert c \rangle \langle d \vert  Q_{1,3} &= \langle z \vert Q_{1,2,3}. \label{eq:hx-eq-3}
\end{align} 
In Eqs. (\ref{eq:hx-eq-1}--\ref{eq:hx-eq-3}), there are three equations and four vectors. We can therefore freely choose one of the vectors. Suppose $\vert a \rangle$ is free a variable, and then 
\begin{equation}
    Q_{2}\vert c\rangle=\frac{1}{\langle d\vert Q_{1,3}\vert\psi_{2}\rangle}\left(\left(E_{FB}-H_{0}\right)\vert\psi_{1}\rangle-Q_{2,3}\vert y\rangle-Q_{2}\vert a\rangle\langle b\vert Q_{1,3}\vert\psi_{2}\rangle\right).
\end{equation} 
Inserting this into \eqref{eq:hx-eq-2} and \eqref{eq:hx-eq-3} we have
\begin{equation}
    \begin{aligned}           
        \langle b\vert Q_{1,3} &= \frac{\langle b\vert Q_{1,3}\vert\psi_{2}\rangle}{\langle\psi_{1}\vert\left(E_{FB}-H_{0}\right)\vert\psi_{1}\rangle-\langle\psi_{1}\vert Q_{2,3}\vert y\rangle}\langle\psi_{2}\vert\left(E_{FB}-H_{0}\right) \\ & +\frac{\langle\psi_{1}\vert\left(E_{FB}-H_{0}\right)\vert\psi_{1}\rangle-\langle\psi_{1}\vert Q_{2,3}\vert y\rangle-\langle\psi_{1}\vert Q_{2}\vert a\rangle\langle b\vert Q_{1,3}\vert\psi_{2}\rangle}{\left(\langle\psi_{1}\vert\left(E_{FB}-H_{0}\right)\vert\psi_{1}\rangle-\langle\psi_{1}\vert Q_{2,3}\vert y\rangle\right)\langle\psi_{3}\vert Q_{2}\vert a\rangle}\langle z\vert Q_{1,2,3},\\
        \langle d\vert Q_{1,3} &= \frac{\langle d\vert Q_{1,3}\vert\psi_{2}\rangle}{\langle\psi_{1}\vert\left(E_{FB}-H_{0}\right)\vert\psi_{1}\rangle-\langle\psi_{1}\vert Q_{2,3}\vert y\rangle}\langle \psi_{2} \vert \left(E_{FB}-H_{0}\right) \\ & -\frac{\langle\psi_{1}\vert Q_{2}\vert a\rangle\langle d\vert Q_{1,3}\vert\psi_{2}\rangle}{\left(\langle\psi_{1}\vert\left(E_{FB}-H_{0}\right)\vert\psi_{1}\rangle-\langle\psi_{1}\vert Q_{2,3}\vert y\rangle\right)\langle\psi_{3}\vert Q_{2}\vert a\rangle}\langle z\vert Q_{1,2,3}.
    \end{aligned}
\end{equation}
Putting these solutions into \eqref{eq:Hx_lshape_proj} gives
 \begin{equation}
     \begin{aligned}
        H_1 &=\frac{1}{\langle\psi_{3}\vert Q_{2}\vert a\rangle}Q_{2}\vert a\rangle\langle z\vert Q_{1,2,3}\\&+\frac{1}{\langle\psi_{1}\vert\left(E_{FB}-H_{0}\right)\vert\psi_{1}\rangle-\langle\psi_{1}\vert Q_{2,3}\vert y\rangle}\left(E_{FB}-H_{0}\right)\vert\psi_{1}\rangle\langle\psi_{2}\vert\left(E_{FB}-H_{0}\right)\\&-\frac{\langle\psi_{3}\vert Q_{2}\vert a\rangle-\langle\psi_{1}\vert Q_{2}\vert a\rangle}{\left(\langle\psi_{1}\vert\left(E_{FB}-H_{0}\right)\vert\psi_{1}\rangle-\langle\psi_{1}\vert Q_{2,3}\vert y\rangle\right)\langle\psi_{3}\vert Q_{2}\vert a\rangle}Q_{2,3}\vert y\rangle\langle\psi_{2}\vert\left(E_{FB}-H_{0}\right)\\&-\frac{\langle\psi_{1}\vert Q_{2}\vert a\rangle}{\left(\langle\psi_{1}\vert\left(E_{FB}-H_{0}\right)\vert\psi_{1}\rangle-\langle\psi_{1}\vert Q_{2,3}\vert y\rangle\right)\langle\psi_{3}\vert Q_{2}\vert a\rangle}\left(E_{FB}-H_{0}\right)\vert\psi_{1}\rangle\langle z\vert Q_{1,2,3}.
     \end{aligned}
     \label{eq:hx-sol-1}
 \end{equation} 
 This solution must satisfy \eqref{eq:hx-eq-1}, which requires 
 \begin{equation}
     \begin{aligned}
        Q_{2,3}\vert y\rangle & \rightarrow\frac{\langle\psi_{3}\vert\left(E_{FB}-H_{0}\right)\vert\psi_{3}\rangle}{\langle\psi_{1}\vert Q_{2,3}\vert y\rangle}Q_{2,3}\vert y\rangle,\\ 
        Q_{2}\vert a\rangle&\rightarrow Q_{1,2}\vert a\rangle,
     \end{aligned}
 \end{equation} 
 and putting into \eqref{eq:hx-sol-1} yields
 \begin{equation}
     \begin{aligned}
         H_1&=\frac{1}{\langle\psi_{3}\vert Q_{1,2}\vert a\rangle}Q_{1,2}\vert a\rangle\langle z\vert Q_{1,2,3} +\frac{\left(E_{FB}-H_{0}\right)\vert\psi_{1}\rangle \langle\psi_{2}\vert\left(E_{FB}-H_{0}\right)}{\langle\psi_{2}\vert\left(E_{FB}-H_{0}\right)\vert\psi_{2}\rangle}\\ 
         & -\frac{\langle\psi_{3}\vert\left(E_{FB}-H_{0}\right)\vert\psi_{3}\rangle}{\langle\psi_{2}\vert\left(E_{FB}-H_{0}\right)\vert\psi_{2}\rangle \langle\psi_{1}\vert Q_{2,3}\vert y\rangle}Q_{2,3}\vert y\rangle \langle\psi_{2}\vert\left(E_{FB}-H_{0}\right).
     \end{aligned}
 \end{equation}

\paragraph{{\bf Resolving $H_2$:}} 

Suppose $H_2$ has the following form: 
\begin{equation}
    H_2=Q_{3}\left(\vert u\rangle\langle v\vert+\vert w\rangle\langle x\vert\right)Q_{1,2}. 
    \label{eq:u221-hy}
\end{equation} 
The inverse eigenvalue problem for this $H_2$ is 
\begin{eqnarray}
    Q_{3}\vert u\rangle\langle v\vert Q_{1,2}\vert\psi_{3}\rangle+Q_{3}\vert w\rangle\langle x\vert Q_{1,2}\vert\psi_{3}\rangle &=& \frac{\langle\psi_{3}\vert\left(E_{FB}-H_{0}\right)\vert\psi_{3}\rangle}{\langle\psi_{1}\vert Q_{2,3}\vert y\rangle}Q_{2,3}\vert y\rangle, \label{eq:hy-1} \\ 
    \langle\psi_{1}\vert Q_{3}\vert u\rangle\langle v\vert Q_{1,2}+\langle\psi_{1}\vert Q_{3}\vert w\rangle\langle x\vert Q_{1,2} &=& \langle\psi_{3}\vert(E_{FB}-H_{0}), \label{eq:hy-2} \\ 
    \langle\psi_{2}\vert Q_{3}\vert u\rangle\langle v\vert Q_{1,2}+\langle\psi_{2}\vert Q_{3}\vert w\rangle\langle x\vert Q_{1,2} &=& -\langle z\vert Q_{1,2,3}. \label{eq:hy-3}
\end{eqnarray} 
Suppose then that $\vert u \rangle$ is a free parameter, so from \eqref{eq:hy-1} we have
\begin{equation}
    Q_{3}\vert w\rangle=\frac{\langle\psi_{3}\vert\left(E_{FB}-H_{0}\right)\vert\psi_{3}\rangle}{\langle\psi_{1}\vert Q_{2,3}\vert y\rangle\langle x\vert Q_{1,2}\vert\psi_{3}\rangle}Q_{2,3}\vert y\rangle-\frac{\langle v\vert Q_{1,2}\vert\psi_{3}\rangle}{\langle x\vert Q_{1,2}\vert\psi_{3}\rangle}Q_{3}\vert u\rangle.
\end{equation} 
Putting this into \eqref{eq:hy-2} and \eqref{eq:hy-3} gives
\begin{equation}
    \begin{aligned}
        \langle x\vert Q_{1,2} &=\frac{\langle x\vert Q_{1,2}\vert\psi_{3}\rangle}{\langle\psi_{3}\vert\left(E_{FB}-H_{0}\right)\vert\psi_{3}\rangle}\langle\psi_{3}\vert(E_{FB}-H_{0})\\ &+\frac{\langle x\vert Q_{1,2}\vert\psi_{3}\rangle\langle\psi_{1}\vert Q_{3}\vert u\rangle}{\langle\psi_{3}\vert\left(E_{FB}-H_{0}\right)\vert\psi_{3}\rangle\langle\psi_{2}\vert Q_{3}\vert u\rangle}\langle z\vert Q_{1,2,3}, \\
        \langle v\vert Q_{1,2} &=\frac{\langle v\vert Q_{1,2}\vert\psi_{3}\rangle}{\langle\psi_{3}\vert\left(E_{FB}-H_{0}\right)\vert\psi_{3}\rangle}\langle\psi_{3}\vert(E_{FB}-H_{0})\\ & +\frac{\langle v\vert Q_{1,2}\vert\psi_{3}\rangle\langle\psi_{1}\vert Q_{3}\vert u\rangle-\langle\psi_{3}\vert\left(E_{FB}-H_{0}\right)\vert\psi_{3}\rangle}{\langle\psi_{3}\vert\left(E_{FB}-H_{0}\right)\vert\psi_{3}\rangle\langle\psi_{2}\vert Q_{3}\vert u\rangle}\langle z\vert Q_{1,2,3}.
    \end{aligned}
\end{equation} 
 Putting the above solutions for $\vert v\rangle,\ \vert w\rangle,\ \vert x\rangle$ into \eqref{eq:u221-hy} we have
 \begin{equation}
 \begin{aligned}
     H_2 &= \frac{1}{\langle\psi_{1}\vert Q_{2,3}\vert y\rangle}Q_{2,3}\vert y\rangle\langle\psi_{3}\vert(E_{FB}-H_{0})\\ 
     &+\frac{1}{\langle\psi_{2}\vert Q_{3}\vert u\rangle}\left(\frac{\langle\psi_{1}\vert Q_{3}\vert u\rangle}{\langle\psi_{1}\vert Q_{2,3}\vert y\rangle}Q_{2,3}\vert y\rangle-Q_{3}\vert u\rangle\right)\langle z\vert Q_{1,2,3}.
\end{aligned}
 \end{equation} 
 
 Therefore, the solution for the U=(2,2,1), $\nu>3$ case is 
 \begin{equation}
     \begin{aligned}
         H_1&=\frac{1}{\langle\psi_{3}\vert Q_{1,2}\vert a\rangle}Q_{1,2}\vert a\rangle\langle z\vert Q_{1,2,3} +\frac{\left(E_{FB}-H_{0}\right)\vert\psi_{1}\rangle \langle\psi_{2}\vert\left(E_{FB}-H_{0}\right)}{\langle\psi_{2}\vert\left(E_{FB}-H_{0}\right)\vert\psi_{2}\rangle}\\ 
         & -\frac{\langle\psi_{3}\vert\left(E_{FB}-H_{0}\right)\vert\psi_{3}\rangle}{\langle\psi_{2}\vert\left(E_{FB}-H_{0}\right)\vert\psi_{2}\rangle \langle\psi_{1}\vert Q_{2,3}\vert y\rangle}Q_{2,3}\vert y\rangle \langle\psi_{2}\vert\left(E_{FB}-H_{0}\right),\\
         H_2&=\frac{1}{\langle\psi_{1}\vert Q_{2,3}\vert y\rangle}Q_{2,3}\vert y\rangle\langle\psi_{3}\vert(E_{FB}-H_{0}) \\ & +\frac{1}{\langle\psi_{2}\vert Q_{3}\vert u\rangle}\left(\frac{\langle\psi_{1}\vert Q_{3}\vert u\rangle}{\langle\psi_{1}\vert Q_{2,3}\vert y\rangle}Q_{2,3}\vert y\rangle-Q_{3}\vert u\rangle\right)\langle z\vert Q_{1,2,3},
     \end{aligned}
     \label{eq:sol-u221-hxy}
 \end{equation} 
 where $E_{FB},\ H_0,\ \vert \psi_i \rangle, i=1,2,3$ are chosen respecting the constraints  \eqref{eq:non-lin-const-L-shape}, and $\vert y \rangle, \vert z \rangle, \vert a \rangle, \vert u \rangle$ are free parameters.  
 
 When $\vert z \rangle = 0$ and $Q_{2,3} \vert y \rangle = (E_{FB} - H_0 ) \vert \psi_1 \rangle$, the solution reduces to that of the special case in Appendix \ref{app:u221-special}.

 \subsubsection{Three-band case} 
 
 In the case of three bands, $Q_{123}=0$, the only possibility that $Q_{1,2,3}\ne0$ thus $Q_{1,2,3}\vert z\rangle\ne0$ is 
\[
c_{1}\vert\psi_{1}\rangle+c_{2}\vert\psi_{2}\rangle+c_{3}\vert\psi_{3}\rangle=0.
\]
 Also, since
\[
Q_{12}\vert\psi_{3}\rangle=Q_{13}\vert\psi_{2}\rangle=Q_{23}\vert\psi_{1}\rangle=0,
\]
any $Q_{12},Q_{13},Q_{23}$ can be used as
$Q_{123}$. We use $Q_{23}$ for $Q_{123}$, and then, for given the inverse eigenvalue problem, suppose 
\[
\vert\psi_{1}\rangle=\alpha\vert\psi_{2}\rangle+\beta\vert\psi_{3}\rangle.
\]
The CLS constraints become 
\begin{equation}
    \begin{aligned}
        \langle\psi_{2}\vert\left(E_{FB}-H_{0}\right)\vert\psi_{2}\rangle&=0,\\\langle\psi_{3}\vert\left(E_{FB}-H_{0}\right)\vert\psi_{3}\rangle&=0,\\\langle\psi_{3}\vert\left(E_{FB}-H_{0}\right)\vert\psi_{2}\rangle&=0,\\\left(E_{FB}-H_{0}\right)\left(\alpha\vert\psi_{2}\rangle+\beta\vert\psi_{3}\rangle\right)&=0.
    \end{aligned}
    \label{eq:u221-cls-cons-3band}
\end{equation}
 Then, for the CLS, $E_{FB}$, and $H_{0}$, satisfying the above constraints,
the eigenvalue problem and destructive interference conditions become
\begin{equation}
\begin{aligned}
\beta\langle\psi_{3}\vert H_1&=\langle\psi_{2}\vert\left(E_{FB}-H_{0}\right),\\\alpha\langle\psi_{2}\vert H_2&=\langle\psi_{3}\vert\left(E_{FB}-H_{0}\right),\\H_1\vert\psi_{2}\rangle&=H_1\vert\psi_{3}\rangle=0,\\H_2\vert\psi_{2}\rangle&=H_2\vert\psi_{3}\rangle=0,\\\langle\psi_{2}\vert H_1&=0,\\\langle\psi_{3}\vert H_2&=0.
\end{aligned}
\label{eq:u221-eig-prob-3band}
\end{equation}

\paragraph{Resolving $H_1$:} The inverse eigenvalue problem for $H_1$ is 
\[
\begin{aligned}\langle\psi_{3}\vert H_1 & =\langle\psi_{2}\vert\left(E_{FB}-H_{0}\right).\end{aligned}
\]
 Suppose $H_1$ is a projector 
\[
H_1=Q_{2}\vert a\rangle\langle b\vert Q_{23},
\]
 and then the inverse eigenvalue problem for $H_1$ becomes
\[
\begin{aligned}
\langle b\vert Q_{23} & =\frac{1}{\langle\psi_{3}\vert Q_{2}\vert a\rangle}\langle\psi_{2}\vert\left(E_{FB}-H_{0}\right),
\end{aligned}
\]
 which gives 
\[
H_1=\frac{Q_{2}\vert a\rangle\langle\psi_{2}\vert\left(E_{FB}-H_{0}\right)}{\langle\psi_{3}\vert Q_{2}\vert a\rangle},
\]
 where $\vert a\rangle$ is fixed freely (i.e. a free parameter) and
not proportional to $\vert \psi_{2} \rangle$. 

Similarly, we can resolve $H_2$ as
\[
H_2=\frac{Q_{3}\vert c\rangle\langle\psi_{3}\vert\left(E_{FB}-H_{0}\right)}{\langle\psi_{2}\vert Q_{3}\vert c\rangle},
\]
 where as above $\vert c\rangle$ is fixed freely and not proportional to $\vert\psi_{2}\rangle$. 

Therefore, the final solution of the U=(2,2,1) three-band problem is 
\begin{equation}
    \begin{aligned}
        H_1 &=\frac{Q_{2}\vert a\rangle\langle\psi_{2}\vert\left(E_{FB}-H_{0}\right)}{\langle\psi_{3}\vert Q_{2}\vert a\rangle}, \\
        H_2 &= \frac{Q_{3}\vert c\rangle\langle\psi_{3}\vert\left(E_{FB}-H_{0}\right)}{\langle\psi_{2}\vert Q_{3}\vert c\rangle},
    \end{aligned}
    \label{eq:u221-hxy-sol-3band}
\end{equation} 
where $E_{FB},\ H_0,\ \vert \psi_1 \rangle,\ \vert \psi_2 \rangle $ are chosen respecting the constraints \eqref{eq:u221-cls-cons-3band} and $\vert a \rangle,  \vert c \rangle$ are free parameters.
 
 \subsection{Special case for U=(2,2,1) with nearest neighbor hoppings} 
 \label{app:u221-special}

In this section we consider the case when $\vert z \rangle = 0$. Then the eigenvalue problem \eqref{eq:lshape-gen-eig-prob-app} reads
\begin{equation}
 \begin{aligned}
    & H_1 \Vec{\psi}_2 + H_2 \vec{\psi}_3 = (E_{FB} -H_0) \vec{\psi}_1, \\
    & H_1^{\dagger} \vec{\psi}_1 = (E_{FB} - H_0) \Vec{\psi}_2, \\
    & H_2^{\dagger} \vec{\psi}_1 = (E_{FB} - H_0) \vec{\psi}_3, \\
    & H_1 \vec{\psi}_1 = 0, \quad H_1^{\dagger} \Vec{\psi}_2 = 0, \\
    & H_1 \vec{\psi}_3 = H_1^{\dagger} \vec{\psi}_3 = 0, \\
    & H_2 \vec{\psi}_1 = 0, \quad H_2^{\dagger} \vec{\psi}_3 = 0, \\
    & H_2 \Vec{\psi}_2 = H_2^{\dagger} \Vec{\psi}_2 = 0.
 \end{aligned}  
 \label{eq:lshap-eig-prob}
\end{equation} 
The CLS constraints here are the same with \eqref{eq:non-lin-const-L-shape}.
Note the bra/ket notations are equivalent to the corresponding vector, i.e. $\vec{\psi}_{i,j}=\vert \psi_{i,j} \rangle$.

Assume that a CLS satisfying \eqref{eq:non-lin-const-L-shape} is given, and then we solve \eqref{eq:lshap-eig-prob}. We write the $H_1,H_2$ as single projector 
\begin{equation}
    H_1 = \vert x \rangle \langle y \vert, \ \ H_2 = \vert v \rangle \langle w \vert. 
\end{equation} 
In order to satisfy the destructive interference conditions we introduce the following operators: 
\begin{equation}
    \begin{aligned}
        Q_{23} & = I - \frac{Q_3 \vert \psi_2 \rangle \langle \psi_2 \vert}{\langle \psi_2 \vert Q_3 \vert \psi_2 \rangle} - \frac{Q_2 \vert \psi_3 \rangle \langle \psi_3 \vert}{\langle \psi_3 \vert Q_2 \vert \psi_3 \rangle}, \\
        Q_{13} & = I - \frac{\vert \psi_1 \rangle \langle \psi_1 \vert Q_3}{\langle \psi_1 \vert Q_3 \vert \psi_1 \rangle} - \frac{ \vert \psi_3 \rangle \langle \psi_3 \vert Q_1}{\langle \psi_3 \vert Q_1 \vert \psi_3 \rangle}, \\
        Q_{12} & = I - \frac{\vert \psi_1 \rangle \langle \psi_1 \vert Q_2}{\langle \psi_1 \vert Q_2 \vert \psi_1 \rangle} - \frac{ \vert \psi_2 \rangle \langle \psi_2\vert Q_1}{\langle \psi_2 \vert Q_1 \vert \psi_2 \rangle}, \\
        Q_i & = I - \frac{\vert \psi_i \rangle \langle \psi_i \vert}{\langle \psi_i \vert \psi_i \rangle}.
    \end{aligned}
\end{equation} 

Then we can write $H_1,H_2$ as 
\begin{equation}
   H_1 = Q_{23} \vert x \rangle \langle y \vert Q_{13},\ \ H_2 = Q_{23}\vert v \rangle \langle w \vert Q_{12}, 
   \label{eq:hx-hy-proj}
\end{equation} 
and putting \eqref{eq:hx-hy-proj} into \eqref{eq:lshap-eig-prob} we get 
\begin{equation}
    \begin{aligned}
       & Q_{23} \vert x \rangle \langle y \vert Q_{13} \vert \psi_2 \rangle +  Q_{23}\vert v \rangle \langle w \vert Q_{12} \vert \psi_3 \rangle = ( E_{FB} - H_0 ) \vert \psi_1 \rangle, \\
       & Q_{13}^{\dagger} \vert y \rangle \langle x \vert Q_{23}^{\dagger} \vert \psi_1 \rangle = ( E_{FB} - H_0 ) \vert \psi_2 \rangle, \\
       & Q_{12}^{\dagger} \vert w \rangle \langle v \vert Q_{23}^{\dagger} \vert \psi_1 \rangle = ( E_{FB} - H_0 ) \vert \psi_3 \rangle.
    \end{aligned}
    \label{eq:lshape-eige-proj}
\end{equation} 
In \eqref{eq:lshape-eige-proj}, there are three equations and four vectors, $\vert x \rangle, \vert y \rangle, \vert v \rangle, \vert w \rangle$, so the vectors can be chosen freely. We choose $\vert v \rangle = \vert x \rangle$, and then 
    \begin{align}
       & Q_{23} \vert x \rangle (\langle y \vert Q_{13} \vert \psi_2 \rangle +  \langle w \vert Q_{12} \vert \psi_3 \rangle) = ( E_{FB} - H_0 ) \vert \psi_1 \rangle, \label{eq:lshape-eige-proj-1}\\
       & Q_{13}^{\dagger} \vert y \rangle  = \frac{( E_{FB} - H_0 ) \vert \psi_2 \rangle}{\langle x \vert Q_{23}^{\dagger} \vert \psi_1 \rangle}, \label{eq:lshape-eige-proj-2}\\
       & Q_{12}^{\dagger} \vert w \rangle = \frac{( E_{FB} - H_0 ) \vert \psi_3 \rangle}{\langle v \vert Q_{23}^{\dagger} \vert \psi_1 \rangle}. \label{eq:lshape-eige-proj-3} 
    \end{align} 
Now putting \eqref{eq:lshape-eige-proj-2} and \eqref{eq:lshape-eige-proj-3} into \eqref{eq:lshape-eige-proj-1}, and using the last identity in the compatibility constraints \eqref{eq:non-lin-const-L-shape}, we get 
\begin{equation}
    Q_{23} \vert x \rangle = \frac{\langle \psi_1 \vert Q_{23} \vert x \rangle}{\langle \psi_1 \vert (E_{FB} - H_0 ) \vert \psi_1 \rangle} (E_{FB} - H_0 ) \vert \psi_1 \rangle. 
    \label{eq:proj-sol-x}
\end{equation} 
Putting \eqref{eq:lshape-eige-proj-2}, \eqref{eq:lshape-eige-proj-3}, and \eqref{eq:proj-sol-x} into \eqref{eq:hx-hy-proj}, we get 
\begin{equation}
    \begin{aligned}
        H_1 & = \frac{(E_{FB} - H_0 ) \vert \psi_1 \rangle \langle \psi_2 \vert ( E_{FB} - H_0 )}{\langle \psi_1 \vert (E_{FB} - H_0 ) \vert \psi_1 \rangle}, \\
        H_2 & = \frac{(E_{FB} - H_0 ) \vert \psi_1 \rangle \langle \psi_3 \vert ( E_{FB} - H_0 )}{\langle \psi_1 \vert (E_{FB} - H_0 ) \vert \psi_1 \rangle}.
    \end{aligned}
\end{equation} 
Note that this solution is not limited to three bands. If we define a transverse operator 
\begin{equation}
    Q_{123} = I - \frac{\vert \psi_1 \rangle \langle \psi_1 \vert Q_{23}^{\dagger}}{\langle \psi_1 \vert Q_{23}^{\dagger} \vert \psi_1 \rangle} - \frac{\vert \psi_2 \rangle \langle \psi_2 \vert Q_{13}}{\langle \psi_2 \vert Q_{13} \vert \psi_2 \rangle} - \frac{\vert \psi_3 \rangle \langle \psi_3 \vert Q_{12}}{\langle \psi_3 \vert Q_{12} \vert \psi_3 \rangle},
\end{equation} 
then we can write 
\begin{equation}
    \begin{aligned}
        H_1 & = \frac{(E_{FB} - H_0 ) \vert \psi_1 \rangle \langle \psi_2 \vert ( E_{FB} - H_0 )}{\langle \psi_1 \vert (E_{FB} - H_0 ) \vert \psi_1 \rangle} + Q_{23} K Q_{123}, \\
        H_2 & = \frac{(E_{FB} - H_0 ) \vert \psi_1 \rangle \langle \psi_3 \vert ( E_{FB} - H_0 )}{\langle \psi_1 \vert (E_{FB} - H_0 ) \vert \psi_1 \rangle} + Q_{23} M Q_{123},
    \end{aligned}
\end{equation} 
where $K, M$ are arbitrary $\nu \times \nu$ matrices. Because of operators $Q_{23},\ Q{123}$, the extra $K, M$ terms do not effect the eigenvalue problem \eqref{eq:lshap-eig-prob}. Note that when $\nu=3$, the operator $Q_{123} =0$; therefore, the extra term vanishes for $\nu=3$.

\subsection{U=(2,2,0) case with three bands}
\label{app:nn-u220-nu3} 

Putting the values $U_2=2,\ s=0$ to Eq. \eqref{eq:u22-gen-eig-prob}, we get the eigenvalue problem and destructive interference conditions in this case as follows
\begin{equation}
\begin{aligned}
H_1 \psi _{2,1} + H_2 \psi _{1,2} &= \left(E_{FB}  - H_0\right) \psi _{1,1}\\
H_1^{\dagger} \psi _{1,1} + H_2 \psi _{2,2} &= \left(E_{FB}  - H_0\right) \psi _{2,1},\\
H_1 \psi _{2,2} + H_2^{\dagger} \psi _{1,1} &= \left(E_{FB}  - H_0\right) \psi _{1,2},\\
H_1^{\dagger} \psi _{1,2} + H_2^{\dagger} \psi _{2,1} &= \left(E_{FB}  - H_0\right) \psi _{2,2},
\end{aligned}
\label{eq:square_u22_eig_prob}
\end{equation}
with destructive interference conditions
\begin{equation}
 \begin{aligned}
  H_1 \psi _{1,1} &= H_1 \psi _{1,2} = 0,\\
  H_1^{\dagger} \psi _{2,1 } &= H_1^{\dagger} \psi _{2,2} = 0,\\ 
  H_2 \psi _{1,1 } &= H_2 \psi _{2,1} = 0,\\ 
  H_2^{\dagger} \psi _{1,2} &= H_2^{\dagger} \psi _{2,2} = 0.
 \end{aligned}
\end{equation}

Since $H_1, H_2$ have two zero modes, we can parameterize the hopping matrices $H_1,H_2$ in the following way: 

\begin{equation}
 H_1=|x\rangle \langle y|=\left(
\begin{array}{ccc}
 a d & a e & a f \\
 b d & b e & b f \\
 c d & c e & c f \\
\end{array}
\right),\ \ H_2=|u\rangle \langle v|=\left(
\begin{array}{ccc}
 g r & g s & g t \\
 h r & h s & h t \\
 l r & l s & l t \\
\end{array}
\right),
\label{eq:Hxy}
\end{equation}

where

\begin{equation}
 |x\rangle =\left(
\begin{array}{c}
 a \\
 b \\
 c \\
\end{array}
\right),\ \ \ |y\rangle =\left(
\begin{array}{c}
 d \\
 e \\
 f \\
\end{array}
\right),\ \ \ |u\rangle =\left(
\begin{array}{c}
 g \\
 h \\
 l \\
\end{array}
\right),\ \ \ |v\rangle =\left(
\begin{array}{c}
 r \\
 s \\
 t \\
\end{array}
\right).
\end{equation}

Matrices $H_1,H_2$ have two zero modes.

\begin{equation}
 \begin{aligned}
  H_1:& \left(
\begin{array}{c}
 -f \\
 0 \\
 d \\
\end{array}
\right) , \left(
\begin{array}{c}
 -e \\
 d \\
 0 \\
\end{array}
\right),\\
H_1^{\dagger}:& \left(
\begin{array}{c}
 -c \\
 0 \\
 a \\
\end{array}
\right) , \left(
\begin{array}{c}
 -b \\
 a \\
 0 \\
\end{array}
\right),\\
H_2:& \left(
\begin{array}{c}
 -t \\
 0 \\
 r \\
\end{array}
\right) , \left(
\begin{array}{c}
 -s \\
 r \\
 0 \\
\end{array}
\right),\\
H_2^{\dagger}:& \left(
\begin{array}{c}
 -l \\
 0 \\
 g \\
\end{array}
\right) , \left(
\begin{array}{c}
 -h \\
 g \\
 0 \\
\end{array}
\right).
 \end{aligned}
\end{equation}

Destructive interference conditions impose numerous constraints on $H_1,H_2$ and the CLS. Here we construct a parameterization of $ H_1,H_2 $ and the
CLS such that, with the parameterization, they satisfy the destructive interference conditions by default.

Consider that $\psi _{1,1}$ is the right zero mode of both $H_1, H_2$, and therefore it is parallel to the cross product of $y$ and $\nu$ and also parallel to one of the right
zero eigenvectors of $H_1,H_2$ as

\begin{equation}
\psi _{1,1}=\alpha  (y \times v) \parallel \left( \begin{array}{c}
 -f \\
 0 \\
 d \\
\end{array} \right) \parallel  \left(
\begin{array}{c}
 -t \\
 0 \\
 r \\
\end{array} \right). 
\end{equation}

Then $\psi _{2,1}$ is the left zero mode of $H_1 $ and the right zero mode of $H_2$, and therefore parallel to the cross product of $x$ and $\mu$ and also parallel
to one of the left (right) zero eigenvectors of $H_1 \left(H_2\right)$ as

\begin{equation}
\psi_2=\beta  (x\times v) \parallel \left(
\begin{array}{c}
 -b \\
 a \\
 0 \\
\end{array}
\right) \parallel \left(
\begin{array}{c}
 -s \\
 r \\
 0 \\
\end{array}
\right). 
\end{equation}

Next, $\psi _{1,2}$ is the right zero mode of $H_1$ and the left zero mode of $H_2$, and therefore parallel to the cross product of $y$ and $\mu$ and also parallel
to one of the right (left) zero eigenvectors of $H_1 \left(H_2\right)$ as

\begin{equation}
 \psi _{1,2}=\gamma  (y\times u)\parallel \left(
\begin{array}{c}
 -e \\
 d \\
 0 \\
\end{array}
\right)\parallel \left(
\begin{array}{c}
 -h \\
 g \\
 0 \\
\end{array}
\right). 
\end{equation}

Finally, $\psi _{2,2}$ is the left zero mode of both $H_1$ and $ H_2$, and therefore parallel to the cross product of $x$ and $\mu$ and also parallel to one of the
left zero eigenvectors of $ H_1,H_2$ as

\begin{equation}
\psi _{2,2}=\eta  (x\times u)\parallel \left(
\begin{array}{c}
 -c \\
 0 \\
 a \\
\end{array}
\right)\parallel \left(
\begin{array}{c}
 -l \\
 0 \\
 g \\
\end{array}
\right). 
\end{equation}

From all the above, we can see that

\begin{equation}
\begin{aligned}
 t&=f,\\
 d&=r=a=g,\\
 b&=s,\\
 e&=h,\\
 c&=l.
\end{aligned}
\label{eq:Hxy_elements}
\end{equation}

For simplicity, we choose all proportionality factors to be 1. Then the expressions for all $\psi$ reduce to the following equations:

\begin{equation}
\begin{aligned}
\psi_1 &= \left(\begin{array}{c} (-b f+e f) \alpha  \\ 
                                    0 \\ 
                                    (a b-a e) \alpha \} \\ \end{array} \right) =
 \left( \begin{array}{c} -f \text{$  \alpha $1} \\ 
                                      0 \\ 
                                      a \text{$\alpha $1} \\ \end{array} \right),\\
\psi_2 &= \left( \begin{array}{c} (-b c+b f) \beta  \\
 (a c-a f) \beta  \\
 0 \\
\end{array} \right) =
\left( \begin{array}{c} -b \text{$\beta $1} \\
                                    a \text{$\beta $1} \\
                                    0 \\ \end{array} \right),\\
\psi_3 &= \left(\begin{array}{c}
 (c e-e f) \gamma  \\
 (-a c+a f) \gamma  \\
 0 \\
\end{array}\right) =
\left(\begin{array}{c}
 -e \text{$\gamma $1} \\
 a \text{$\gamma $1} \\
 0 \\
\end{array}\right),\\
\psi_4 &= \left( \begin{array}{c}
 (b c-c e) \eta  \\
 0 \\
 (-a b+a e) \eta  \\
\end{array} \right) =
\left( \begin{array}{c}
 -c \text{$\eta $1} \\
 0 \\
 a \text{$\eta $1} \\
\end{array} \right).
\end{aligned}
\label{eq:psies}
\end{equation}

We solve the above equations to get

\begin{equation}
\text{$\alpha $1}=(b-e) \alpha ,\ \ \text{$\beta $1}=(c-f) \beta ,\ \ \text{$\gamma $1}=-c \gamma +f \gamma ,\ \ \text{$\eta $1}=-b \eta +e \eta. 
\label{eq:psies_coef}
\end{equation}

We can set one of the pre-factors to be 1; here, $\eta =1$. Then \eqref{eq:psies} becomes

\begin{equation}
\begin{aligned}
\psi _{1,1} &=\left(
\begin{array}{c}
 -(b-e) f \alpha  \\
 0 \\
 a (b-e) \alpha  \\
\end{array}
\right),\\
\psi_2 &=\left(
\begin{array}{c}
 -b (c-f) \beta  \\
 a (c-f) \beta  \\
 0 \\
\end{array}
\right),\\
\psi_3 &=\left(
\begin{array}{c}
 -e (-c \gamma +f \gamma ) \\
 a (-c \gamma +f \gamma ) \\
 0 \\
\end{array}
\right)\\
\psi_4 &= \left(
\begin{array}{c}
 -c (-b+e) \\
 0 \\
 a (-b+e) \\
\end{array} \right).
\end{aligned}
\end{equation}

We define $H_0$ as

\begin{equation}
H_0=\left(
\begin{array}{ccc}
 0 & 0 & 0 \\
 0 & 1 & 0 \\
 0 & 0 & \epsilon  \\
\end{array}
\right).
\end{equation}

Putting \eqref{eq:Hxy_elements} and \eqref{eq:psies_coef} into \eqref{eq:Hxy}, we get

\begin{equation}
H_1=\left(
\begin{array}{ccc}
 a^2 & a e & a f \\
 a b & b e & b f \\
 a c & c e & c f \\
\end{array}
\right),\\ \ _y=\left(
\begin{array}{ccc}
 a^2 & a b & a f \\
 a e & b e & e f \\
 a c & b c & c f \\
\end{array}
\right).
\label{eq:Hxy_new}
\end{equation}

Now the eigenvalue problem \eqref{eq:square_u22_eig_prob} becomes

\begin{equation}
\begin{aligned}
\left(
\begin{array}{c}
 (b-e) \left(-a^2 (c-f) (\beta +\gamma )+f \alpha  E_{FB} \right) \\
 -a (b-e) (c-f) (b \beta +e \gamma ) \\
 a (b-e) (-c (c-f) (\beta +\gamma )+\alpha  (\epsilon -E_{FB} )) \\
\end{array}
\right)==\left(
\begin{array}{c}
 0 \\
 0 \\
 0 \\
\end{array}
\right),\\
\\
\left(
\begin{array}{c}
 (c-f) \left(a^2 (b-e) (1+\alpha )+b \beta  E_{FB} \right) \\
 a (c-f) ((b-e) e (1+\alpha )+\beta -\beta  E_{FB} ) \\
 a (b-e) (c-f) (c+f \alpha ) \\
\end{array}
\right) &= \left( \begin{array}{c}
 0 \\
 0 \\
 0 \\
\end{array} \right),\\
\left( \begin{array}{c}
 (c-f) \left(a^2 (b-e) (1+\alpha )-e \gamma  E_{FB} \right) \\
 a (c-f) (b (b-e) (1+\alpha )+\gamma  (-1+E_{FB} )) \\
 a (b-e) (c-f) (c+f \alpha ) \\
\end{array}
\right) &=\left( \begin{array}{c}
 0 \\
 0 \\
 0 \\
\end{array} \right),\\
\left( \begin{array}{c}
 (b-e) \left(-a^2 (c-f) (\beta +\gamma )-c E_{FB} \right) \\
 -a (b-e) (c-f) (b \beta +e \gamma ) \\
 a (b-e) (-(c-f) f (\beta +\gamma )-\epsilon +E_{FB} ) \\
\end{array}
\right) &=\left( \begin{array}{c}
 0 \\
 0 \\
 0 \\
\end{array} \right).
\end{aligned}
\end{equation}

Assuming that $a\neq 0,\ c\neq f,\ b\neq e$, we solve the above equations. Also in the eigenvalue problem we can set $ a=1$. (We have set $\eta =1$ already. 
This can still be done by dividing all equations by $a$, then dividing all equations by $\frac{\eta }{a}$; thus, we can set $\eta $=1 as well.) They read:

\begin{equation}
\begin{aligned}
b & =\frac{\sqrt{2} (-1+E_{FB} )}{E_{FB}  \sqrt{\frac{-\sqrt{-\alpha  (-1+E_{FB} )^2 E_{FB} ^4 \left(4 (1+\alpha )^2-\alpha  E_{FB} ^2\right)}+(-1+E_{FB}
) E_{FB}  \left(-2 (1+\alpha )^2+\alpha  E_{FB} ^2\right)}{(1+\alpha )^2 E_{FB} ^2}}}, \\
c &= \frac{\sqrt{\alpha } \sqrt{-\epsilon +E_{FB} }}{\sqrt{E_{FB} }}, \\
e &=-\frac{\sqrt{\frac{-\sqrt{-\alpha  (-1+E_{FB} )^2 E_{FB} ^4 \left(4 (1+\alpha
)^2-\alpha  E_{FB} ^2\right)}+(-1+E_{FB} ) E_{FB}  \left(-2 (1+\alpha )^2+\alpha  E_{FB} ^2\right)}{(1+\alpha )^2 E_{FB} ^2}}}{\sqrt{2}}, \\
f &=-\frac{\sqrt{-\epsilon +E_{FB} }}{\sqrt{\alpha } \sqrt{E_{FB} }}, \\
\beta &= \frac{-\alpha  (-1+E_{FB} ) E_{FB} ^3+\sqrt{-\alpha  (-1+E_{FB} )^2 E_{FB} ^4 \left(4 (1+\alpha )^2-\alpha  E_{FB} ^2\right)}}{2 (1+\alpha
) (-1+E_{FB} ) E_{FB} ^2},\\ \ \ \
\gamma &=-\frac{\alpha  (-1+E_{FB} ) E_{FB} ^3+\sqrt{-\alpha  (-1+E_{FB} )^2 E_{FB} ^4 \left(4 (1+\alpha )^2-\alpha  E_{FB} ^2\right)}}{2 (1+\alpha ) (-1+E_{FB} ) E_{FB} ^2}.
\end{aligned}
\end{equation}

There are many possible solutions, with the above solution representing one of them.

\section{Next nearest neighbor hoppings}
\label{app:2d-nnn}

\subsection{U=(2,2,1) case}
\label{app:nnn-u221}

The configuation of this case is shown in Fig. \ref{fig:u22-nnn-config} (c) in the main text. Putting the values $U_2=2,\ s=1$ to Eq. \eqref{eq:u22-nnn-eig-prob}, we get the eigenvalue problem and destructive interference conditions in this case as follows 
\begin{equation}
    \begin{aligned}
       H_1 \vert \psi_2 \rangle + H_2 \vert \psi_3 \rangle &= (E_{FB} - H_0) \vert \psi_1 \rangle, \\
       H_1^\dagger \vert \psi_1 \rangle + H_3^\dagger \vert \psi_3 \rangle &= (E_{FB} - H_0) \vert \psi_2 \rangle , \\
       H_2^\dagger \vert \psi_1 \rangle + H_3 \vert \psi_2 \rangle &= (E_{FB} - H_0) \vert \psi_3 \rangle ,\\
       H_1 \vert \psi_1 \rangle &= 0, \\
       \langle \psi_2 \vert H_1 &= 0, \\
       H_2 \vert \psi_1 \rangle &= 0, \\
       \langle \psi_3 \vert H_2 &= 0, \\
       H_3 \vert \psi_3 \rangle &= 0, \\
       \langle \psi_2 \vert H_3 &= 0, \\
       H_1 \vert \psi_3 \rangle + H_3 \vert \psi_1 \rangle &= 0, \\
       H_2 \vert \psi_2 \rangle + H_3^\dagger \vert \psi_1 \rangle &= 0, \\
       \langle \psi_3 \vert H_1 + \langle \psi_2 \vert H_2 &= 0 \; .
    \end{aligned}
    \label{eq:nnn-triang-eig-prob}
\end{equation} 

We define 
\[
\begin{aligned}H_1^{\dagger}\vert\psi_{3}\rangle & =-H_2^{\dagger}\vert\psi_{2}\rangle=Q_{1}\vert x\rangle,\\
H_2\vert\psi_{2}\rangle & =-H_3^{\dagger}\vert\psi_{1}\rangle=Q_{3}\vert y\rangle,\\
H_1\vert\psi_{3}\rangle & =-H_3\vert\psi_{1}\rangle=Q_{2}\vert z\rangle,\\
H_2\vert\psi_{3}\rangle & =Q_{3}\vert u\rangle,\\
H_3\vert\psi_{2}\rangle & =Q_{2}\vert v\rangle,\\
H_3^{\dagger}\vert\psi_{3}\rangle & =Q_{3}\vert w\rangle.
\end{aligned}
\]
 Then the eigenvalue problem \eqref{eq:nnn-triang-eig-prob} decouples
into inverse eigenvalue problems of $H_1,H_2,H_3$ respectively.
The inverse eigenvalue problem for $H_1$ is 
\begin{equation}
\begin{aligned}H_1\vert\psi_{2}\rangle & =(E_{FB}-H_{0})\vert\psi_{1}\rangle-Q_{3}\vert u\rangle,\\
H_1\vert\psi_{3}\rangle & =Q_{2}\vert z\rangle,\\
\langle\psi_{1}\vert H_1 & =\langle\psi_{2}\vert(E_{FB}-H_{0})-\langle w\vert Q_{3},\\
\langle\psi_{3}\vert H_1 & =\langle x\vert Q_{1},\\
H_1\vert\psi_{1}\rangle & =0,\\
\langle\psi_{2}\vert H_1 & =0,
\end{aligned}
\label{eq:diag-l-shape-hx-eig-prob}
\end{equation}
 and the inverse eigenvalue problem for $H_2$ is 
\begin{equation}
\begin{aligned}H_2\vert\psi_{2}\rangle & =Q_{3}\vert y\rangle,\\
H_2\vert\psi_{3}\rangle & =Q_{3}\vert u\rangle,\\
\langle\psi_{1}\vert H_2 & =\langle\psi_{3}\vert(E_{FB}-H_{0})-\langle v\vert Q_{2},\\
\langle\psi_{2}\vert H_2 & =-\langle x\vert Q_{1},\\
H_2\vert\psi_{1}\rangle & =0,\\
\langle\psi_{3}\vert H_2 & =0,
\end{aligned}
\label{eq:diag-l-shape-hy-eig-prob}
\end{equation}
while the inverse eigenvalue problem for $H_3$ is 
\begin{equation}
\begin{aligned}H_3\vert\psi_{1}\rangle & =-Q_{2}\vert z\rangle,\\
H_3\vert\psi_{2}\rangle & =Q_{2}\vert v\rangle,\\
\langle\psi_{1}\vert H_3 & =-\langle y\vert Q_{3},\\
\langle\psi_{3}\vert H_3 & =\langle w\vert Q_{3},\\
H_3\vert\psi_{3}\rangle & =0,\\
\langle\psi_{2}\vert H_3 & =0.
\end{aligned}
\label{eq:diag-l-shape-hyx-eig-prob}
\end{equation}

By multiplying $\vert\psi_{1}\rangle,\vert\psi_{2}\rangle,\vert\psi_{3}\rangle$
from the left and right to the above inverse eigenvalue problems and comparing
them, we can show that $\vert x\rangle,\vert y\rangle,\vert z\rangle,\vert u\rangle,\vert v\rangle,\vert z\rangle,\vert\psi_{1}\rangle,\vert\psi_{2}\rangle,\vert\psi_{3}\rangle$
should satisfy the following constraints 
\[
\begin{aligned}
\langle\psi_{1}\vert(\lambda-H_{0})\vert\psi_{1}\rangle-\langle\psi_{1}\vert Q_{3}\vert u\rangle&=\langle\psi_{2}\vert(\lambda-H_{0})\vert\psi_{2}\rangle-\langle w\vert Q_{3}\vert\psi_{2}\rangle=\langle\psi_{1}\vert H_{x}\vert\psi_{2}\rangle,\\\langle\psi_{2}\vert(\lambda-H_{0})\vert\psi_{3}\rangle&=\langle\psi_{1}\vert Q_{2}\vert z\rangle=\langle\psi_{1}\vert H_{x}\vert\psi_{3}\rangle,\\\langle\psi_{2}\vert(\lambda-H_{0})\vert\psi_{1}\rangle&=\langle w\vert Q_{3}\vert\psi_{1}\rangle,\\\langle\psi_{2}\vert(\lambda-H_{0})\vert\psi_{1}\rangle&=\langle\psi_{2}\vert Q_{3}\vert u\rangle,\\\langle\psi_{3}\vert(\lambda-H_{0})\vert\psi_{1}\rangle&=\langle x\vert Q_{1}\vert\psi_{2}\rangle=\langle\psi_{3}\vert H_{x}\vert\psi_{2}\rangle,\\\langle x\vert Q_{1}\vert\psi_{3}\rangle&=\langle\psi_{3}\vert Q_{2}\vert z\rangle=\langle\psi_{3}\vert H_{x}\vert\psi_{3}\rangle,\\\langle\psi_{3}\vert(\lambda-H_{0})\vert\psi_{2}\rangle&=\langle\psi_{1}\vert Q_{3}\vert y\rangle=\langle\psi_{1}\vert H_{y}\vert\psi_{2}\rangle,\\\langle\psi_{3}\vert(\lambda-H_{0})\vert\psi_{3}\rangle-\langle v\vert Q_{2}\vert\psi_{3}\rangle&=\langle\psi_{1}\vert Q_{3}\vert u\rangle=\langle\psi_{1}\vert H_{y}\vert\psi_{3}\rangle,\\\langle\psi_{3}\vert(\lambda-H_{0})\vert\psi_{1}\rangle&=\langle v\vert Q_{2}\vert\psi_{1}\rangle,\\\langle\psi_{2}\vert Q_{3}\vert y\rangle&=-\langle x\vert Q_{1}\vert\psi_{2}\rangle=\langle\psi_{2}\vert H_{y}\vert\psi_{2}\rangle,\\\langle\psi_{2}\vert Q_{3}\vert u\rangle&=-\langle x\vert Q_{1}\vert\psi_{3}\rangle=\langle\psi_{2}\vert H_{y}\vert\psi_{3}\rangle,\\\langle y\vert Q_{3}\vert\psi_{1}\rangle&=\langle\psi_{1}\vert Q_{2}\vert z\rangle=-\langle\psi_{1}\vert H_{yx}\vert\psi_{1}\rangle,\\\langle\psi_{1}\vert Q_{2}\vert v\rangle&=-\langle y\vert Q_{3}\vert\psi_{2}\rangle=\langle\psi_{1}\vert H_{yx}\vert\psi_{2}\rangle,\\\langle w\vert Q_{3}\vert\psi_{1}\rangle&=-\langle\psi_{3}\vert Q_{2}\vert z\rangle=\langle\psi_{3}\vert H_{yx}\vert\psi_{1}\rangle,\\\langle\psi_{3}\vert Q_{2}\vert v\rangle&=\langle w\vert Q_{3}\vert\psi_{2}\rangle=\langle\psi_{3}\vert H_{yx}\vert\psi_{2}\rangle.
\end{aligned}
\]

If we assume $H_1$ has the following form 
\[
H_1=\frac{Q_{2}\vert z\rangle\langle x\vert Q_{1}}{\langle\psi_{3}\vert Q_{2}\vert z\rangle}+Q_{23}\vert a\rangle\langle b\vert Q_{13},
\]
 where $\langle\psi_{3}\vert Q_{2}\vert z\rangle=\langle x\vert Q_{1}\vert\psi_{3}\rangle=\langle\psi_{3}\vert H_1\vert\psi_{3}\rangle$
(this is achieved by multiplying $\vert\psi_{3}\rangle$ to the second and
fourth equation of (\ref{eq:diag-l-shape-hx-eig-prob})), then all
equations except for the first and third are satisfied in
the inverse eigenvalue problem (\ref{eq:diag-l-shape-hx-eig-prob}) for
$H_1$. The first and third equations become
\begin{equation}
\begin{aligned}\frac{Q_{2}\vert z\rangle\langle x\vert Q_{1}\vert\psi_{2}\rangle}{\langle\psi_{3}\vert Q_{2}\vert z\rangle}+Q_{23}\vert a\rangle\langle b\vert Q_{13}\vert\psi_{2}\rangle & =Q_{2}\vert U\rangle,\\
\frac{\langle\psi_{1}\vert Q_{2}\vert z\rangle\langle x\vert Q_{1}}{\langle\psi_{3}\vert Q_{2}\vert z\rangle}+\langle\psi_{1}\vert Q_{23}\vert a\rangle\langle b\vert Q_{13} & =\langle W\vert Q_{1},
\end{aligned}
\label{eq:diag-l-shape-hx-eig-reduced}
\end{equation}
 where 
\begin{equation}
\begin{aligned}Q_{2}\vert U\rangle & =H_1\vert\psi_{2}\rangle=(E_{FB}-H_{0})\vert\psi_{1}\rangle-Q_{3}\vert u\rangle,\\
\langle W\vert Q_{1} & =\langle\psi_{1}\vert H_1=\langle\psi_{2}\vert(E_{FB}-H_{0})-\langle w\vert Q_{3}.
\end{aligned}
\label{eq:UW}
\end{equation}
 Solving (\ref{eq:diag-l-shape-hx-eig-reduced}) for $Q_{23}\vert a\rangle,\ \langle b\vert Q_{13}$
we get 
\[
\begin{aligned}Q_{23}\vert a\rangle & =\frac{1}{\langle b\vert Q_{13}\vert\psi_{2}\rangle}\left(\frac{\langle\psi_{3}\vert Q_{2}\vert z\rangle Q_{2}\vert U\rangle-\langle x\vert Q_{1}\vert\psi_{2}\rangle Q_{2}\vert z\rangle}{\langle\psi_{3}\vert Q_{2}\vert z\rangle}\right),\\
\langle b\vert Q_{13} & =\frac{\langle b\vert Q_{13}\vert\psi_{2}\rangle\left(\langle\psi_{3}\vert Q_{2}\vert z\rangle\langle W\vert Q_{1}-\langle\psi_{1}\vert Q_{2}\vert z\rangle\langle x\vert Q_{1}\right)}{\langle\psi_{1}\vert Q_{2}\vert U\rangle\langle\psi_{3}\vert Q_{2}\vert z\rangle-\langle\psi_{1}\vert Q_{2}\vert z\rangle\langle x\vert Q_{1}\vert\psi_{2}\rangle}.
\end{aligned}
\] 
Thus the solution for $H_1$ is 
{\footnotesize
\begin{equation}
    \begin{aligned}
        H_1&=\frac{Q_{2}\vert z\rangle\langle x\vert Q_{1}}{\langle\psi_{3}\vert Q_{2}\vert z\rangle}+\frac{\left(\langle\psi_{3}\vert Q_{2}\vert z\rangle(E_{FB}-H_{0})\vert\psi_{1}\rangle-\langle\psi_{3}\vert Q_{2}\vert z\rangle Q_{3}\vert u\rangle-\langle x\vert Q_{1}\vert\psi_{2}\rangle Q_{2}\vert z\rangle\right)}{\langle\psi_{3}\vert Q_{2}\vert z\rangle\left(\left(\langle\psi_{1}\vert(E_{FB}-H_{0})\vert\psi_{1}\rangle-\langle\psi_{1}\vert Q_{3}\vert u\rangle\right)\langle\psi_{3}\vert Q_{2}\vert z\rangle-\langle\psi_{1}\vert Q_{2}\vert z\rangle\langle x\vert Q_{1}\vert\psi_{2}\rangle\right)} \\
        & \times \left(\langle\psi_{3}\vert Q_{2}\vert z\rangle\langle\psi_{2}\vert(E_{FB}-H_{0})-\langle\psi_{3}\vert Q_{2}\vert z\rangle\langle w\vert Q_{3}-\langle\psi_{1}\vert Q_{2}\vert z\rangle\langle x\vert Q_{1}\right) \; .
    \end{aligned}
\end{equation}
}

Using the same procedure, we get the solution for $H_2,H_3$ as 
{\footnotesize
\begin{equation}
    \begin{aligned}
        H_2&=-\frac{Q_{3}\vert y\rangle\langle x\vert Q_{1}}{\langle\psi_{2}\vert Q_{3}\vert y\rangle}+\frac{\left(\langle\psi_{2}\vert Q_{3}\vert y\rangle Q_{3}\vert u\rangle+\langle x\vert Q_{1}\vert\psi_{3}\rangle Q_{3}\vert y\rangle\right)}{\langle\psi_{2}\vert Q_{3}\vert y\rangle\left(\langle\psi_{2}\vert Q_{3}\vert y\rangle\langle\psi_{1}\vert Q_{3}\vert u\rangle+\langle x\vert Q_{1}\vert\psi_{3}\rangle\langle\psi_{1}\vert Q_{3}\vert y\rangle\right)} \\
        & \times \left(\langle\psi_{2}\vert Q_{3}\vert y\rangle\langle\psi_{3}\vert(E_{FB}-H_{0})-\langle\psi_{2}\vert Q_{3}\vert y\rangle\langle v\vert Q_{2}+\langle\psi_{1}\vert Q_{3}\vert y\rangle\langle x\vert Q_{1}\right), \\
        H_3&=-\frac{Q_{2}\vert z\rangle\langle y\vert Q_{3}}{\langle\psi_{1}\vert Q_{2}\vert z\rangle}+\frac{\left(\langle\psi_{1}\vert Q_{2}\vert z\rangle Q_{2}\vert v\rangle+\langle y\vert Q_{3}\vert\psi_{2}\rangle Q_{2}\vert z\rangle\right)\left(\langle\psi_{1}\vert Q_{2}\vert z\rangle\langle w\vert Q_{3}+\langle\psi_{3}\vert Q_{2}\vert z\rangle\langle y\vert Q_{3}\right)}{\langle\psi_{1}\vert Q_{2}\vert z\rangle\left(\langle\psi_{1}\vert Q_{2}\vert z\rangle\langle\psi_{3}\vert Q_{2}\vert v\rangle+\langle y\vert Q_{3}\vert\psi_{2}\rangle\langle\psi_{3}\vert Q_{2}\vert z\rangle\right)} \; .
    \end{aligned}
\end{equation}
}

\subsection{U=(2,1,0) case} 
\label{app:nnn-u21}

This case is shown in Fig. \ref{fig:u22-nnn-config} (d). Putting the values $U_2=1,\ s=0$ to Eq. \eqref{eq:u22-nnn-eig-prob}, we get the eigenvalue problem and destructive interference conditions in this case as follows 
\begin{equation}
    \begin{aligned}
       H_1 \vert \psi_2 \rangle &= (E_{FB} -H_0) \vert \psi_1 \rangle , \\
       \langle \psi_1 \vert H_1 &= \langle \psi_2 \vert (E_{FB} - H_0) , \\
       H_1 \vert \psi_1 \rangle &=0 , \\
       \langle \psi_2 \vert H_1 &= 0, \\
       H_2 \vert \psi_1 \rangle &= 0, \\
       \langle \psi_2 \vert H_2 &= 0, \\
       H_3 \vert \psi_1 \rangle &= 0 , \\
       \langle \psi_2 \vert H_3 &= 0 , \\
       H_2 \vert \psi_2 \rangle + H_3^\dagger \vert \psi_1 \rangle &= 0, \\
       \langle \psi_1 \vert H_2 + \langle \psi_2 \vert H_3^\dagger &= 0. \\
    \end{aligned}
    \label{eq:u21-nnn-eig-prob}
\end{equation}
Using destructive interference conditions, we eliminate $H_1,H_2,H_3$ to get the following CLS constraints  
\begin{equation}
    \begin{aligned}
       \langle \psi_2 \vert (E_{FB} - H_0) \vert \psi_1 \rangle &= 0 , \\
       \langle \psi_1 \vert (E_{FB} - H_0) \vert \psi_1 \rangle &= \langle \psi_2 \vert (E_{FB} - H_0) \vert \psi_2 \rangle \; .
    \end{aligned}
    \label{eq:u21-nnn-cls-cons}
\end{equation} 
The destructive interference conditions (third to eighth equations in \eqref{eq:u21-nnn-eig-prob}) suggest that
\begin{equation}
    \begin{aligned}
       H_1 &= Q_2 \vert a \rangle \langle b \vert Q_1 , \\ 
       H_2 &= Q_2 \vert c \rangle \langle d \vert Q_1 , \\
       H_3 &= Q_2 \vert e \rangle \langle f \vert Q_1, \\
    \end{aligned}
\end{equation}
where $\vert a \rangle, \vert b \rangle, \vert c \rangle, \vert d \rangle, \vert e \rangle, \vert f \rangle$ are arbitrary vectors. Then, the last two equations in \eqref{eq:u21-nnn-eig-prob} give 
\begin{equation}
    \begin{aligned}
       Q_2 \vert c \rangle \langle d \vert Q_1 \vert \psi_2 \rangle &= - Q_1 \vert f \rangle \langle e \vert Q_2 \vert \psi_1 \rangle , \\ 
       \langle \psi_1 \vert Q_2 \vert c \rangle \langle d \vert Q_1 &= - \langle \psi_2 \vert Q_1 \vert f \rangle \langle e \vert Q_2 \; .
    \end{aligned}
\end{equation} 
This indicates that 
\begin{equation}
    \begin{aligned}
       Q_2 \vert a \rangle & \propto Q_1 \vert b \rangle ,\quad \forall a,b, \\
       Q_1 \vert c \rangle & \propto Q_2 \vert d \rangle, \quad \forall c,d. \\
    \end{aligned}
\end{equation} 
The above condition implies that $\vert a \rangle, \vert b \rangle, \vert c \rangle, \vert d \rangle$ are perpendicular to $\vert \psi_1 \rangle, \vert \psi_2 \rangle$ at the same time. Therefore, we can write $H_2, H_3$ as
\begin{equation}
    \begin{aligned}
       H_2 &= Q_{12} \vert c \rangle \langle d \vert Q_{12} , \\
       H_3 &= Q_{12} \vert e \rangle \langle f \vert Q_{12} \; .
    \end{aligned}
\end{equation} 
Then the last destructive interference conditions become
\begin{equation}
    \begin{aligned}
       H_2 \vert \psi_2 \rangle &= 0, \\
       \langle \psi_1 \vert H_2 &= 0, \\
       H_3 \vert \psi_2 \rangle &= 0, \\
       \langle \psi_1 \vert H_3 &=0. \\
    \end{aligned}
\end{equation} 

For given $\vert \psi_1 \rangle,\ \vert \psi_2 \rangle,\ H_0,\ E_{FB}$ satisfying constraints \eqref{eq:u21-nnn-cls-cons}, the eigenvalue problem \eqref{eq:u21-nnn-eig-prob} then becomes an inverse eigenvalue problem of $H_1$ as
\begin{equation}
    \begin{aligned}
      H_1 \vert \psi_2 \rangle &= (E_{FB} -H_0) \vert \psi_1 \rangle , \\
      \langle \psi_1 \vert H_1 &= \langle \psi_2 \vert (E_{FB} - H_0) \; ,
    \end{aligned}
\end{equation} 
which gives the following solution 
\begin{equation}
      H_1 = \frac{(E_{FB} - H_0) \vert \psi_1 \rangle \langle \psi_2 \vert (E_{FB} - H_0) }{\langle \psi_1 \vert (E_{FB} - H_0) \vert \psi_1 \rangle}.
\end{equation} 

The final solution is 
\begin{equation}
    \begin{aligned}
        H_1 &= \frac{(E_{FB} - H_0) \vert \psi_1 \rangle \langle \psi_2 \vert (E_{FB} - H_0) }{\langle \psi_1 \vert (E_{FB} - H_0) \vert \psi_1 \rangle}, \\
       H_2 &= Q_{12} \vert c \rangle \langle d \vert Q_{12} , \\
       H_3 &= Q_{12} \vert e \rangle \langle f \vert Q_{12}. \\
    \end{aligned}
\end{equation}


\chapter{Supplementary materials for the non-Hermitian flatband generator} 

\section{CLS-based generator} 
\label{app:cls-gen}

This method is based on compact localized states as a direct extension of the method introduced in Ref. \cite{maimaiti2017compact}.  Unlike the Hermitian case, the non-Hermitian FB does not necessarily host CLSs \cite{maimaiti2017compact,maimaiti2019universal}. In the case that the non-Hermitian FB does host a CLS of class $U$ \cite{maimaiti2017compact,maimaiti2019universal}, then the band is fully flat, i.e. both real and imaginary parts are $k$ independent. More precisely, the CLS $\Psi_{CLS}=(\vec{\psi}_1,\vec{\psi}_2,\dots,\vec{\psi}_U)$ is the eigenvector of the $U \times U$ tri-diagonal block matrix
\begin{equation}
    \mathcal{H}_U = \left(\begin{array}{cccccc}
        H_0 & H_l & 0 & 0 & \dots & 0\\
        H_r & H_0 & H_l & 0 & \dots & 0\\
        0 & \ddots & \ddots & \ddots & \ddots & \vdots\\
        \vdots & ~ & ~ & ~ & ~ & \vdots\\
        0 & \dots & 0 & H_r & H_0 & H_l\\
        0 & \dots & 0 & 0 & H_r & H_0
    \end{array}\right), \; 
    \label{eq:cls-def-HU}
\end{equation}
with eigenenergy $E_{FB}=E_1 + i E_2,\ E_1,E_2 \in \mathcal{R}$. The following destructive interference condition must be satisfied 
\begin{equation}
    H_l \vec{\psi}_1 = H_r \vec{\psi}_U = 0.
    \label{eq:dest-intf-cond}
\end{equation}
Therefore, a \emph{necessary} condition for the existence of a non-Hermitian CLS reads
\begin{equation}
    \det H_l = \det H_r = 0.
    \label{eq:necessary-cond-cls}
\end{equation} 
We rewrite the CLS problem~(\ref{eq:cls-def-HU}--\ref{eq:dest-intf-cond}) as 
\begin{eqnarray}
    H_l \vec{\psi}_2 &=& (E_{FB} - H_0)\vec{\psi}_1 \label{eig-1-app}, \\
    H_r \vec{\psi}_{j-1} + H_l \vec{\psi}_{j+1} &=& (E_{FB} - H_0)\vec{\psi}_j , 2 \le j \le U-1
    \label{eig-2-app},\\
    H_r \vec{\psi}_{U-1} &=& (E_{FB} - H_0)\vec{\psi}_U
     \label{eig-3-app}, \\
    H_l \vec{\psi}_1 &=& H_r \vec{\psi}_U = 0
    \label{eig-4-app},\\
    \vec{\psi}_j &=& 0 \;,\;  j<0,\,j>U. \;
    \label{eig-5-app}
\end{eqnarray}
 The non-Hermitian FB generator is the set of all possible matrices $H_0,H_r,H_l$ that satisfy Eqs. (\ref{eig-1-app}--\ref{eig-5-app}).

 We consider a two-band problem, i.e. $\nu=2$ sites per unit cell. In this case, using the same argument as the previous section, $H_0$ can take the form as given in \eqref{eq:def-H0}. With the proper choice of CLS $\Psi_{CLS}$ and FB energy $E_{FB}$, solving Eqs. (\ref{eig-1-app}--\ref{eig-5-app}) becomes an inverse eigenvalue of finding $H_r,H_l$.

\subsection{$U=1$ case}

In this case, the eigenvalue problem (\ref{eig-1-app}--\ref{eig-5-app}) becomes
\begin{equation}
 \begin{aligned}
  & H_0 \vec{\psi} = \lambda \vec{\psi},\\
  & H_r \vec{\psi} = 0,\\
  & H_l \vec{\psi} = 0.
 \end{aligned}
 \label{eq:cls-eq-u1}
\end{equation} 
Suppose $\vec{\psi}=(x,y) \ne 0$, then \eqref{eq:cls-eq-u1} becomes 
\begin{equation}
\begin{aligned}
    \nu y & = \lambda x, \\
    \mu y & = \lambda y, \\
    a x + b y & = \lambda x, \\
    c x + d y & = \lambda y, \\
    f x + g y & = \lambda x, \\
    h x + l y & = \lambda y. \; 
\end{aligned}
\label{eq:cls-eq-u1-1}
\end{equation}
We solve \eqref{eq:cls-eq-u1} for different forms of $H_0$. 

\paragraph{Degenerate $H_0$} 

In this case $\mu=\nu=0$ and \eqref{eq:cls-eq-u1-1} yields the following solution:
\begin{equation*}
 \begin{aligned}
  & \lambda = 0,\\
  & a  = -\frac{y}{x} b,\\
  & c  = -\frac{y}{x} d,\\
  & f  =-\frac{y}{x} g,\\
  & h  =-\frac{y}{x} l.\\
 \end{aligned}
\end{equation*} 
Thus the hopping matrices, CLS, and FB energy are
 \begin{gather*}
   H_{0}=\begin{pmatrix}0 & 0\\
                        0 & 0
         \end{pmatrix},\\
   H_{l}=\begin{pmatrix}-\frac{y}{x} g & g\\
                        -\frac{y}{x} l & l
         \end{pmatrix},\\
   H_{r}=\begin{pmatrix}-\frac{y}{x} b & b\\
                        -\frac{y}{x} d & d
         \end{pmatrix},\\
   \vec{\psi}=(x,y),\\
   \lambda = 0.
\end{gather*}

\paragraph{Abnormal $H_0$} 
In this case $\mu=0,\ \nu=1$ and \eqref{eq:cls-eq-u1-1} gives 
\begin{equation*}
 \begin{aligned}
  & y = 0,\\
  & a  = 0,\\
  & c  = 0,\\
  & f  =0,\\
  & h  =0, \; 
 \end{aligned}
\end{equation*}
which gives the following hopping matrices, CLS, and FB energy:
 \begin{gather*}
   H_{0}=\begin{pmatrix}0 & 1\\
                        0 & 0
         \end{pmatrix},\\
   H_{l}=\begin{pmatrix}0 & g\\
                        0 & l
         \end{pmatrix},\\
   H_{r}=\begin{pmatrix} 0 & b\\
                         0 & d
         \end{pmatrix},\\
   \vec{\psi}={x,0},\\
   \lambda =0.
\end{gather*} 

\paragraph{Non-degenerate $H_0$}
In this case $\mu=1,\ \nu=0$ and \eqref{eq:cls-eq-u1} has the following solution:
 \begin{gather*}
   H_{0}=\begin{pmatrix}0 & 0\\
                        0 & 1
         \end{pmatrix},\\
   H_{l}=\begin{pmatrix}f & 0\\
                        g & 0
         \end{pmatrix},\\
   H_{r}=\begin{pmatrix} a & 0\\
                         c & 0
         \end{pmatrix},\\
   \vec{\psi}=(0,1), \\
   \lambda = 1 \; .
\end{gather*} 
Here, $y$ is normalized to be 1.

\subsection{U=2 case}
\label{app:cls-gen-u2}

In this case the eigenvalue problem (\ref{eig-1-app}--\ref{eig-5-app}) becomes
\begin{equation}
 \begin{aligned}
  & H_0 \psi_1 + H_l \psi_2 = \lambda \psi_1, \\
  & H_0 \psi_2 + H_r \psi_1 = \lambda \psi_2, \\
  & H_l \psi_1 = 0,\ \ H_r \psi_2 =0.
 \end{aligned}
 \label{u2_eigv_prb}
\end{equation} 
We can parameterize $H_r, H_l$ in the following way
\begin{equation}
    H_r = \begin{pmatrix} a & b \\
    c & \frac{b c}{a} \end{pmatrix}, \quad 
    H_l = \begin{pmatrix} f & g \\
    h & \frac{g h}{f} \end{pmatrix},
    \label{Hr-Hl-dif-para}
\end{equation} 
 which satisfies the destructive interference conditions by definition. Then we can choose $\psi_1, \psi_2$ to be zero eigenvectors of $H_l, H_r$, respectively, as
 \begin{equation}
     \psi_1 = \begin{pmatrix} -\frac{g}{f} \\ 1 \end{pmatrix}, \quad \psi_2 = \begin{pmatrix} -\frac{b}{a} \\ 1 \end{pmatrix}.
 \end{equation} 
 Then \eqref{u2_eigv_prb} becomes 
 \begin{equation}
     \begin{aligned}
         \begin{pmatrix} -\frac{b f}{a}+g+\nu \\ -\frac{b h}{a}+\frac{g h}{f}+\mu \end{pmatrix} &= \begin{pmatrix} -\frac{g \lambda }{f} \\ \lambda \end{pmatrix}, \\ 
         \begin{pmatrix} -\frac{a g}{f}+b+\nu  \\ \frac{b c}{a}-\frac{c g}{f}+\mu \end{pmatrix} &= \begin{pmatrix} -\frac{b \lambda }{a} \\ \lambda \end{pmatrix}.
     \end{aligned}
     \label{u2_eigv_prob_1}
 \end{equation} 
 Solving the above we get 
 \begin{equation}
     \begin{aligned}
         c &= \frac{a^2 \mu -a^2 \lambda }{a \nu +b \lambda },\\ 
         f &= -a-\lambda, \\ 
         g &= -\frac{(a+\lambda ) (a b+a \nu +b \lambda )}{a^2},\\ 
         h &= \frac{a^2 (\lambda -\mu )}{a \nu +b \lambda }.
     \end{aligned}
     \label{u2_sol}
 \end{equation} 
 Putting the corresponding values of $\mu,\nu$ in \eqref{u2_sol} we get solutions for degenerate, non-degenerate, and abnormal cases. The band structure for solution \eqref{u2_sol} is
 \begin{equation*}
     \begin{aligned}
         E_{FB} &= \lambda, \\ 
         E_k &= \frac{a e^{i k} (a \nu +b \mu )}{a \nu +b \lambda }+\frac{e^{-i k} (-b \mu  (a+\lambda )-a \nu  (a+\mu ))}{a \nu +b \lambda }-\lambda +\mu.
     \end{aligned}
 \end{equation*}
 
 \subsection{U=3 case} 
 \label{app:cls-gen-u3}
 
 In this case we have 
\begin{equation}
 \begin{aligned}
  & H_0 \psi_1 + H_l \psi_2 = \lambda \psi_1, \\
  & H_0 \psi_2 + H_r \psi_1 + H_l \psi_3 = \lambda \psi_2, \\
  & H_0 \psi_3 + H_r \psi_2 = \lambda \psi_3, \\
  & H_l \psi_1 = 0,\ \ H_r \psi_3 =0.
 \end{aligned}
 \label{u3_eigv_prb}
\end{equation} 
We use the same parameterization as in \eqref{Hr-Hl-dif-para} and choose the CLS as 
 \begin{equation}
     \psi_1 = \begin{pmatrix} -\frac{g}{f} \\ 1 \end{pmatrix}, \quad \psi_2 = \begin{pmatrix} \alpha \\ \beta \end{pmatrix}, \quad \psi_3 = \begin{pmatrix} -\frac{b}{a} \\ 1 \end{pmatrix},
 \end{equation} 
which satisfies the destructive interference conditions by definition. Then \eqref{u3_eigv_prb} becomes
\begin{equation}
  \begin{aligned}
     \begin{pmatrix} \alpha  f+\beta  g+\nu \\ \frac{\beta  g h}{f}+\alpha  h+\mu \end{pmatrix} &= \begin{pmatrix} -\frac{g \lambda }{f} \\ \lambda \end{pmatrix}, \\
     \begin{pmatrix} -\frac{b f}{a}-\frac{a g}{f}+\beta  \nu +b+g \\ \frac{a (-c g+\beta  f \mu +g h)+b f (c-h)}{a f}\end{pmatrix} &= \begin{pmatrix}  \lambda \\ \beta  \lambda \end{pmatrix}, \\
     \begin{pmatrix} a \alpha +b \beta +\nu \\  \frac{b \beta  c}{a}+\alpha  c+\mu \end{pmatrix} &= \begin{pmatrix} -\frac{b \lambda }{a} \\ \lambda \end{pmatrix}.
 \end{aligned}
 \label{u3_eigv_prob_1}
\end{equation}
Solving the above equation we have 
 \begin{equation}
     \begin{aligned}
         b &= \frac{\left(-\lambda \pm \sqrt{\lambda ^2-4 a f}\right) (f \nu +g \lambda )+2 a f g}{2 f^2},\\ 
         c &= \frac{(\lambda -\mu ) \left(\lambda  \left(\lambda \pm \sqrt{\lambda ^2-4 a f}\right)-2 a f\right)}{2 (f \nu +g \lambda )}, \\ 
         h &= \frac{f^2 (\mu -\lambda )}{f \nu +g \lambda },\\ 
         \alpha &= \frac{g \lambda  (f-a)-\left(2 a f \nu \pm g (a-f) \sqrt{\lambda ^2-4 a f}\right)}{2 a f^2},\\ 
         \beta &= -\frac{\lambda  (a+f)\pm (a-f) \sqrt{\lambda ^2-4 a f}}{2 a f}.
     \end{aligned}
     \label{u3_sol}
 \end{equation} 
 This solution \eqref{u3_sol} is the same as the solution \eqref{eq:single-flat-gen-sol-1} with the band calculation method. By inserting corresponding values of $\mu,\nu$, we can get solutions for different cases. 
 
 When $a=-f-\lambda$, the solution \eqref{u3_sol} for the $U=3$ case reduces to solution \eqref{u2_sol} for the $U=2$ case.

\subsection{Inverse eigenvalue method for CLS approach}
\label{app:nh-inv-eig-method}

We can write the inverse eigenvalue problem for the $U=2$ case as
\begin{equation}
 \begin{aligned}
  & H_r \vert \psi_2 \rangle  = \left( \lambda - H_0 \right \vert  \psi_1 \rangle, \\
  & H_l \vert \psi_1 \rangle = \left( \lambda - H_0 \right) \vert \psi_2 \rangle, \\
  & H_r \vert \psi_1 \rangle = 0,\ \ H_l \vert \psi_2 \rangle  = 0,
 \end{aligned}
\end{equation}
with solution
\begin{equation}
 \begin{aligned}
  & H_r = \frac{\left( \lambda - H_0 \right) \vert \psi_1 \rangle \langle \psi_2 \vert}{\langle \psi_2 \vert Q_1 \vert \psi_2 \rangle} Q_1,\\
  & H_l = \frac{\left( \lambda - H_0 \right) \vert \psi_2 \rangle \langle \psi_1 \vert}{\langle \psi_1 \vert Q_2 \vert \psi_1 \rangle} Q_2.
 \end{aligned}
\end{equation} 
If we use the following ansatz
\begin{equation}
  \psi_1 = \vert \theta \rangle = \begin{pmatrix} \cos \theta \\ \sin \theta \end{pmatrix},\ \psi_2 = \alpha \vert \varphi \rangle = \alpha \begin{pmatrix} \cos \varphi \\ \sin \varphi \end{pmatrix} \; ,
\end{equation}
 the solution is given for parameters $\theta,\ \varphi,\ \alpha,\ \lambda$, and the FB energy is $E_{FB}=\lambda$.

\section{Solving completely flat bands} 

A completely flat band has $k$-independent real and imaginary parts. Our starting point solving this case is Eq. \eqref{eq:band-eqn}, which is 
\begin{equation}
\begin{aligned}
      x_k + y_k & =\mu+e^{ik}(a+d)+e^{-ik}(f+l),\\
      x_k y_k & =e^{2ik}\det H_{r}+e^{-2ik}\det H_{l} \\ & + (\nu f-            
      \mu h)e^{ik} + (\nu a-\mu c)e^{-ik} \\ & +df-cg-bh+al. \; 
\end{aligned}
\label{eq:band-eqn-app} 
\end{equation} 
 We assume one of $x_k, y_k$ or both are $k$ independent and solve \eqref{eq:band-eqn-app} to find the FB Hamiltonian.

\subsection{Both bands are completely flat}
\label{app:both-fb}

In this case, both $x_k, y_k$ in \eqref{eq:band-eqn-app} are $k$ independent. We assume $x_{k}=x$ and $y_{k}=y$, then \eqref{eq:band-eqn-app} becomes
\begin{equation}
  \begin{aligned}
      x+y & =\mu+e^{ik}(a+d)+e^{-ik}(f+l),\\
      xy & =e^{2ik}\det H_{r}+e^{-2ik}\det H_{l}+(\nu f-\mu h)e^{ik} \\ & +(\nu a - \mu c)e^{-ik} + df - cg - bh + al.
  \end{aligned}
\label{eq:bands-eqn-allf-xy}    
\end{equation} 
Requiring the polynomial of $e^{ik}$ to vanish
gives the following equations:
\begin{equation}
\begin{aligned}
&a+d=0,\\
&f+l=0,\\
&\det\,H_{r}=ad-bc=0,\\
&\det\,H_{l}=fl-hg=0,\\
&\nu f-\mu h=0,\\
&\nu a-\mu c=0,\\
&xy=df-cg-bh+al,\\
&x+y= \mu .
\end{aligned}
\label{eq:bands-eqn-allf-xy-eqs-app}
\end{equation} 
From these it follows $d=-a$, $l=-f$, and either $f=a=0$ or $h=c=0$,
or none.

Solving \eqref{eq:bands-eqn-allf-xy-eqs-app} for degenerate, non-degenerate, and abnormal cases separately, we can get $H_l,H_r$ that gives both bands as completely flat.

\paragraph{Degenerate case:}

Here we have $\mu=\nu=0$, and $bc+a^{2}=0$, $hg+f^{2}=0$, and $y=-x$. Therefore, $c=-a^{2}/b$, $h=-f^{2}/g$, and
\begin{gather}
H_{r}=\begin{pmatrix}a & b\\
-\frac{a^{2}}{b} & -a
\end{pmatrix},\label{eq:bands-eqn-allf-xy-Hlr-sol-zero}\\
H_{l}=\begin{pmatrix}f & g\\
-\frac{f^{2}}{g} & -f
\end{pmatrix},\\
x^{2}=2af-\frac{a^{2}g}{b}-\frac{bf^{2}}{g}.
\end{gather}

\paragraph{Non-degenerate $H_{0}$}

In this case, $\mu=1,\ \nu=0$ and $h=c=0$, $ad=fl=0$, and $y=1-x$.
Therefore, $a=d=f=l=0$ and either $b$ or $c$ are zero and either
$h$ or $g$ are zero, giving
\begin{gather}
H_{r}=\begin{pmatrix}0 & b\\
0 & 0
\end{pmatrix}\quad\begin{pmatrix}0 & 0\\
c & 0
\end{pmatrix},\label{eq:bands-allf-xy-Hlr-sol-ndeg}\\
H_{l}=\begin{pmatrix}0 & g\\
0 & 0
\end{pmatrix}\quad\begin{pmatrix}0 & 0\\
h & 0
\end{pmatrix},\\
x(1-x)=-cg-bh.
\end{gather}
There are four possible solutions of $bc=0,\ hg=0$. Interestingly,
two out of the four imply $x=0,\ y=1$ or $x=1,\ y=0$, that are very
likely $U=1$ class.

\paragraph{Abnormal $H_{0}$}

Here, $\mu=0,\ \nu=1$, $f=a=0$, and $y=-x$. Thus
$d=l=0$ and $bc=hg=0$, and
\begin{gather}
H_{r}=\begin{pmatrix}0 & b\\
0 & 0
\end{pmatrix}\quad\begin{pmatrix}0 & 0\\
c & 0
\end{pmatrix},\label{eq:bands-eqn-allf-xy-Hlr-sol-an}\\
H_{l}=\begin{pmatrix}0 & g\\
0 & 0
\end{pmatrix}\quad\begin{pmatrix}0 & 0\\
h & 0
\end{pmatrix},\\
x^{2}=cg+bh.
\end{gather} 

\subsection{One band is completely flat}
\label{app:one-band-fb}

Requiring either $x$ or $y$ in \eqref{eq:band-eqn-app} will yield $\det H_r =0,\ \det H_l = 0$, and therefore we can parameterize $H_r,H_l$ as
\begin{equation}
    H_r = \begin{pmatrix} a & b\\
                          c & \frac{b c}{a} \end{pmatrix}, 
    H_r = \begin{pmatrix} f & g\\
                          h & \frac{g h}{f} \end{pmatrix}, 
\end{equation} 
which makes $H_r,H_l$ to be singular by definition. We assume that only $x_{k}=x$ is flat, so then \eqref{eq:band-eqn-app} becomes
\begin{equation}
  \begin{aligned}
      x+y_k &=\frac{b c e^{i k}}{a}+a e^{i k}+e^{-i k} \left(\frac{g h}{f}+f\right)+\mu),\\
      x y_k & =\frac{(a g-b f) (a h-c f)}{a f} +e^{i k} (a \mu -c \nu )\\ &+e^{-i k} (f \mu -h \nu ).
  \end{aligned}
\label{eq:bands-eqn-allf-xy-1}    
\end{equation} 
This results in the following equations:
\begin{equation}
  \begin{aligned}
      y_{k} & =\frac{b c e^{i k}}{a}+a e^{i k}+e^{-i k} \left(\frac{g h}{f}+f\right)+\mu -x,\\
      y_{k} & =\frac{e^{i k} (a \mu -c \nu )+e^{-i k} (f \mu -h \nu )}{x}, \\
      & + \frac{(a g-b f) (a h-c f)}{a f x}.  \;
  \end{aligned}
\label{eq:bands-eqn-of-generic-2}    
\end{equation}
Consequently, equating powers of $e^{ik}$, we find 
\begin{equation}
\begin{aligned}
  & \frac{b c}{a}+a=\frac{a \mu -c \nu }{x},\\  & \frac{g h}{f}+f=\frac{f \mu -h \nu }{x}\\ & \mu -x=\frac{(a g-b f) (a h-c f)}{x (a f)}
\end{aligned}
\label{eq:bands-of-generic-3}
\end{equation}

Solving \eqref{eq:bands-of-generic-3} gives
\begin{equation}
 \begin{aligned}
 b & =\frac{\left(-x\pm \sqrt{x^2-4 a f}\right) (f \nu +g x)+2 a f g}{2 f^2},\\ 
 c &=\frac{(x-\mu ) \left(\left(x^2\pm x \sqrt{x^2-4 a f}\right)-2 a f\right)}{2 (f \nu +g x)},\\
  h & =\frac{f^2 (\mu -x)}{f \nu +g x}.
    \end{aligned} 
    \label{eq:single-flat-gen-sol-1}
\end{equation}  
Then the band structure is 
\begin{equation}
\begin{aligned}
    E_{FB}&= x,\\ 
    E_k &= \frac{2 (f \nu +g \mu ) \left(-(x-\mu )+\left(a e^{i k}+f e^{-i k}\right)\right)}{2 (f \nu +g x)} \\ 
    & \mp \frac{ e^{i k} \nu  (x-\mu ) \left(\sqrt{x^2-4 a f}-x\right)}{2 (f \nu +g x)} \; .
\end{aligned}
\label{eq:single-flat-gen-band}
\end{equation}

Putting corresponding values of $\mu,\nu$ into \eqref{eq:single-flat-gen-sol-1} and \eqref{eq:single-flat-gen-band}, we can get solutions for degenerate, non-degenerate, and abnormal cases with corresponding band structures.

\section{Solving partially flat bands}
\label{app:partial-fb} 

We assume that $x_k=x_1 + i x_2,\ y=y_1 + i y_2$, and $x_1,x_2,y_1,y_2 \in \mathcal{R}$, so then expanding \eqref{eq:band-eqn-app} yields
  \begin{equation}
    \begin{aligned} 
      x_k + y_k &= x_{1} +y_{1} +i\left( x_{2} + y_{2} \right) = \mu \\ &+\cos(k)(a+d+f+l)\\ &-   
      i\sin(k)(a+d-f-l),\\ 
      x_k y_k &=x_{1}y_{1}-
      x_{2}y_{2}+i\left(x_{2}y_{1}+x_{1}y_{2}\right) \\ &=al-bh-cg+df\\&+(a\mu-
      c\nu+f\mu-h\nu)\cos(k)\\ &+(\det H_{l}+\det H_{r})\cos(2k)\\ &+i(\left(-
      a\mu+c\nu+f\mu-h\nu\right)\sin(k)\\ & +(\det H_{l}-\det H_{r})\sin(2k)). \;
    \end{aligned}
    \label{eq:re-im-gen-eq}
  \end{equation} 
Equating real and imaginary parts of \eqref{eq:re-im-gen-eq} gives
   \begin{eqnarray}
    x_{1}+y_{1}&=&\mu+(a+d+f+l)\cos(k), \label{eq:complex-exp-band-eq-1}\\
    x_{2}+y_{2}&=&-(a+d-f-
    l)\sin(k),\label{eq:complex-exp-band-eq-1-1}\\
    x_{1}y_{1}-x_{2}y_{2}&=&al-bh-cg+df \notag \\   
    && +(a\mu-c\nu+f\mu-h\nu)\cos(k) \label{eq:complex-exp-band-eq-1-2} \\ 
    &&+(\det H_{l}+\det H_{r})\cos(2k),  \notag \\
    x_{2}y_{1}+x_{1}y_{2}&=&\left(-a\mu+c\nu+f\mu-
    h\nu\right)\sin(k) \label{eq:complex-exp-band-eq-2} \\ 
    &&  +(\det H_{l}-\det H_{r})\sin(2k).  \notag \; 
  \end{eqnarray} 
Solving Eqs. (\ref{eq:complex-exp-band-eq-1}--\ref{eq:complex-exp-band-eq-2}) under the condition that some of $x_1,\ x_2,\ y_1,\ y_2$ are $k$ independent, we get the solution for partially flat bands.

\subsection{Real parts of both bands are flat}
\label{app:both-re-fb}

In this case, $x_1,\ y_1$ are $k$ independent, so Eqs. (\ref{eq:complex-exp-band-eq-1}--\ref{eq:complex-exp-band-eq-1-1}) give
\begin{equation}
    \begin{aligned}
        y_1 &= \mu- x_1 + (a+d+f+l)\cos(k),\\
        y_2 &= -x_2-(a+d-f-l)\sin(k).\; 
    \end{aligned}
    \label{eq:both-re-fb-y1-y2}
\end{equation}
We put this equation into (\ref{eq:complex-exp-band-eq-1-2}--\ref{eq:complex-exp-band-eq-2}) and solve for $x_2$. Then, requiring $x_1$ to be $k$ independent by setting the coefficients of $k$-independent terms to zero, we get the following equations: 
\begin{equation}
    \begin{aligned}
       a+d+f+l&= 0, \\ 
       a  -c +f  -h & =0, \\ 
       a d+f l -b c-g h & =\frac{\left(X + Y \right) \left( X + Z \right) }{2 (\mu -2 x_1)^2},\\
       -a d+b c+f l-g h & =0,  \; 
    \end{aligned}
    \label{eq:both_real_fb}
\end{equation}
where $X=x_1 (-a-d+f+l),\ Y=a \mu -c \nu -f \mu +h \nu,\ Z=c \nu +d \mu -h \nu -l \mu$. 

Solving \eqref{eq:both_real_fb} for degenerate, non-degenerate, and abnormal cases separately, we get the $H_r,H_l$ that gives the real parts of both bands as flat.  

\paragraph{Degenerate case} 

In this case $\mu=\nu=0$ and \eqref{eq:both_real_fb} becomes
\begin{equation}
    \begin{aligned}
        a+d+f+l &=0, \\
        \frac{1}{8} (a+d-f-l)^2+b c+g h &=a d+f l, \\
        -a d+b c+f l-g h &=0,\\
        \frac{1}{8} (a+d-f-l)^2+a l+d f+x_1^2 &=b h+c g.
    \end{aligned}
    \label{eq:both-real-degen-eqs}
\end{equation}
If we consider $x_1$ as a parameter, then the solution for \eqref{eq:both-real-degen-eqs} is 
  \begin{align*}
      H_{0}&=\begin{pmatrix}0 & 0\\
                        0 & 0
                \end{pmatrix},\\
         H_{l}&=\begin{pmatrix} -f & \frac{\pm 2 \sqrt{A} + B}{(a-d)^2}\\
                        -\frac{\pm 2 \sqrt{A} + B}{4 b^2} & -a-d-f
               \end{pmatrix},\\
         H_{r}&=\begin{pmatrix} a & b\\
                -\frac{(a-d)^2}{4 b} & d
               \end{pmatrix},\; 
  \end{align*}
where $A=b^2 x_1^2 \left((d-a) (a+d+2                f)+x_1^2\right),\ B=b (a-d) (a+d+2 f)-2 b x_1^2$. 
Then the band structure is 
  \begin{align}
     x_k = -x_1+i (a+d) \sin (k),\\
     y_k = x_1+i (a+d) \sin (k).
  \end{align}

On the other hand, if we consider $x_1$ as a function of $H_r,H_l$, then the solution for \eqref{eq:both-real-degen-eqs} is
  \begin{align}
      H_{0}=\begin{pmatrix}0 & 0\\
                        0 & 0
                \end{pmatrix},\\
      H_l = \begin{pmatrix} f & -\frac{(a+d+2 f)^2}{4 h} \\
                            h & -a-d-f 
            \end{pmatrix},\\
      H_r = \begin{pmatrix}a & b \\
        -\frac{(a-d)^2}{4 b} & d \\
            \end{pmatrix},\\
     x_1 = \pm \frac{(a-d) (a+d+2 f)+4 b h}{4 \sqrt{b} \sqrt{h}},\; 
  \end{align}
and the band structure is 
  \begin{align}
      x_k &= -\frac{\sqrt{b h ((a-d) (a+d+2 f)+4 b h)^2}}{4 b h} \notag \\ & +\frac{4 i b h (a+d) \sin (k)}{4 b h} ,\\
      y_k &= \frac{\sqrt{b h ((a-d) (a+d+2 f)+4 b h)^2}}{4 b h} \notag \\ & + \frac{4 i b h (a+d) \sin (k)}{4 b h}.
  \end{align}

\paragraph{Non-degenerate case} 

In this case $\mu=1,\ \nu=0$ and \eqref{eq:both_real_fb} becomes 
        \begin{align}
        & a+d+f+l= 0, \\ 
        & a + f =0, \\ 
        & \frac{\left(X+a-f\right) \left(X+d-l\right)}{2 \left(1-2 x_1\right){}^2}+b c+g h=a d+f l\\ 
        & -a d+b c+f l-g h = 0,\\
        & \frac{ x_1\left(1- x_1\right)(a-d-f+l)^2}{2 \left(1-2 x_1\right)^2}=-(b h+c g+x_1) \\ 
        & +a l+d f+x_1^2 \notag \; ,
      \end{align} 
where $X=x_1 (-a-d+f+l)$. Then the solution is 
  \begin{equation}
      \begin{aligned}
         H_{0}&=\begin{pmatrix}0 & 0\\
                        0 & 1
                \end{pmatrix},\\
         H_{r}&=\begin{pmatrix} -f & -\frac{\left(x_1-1\right) x_1 (f-l)^2}{c  
                           \left(1-2 x_1\right){}^2}\\
                        c & -l
               \end{pmatrix},\\
         H_{l}&=\begin{pmatrix}f & \frac{\left(x_1-1\right) x_1 (C+D)}{2 c \left(1-2 x_1\right){}^2}\\
                        \frac{c (C-D)}{2 (f-l)^2} & l
               \end{pmatrix},
      \end{aligned}
  \end{equation}
where $C=\sqrt{4(f-l)^2+(1-2x_1)^2} \left(1-2 x_1\right),\ D=2 (f-l)^2 + (1 - 2x_1)^2$.

Then the band structure is 
\begin{equation}
    \begin{aligned}
        x_k & = x_1-\frac{2 i \sin (k) \left(x_1 (f+l)-f\right)}{2 x_1-1},\\
        y_k & = -\frac{2 i \sin (k) \left(x_1 (f+l)-l\right)}{2 x_1-1}-x_1+1 \; .
    \end{aligned}
\end{equation}
Obviously, the real parts $x_1,1-x_1$ are $k$ independent, i.e flat.

\paragraph{Abnormal case}

In this case $\mu=0,\ \nu=1$ and  \eqref{eq:both_real_fb} becomes     
\begin{align}
        & a+d+f+l =0, \\
        & -c - h = 0, \\
        & \frac{\left(X+c-h\right) \left(X-c+h\right)}{8 x_1^2}  = a d + f l - b c - g h\\
        & -a d+b c+f l-g h  = 0, \\
        & b h+c g+\frac{(c-h)^2}{8 x_1^2} =\frac{1}{8} (a+d-f-l)^2 \\ & +a l+d f+x_1^2 \notag \; ,
    \end{align}  
where $X=x_1 (-a-d+f+l)$. Then the solution is 
  \begin{equation}
      \begin{aligned}
         H_{0}&=\begin{pmatrix}0 & 1\\
                        0 & 0
                \end{pmatrix},\\
         H_{l}&=\begin{pmatrix} f & \frac{h^2-x_1^2 \left(d+f-x_1\right){}^2}{4 h x_1^2}\\
                        h & x_1-d
               \end{pmatrix},\\
         H_{r}&=\begin{pmatrix}-f-x_1 & \frac{x_1^2 \left(d+f+x_1\right){}^2-h^2}{4 h x_1^2}\\
                        -h & d
               \end{pmatrix}.\;
      \end{aligned}
  \end{equation} 
The band structure is 
\begin{equation}
    \begin{aligned}
        x_k & = -x_1+i \sin (k) \left(d-f+\frac{h}{x_1}-x_1\right),\\
        y_k & = x_1-i \sin (k) \frac{ \left(x_1 \left(-d+f+x_1\right)+h\right)}{x_1} \; .
    \end{aligned}
\end{equation} 

\subsection{Real part of one band is flat} 

In this case, either $x_1$ or $y_1$ is $k$ independent in Eqs. \eqref{eq:complex-exp-band-eq-1}--\eqref{eq:complex-exp-band-eq-2}. If we assume $x_1$ is $k$ independent and solve equations \eqref{eq:complex-exp-band-eq-1}--\eqref{eq:complex-exp-band-eq-2} for $x_1$, then according to \eqref{eq:complex-exp-band-eq-1} we have 
\begin{equation}
    \begin{aligned}
        x_1 &= \mu, \\
        y_1 &= (a+d+f+l)Cos (k).
    \end{aligned}
\end{equation}
(In a similar way we can assume $y_1$ is $k$ independent, and then $y_1=\mu,\ x_1=(a+d+f+l)Cos (k)$.) Equations (\ref{eq:complex-exp-band-eq-1}--\ref{eq:complex-exp-band-eq-2}) become 
\begin{equation}
    \begin{aligned}
        x_2+y_2 &= \sin (k) (a+d-f-l), \\ 
        x_2 y_2 &= -a l+b h+\cos (k) (\nu  (c+h)+d \mu +l \mu ) \\ 
        & +c g-d f-(\det H_l + \det H_r) \cos (2 k), \\
        \mu  y_2 &= \sin (k) (a \mu -c \nu -f \mu +h \nu ) \\ 
        & -x_2 \cos (k) (a+d+f+l), \\ 
        & +(\det H_r - \det H_l) \sin (2 k).
    \end{aligned}
    \label{eq:band-eq-single-re-flat}
\end{equation}
For convenience, we make the following replacement of variables: 
\begin{equation}
    \begin{aligned}
        \det H_r &= a d-b c,\\
        \det H_l &= f l-g h,\\
        V &= a+d+f+l,\\
        X &= a+d-f-l,\\
        Y &= a \mu -c \nu +f \mu -h \nu,\\
        Z &= a \mu -c \nu -f \mu +h \nu, \\
        W &= a l-b h-c g+d f.
    \end{aligned}
    \label{eq:rep-var}
\end{equation}
Equation \eqref{eq:band-eq-single-re-flat} then becomes 
    \begin{align}
        x_1+y_1 &=V \cos (k)+\mu, \label{eq:single-re-fb-1}\\ 
        x_2+y_2 &=X \sin (k), \label{eq:single-re-fb-2}\\ 
        x_1 y_1-x_2 y_2 &= (\det H_l+\det H_r) \cos (2 k) \label{eq:single-re-fb-3 }\\ 
        & + Y \cos (k) +W, \notag\\
        x_2 y_1+x_1 y_2 &=-(\det H_l-\det H_r)\sin (2k). \label{eq:single-re-fb-4} \\ 
        & + Z \sin (k) \notag
    \end{align}
We can solve \eqref{eq:single-re-fb-4} for three variables: $x_2,\ y_2$, and a third variable that is one of $X,Y,Z,U,V,W$. We require the third variable to be $k$ independent by zeroing the coefficients of all $k$-dependent terms. This gives a set of equations, and solving them gives the solution for this case (the real part of one band is flat). The following are our results.

\paragraph{Degenerate case:} In this case $\mu=0,\ \nu=0$ and the solution is
\begin{equation}
    \begin{aligned}
        H_r & = \left(
        \begin{array}{cc}
            a & b \\
            \frac{(a+f) (b (d+f)+(d-a) g)}{(b+g)^2} & d \\
        \end{array}
        \right), \\
        H_l &= \left(
            \begin{array}{cc}
            f & g \\
            \frac{(a+f) (g (a+l)+b (l-f))}{(b+g)^2} & l \\
            \end{array}
        \right). \; 
    \end{aligned}
    \label{eq:single-re-fb-degen-sol}
\end{equation} 
The band structure is
\begin{equation}
    \begin{aligned}
        E_1 &= -\frac{2 i \sin (k) (b f-a g)}{b+g},\\ 
        E_2 &= \frac{\left(e^{ i k} (b (a+d+f)+d g)\right)}{b+g}, \\ 
        & + \frac{\left(e^{-i k}(a g+b l+f g+g l)\right)}{b+g}.
    \end{aligned}
\end{equation} 

\paragraph{Abnormal case:} In this case $\mu=0,\ \nu=1$ and the solution is 
\begin{equation}
    H_r = \left(
\begin{array}{cc}
 a & b \\
 0 & d \\
\end{array}
\right), \quad H_l = \left(
\begin{array}{cc}
 -a & g \\
 0 & l \\
\end{array}
\right).
\label{eq:single-re-fb-abn-sol}
\end{equation} 
The band structure is 
\begin{equation}
    \begin{aligned}
        E_1 &= 2 i a \sin (k),\\ 
        E_2 &= d e^{i k}+e^{-i k} l+1.
    \end{aligned}
\end{equation}
Note that when $a=-f$, the solution \eqref{eq:single-re-fb-degen-sol} for the degenerate case reduces to the abnormal case \eqref{eq:single-re-fb-abn-sol}. 

\paragraph{Non-degenerate case:} In this case $\mu=1,\ \nu=0$ and the solution is
\begin{equation}
    H_r = \left(
\begin{array}{cc}
 a & 0 \\
 c & d \\
\end{array}
\right), \quad H_l = \left(
\begin{array}{cc}
 f & 0 \\
 h & -d \\
\end{array}
\right).
\end{equation} 
The band structure is
\begin{equation}
    \begin{aligned}
        E_1 &= 1+2 i d \sin (k),\\ 
        E_2 &= e^{-i k} \left(f+a e^{2 i k}\right).
    \end{aligned}
\end{equation} 

\subsection{Imaginary parts of both bands are flat}
\label{app:both-im-fb}

In this case, $x_2,\ y_2$ are $k$ independent in Eqs. \eqref{eq:complex-exp-band-eq-1}--\eqref{eq:complex-exp-band-eq-2}. Using the same procedure as in Section \ref{app:both-re-fb}, solving \eqref{eq:complex-exp-band-eq-1}--\eqref{eq:complex-exp-band-eq-2} for $y_2$ and requiring $y_1$ to be $k$ independent, we obtain 
\begin{equation}
    \begin{aligned}
      a+d & = f+l,\\
      b c+f l &= a d+g h,\\
      16 x_2^2 (b c-a d) &= -(a \mu +\nu  (h-c)-f \mu )^2  \\ 
      & - x_2^2(a+d+f+l)^2,\\
      \mu  x_2 (a+d+f+l) &= 2 x_2 (a \mu -\nu  (c+h)+f \mu ),\\
      8 x_2^4 + 2 \mu^2 x_2^2 &= 8 x_2^2 (a l-b h-c g+d f) \\ 
      & - x_2^2 (a+d+f+l)^2, \\ 
      & +(a \mu +\nu  (h-c)-f \mu )^2.
    \end{aligned}
    \label{eq:both_im_fb_cond}
\end{equation} 

\paragraph{Degenerate case} 

In this case $\mu=\nu=0$ and \eqref{eq:both_im_fb_cond} becomes 
\begin{equation}
    \begin{aligned}
       a+d &=f+l,\\
       b c+f l &=a d+g h,\\
       16 x_2^2 (b c-a d) &= -x_2^2(a+d+f+l)^2,\\
        x_2^2(a+d+f+l)^2 &= 8 x_2^2(a l-b h-c g+d f)\\ & -8 x_2^4.
    \end{aligned}
\end{equation}
The solution is 
  \begin{equation}
      \begin{aligned}
         H_{0}&=\begin{pmatrix}0 & 0\\
                        0 & 0
                \end{pmatrix},\\
         H_{l}&=\begin{pmatrix} f & \frac{b \left(F+2 x_2^2\right)}{(a-d)^2}\\
                        \frac{F-2 x_2^2}{4 b} & a+d-f
               \end{pmatrix},\\
         H_{r}&=\begin{pmatrix} a & b\\
                        -\frac{(a-d)^2}{4 b} & d
               \end{pmatrix}\; ,
      \end{aligned}
  \end{equation}
where $F=-2 \sqrt{x_2^2 (d-a) (a+d-2 f)+x_2^4}+(d-a) (a+d-2 f)$. Then the band structure is
\begin{equation}
    \begin{aligned}
       x_k &= \frac{ 2 b (a-d)^2 (a+d)\cos k }{2 b (a-d)^2}\\
       &- \frac{ 2 \sqrt{-b^2 x_2^2 (a-d)^4}}{2 b (a-d)^2},\\
       y_k &= \frac{2 b (a+d) (a-d)^2 \cos k}{2 b (a-d)^2} \\
       &+ \frac{2 \sqrt{-b^2 x_2^2 (a-d)^4}}{2 b (a-d)^2}.
    \end{aligned}
\end{equation} 

\paragraph{Non-degenerate case:} In this case $\mu=1,\ \nu=0$ and \eqref{eq:both_im_fb_cond} becomes 
\begin{equation}
    \begin{aligned}
       & a+d=f+l,\\
       & b c+f l=a d+g h,\\
       & x_2^2 \left(16 (b c-a d)+(a+d+f+l)^2\right)+(a-f)^2=0,\\
       & x_2 (a-d+f-l)=0,\\
       & x_2^2 \left(8 (a l-b h-c g+d f)-(a+d+f+l)^2-2\right)+(a-f)^2=8 x_2^4.
    \end{aligned}
\end{equation}
The solution is 
  \begin{equation}
      \begin{aligned}
         H_{0}&=\begin{pmatrix}0 & 0\\
                        0 & 1
                \end{pmatrix},\\
         H_{l}&=\begin{pmatrix} f & \frac{-G}{(a-f)^2}\\
                        \frac{\left(4 x_2^2+1\right) (a-f)^4}{16 x_2^2 G} & a
               \end{pmatrix},\\
         H_{r}&=\begin{pmatrix} a & b\\
                        -\frac{\left(4 x_2^2+1\right) (a-f)^2}{16 b x_2^2} & f
               \end{pmatrix},\; 
      \end{aligned}
  \end{equation} 
where $G=2 \sqrt{b^2 x_2^2 \left(x_2^2-(a-f)^2\right)}+b (a-f)^2-2 b x_2^2$. Then the band structure is 
\begin{equation}
    \begin{aligned}
        x_k & =-\frac{K}{4 b G x_2^2 (a-f)^2}+(a+f) \cos (k)+\frac{1}{2},\\
        y_k &= \frac{K}{4 b G x_2^2 (a-f)^2}+(a+f) \cos (k)+\frac{1}{2},
    \end{aligned}
\end{equation}
where 
\begin{equation*}
    \begin{aligned}
        K &= \Bigg[b G x_2^2 (a-f)^4 \Bigg(4 x_2^2 \Big(b^2 (a-f)^4-4 i b G (a-f) \sin (k) \\ 
        & +b G \left(1-2 (a-f)^2\right)+G^2\Big)+b^2 (a-f)^4 \\ 
        &-2 b G (a-f)^2 \cos (2 k)+G^2\Bigg)\Bigg]^{\sfrac{1}{2}}.
    \end{aligned}
\end{equation*} 

\paragraph{Abnormal case}

In this case $\mu=0,\ \nu=1$ and \eqref{eq:both_im_fb_cond} becomes 
\begin{equation}
    \begin{aligned}
       a+d &=f+l,\\
       b c+f l &=a d+g h,\\
       (c-h)^2 &=-16 x_2^2  (b c-a d)\\ 
       &-x_2^2 (a+d+f+l)^2,\\
       x_2 (c+h) &=0,\\ 
       (c-h)^2 - 8 x_2^4 &= x_2^2 (a+d+f+l)^2 \\ 
       & - 8 x_2^2 (a l-b h-c g+d f).
    \end{aligned}
\end{equation}
The solution is 
  \begin{equation}
      \begin{aligned}
         H_{0}&=\begin{pmatrix}0 & 1\\
                        0 & 0
                \end{pmatrix},\\
         H_{l}&=\begin{pmatrix} d-i x_2 & \frac{c}{4 x_2^2}-\frac{x_2^2}{c}\\
                        -c & d+i x_2
               \end{pmatrix},\\
         H_{r}&=\begin{pmatrix} d & -\frac{c}{4 x_2^2}\\
                                c & d
               \end{pmatrix}\; .
      \end{aligned}
  \end{equation}
The band structure is
\begin{equation}
    \begin{aligned}
       x_k &= 2 d \cos k -\frac{ \sqrt{c^2 x_2^2 \left(c \sin (k)+i x_2^2\right){}^2}}{c x_2^2}, \\
       y_k &= 2 d \cos k + \frac{ \sqrt{c^2 x_2^2 \left(c \sin (k)+i x_2^2\right){}^2}}{c x_2^2}.
    \end{aligned}
\end{equation} 

\subsection{Imaginary part of one band is flat}

Using the same method as the case in which the real part of one band is flat, we can also solve the case in which the imaginary part of one band is flat. To do so, we require $x_2$ ($y_2$) to be $k$ independent in Eqs. \eqref{eq:complex-exp-band-eq-1}--\eqref{eq:complex-exp-band-eq-2}. The only possibility is $x_2=0,\ y_2=(a+d-f-l)$, and then following the same steps as the real part of one band case we get the following results. 
\paragraph{Degenerate case:}
\begin{equation}
    \begin{aligned}
        H_r &= \begin{pmatrix}a & b\\
                        \frac{(a-f) (g (a-d)+b (d-f))}{(b-g)^2} & d
                \end{pmatrix}, \\ 
        H_l &= \begin{pmatrix}f & g\\
                       \frac{(a-f) (g (a-l)+b (l-f))}{(b-g)^2} & l
                \end{pmatrix}.
    \end{aligned}
\end{equation} 
The band structure is 
\begin{equation}
  \begin{aligned}
      E_1 &= \frac{2 \cos (k) (b f-a g)}{b-g}, \\ 
      E_2 &= \frac{ e^{ i k} (b (a+d-f)-d g)}{b-g},\\
      &+\frac{ e^{-i k} (a g+b l-f g-g l)}{b-g}.
  \end{aligned}
\end{equation} 

\paragraph{Abnormal case:}
\begin{equation}
    H_r = \left(
\begin{array}{cc}
 a & b \\
 0 & d \\
\end{array}
\right), \quad 
    H_l = \left(
\begin{array}{cc}
 a & g \\
 0 & l \\
\end{array}
\right).
\end{equation} 
The band structure is 
\begin{equation}
  \begin{aligned}
      E_1 &= 2 a \cos (k), \\ 
      E_2 &= e^{-i k} \left(l+d e^{2 i k}\right).
  \end{aligned}
\end{equation} 

\paragraph{Non-degenerate case}
In this case, the solution is the same as the abnormal case, and the band structure is 
\begin{equation}
  \begin{aligned}
      E_1 &= 2 a \cos (k), \\
      E_2 &= d e^{i k}+e^{-i k} l+1.
  \end{aligned}
\end{equation}

\subsection{Modulus of a band is flat}
\label{app:mod-fb}

Suppose $x=r e^{i\theta_{k}}$ in Eq. \eqref{eq:bands-eqn-allf-xy}, then 
\begin{equation}
  \begin{aligned}
      y	& =	e^{ik}(a+d)+e^{-ik}(f+l)+\mu-re^{i\theta_{k}}, \\
      y	& =	\frac{1}{r}e^{-i\theta_{k}}\Big(e^{ik}(a\mu-c\nu)+al-bh-cg+df \\ & +\det H_{l}e^{-2ik}+\det H_{r}e^{2ik}+e^{-ik}(f\mu-h\nu)\Big).
  \end{aligned}
\end{equation} 

If we assume $\theta_k = m k ,\ m \ne 0$, then 
\begin{equation}
  \begin{aligned}
     & e^{i\left(m-1\right)k}(f+l) +e^{i\left(m+1\right)k}(a+d)  +\mu e^{imk} -re^{2imk} \\ 
     &=\frac{1}{r} \Big(e^{-ik}(f\mu-h\nu) \\ & +al-bh-cg+df+e^{ik}(a\mu-c\nu) \\ & + \det H_{l}e^{-2ik} +\det H_{r}e^{2ik} \Big).
  \end{aligned}
  \label{eq:generic-m}
\end{equation} 
By equating the same powers of $e^{i k}$ in \eqref{eq:generic-m}, we can show that, when $m>1$ or $m<-1$, the only possible solution for \eqref{eq:generic-m} is $r=0$. Therefore, only when $m=\pm 1$ do we have a non-trivial solution. 
\\

\subsubsection{m=1 case} 

In this case, equating the coefficients of the same powers of $e^{ik}$ on the two sides of \eqref{eq:generic-m} we get 
\begin{equation}
  \begin{aligned}
        f+l & =\frac{1}{r}\left(al-bh-cg+df\right),\\
        \mu & =\frac{1}{r}\left(a\mu-c\nu\right),\\
        a+d-r & =\frac{1}{r}\det H_{r},\\
        \det H_{l} & =0,\\
        f\mu-h\nu & =0.
  \end{aligned}
  \label{eq:mod-fb-m1-eqs}
\end{equation} 

\paragraph{Degenerate case:} In this case $\mu=0,\ \nu=0$ and \eqref{eq:mod-fb-m1-eqs} gives the following solution: 
\begin{equation}
   \begin{aligned}
     H_{r} & =\left(\begin{array}{cc}
         a & b\\
        -\frac{(a-r)(r-d)}{b} & d
       \end{array}\right),\\
    H_{l} & =\left(\begin{array}{cc}
             f & \frac{bf}{a-r}\\
           \frac{l(a-r)}{b} & l
           \end{array}\right). \;
\end{aligned}  
\end{equation}
 Then the band structure is 
 \begin{equation}
   \begin{aligned}
     E_{1} & =e^{ik}r,\\
     E_{2} & =e^{-ik}\left(ae^{2ik}+de^{2ik}+f-e^{2ik}r+l\right).
   \end{aligned}
\end{equation} 

\paragraph{Non-degenerate case:} In this case $\nu=0,\ \mu=1$ and \eqref{eq:mod-fb-m1-eqs} gives
\begin{equation}
  \begin{aligned}H_{r} & =\left(\begin{array}{cc}
r & 0\\
c & d
\end{array}\right),\\
H_{l} & =\left(\begin{array}{cc}
0 & 0\\
h & l
\end{array}\right) \; ,
\end{aligned}  
\end{equation}
with band structure 
\begin{equation}
  \begin{aligned}
    E_{1} & =e^{ik}r,\\
    E_{2} & =de^{ik}+e^{-ik}l+1.
  \end{aligned}  
\end{equation} 

\paragraph{Abnormal case:} In this case $\mu=0,\ \nu=1$ and \eqref{eq:mod-fb-m1-eqs} gives
\begin{equation}
    \begin{aligned}H_{r} & =\left(\begin{array}{cc}
              r & b\\
             0 & d
            \end{array}\right),\\
         H_{l} & =\left(\begin{array}{cc}
             0 & g\\
            0 & l
            \end{array}\right),
     \end{aligned}
\end{equation} 
 with band structure
\begin{equation}
  \begin{aligned}
    E_{1} & =e^{ik}r,\\
    E_{2} & =de^{ik}+e^{-ik}l+1.
  \end{aligned}  
\end{equation} 

\subsubsection{ $m=-1$ case} 

Similar to the $m=1$ case, by equating the same powers of $e^{i k}$ in \eqref{eq:generic-m}, we have 
\begin{equation}
   \begin{aligned}
       f+l-r & =\frac{1}{r}\det H_{l},\\
       \frac{1}{r}\left(f\mu-h\nu\right) & =\mu,\\
       a+d & =\frac{1}{r}\left(al-bh-cg+df\right),\\
       \det H_{r} & =0,\\
       a\mu-c\nu & =0.
   \end{aligned} 
\end{equation} 

\paragraph{Degenerate case:} In this case $\mu=0,\ \nu=0$ and \eqref{eq:generic-m} gives
\begin{equation}
   \begin{aligned}
        H_{r} & =\left(\begin{array}{cc}
        a & b\\
        \frac{ad}{b} & d
        \end{array}\right),\\
        H_{l} & =\left(\begin{array}{cc}
        f & \frac{b(f-r)}{a}\\
        \frac{a(l-r)}{b} & l
        \end{array}\right).
   \end{aligned}
\end{equation}
 Then the band structure is 
\begin{equation}
  \begin{aligned}
    E_{1} & =e^{-ik}r,\\
    E_{2} & =e^{-ik}\left(ae^{2ik}+de^{2ik}+f+l-r\right).
   \end{aligned}  
\end{equation}

\paragraph{Non-degenerate case:} In this case $\nu=0,\ \mu=1$ and \eqref{eq:generic-m} gives the following solution:
\begin{equation}
    \begin{aligned}
       H_{r} & =\left(\begin{array}{cc}
       0 & 0\\
       c & d
       \end{array}\right),\\
       H_{l} & =\left(\begin{array}{cc}
       r & 0\\
       h & l
       \end{array}\right).
   \end{aligned}
\end{equation}
 Then the band structure is 
\begin{equation}
    \begin{aligned}
       E_{1} & =e^{-ik}r,\\
       E_{2} & =de^{ik}+e^{-ik}l+1.
    \end{aligned}
\end{equation} 

\paragraph{Abnormal case:} In this case $\nu=1,\ \mu=0$ and \eqref{eq:generic-m} gives
\begin{equation}
    \begin{aligned}
       H_{r} & =\left(\begin{array}{cc}
       0 & b\\
       0 & d
       \end{array}\right),\\
       H_{l} & =\left(\begin{array}{cc}
       r & g\\
       0 & l
       \end{array}\right).
   \end{aligned}
\end{equation}
 Then the band structure is 
\begin{equation}
    \begin{aligned}
       E_{1} & =e^{-ik}r,\\
       E_{2} & =e^{-ik}\left(l+de^{2ik}\right).
    \end{aligned}
\end{equation}



\thispagestyle{empty}

\end{document}